%% file: thesis_final_PH.tex
\documentclass[11pt]{report}
\usepackage[a4paper, left=2 cm, right=2 cm, top= 2 cm, bottom=2 cm]{geometry}
\usepackage[utf8]{inputenc}
\usepackage[T1]{fontenc}
\usepackage[sc]{mathpazo}
\usepackage[x11names]{xcolor}
\usepackage{csquotes}
\usepackage{youngtab}
\usepackage{fancyhdr}
\usepackage[Glenn]{fncychap}
\ChNumVar{\fontsize{60}{62}}
\ChNameVar{\Huge\color{RoyalBlue4}}
\ChNameAsIs
\ChTitleVar{\Huge\bf\color{Sienna4}} 
\ChTitleAsIs
\ChRuleWidth{ 2pt}
\usepackage[bottom]{footmisc}
\usepackage[most]{tcolorbox}

\usepackage{appendix}
\usepackage{relsize}
\usepackage{mwe}
\usepackage{bbold}
\usepackage{hyperref}
\usepackage{subcaption}
\usepackage{tikz-feynman}
\usepackage{bookmark}
\usepackage{mathtools}
\usepackage{setspace}
\usepackage{amsmath, amssymb, amsthm, float, graphicx, amsfonts, braket, mathrsfs}
\usepackage{array, boldline, makecell, booktabs}

\numberwithin{equation}{section}
\usepackage{multirow}

\hypersetup{colorlinks=true, linkcolor=SlateBlue4, citecolor=Firebrick4,urlcolor=Green4}

\def\beq{\begin{eqnarray}}\def\eeq{\end{eqnarray}}
\def\be{\begin{equation}}\def\ee{\end{equation}}
\def\x{\xi}
\def\p{\pi}

\def\g{\gamma}
\def\r{\rho}
\def\s{\sigma}
\def\m{\mu}
\def\n{\nu}
\def\a{\alpha}
\def\e{\epsilon}
\def\ve{\varepsilon}
\def\k{\kappa}
\def\b{\beta}
\def\d{\delta}

\def\f{\phi}
\def\vf{\varphi}

\def\th{\theta}
\def\t{\tau}
\def\D{\Delta}
\def\G{\Gamma}
\def\l{\lambda}
\def\pd{\partial}

\def\ma{{\mathcal{A}}}
\def\mb{{\mathcal{B}}}
\def\mc{{\mathcal{C}}}
\def\md{{\mathcal{D}}}
\def\me{{\mathcal{E}}}
\def\mf{{\mathcal{F}}}
\def\mg{{\mathcal{G}}}
\def\mh{{\mathcal{H}}}
\def\mi{{\mathcal{I}}}
\def\mj{{\mathcal{J}}}
\def\mk{{\mathcal{K}}}
\def\ml{{\mathcal{L}}}
\def\mm{{\mathcal{M}}}
\def\mn{{\mathcal{N}}}
\def\mo{{\mathcal{O}}}
\def\mmp{{\mathcal{P}}}

\def\mr{{\mathcal{R}}}
\def\ms{{\mathcal{S}}}
\def\mt{{\mathcal{T}}}
\def\mw{{\mathcal{W}}}

\def\fh{\mathfrak{H}}

\def\G{\Gamma}

\def\w{\omega}
\def\zt{\tilde{z}}
\def\immp{\mathcal{A}\left(s_{1}^{\prime} ; s_{2}^{(+)}\left(s_{1}^{\prime}, a\right)\right)}
\def\imm{\mathcal{A}\left(s_{1} ; s_{2}^{(+)}\left(s_{1}, a\right)\right)}
\def\imeas{ \int_{\frac{2 \mu}{3}}^{\infty} \frac{d s_{1}^{\prime}}{s_{1}^{\prime}}}
\newcommand{\Abs}[1]{\left| #1 \right|}
\newcommand{\floor}[1]{\lfloor #1 \rfloor}

\newcommand{\eq}[1]{eq.\eqref{#1}}
\def\spart{S_\ell}
\def\tpart{T_\ell}

\def\rb{\mathbb{R}}
\def\tell{\tilde{\ell}}

\def\mW{{\mathcal{W}}}
\def\nn{\nonumber}
\def\lan{\langle}
\def\ran{\rangle}
\def\Mack{\widehat{P}}

\def\dphi{\D_\f}

\def\lsumt{\sum_{\substack{\ell\\ \ell~\text{even}}}^L}

\def\tenp{\otimes}

\newcommand{\bea}{\begin{eqnarray}}
	\newcommand{\eea}{\end{eqnarray}}
\newcommand{\ba}{\begin{equation}\begin{aligned}}
		\newcommand{\ea}{\end{aligned}\end{equation}}
	\newcommand{\bes}{\begin{equation*}}
		\newcommand{\ees}{\end{equation*}}

\newcommand{\Z}{\mathbb{Z}}
\newcommand{\tyng}{\tiny\yng}

\usepackage{bm}
\usepackage{bm}

\usepackage{amsthm}
\newtheorem{theorem}{Theorem}[section]
\theoremstyle{definition}

\newtheorem{corollary}{Corollary}[theorem]
\newtheorem{lemma}[theorem]{Lemma}

\begin{document}
	\pagenumbering{gobble}	
	\pagestyle{empty}	
	\begin{center}
	
	\null 
	{\Huge\textbf{\textcolor{Coral4}{Explorations in the Space of S-Matrices}}}
\vspace{3.6cm}
	
	\Large{\textit{A thesis submitted for the degree of} \\ \textit{Doctor of Philosophy} \\\textit{in the Faculty of Sciences}}
	
	\vspace{2 cm}
	
	{\LARGE \textcolor{DodgerBlue3}{\bf Parthiv Haldar}}
	
	\vspace{3cm}
	\begin{figure}[h]
		\centering
		\includegraphics[scale=0.1]{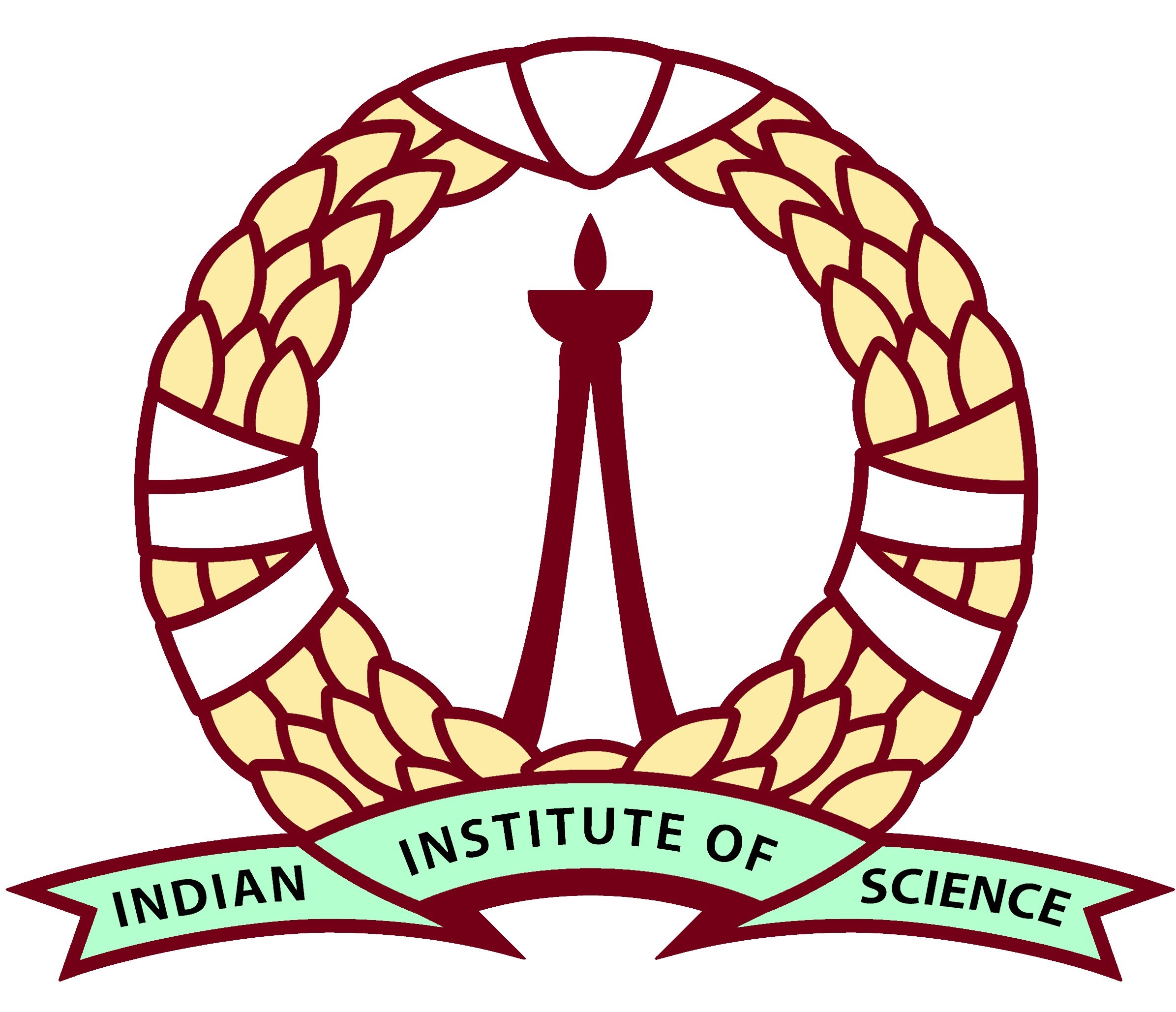}
	\end{figure}
	
	\vspace{3cm}
	{\Large \textit{Centre for High Energy Physics\\
			Indian Institute of Science \\
			\ Bangalore - 560012. India. \\}}

\end{center}

\pagebreak

\pagestyle{empty}

\begin{flushleft}
\textbf{{\LARGE Declaration}}
\end{flushleft}

\vskip 2 cm

{\large I hereby declare that this thesis ``Explorations in the Space of S-Matrices" is based on my own research work, that I carried out with my collaborators in the Centre for High Energy Physics, Indian Institute of Science, during my tenure as a PhD student under the supervision of Prof. Aninda Sinha. The work presented here does not feature as the work of someone else in obtaining a degree in this or any other institute. Any other work that may have been quoted or used in writing this thesis has been duly cited and acknowledged.

	\vskip 2cm

	\ \\
	Date:   \hspace{10.5cm} Parthiv Haldar
	
	\vskip 1.5cm

	\ \\
	Certified by:
	
	\vskip 1.5cm
	
	\ \\
	Prof. Aninda Sinha\\
	Centre for High Energy Physics\\
	Indian Institute of Science\\
	Bangalore: 560012\\
	India}

\pagebreak

\

\vskip 1cm

\begin{flushright}
	{\color{Sienna4}\textbf{\LARGE Acknowledgement}}	
\end{flushright}
\vspace{0.1 pt}
{\color{RoyalBlue3}\hrule}
\vspace{0.07 cm}
{\color{RoyalBlue3}\hrule} 
\vspace{0.5 cm}	
First and foremost, I would like to express my gratitude towards my advisor \emph{Prof. Aninda Sinha}. Without his constant encouragement and support, it would not have been possible to complete this journey. His encouragement has provided me with the courage and conviction to explore less traversed paths in my journey for knowledge. Most importantly, he kept his unwavering trust and belief in me when I myself lost them. I couldn't thank him more for being the emotional support during the most vulnerable moments of my life! I count myself luckiest because I went for thesis advisor  but at the end of $5$ years \emph{Aninda} became much more than just thesis advisor!\par

 I was fortunate to work with outstanding collaborators Subham Dutta Chowdhury, Kallol Sen, Kausik Ghosh, Parijat Dey, Prashanth Raman and Ahmadullah Zahed. Their zeal for explorations inspired me to push myself towards betterment.

 I am indebted to some of the excellent courses that I was part of during my Integrated PhD programme. In particular, the quantum mechanics courses taught by \emph{Prof. Diptiman Sen} and \emph{Prof. Rohini Godbole} had been my asset during my entire PhD journey. The long journey at IISc would not have been possible without the companionship and friendship of \emph{Sahel Dey}, \emph{Prerana Biswas}, \emph{Sayantan Ghosh}.  The many a precious moments I made with you kept me going even in the darkest hours! Thank you guys for bearing with me for so long!

 Last but not the least, without my parents' unrelenting support I would not have dared to walk the path! They stood beside me in my every decision, right and wrong! In their many sacrifices, they gave wings to my dream!

	\newpage
	\pagestyle{empty}
	\vspace{5 cm}
	\begin{tikzpicture}
		\node at (0,0)
		{\includegraphics[scale=0.25]{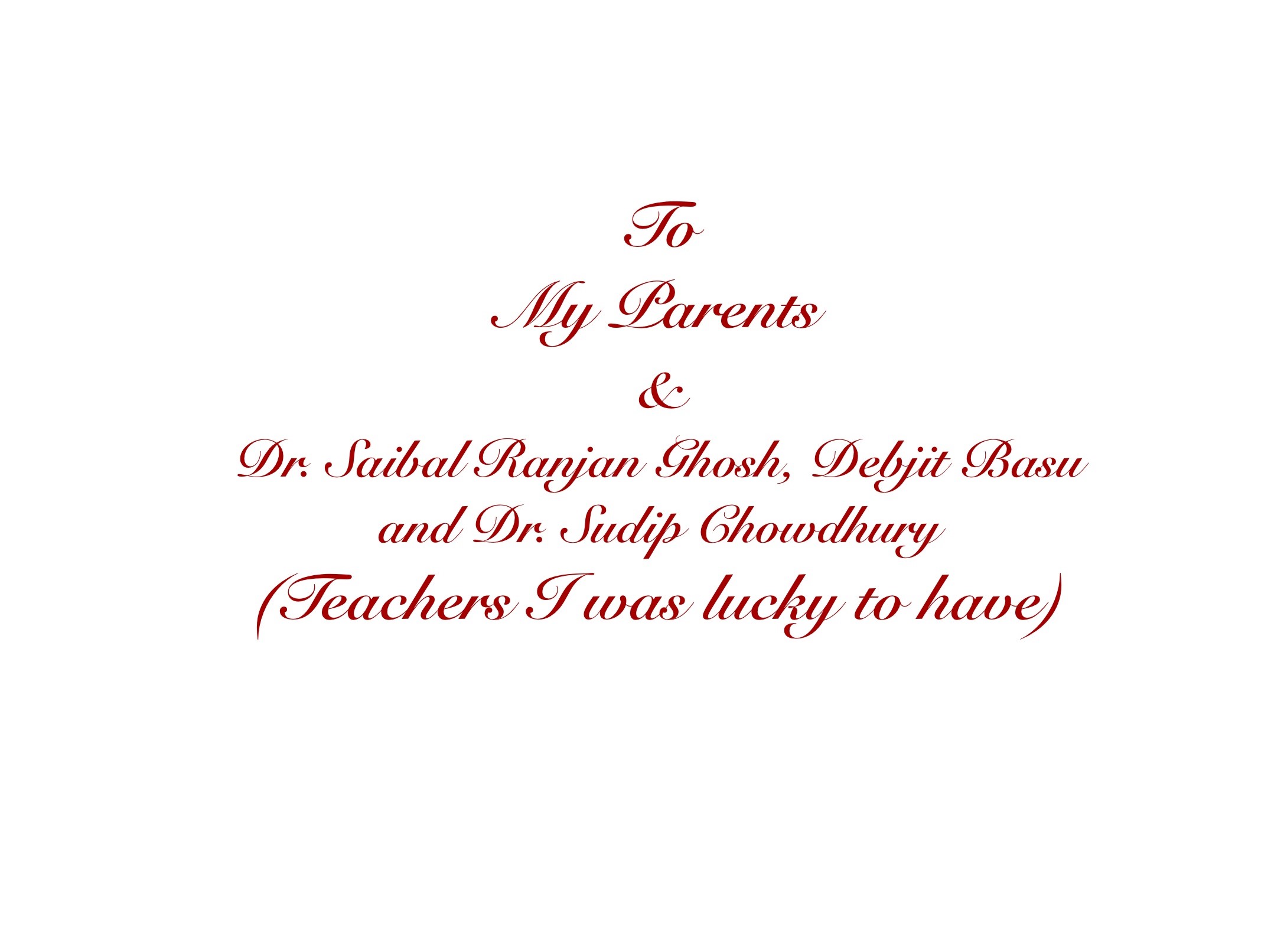}};
		\node at (0,12){};
		\node at (25,0){};
	\end{tikzpicture}
	
	
	\newpage
	\pagestyle{empty}
\chapter*{\bf Synopsis} 
The quantum theory of relativistic particles is centred on studying quantum amplitudes for the scattering of quantum particles. Such amplitudes can be viewed as elements of an abstract matrix called S-matrix. Thus S-matrix is one of the most important observables in the theory of quantum dynamics of relativistic particles. The traditional way of constructing S-matrix is via introducing non-observable entities called quantum fields with infinite degrees of freedom. Quantum field theory succeeds impressively in this endeavour. However, due to the infinite number of degrees of freedom of the quantum fields, one has to deal with various infinities along the way. While there are well developed mathematical techniques called renormalization to tackle such infinities and extract finite observable results, it would be impressive if we could find ways to work around such infinities. Heisenberg first proposed such a way of business built upon entirely in terms of S-matrix guided by fundamental physical principles of unitarity, relativistic causality and relativistic covariance and other symmetry principles as warranted. In a modern revival of that idea, S-matrix can be used to define an abstract theory space. Thus, we can explore the space of physically consistent theories by studying the consequences of the aforementioned physical principles on the S-matrix. This thesis is devoted to such explorations for S-matrix amplitudes for $2-2$ scattering of identical particles. 

\paragraph{} In the first part of
the thesis, we discuss a novel mathematical way of exploring the space of S-matrices
using the tools from the mathematical field of geometric function theory (GFT). GFT is the study of geometric properties of complex analytic functions regarded as mappings between complex planes. We use two particular types of functions that play prominent roles in GFT: univalent functions and typically real functions. Univalent functions are analytic functions which are also injective. Geometrically, univalent functions are conformal mappings. These functions are known to satisfy various bounding relations. The most famous of them is  \emph{de Branges’ theorem}, also known as 
\emph{Bieberbach conjecture}, which bounds the Taylor coefficient of univalent functions. Not only the Taylor coefficients but the functions themselves are also bounded described by \emph{Koebe growth theorem}. The other kind, namely the typically real functions, are functions with real Taylor coefficients. These functions have positive imaginary part in the upper half-plane and negative imaginary part in the lower half-plane. These functions also satisfy interesting bounds similar to those for the univalent functions. However, the connection between these functions and the scattering amplitudes is not obvious. A novel dispersive representation of scattering amplitudes enables us to unearth these connections with univalent and typically real functions. 

Dispersion relation and crossing symmetry are two important consequences of the causality principle. The usual fixed transfer dispersion relation is not manifestly crossing symmetric, and therefore, the crossing symmetry has to be imposed separately, which is quite a non-trivial task. It would be nice to have a \emph{manifestly crossing symmetric dispersive representation}. Auberson and Khuri came up with such a dispersive representation with a kernel which is manifestly crossing symmetric. In the process of deriving this dispersion relation, the Mandelstam variables are parametrized as rational functions of a complex variable  $\tilde{z}$ and another parameter $a$. When expressed as functions of the Mandelstam variables, $\tilde{z}$ and $a$ are crossing symmetric. The dispersion kernel turns out to be a rational function of $\tilde{z}$ and $a$. We establish that the kernel is univalent as well as typically real inside unit disc $|\tilde{z}|<1$ when $a$ takes value in a certain real interval. Further, using unitarity and typically realness of the crossing symmetric dispersion kernel, it turns out that the scattering amplitude itself is a typically real function inside the unit disc in complex $\tilde{z}$ plane for the same real values of $a$. We would like to emphasize that the crossing symmetric dispersion relation is crucial for establishing this connection. Now employing the bounding properties of univalent and typically real functions as mentioned above, we obtain bounds on Taylor coefficients of the expansion of the scattering amplitude about $\tilde{z}=0$ as well as upper and lower bounds on the amplitude itself. Physically,  expanding the amplitude about $\tilde{z}=0$ corresponds to having a low energy expansion in a basis of  crossing symmetric functions. These expansions can be interpreted as EFT amplitude, and the corresponding Taylor coefficients can be expressed as ratios of Wilson coefficients of the EFT. Thus we obtain \emph{$2-$sided bounds on ratios of Wilson coefficients}. Such $2-$sided bounds have been obtained before numerically. Our analysis sheds light on a \emph{concrete mathematical origin} of such bounds for the first time, which was not there before. We apply this analysis to scattering amplitudes of identical massive bosons, pions, effective string amplitudes, $4-$photon amplitude and $4-$graviton amplitudes. 

\paragraph{} In the second part of the thesis we turn our attention to holographic S-matrices. The conjectural $AdS/CFT$ holography provides a way to construct flat space scattering amplitudes from the Mellin amplitudes of a conformal field theory (CFT) by taking a large radius limit of the dual $AdS$ space. Various analytic properties of flat space scattering amplitudes  are encoded in corresponding properties of the CFT Mellin amplitude. Flat space $2-2$ scattering amplitudes are known to satisfy high energy bounds called the Froissart-Martin bound which follows from axiomatic analyticity and unitarity properties of the S-matrix.  Froissart-Martin bound is one of the robust consistency tests for a flat space scattering amplitude. Therefore if a holographic construction of the S-matrix is to work, one should be able to obtain a systematic derivation of the Froissart-Martin bound starting with $4-$point Mellin amplitude for a holographic CFT. We provide such a derivation in the second part of the thesis. We find that our holographic derivation gives the exact Froissart-Martin bound in $4$ spacetime dimensions, while in greater spacetime dimensions, we get weaker bounds. We attempt to argue the possible reason for this behaviour.

\newpage 
\pagestyle{empty}
\begin{flushright}
{\color{Firebrick4} \Large \emph{ Great things are not accomplished by those who yield to\\
 trends and fads and popular opinion.}\\}
\vspace{0.5 cm}
{\color{SlateBlue4} \Large \emph{Jack Kerouac}}
	
\end{flushright}

\newpage 
	\pagestyle{empty}
	\tableofcontents
	
	\chapter*{\textbf{{ Publications and Preprints}}}
	\addcontentsline{toc}{chapter}{{{ Publications and Preprints}}}

	\begin{flushleft}
		{\textbf{\Large The thesis is based on the following works}}
	\end{flushleft}
	\vspace{0.3 cm}
	\begin{enumerate}
		
		\item  
		({\color{Blue3}Chapter 3 and Chapter 4}) {\bf Quantum field theory and the Bieberbach conjecture},\\
		P.~Haldar, A.~Sinha and A.~Zahed,\\
		SciPost Phys. \textbf{11}, 002 (2021), [arXiv:2103.12108 [hep-th]].
		\vspace{0.5 cm}
		
		\item 
		({\color{Blue3}Chapter 5}) \textbf{Crossing Symmetric Spinning S-matrix Bootstrap: EFT bounds},\\
		S.~D.~Chowdhury, K.~Ghosh, P.~Haldar, P.~Raman and A.~Sinha,\\
		Submitted to SciPost Phys., [arXiv:2112.11755 [hep-th]].
		\vspace{0.5 cm}

		\item 
		({\color{Blue3}Chapter 6}) \textbf{	Froissart bound for/from CFT Mellin amplitudes},\\
		P.~Haldar and A.~Sinha,\\
		SciPost Phys. \textbf{8}, 095 (2020), [arXiv:1911.05974 [hep-th]].

	\end{enumerate}
	\vspace{1 cm}
	\begin{flushleft}
		{\bf \Large Other works that were completed during PhD, not included in the thesis}
	\end{flushleft}
	\vspace{0.3 cm}
	\begin{enumerate}
		\item {\bf Relative entropy in scattering and the S-matrix bootstrap,
		}\\
		A.~Bose, P.~Haldar, A.~Sinha, P.~Sinha and S.S.~Tiwari,\\
		SciPost Phys. \textbf{9}, 081 (2020), [arXiv:2006.12213 [hep-th]].
		\vspace{0.5 cm}
		\item {\bf Regge amplitudes in generalized fishnet and chiral fishnet theories,
		}\\
		S.~Dutta Chowdhury, P.~Haldar and K.~Sen,\\
		JHEP \textbf{12}, 117 (2020), [arXiv:2008.10201 [hep-th]].
		\vspace{0.5 cm}
		
		\item {\bf On the Regge limit of Fishnet correlators
			,
		}\\
		S.~Dutta Chowdhury, P.~Haldar and K.~Sen,\\
		JHEP \textbf{10}, 249 (2019), [arXiv:1908.01123 [hep-th]].
		\vspace{0.5 cm}
	\end{enumerate}
	
	\pagenumbering{arabic}

	\pagebreak
	\pagestyle{fancy}
	\include{Introduction}

	\include{smatrev}

	\include{GFTrev}
\include{bieberbach}

    \include{spinbootnew}

	\include{mellin_bound}

	\include{conclusion}

\end{document}

%% file: Introduction.tex
\chapter{ {Introduction}}\label{intro}
\begin{flushright}
	{\large\emph{\color{Blue4}{``I think the fact that one could deal with physical matrix elements in a \\
	largely model independent way was the most attractive aspect to me of\\
	S-matrix theory."}\\
\vspace{0.3 cm}
	Marvin Leonard Goldberger}}	
\end{flushright}

The central problem of quantum dynamics is to calculate quantum amplitude of scattering of particles because it is though scattering experiments that the properties of the particles like mass, charge, magnetic moment and interaction with other particles is studied. The tradition way of computing the scattering amplitudes is via  Lagrangian, or equivalently Hamiltonian, framework of quantum field theory. Here one associates one quantum field to each particles and writes a Lagrangian (Hamiltonian) for describing the particle interactions. Each Lagrangian (Hamiltonian) corresponds to a particular model of interaction. Then scattering amplitudes are calculated usually in perturbation theory. Due to the fact that the field variables are functions of continuous spacetime, they have infinite number of degrees freedom. Due to this one encounter various infinities in the perturbative evaluation of th amplitude. While mathematical formalism of  renormalization theory  is well suited to tackle these infinities, one wonders whether it is possible to build up a formalism entirely in terms of finite observables in a model independent way. S-Matrix theory was born as an answer to this curiosity. In its modern incarnation S-Matrix is considered as navigator in the space of qauntum field theories. In this introductory chapter we present a discussion of past and present aspects of S-Matrix theory and give an overview of where the work undertaken in this thesis stands in the larger context of S-Matrix theory.

\section{{S-Matrix theory: Historical developments} } 
It is impossible to understand the modern perspective of S-Matrix theory without an understanding of its historical development. Therefore we give a brief review of different phases of the historical development of S-Matrix theory. Our review closely follows the discussion in \cite{Cushing} which is an excellent exposition on how S-Matrix theory developed during 1950-1960s. 
\subsection*{Heisenberg's S-Matrix theory}
During the early days of perturbative quantum field theories through the 1930s, the appearance of infinities was extremely disturbing because it contradicted the experimental results, which were finite. In the background of a strong sense of doubt regarding the correctness of the then formulation of quantum field theory, Heisenberg proposed an alternative way to study the quantum dynamics of relativistic particles, where he wanted to avoid any reference to Hamiltonian to base his theory entirely on observable quantities which are known to be finite. It is to be emphasized that Heisenberg didn't doubt the status of Hamiltonian (in the sense of total energy) as a quantum-mechanical observable. Rather, the computational framework of calculating observables, like bound-state energy levels or scattering cross-sections, was mathematically ill-defined because they often produced infinities. To avoid this, Heisenberg, in a series of papers \cite{Heisen1, Heisen2, Heisen3, Heisen4} proposed to build the theory entirely based upon S-Matrix\footnote{ While S-Matrix was introduced in the context of nuclear physics by Wheeler in \cite{Wheeler}, Heisenberg introduced it independently in his work. At least, he did not acknowledge Wheeler's work. Whether Heisenberg was aware of Wheeler's work is a point of historical debate. See \cite{Cushing} for details}. In these works, his goal was to find as many general, model-independent properties of S-Matrix as possible. In particular, Heisenberg first provided a general proof of the unitarity of S-Matrix in his work. However,  what should be desired analytic structure of the S-Matrix was not clear to him. Also, he realized that one needs a rather complicated mathematical framework to tackle S-Matrix theory fruitfully. By the late 1940s, Heisenberg himself abandoned the S-Matrix theory because, by then, perturbative renormalization theory provided the means to tame the infinities of the quantum field theoretical calculations. While Heisenberg's S-Matrix theory had lost steam by the end of 1940, the ideas and questions raised in response to his work are of profound significance. In particular, the question about how the causality requirement, suggested in the mid-1940s by Kronig as a constraint on S-Matrix, was to be wired into S-Matrix theory led to later development of more fruitful endeavours of S-Matrix theory in the form of dispersion relations as we discuss next.

\subsection*{ Dispersion Relations and Axiomatic QFT}
The next phase of development in the  S-Matrix theory occurred during the 1950s post World War II. The inability of quantum field theory to explain the huge amount of experimental data regarding \emph{strong interactions} led to this development. The key player in this phase was \emph{dispersion relation} for scattering amplitude. The first relativistic dispersion relation was obtained by Godberger, Gell-Mann and Thirring  (GGT) \cite{GGT} perturbatively for forward scattering amplitudes of photon. Their motivation was to find a computational tool that involved only finite, experimentally measurable quantities. The key point in their derivation of dispersion relation was to utilize the principle of micro-causality, which states that (anti)commutators of spacelike separated quantum fields vanish. Again, we should emphasize that when we say that this work did not use QFT, we mean that any specific Hamiltonian or associated equation of motion was used in deriving the dispersion relations. A short time after GGT, Goldberger \cite{GoldB1} was able to obtain dispersion relation without using perturbation theory employing the notion.   

While these works were pragmatist in nature, they were not in particular rigorous. An attempt to derive results in a rigorous mathematical framework of quantum field theory was first undertaken by  Lehmann, Symanzik and Zimmermann (LSZ). Their attempt was probably most in the spirit of Heisenberg's original S-Matrix programme. LSZ \cite{LSZ2} explicitly cites Heisenberg. They assumed the existence of an S-Matrix which connected the causal asymptotic in and out fields and then sought to prove the existence of a causal field interpolating between the asymptotic fields. By causal fields are meant field operators whose commutators vanish for spacelike separated points, i.e. satisfies micro-causality. Many impressive and deep results were obtained within and as an outgrowth of LSZ programme. In particular, Lehmann used LSZ formalism and micro-causality to establish the analyticity of $2-2$ scattering amplitude in momentum transfer variables at fixed angle. It is worth mentioning here that a Soviet school led by Bogoliubov also took up a similar approach. In fact, Bogoliubov, Medvedev and Polivanov (BMP) \cite{BMP} first gave a mathematical proof of dispersion relation in a similar mathematical framework of quantum field theory, which was later established to be equivalent to the LSZ formalism.  

\subsection* {S-Matrix Bootstrap}
While the above dispersion relation-centric development was deeply rooted in a general framework of quantum field theory in spite of not using any explicit Hamiltonian, there was an independent development of S-Matrix programme which aimed to reject the quantum field theory construct altogether. In a way, this development was the direct successor of Heisenberg's old S-Matrix programme. There are three main schools of thought driving this development: the \emph{Soviet school} headed by Landau, the \emph{Cambridge school} leaded by Polkinghorne and the \emph{Berkeley school} headed by Chew. \cite{ELOP} is the definitive reference of the Cambridge school, and \cite{Chew1, Chew2} give an overview of Chew's analysis. Chew championed the term \emph{S-Matrix bootstrap} for this line of development. The word `bootstrap' means lifting oneself up by pulling one's bootstrap. This autonomous S-Matrix programme aimed to solve the entire theory, in particular, obtain  the entire particle content and their interaction properties from the internal consistency of the theory, which was taken to be unitarity! At the time, the idea was being pushed for strong interactions and Chew coined the term \emph{nuclear democracy} to describe this self-consistent picture of particles. The nuclear democracy proposed by Chew was an extreme utopia which was impossible to achieve because, to do so, even the basic principles themselves have to be generated by some notion of self-consistency! Nevertheless, the ideas were radical and profound, which would continue inspiring various future development. 

\subsection*{ QCD, String theory and a setback for S-Matrix theory}
As impressive as the S-Matrix theory developments during 1950-60s, the success of non-Abelian gauge field theory in the form of quantum chromodynamics (QCD) pushed back S-Matrix theory. On another front quantum string theory, which ironically emerged from S-Matrix bootstrap itself, emerged as a strong candidate of quantum gravity theory. But string theory found its impressive success as a gauge quantum field theory. Thus, quantum field theory came triumphant during 1970s onwards, and S-matrix theory laid dormant. However, dispersion relation theory technique continued to be used in various phenomenological analysis regarding strong interactions because of lack of complete understanding of QCD in reference to hadron physics. Nevertheless, by and large, the S-Matrix theory lay dormant in the rest of the 20th century, waiting to be reinvigorated in the 21st-century scenario of fundamental physics. 
 
\section{ S-Matrix theory: 21st century developments}
S-Matrix theory has seen an resurgence in interest in 21st century. However it is less radical than the 20th century scenario. S-Matrix theory is being conceived as a parallel way to look for interesting theories. Well this mood was already there in the previous century.
But the term S-Matrix bootstrap has been tamed to fit into this more moderate picture. The main impetus to look into S-Matrix bootstrap came from the impressive success of a similar programme in the realm of conformal field theory, the \emph{conformal bootstrap}. Conformal bootstrap \cite{Poland:2018epd, Bissi:2022mrs} is the progarmme of solving conformal field theory using the internal consistency of conformal symmetry, uniarity, crossing and operator product expansion (OPE). During the past decade, powerful analytical and numerical tools have been developed to extract fruitful results out of conformal bootstrap. The practitioners of this programme soon realized that these same tools might be useful to resurrect the long dormant S-Matrix programme, possibly with suitable modification. Apart from this, an independent exploration into effective field theories have brought the spotlight on S-Matrix theory. 

The central theme of modern S-Matrix theory, or S-Matrix bootstrap (we will use the terms interchangeably), can be thought of as exploration in the space of S-matrices. As opposed to finding the unique S-Matrix for a particular interaction as concieved in the original bootstrap programme, the modern development takes a comparatively moderate approach.    

\subsection*{Numerical S-Matrix bootstrap} Inspired by the impressive success of the numerical analysis based on semi-definite programming (SDP), a numerical scheme for S-Matrix bootstrap was proposed in \cite{SBoot3}. Maximal analyticity and unitarity in the form of positivity was taken as inputs to form a semi-definite linear programming. What do we look for through such numerical explorations? These numerical explorations are used to put constraints on the parameters of an associated quantum field theoretical model. For example, \cite{SBoot3} obtained constraints on cubic coupling in a quantum field theory of scalars by relating the cubic coupling to a particular element of the S-Matrix and constraining that parameter. Such numerical exploration was carried out for pion S-Matrix to explore the theory space of pion physics and constrain the parameters of chiral perturbation theory, the quantum field theoretical framework used to tackle pion physics. Similar explorations have been undertaken  to put bounds on other physically interesting parameters related to quantum field theories, like anamoly coefficients \cite{Karateev:2022jdb}.    

\subsection*{Explorations in the space of EFTs} 
Effective field theory (EFT) is another front where S-Matrix theory, particularly the dispersion relation formalism, is playing important role.  EFT is a framework to work with non-renormalizable quantum field theories. One of the key lesson of Wilson's theory of renormalization is that physics decouples at different scales. Thus if we are interested at a particular energy scale, we can as well work with a theory that is valid upto that energy scale. This introduces a cutoff $\Lambda_0$ into the theory and then we can even work with non-renormalizable interactions, the higher dimensional operators in the Lagrangian.  The efficiency of such a description crucially depends on the size of the dimensionfull coupling parameters multiplying these higher dimensional operators. S-Matrix theory can be used to bound \cite{nima1, Caron-Huot:2020cmc, Caron-Huot:2021enk, sasha} these coupling parameters . 

To understand the basic working principle, we first recall that an effective field theory can be obtained by integrating out the energy modes higher than $\Lambda_0$ of some renormalizable quantum field theory. This renormalizable quantum field theory is called the \emph{UV completion} of the EFT concerned.  In general, given an arbitrary EFT, it is not possible to know the UV completion. However, an S-Matrix of the UV-complete theory must satisfy the basic requirements of unitarity and analyticity. An EFT S-Matrix amplitude can be obtained as  a low energy expansion of the corresponding S-Matrix amplitude in the UV-complete theory. The constraints that follow from the requirements of uniatrity and causality then translates bounds, in particular $2-$sided bounds, on the coefficients of the low energy expansion of the scattering amplitude. These coefficients being related to the coefficients of the EFT Lagrangian, we ultimately constrain the space of physical effective field theories in this way.

\subsection*{Holographic S-Matrix} 
A new insight into flat space S-matrix theory comes from $AdS/CFT$ holographic duality. In $AdS/CFT$ holographic duality,  $AdS$ physics  is encoded into a CFT at the $AdS$ boundary. In particular, correlation functions in the boundary CFT encodes information about scattering in the bulk $AdS$. Thus, we can describe scattering in $AdS$ entirely in terms of boundary CFT. Now it is possible to take a zero curvature limit of $AdS$ scattering amplitudes to obtain flat space scattering amplitudes. Then via $AdS/CFT$ duality, the properties the flat space scattering amplitude  will be encoded in the boundary CFT data. We call the S-Matrix obtained this way \emph{holographic S-Matrix}. We can extract various properties of an holographic S-matrix using the basic CFT structures like the spectrum of the operators, the OPE coefficients etc. Since conformal bootstrap has been setup one firmer ground, one can set up a holographic S-Matrix   bootstrap in terms of CFT bootstrap

\section{Explorations in the thesis}
This thesis undertakes explorations into S-Matrices with two objectives. First objective is to develop novel mathematical tools for analyzing the S-Matrix properties and then use those tools to constrain low energy EFT data. In particular, we borrow the study of univalent and typically real functions from the mathematical field of geometric function theory. The key to this connection is an manifestly crossing symmetric dispersion relation (CSDR) for the scattering amplitude due to Auberson and Khuri \cite{AKi}. The dispersion kernel turns out to be univalent and typically real functions.  Now these functions are known to satisfy various bounds. In particular, the Taylor coefficients of these functions are bounded. We translate these bounds into $2$-sided bounds on EFT parameters. In recent times \cite{Caron-Huot:2020cmc}, such $2-$sided bounds have been explored into. However, most of those explorations are numerical. The analysis presented in this thesis establishes robust mathematical foundation for these bounds.  We check these bounds against numerical S-Matrices for pion scattering, low energy string amplitudes. 
Next we extend the analysis to spinning amplitudes. For spinning amplitudes, we consider the so-called helicity amplitudes. The advantage of working with helicity amplitudes are that they diagonalize uniatrity. However the price we pay is that they are not crossing-symmetric by themselves. We overcame this technical hurdle by constructing crossing-symmetric linear combinations of the helciity amplitudes using the representation theory of $S_3$, the group implementing crossing symmetry. Then these manifestly crossing symmetric amplitudes admit the CSDR, and we repeat the analysis that we did for the scalar amplitude. We consider $4-$photon and $4-$graviton scattering amplitudes and bound the parameters of low-energy electro-magnetic and gravitational EFTs.  

In a parallel objective, we take up exploration into holographic S-Matrices. In particular, we obtain a Froissart-Martin bound for holographic S-Matrix starting from Mellin amplitude of the CFT $4-$point correlator. The key point of this analysis is that we show how conformal field theory structure gives rise to properties of flat space scattering amplitudes. One of our curious findings is that in space-time dimensions higher than $6$, the holographic Froissart-Martin bound is weaker than the usual flat space bound.

\section{Organization of the thesis}
 The thesis is organized as follows. We start with a panoramic review of S-matrix theory. We review various results of regarding S-matrices in the framework of aaxiomatic quantum field theory in chapter \ref{smatrev}. Next we review the some essential tools from the mathematical field of geometric function theory (GFT) in chapter \ref{GFTRev}. Next we analyze scattering amplitudes of scalar particles with tools from GFT in chapter \ref{biberbach}. We extend this analysis to spinning amplitude in chapter \ref{spinboot}. Next we turn to holographic S-matrices in chapter \ref{mellinbound} where we derive a Froissart bound on holographic S-matrix amplitudes. We conclude discussing some possible future endeavours  in chapter \ref{conclusion}. 



%% file: smatrev.tex
\chapter{ S-Matrix : A Bird's Eye Review}\label{smatrev}
\section{Introduction}
In this chapter, we review some basics of S-Matrix theory as analyzed in the axiomatic quantum field theory framework. We restrict ourselves to the discussion of $2-2$ scattering, where most of the results have been obtained. We won't be able to review all the results in gory detail. Instead, we have reviewed some key arguments and results. In particular, our focus has been on spelling out the basic argument behind how causality and unitarity give rise to various analyticity properties of the scattering amplitude. We also spell out a derivation of the famed Froissart-Martin bound. This bird's eye review aims to set the foundation upon which our works in subsequent chapters are built. We suggest the reader consult more extensive reviews like \cite{Sommer1}, \cite{MartinLec} for all the detailed arguments which are lacking here. 

Finally, in this exposition, we could not say anything about the rich non-quantum field-theoretic analysis of S-Matrix, which was developed during 1960s. See \cite{ELOP}, \cite{Iagol} for authoritative expositions into this abstract theory of S-Matrix. 

Let us give a heads up as to what we review in this chapter. We consider the $2-2$ elastic scattering of identical massive scalar bosons in Minkowski spacetime $\mathbb{M}^{3,1}$. We want to emphasize that all the properties we discuss generalize straightforwardly, often unaltered, to general spacetime dimensions. At places, we will quote the result in general spacetime dimension without any proof for later convenience, in particular chapter \ref{biberbach} where our analysis is valid in general spacetime dimensions. Further, the generalizations to scattering amplitudes of external spinning particles will be discussed in chapter \ref{spinboot} where we take up analysis of spinning amplitudes. We start with a short description of particles, states and the corresponding state spaces. After this, we have reviewed an abstract definition of the S-Matrix. Next, we have delineated the basic assumptions of the Lehmann-Symanzik-Zimmermann axiomatic quantum field theory of S-Matrix. Finally, we have reviewed various analyticity properties and bounds that follow as consequences of causality and unitarity through multiple sections.  
\section{Particles, states and Fock space}
S-Matrix describes scattering of quantum particles. Thus we need a proper idea of quantum states of relativistic particles. Wigner \cite{Wigner1} first gave concrete meaning to a relativistic quantum particle in terms of the symmetry group of Minkowski spacetime, the Poincar\'{e} group. Let us first discuss briefly Wigner's central result. The discussion below closely follows \cite{WeinQFT}.
\subsection*{Poincar\'{e} invariance and one particle states: Wigner classification}
We consider a relativistic theory in $4-$dimensional Minkowski spacetime $\mathbb{M}^{3,1}$. We choose mostly positive signature for the metric. The canonical metric tensor is given  
\be
\eta_{\mu \nu}\equiv\text{diag.}\{-1,+1,+1,+1\}.
\ee  A Minkowski vector $x\in \mathbb{M}^{3,1}$ is defined to be \emph{timelike} if $x^2:=\eta_{\mu \nu}x^\m x^\n<0$, \emph{spacelike} if $x^2:=\eta_{\mu \nu}x^\m x^\n>0$, and \emph{null} or \emph{lightlike} if $x^2=0$. 

The isometry group of $\mathbb{M}^{3,1}$ is the Poincar\'{e} group which is the semi-direct product of Lorentz group, $SO(3,1)$ and spacetime translation group, $\mathbb{R}^{3,1}$, 
\be 
\text{ISO}(3,1):=\mathbb{R}^{3,1}\rtimes SO(3,1).
\ee  The generators of the Poincar\'{e} groups are 
\be 
\{\mmp^\m,\, \mj^{\r\s}\},
\ee where $\mmp^\m$ generates spacetime translation and $\mj^{\r\s}=-\mj^{\s\r}$ generates Lorentzian spacetime rotations. These generators satisfy the Poincar\'{e} algebra $\mathfrak{iso}(3,1)$ 
\begin{align}
	\begin{split}
	&[\mmp^\m, \mmp^\n]=0,\qquad [\mmp^\m,\mj^{\r\s}]=i\, \left(\eta^{\m\s}\mmp^\r-\eta^{\m\r}\mmp^\s\right),\\
	&[\mj^{\m\n}, \mj^{\r\s}]=\,i\,\left(\eta^{\m\r}\mj^{\n\s}-\eta^{\n\r}\mj^{\m\s}+\eta^{\m\s}\mj^{\r\n}-\eta^{\s\n}\mj^{\r\m}\right).
	\end{split}
\end{align} 	
$\mathfrak{iso}(3,1)$ has two Casimirs $\mc_1,\, \mc_2$.
\begin{enumerate}
	\item  $\mc_1=\mmp^2:=\eta_{\m\n}\mmp^\m \mmp^\n$.  
	\item   $\mc_2=\mw^2:=\eta_{\m\n}\mw^\m\mw^\n,$ where $\mw^\m$ is the Pauli-Ljubanski pseudovector defined as 
	\be 
	\mw^\m:=\frac{1}{2}\, \e^{\m\n\r\s}\, \mmp_\n\, \mj_{\r\s},
	\ee $\e^{\m\n\r\s}$ being the totally anti-symmetric Levi-Civita tensor  with $\e^{0123}=-1$. 
\end{enumerate} 
Wigner defined the one-particle quantum relativistic states to be  transforming in non-negative energy irreducible unitary representations of the Poincar\'{e} group.   Unitary irreducible representations of Poincar\'{e} group are infinite-dimensional because the Poincar\'{e} group is non-compact. These representations can labelled by the eigenvalues of the two Casimirs. According to these labels, quantum states of relativistic particles are in the following representations. 
\begin{enumerate}
	\item [(a)] \be \mmp^2=-m^2,\qquad \mw^2=m^2\, s(s+1),\qquad m\in\mathbb{R}^+.
	\ee
	These representations describe  \emph{massive particles} with mass $m>0$. $s$ is the spin of the particle which labels the unitary irreducible representation of the spin group $SO(3)$, and thus, takes discrete positive half-integer values, $s\in \{p/2\,|\, p=0,\,1,\,2,\,\dots\}.$ Since the spin projection $s_3$ can take on values between $-s$ and $+s$, spin $s$ massive particles fall into $(2s+1)-$dimensional multiplets .

	\item [(b)] \be  \mmp^2=0,\,\qquad \mw^2=0.\ee 
	Since both $\mmp^\m$ and $\mw^\n$ are light-like, they are proportional
	\be 
	\mw^\m=\,h\,\mmp^\m.
	\ee These representations describe \emph{massless particles}. The proportionality factor $h$ is called the \emph{helicity} and equals to $\pm s$, where $s$ is the spin of the particles labeling the unitary irreps of the little group which is, for massless particles, $ISO(2)$ instead of $SO(3)$. $s$ takes on positive half-integer values as before. 
	
\end{enumerate} 	  
In this chapter, we will only work with massive scalar, $s=0$, particles.  
\subsection*{Fock space}  
In relativistic quantum theory, a generic scattering event is described in terms of a system of an arbitrary but finite number of \emph{non-interacting} relativistic particles. One natural way to tackle this is in terms of Fock space states. Formally a Fock space is a tensor algebra over a single-particle Hilbert space,$\mh_1$, i.e. 
\be 
\mf(\mh_1)=\bigoplus_{n=0}^\infty  \mh_1^{\tenp^n},
\ee where $\mh_1^{\tenp^n}$ is just $n-$fold tensor product space of $\mh_1$, also called $n$th tensor power of $\mh_1$. If we consider identical Bosons or Fermions then the tensor power will be either symmetric or anti-symmetric power respectively. More concretely, Fock space  corresponding to a single kind of Bosonic particle is simple symmetric algebra over $\mh_1$
\be 
\mf_B(\mh_1)=\bigoplus_{n=0}^\infty  \mh_1^{\vee^n},
\ee where $\mh_1^{\vee^n}$ is $n$th symmetric power of $\mh_1$ spanned by symmetric product states 
\be 
\Ket{\psi_{\s_1}}\vee\Ket{\psi_{\s_2}}\vee\dots\vee \Ket{\psi_{\s_n}}:= \frac{1}{n!}\sum_\p \Ket{\psi_{\s_{\p(1)}}}\tenp\Ket{\psi_{\s_{\p(2)}}}\tenp\dots\tenp \Ket{\psi_{\s_{\p(n)}}},\qquad \p\in S_n,
\ee   $S_n$ being the group of permutations of $\{1,2,\dots,n\}$. The $n=0$ subspace, $\mh_1^{\vee 0}\cong \mathbb{C}$ is called the \emph{vacuum sector} with the \emph{Bosonic Fock vacuum} $\ket{0_B}$ normalized to $ \braket{_B 0|0_B}=1$. Note that the vacuum state \emph{is not} the null vector.  Similarly, a  Fermionic Fock space is an exterior algebra over the single particle space $\mh_1$
\be 
\mf_F(\mh_1)=\bigoplus_{n=0}^\infty  \mh_1^{\wedge^n},
\ee  the $n-$particle Fermionic subspace being spanned by exterior product states of the form 

\be 
\Ket{\psi_{\s_1}}\wedge\Ket{\psi_{\s_2}}\wedge\dots\wedge \Ket{\psi_{\s_n}}:= \frac{1}{n!}\sum_\p\, \e(\p) \Ket{\psi_{\s_{\p(1)}}}\tenp\Ket{\psi_{\s_{\p(2)}}}\tenp\dots\tenp \Ket{\psi_{\s_{\p(n)}}},\qquad \p\in S_n,
\ee  $\e(\p)$ being the sign of the permutation, $+1$ for even permutation and $-1$ for odd permutations. The \emph{Fermionic Fock vacuum} is again defined by $\braket{_F 0|0_F}=1$.

The algebraic structure of Fock space allows one to introduce a pair of linear operators for each single-particle state called \emph{creation} and \emph{annihilation} operators. The creation operator $a^\dagger(\psi)$ creates a single particle in the state $\ket{\psi}$ and the annihilation operator $a(\psi)$ destroys a single particle in the state $\ket{\psi}$. These operators are hermitian conjugates of each other. 
These operators satisfy canonical commutation or anti-commutation algebra, respectively, for Bosons and Fermions 
\begin{align} 
	[a^\dagger(\psi), a^\dagger(\f) ]_{\mp}=0,&\qquad [a(\psi), a(\f) ]_{\mp}=0,\\
	[a(\psi), a^\dagger(\f) ]_{\mp}&=\Braket{\psi|\f}\mathbb{1}.
\end{align} The $``-''$ sign corresponds to commutation relation for Bosons, and the $``+''$ sign correspond to anti-commutation relation for Fermions. We will call this algebra of operators as \emph{Fock algebra}.
One can construct an $n-$particle state with repeated action of creation operators on the vacuum state. In particular, the single particle states are given by 
\be 
\ket{\psi}=a^\dagger (\psi)\ket{0}.
\ee Further, a Bosonic $n-$particle state can be obtained as 
\be 
\Ket{\Psi_n^B(\s_1,\s_2,\dots,\s_n)}:=a_B^\dagger(\psi_{\s_1})a_B^\dagger(\psi_{\s_2})\dots a_B^\dagger(\psi_{\s_n})\ket{0_B}=\sqrt{n!}\,\Ket{\psi_{\s_1}}\vee\Ket{\psi_{\s_2}}\vee\dots\vee \Ket{\psi_{\s_n}}\,
\ee and, similarly for Fermionic states 
\be 
\Ket{\Psi_n^F(\s_1,\s_2,\dots,\s_n)}:=a_F^\dagger(\psi_{\s_1})a_F^\dagger(\psi_{\s_2})\dots a_F^\dagger(\psi_{\s_n})\ket{0_F}=\sqrt{n!}\Ket{\psi_{\s_1}}\wedge\Ket{\psi_{\s_2}}\wedge\dots\wedge \Ket{\psi_{\s_n}}.
\ee 

In scattering problem one is interested in transition among states of particles of definite momenta, spin and other possible charges. 
Thus a generic single particle state is of the form $\ket{k_\a,\,s_\a\,,\,\l_\a}$, where $k$ is the Lorentzian on-shell $4-$momentum normalized by $k_\a^2=m_\a^2$, $s$ is the suitable spin degree of freedom, and $\l$ labels all other degrees of freedom collectively. The subscript $\a$ runs over different kinds of particles. The total Fock space now will be tensor product of Fock space corresponding to each kind of particles
\be\label{FreeFock} 
\mf=\bigotimes_{\a}\mf\left(\mh^{(\a)}_1\right),
\ee $\mh_1^{(\a)}$ being the single particle space for $\a$th particle spanned by the states $\{\ket{k_\a,\,s_\a,\,\l_\a\}}$. The Fock vacuum is the unique Poincar\'{e} invariant state 
\be 
U(\Lambda,a)\Ket{0}=\Ket{0}.
\ee The Fock vacuum can be written as tensor product of the Fock vacuua of all the \emph{different} species of particles 
\be 
\Ket{0}=\bigotimes_\a \Ket{0_\a}.
\ee 
The single-particle states are normalized by 
\be \label{mominner}
\Braket{k_\a,\,s_\a,\,\l_\a|k_\b,\,s_\b,\,\l_\b}=\sqrt{2E_{k_\a}2E_{k_\b}}\,\d_{\a\b}\, \d^{3}(\mathbf{k}_\a-\mathbf{k}_\b).
\ee Here $\d_{ij}$ takes care of orthogonalization in all the other degrees of freedom. The Fock algebra for these states are then given by 
\begin{align} 
	[a^\dagger(k_\a,\,s_\a,\,\l_\a), a^\dagger(k_\b,\,s_\b,\,\l_\b) ]_{\mp}=0,&\qquad [a(k_\a,\,s_\a,\,\l_\a), a(k_\b,\,s_\b,\,\l_\b) ]_{\mp}=0,\\
	[a(k_\a,\,s_\a,\,\l_\a), a^\dagger(k_\b\,s_\b,\,\l_\b) ]_{\mp}&=\sqrt{2E_{k_\a}2E_{k_\b}}\,\d_{\a\b}\, \d^{3}(\mathbf{k}_\a-\mathbf{k}_\b).
\end{align} 
The completeness relation for the single-particle states are given by 
\be 
\sum_\b\int d\m_1(k_\a)\,\Ket{k_\a,\,s_\a,\,\l_\a^\b}\Bra{k_\a,\,s_\a,\,\l_{\a}^\b}=\mathbb{1}_{\a 1},\qquad d\m_1(k_\a):=d^4k_\a\,\th(k_\a^0)\d(k_\a^2-m_\a^2).
\ee Here the sum over $\b$ means sum and integrals overall degrees of freedom of the $\a$th particle. $\mathbb{1}_{\a 1}$ is the identity operator over the single-particle subspace corresponding to $\a$th particle type. 

The completeness relation on the full Fock space $\mf$ is given by 
\begin{align}\label{genmeas}
	\begin{split}
		\mathbb{1}=\ket{0}\bra{0}+\sum_{N\ge 1} \int \,& d\m_N(p) \Ket{\{p_{\a_i},\,s_{\a_i},\,\l_{\a_i}\}_{\text{card}(A)=N}}\Bra{\{p_{\a_i},\,s_{\a_i},\,\l_{\a_i}\}_{\text{card}(A)=N}},\quad A=\{\a_i\}\\
		&d\m_N(p)=\left(\prod_{\sum_\a n_\a=N}\,\frac{1}{n_{\a}!}\right)\prod_{A}\, d^4p_{\a_i}\,\th(p_{\a_i}^0)\d\left(p_{\a_i}^2-m_{\a_i}^2\right)\,\w(s_{\a_i},\,\l_{\a_i}).
	\end{split}
\end{align}  Let us explain the notation a bit. $\a_i$ corresponds to $i$th particle of type $\a$.      $\text{card}(A)$ is the cardinality of the set $A$ which index the individual particles. The label $\text{card}(A)$ in the states signify that we the \emph{total} no. of particles in the corresponding state is $N$. Further, $n_\a$ is total number of particles of type $\a$, so that 
\bes 
\sum_\a\,n_\a=N.
\ees 
Finally, $\mathbb{1}$ is the identity operator on the total Fock space $\mf$. The integral stands for summation over all the quantum numbers corresponding to fixed $N$, comprising of integrals over continuous quantum degrees of freedom including the momenta as well as the sum over all other discrete quantum numbers. $\w(s,\l)$ stands for the measure for all other qauntum numbers.

 We will work with a much simpler system comprising of single neutral scalar massive Bosons with mass $m>0$. The measure $d\m_N(p)$ then becomes 
 \be\label{idBos} 
 d\m(p_1,\dots, p_N)=\frac{1}{N!}\,\prod_{i=1}^{N}\, d^4p_{i}\,\th(p_{i}^0)\d\left(p_{i}^2-m^2\right)
 \ee  
\subsubsection*{A mathematical pointer} 
We observe that the space of single particle states, $\mh_1^{\a}$ is \emph{not} a Hilbert space. This is because norm of the states $\{\ket{k_\a,\, s_\a,\, \l_\a}\}$ do not exists as evident from the inner product expression \eqref{mominner}. Thus the state of definite momentum is not a physically realizable state. However, we can construct a wavepacket state with a small spread around some momentum value by taking a superposition of the states $\{\ket{k_\a,\, s_\a,\, \l_\a}\}$. Define a state 
\be 
\ket{\psi(k_\a, \, \d k_\a)}=\int d^4 p_\a\,\, \th(p_\a^0)\, \d(\p_\a^2-m_\a^2)\, f(p_\a)\, \ket{p_\a,\, s_\a,\, \l_\a},
\ee   where $f(p_\a)$ is a function of momentum which is peaked at $p_\a=k_\a$ and has a finite support of measure $\d k_\a$ such that  $ \braket{\psi(k_\a, \, \d k_\a)|\psi(k_\a, \, \d k_\a)}<\infty$. $\ket{\psi(k_\a, \, \d k_\a)}$ describes a physical state. In principle, we should work with such wavepacket states. However, we will work with the states $\ket{k_\a, \, s_\a,\, \l_\a}$ with impunity. While we will not delve into mathematical details, such states can be accommodated within a more general structure called the \emph{rigged Hilbert space}. Thus the precise mathematical identity of the space $\mh_1^\a$ is that of a rigged Hilbert space.  


\section{S-Matrix: Basic construction }\label{Smatbasic}

The physics of relativistic particle interaction is probed with the help of scattering experiments. In a typical scattering experiment, one starts with particles at \emph{asymptotic past infinity}, $t\to-\infty$,  far apart so that there is no interaction among them, and ends with particles at \emph{asymptotic future infinity}, $t\to+\infty$, so far apart that interactions among them have died off. Following \cite{BogoQFT}, we start by giving a formal definition of this physical picture. Then we will delineate the practical implementation of the formal structure. 
\subsection{ Mathematical definition}
In order to analyze this situation, we introduce two sets of Hilbert states: \emph{`in' states, $\ket{\Psi^+_\a}$} and \emph{`out' states, $\ket{\Psi_\a^-}$}. Let the corresponding state spaces be $\mh^+$ and $\mh^-$, respectively. These states are \emph{defined} by that \emph{if observations are made at $t\to-\infty$ or $t\to+\infty$, then states corresponding to the free particles described by the label $\a$ will be $\ket{\Psi^+_\a}$ or $\ket{\Psi_\a^-}$, respectively.} It is to be emphasized that states $\ket{\Psi^+_\a}$ and $\ket{\Psi_\a^-}$ are defined as \emph{Heisenberg states}, i.e. they don't evolve with time. In particular, these states are supposed to encode the entire spacetime history of the system of particles. Most importantly, they are \emph{not} to be considered as asymptotic limits of some time-dependent Schrodinger state $\ket{\Psi(t)}$ at $t\to\mp\infty$. The use of Heisenberg states is crucial to maintaining manifest Lorentz invariance. The \emph{explicit} time evolution of Schrodinger states treats time on a different footing than the space variables. This destroys the manifest Lorentz invariance. Since a Heisenberg state encapsulates the entire spacetime history of the system, one can keep Lorentz invariance of the system manifest with such states. 

Next, we introduce a state-space of non-interacting particles, $\fh$. In a scattering experiment, one is usually interested in the transition between states with an arbitrary number of particles. Therefore, it is best to use the second-quantization picture. For \emph{non-interacting} particles, we can use the Fock states to describe such states. Thus $\mf$ is the Fock space of relativistic non-interacting particles. It is to be noted that the Fock states are again Heinseberg states. How is this Fock space related to the spaces $\mh^+$ and $\mh^-$? The relation is expressed by a pair of linear operators $\Omega^{\pm}$. These operators define \emph{isometric embeddings}, they preserve inner products. For a given $\ket{\Psi}\in \fh$ one has 
\be \label{MollerOp} 
\Omega^{\pm}\ket{\Psi_\a}:=\ket{\Psi_\a^\pm}\in\mh^\pm.
\ee These operators are called \emph{M\o ller operators}. We want to emphasize that there are (in general) \emph{two distinct} vectors corresponding to the one and the same vector $\ket{\Psi_\a}$ of the Fock space $\fh$. Physically, the operators $\Omega^{+}$ and $\Omega^-$ can be thought to define two \emph{different} choices of frames of references from which the observer views the system; different observers observe \emph{equivalent} states, not \emph{same} states. At this point, we would like to stress that the spaces $\mh^+$ and $\mh^-$ are spaces of interacting particles, \emph{not} free particles. We will make this clearer in the following subsection.

Finally, we require the \emph{asymptotic completeness condition}
\be \label{asympcompl}
\mh^+=\mh^-.
\ee  This expresses the physical picture that we are considering the evolution of the same system of particles; nothing is being created or destroyed.  

With this formal construct, the central problem of scattering becomes the evaluation of the transition amplitude 
\be 
\Braket{\Psi_\a^-|\Psi_\b^+}.
\ee Using \eqref{MollerOp} this amplitude becomes equal to $\Braket{\Psi_\a|\left(\Omega^-\right)^\dagger\Omega^+|\Psi_\b}$, which prompts defining a linear operator the Fock space $\fh$, the S-Matrix
\begin{align}
	\begin{split} 
		\widehat{S}&:\fh\to\fh,\\
		\widehat{S}&:=\left(\Omega^-\right)^\dagger\Omega^+.
	\end{split}
\end{align} The elements of the S-Matrix are then clearly
\be 
S_{\a\b}=\Braket{\Psi_\a^-|\Psi_\b^+}.
\ee 

\subsection{A formal construction}
We gave a rather abstract mathematical definition of scattering in the previous section. How can one realize this? We give a formal construction in this section following \cite{WeinQFT}. We start with dividing the time-translation generator $H$ for the entire system into two parts, a free-particle Hamiltonian $H_0$ and an interaction potential $V$,
\be 
H=H_0+V.
\ee $H_0$ generates the time-translation of the Fock space $\mf$ in \eqref{FreeFock}. Now we assume that we can make this division in such a fashion that the asymptotic states $\{\Ket{\Psi}^{\pm}_\a\}$ corresponding to the eigenstates $\{\Ket{\Psi}_\a\}$ of the free Hamiltonian $H_0$ are eigenstates of the full Hamiltonian with \emph{same} eigenvalues, i.e. given 
\begin{align}
	H_0\Ket{\Psi\a}&=E_\a\Ket{\Psi_\a},\\
	\Braket{\Psi_\a|\Psi_\b}&=\d(\a-\b),
\end{align} the asymptotic states  $\Ket{\Psi_\a^\pm}=\Omega^\pm\Ket{\Psi}$ satisfy
\be \label{Heigen}
H\Ket{\Psi_\a^\pm}=E_\a\Ket{\Psi_\a^\pm}.
\ee This can be more physically thought as the following limiting behavior 
\be \label{asymplim}
\int d\a \,e^{-iE_\a \t}g(\a)\Ket{\Psi_\a^\pm} \rightarrow \int d\a\, e^{-iE_\a \t}g(\a)\Ket{\Psi_\a},\quad \text{as}\,\,\, \t\to \mp\infty.
\ee Here $g(\a)$ is some ampltiude which is non-zero and smoothly varying over some finite range $\D E_\a$ of energies, thus defining a wavepacket.  

\eqref{asymplim} actually defines the embeddings of \eqref{MollerOp}. To see this, we rewrite the equation as the requirement that  
\be \label{asympcond}
e^{-i H \t}\int d\a\, g(\a)\Ket{\Psi_\a^\pm} \rightarrow e^{-iH_0 \t}\int d\a\, g(\a)\Ket{\Psi_\a},\quad \text{as}\, \t\to \mp\infty.
\ee From this we can now obtain concrete expressions for the  M\o ller operators as 
\be 
\Omega^{\pm}=\lim_{\t\to\mp\infty}e^{+i H\t}e^{-iH_0\t}.
\ee 
It should be kept in mind that these M\o ller operators give only meaningful results only when acting on smooth superposition of energy eigenstates. 

One of the crucial consequences of the requirement \eqref{asymplim} is that M\o ller operators define isometric embeddings, i.e. the asymptotic states are normalized just like the free-particle states.

\be\label{Asympcomp} 
\Braket{\Psi_\a^\pm|\Psi_\b^\pm}=\d(\a-\b).
\ee

\subsection{Kinematical structures}\label{kinem}
With formal aspects of the scattering introduced, let us turn to the more specific problem. We will be interested in $2\to2$ elastic scattering in this thesis. For simplicity of the discussion, we will restrict ourselves to scattering of identical neutral scalar Bosons of mass $m>0$. The discussion can be generalized straightforwardly to incorporate spin and charge. Thus the only degrees of freedom considered are those of the momenta. Let us now describe the kinematics of $2\to2$ scattering. 
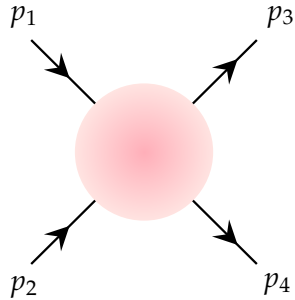
\begin{figure}[H]
	\centering
	\begin{tikzpicture}[scale=0.7]
		\shade[shading=radial, outer color=MistyRose1, inner color=LightPink1] (5,0) circle (1.3 cm);
		\draw (2.738,2.562) node{$p_1$};
		\draw (2.738,-2.562) node{$p_2$};
		\draw (7.562,2.562) node{${p}_3$};
		\draw (7.5,-2.5) node{${p}_4$};
		\draw
		[
		line width=0.3mm,
		decoration={markings, mark=at position 0.70 with {\arrow[line width=1 mm]{stealth}}},
		postaction={decorate}
		] (5.919,0.919)--(7.121,2.121);
		\draw
		[
		line width=0.3mm,
		decoration={markings, mark=at position 0.70 with {\arrow[line width=1 mm]{stealth}}},
		postaction={decorate}
		] (5.919,-0.919)--(7.121,-2.121);
		\draw
		[
		line width=0.3mm,
		decoration={markings, mark=at position 0.40 with {\arrowreversed[line width=1 mm]{stealth}}},
		postaction={decorate}
		] (4.081,0.919)--(2.879,2.121);
		\draw
		[
		line width=0.3mm,
		decoration={markings, mark=at position 0.40 with {\arrowreversed[line width=1 mm]{stealth}}},
		postaction={decorate}
		] (4.081,-0.919)--(2.879,-2.121);
	\end{tikzpicture}
	\caption{Scattering configuration}\label{Sconfig1}
\end{figure}
The S-Matrix can be decomposed as 
\be 
\widehat{S}=\mathbb{1}+i\widehat{\mt}.
\ee The non-trivial scattering is described by the matrix $\widehat{T}$. Thus, 
referred to the scattering configuration above, the matrix elements of interests are given by 
\be 
\Braket{{p}_3\,;\,{p}_4|\widehat{\mt}|{p}_1\,;\,{p}_2}
\ee
Next, using the spacetime translation invariance the momentum conserving delta function can be factored  out, and, further making use of the Lorentz invariance to write
\be \label{momconsv}
\Braket{p_3\,;\,p_4|\widehat{\mt}|p_1\,;\,p_2}=(2\p)^4\, \d^{4}\left({p}_1+{p}_2-{p}_3-{p}_4\right)\mt(s,t,u),
\ee where $(s,t,u)$ are the Mandelstam invariants defined by 
\begin{align}
	\begin{split}\label{Mandeldef}
		s&=(p_1+p_2)^2,\\
		t&=(p_1-p_3)^2,\\
		u&=(p_1-p_4)^2.
	\end{split}
\end{align} 	The on-shellness condition now translates to 
\be\label{onshell} 
s+t+u\,=\,4m^2,
\ee which implies that the scattering amplitude $\mt$ is really a function of two variables, $\mt(s,t)$ when restricted to the mass-shell. We can now choose to be in Center of Mass (CoM) frame, defined by 
\be 
\mathbf{p}_1+\mathbf{p}_2=0. 
\ee In this frame, the Mandelstam invariants have simple algebraic expressions 
\begin{align}
	s=4(m^2+\mathbf{k}^2),\qquad t=2\mathbf{k}^2\,(\cos^2\th-1),\qquad u=4m^2-s-t=-2\mathbf{k}^2(1+\cos\th);
\end{align}	with 
$
\mathbf{k}^2:=\mathbf{p}_i^2,\,\forall\,\, 1\le i\le 4,$ and $\th$ being th angle between $\mathbf{p}_1$ and $\mathbf{p}_3$, the \emph{scattering angle}. Clearly
\be \label{scatangl}
\cos\th=1+\frac{2t}{s-4m^2}\,.
\ee 

Based on the kinematical configuration above, we can talk about three channels:
\begin{eqnarray}
	s-\text{channel:}\hspace{1 cm}& 1,2\,\to\, 3,4,\hspace{1 cm}& s\ge 4m^2,\qquad 4m^2-s\le t\le0,\\
	t-\text{channel:}\hspace{1 cm}& 1,3\,\to\, 2,4,\hspace{1 cm}& t\ge 4m^2,\qquad 4m^2-t\le s\le0,\\
	u-\text{channel:}\hspace{1 cm}& 1,4\,\to\, 2,3,\hspace{1 cm}& u\ge 4m^2,\qquad s\,,t\le0,
\end{eqnarray}  with $s,\,t,\,u\in\mathbb{R}$. 

The three channels define three \emph{physical amplitudes} $\mt^s(s,t),\,\mt^t(s,t),\,\mt^u(s,t)$ whose domains of definition, $\md^s,\, \md^t,\, \md^u$, are respectively given by the physical domains for the channels as mentioned above, i.e.
\begin{align}
	\begin{split}\label{physamp}
		\md^s&=\{(s,t)\in\mathbb{R}^2\,|\,\,s\ge 4m^2,\,\, 4m^2-s\le t\le0\},\\
		 \md^t&=\{(s,t)\in\mathbb{R}^2\,|\,\,t\ge 4m^2,\, \,4m^2-t\le s\le0\},\\
		\md^u&=\{(s,t)\in\mathbb{R}^2\,|\,\,u\ge 4m^2,\,\, s\,,t\le0\}.
	\end{split}
\end{align}

\subsection*{Partial wave expansion}
The $2-$particle states do not transform irreducibly under the Poincar\'{e} group. However, they can be decomposed into Poincar\'{e} irreps which are states of definite \emph{total} $4-$momenta and definite {\emph{total}} angular momenta. Thus these states are of the generic form (we are considering massive scalar particles only) 
\be 
\ket{E, \vec{P}\,;\, \ell, \m\,;\,\a},\qquad \m\in\{-\ell,\,-\ell+1,\,\dots\,,\ell-1,\,\ell\},\,\quad \ell\in \mathbb{Z}^{\ge}.
\ee Further, the physical, or on-shell, states satisfy the constraints 
$E>0,\,E^2>\vec{P}^2$. 
We consider these states to be normalized by 
\be 
\braket{E', \vec{P}'\,;\, \ell', \m'\,|E, \vec{P}\,;\, \ell, \m}=(2\p)^4\, \d(E-E')\, \d^{(3)}(\vec{P}-\vec{P}')\, \d_{\ell\ell'}\d_{\m\m'},
\ee so that the $2-$particle completness relation in terms  of these states read
\be \label{2completeness}
\sum_{\ell}\sum_{\m=-\ell}^{\ell}\,\int\,\frac{dE\, d^3\vec{P}}{(2\p)^4}\,\th(E)\,\th(E^2-\vec{P}^2)\, \Ket{E, \vec{P}\,;\, \ell, \m\,}\Bra{E, \vec{P}\,;\, \ell, \m\,}=\mathbb{1}_2,
\ee $\mathbb{1}_2$ being the $2-$particle identity. Now one can decompose the $2-$particle state $\ket{p_1,\, p_2}$ in terms of these states as following 
\be\label{CGdcm} 
\Ket{p_1,\, p_2}=\int\frac{d^4 P}{(2\p)^4}\, \th(P^0)\, \th(P^2)\, \sum_{\ell, \m}\, \ket{E, \vec{P}\,;\, \ell, \m}\braket{E, \vec{P}\,;\, \ell, \m|p_1, p_2},
\ee  where $P\equiv (E,\vec{P})$ is the total $4-$momentum corresponding to the state $\ket{E, \vec{P}\,;\, \ell, \m}$, i.e.
\begin{align} 
	\begin{split}
		\mathcal{P}^0\,\Ket{E, \vec{P}\,;\, \ell, \m}&=E\, \Ket{E, \vec{P}\,;\, \ell, \m},\\
		\mathcal{P}^i\,\Ket{E, \vec{P}\,;\, \ell, \m}&=P^i\, \Ket{E, \vec{P}\,;\, \ell, \m}.
	\end{split} 
\end{align} 

Due to translation invariance, the Clebsch-Gordon (CG) coefficients of this decomposition obey 
\be \label{CG1}
\braket{E, \vec{P}\,;\, \ell, \m|p_1, p_2}\propto \d(E-p_1^0-p_2^0)\,\d^{(3)}(\vec{P}-\vec{p}_1-\vec{p}_2).
\ee We won't need the most general expression. In a while, we will work with CG coefficient in the centre of mass frame where the result is considerably simpler.  

It  follows from Wigner-Eckart theorem that, due to the Poincar\'{e} invariance, the S-Matrix will be diagonal in the basis of these states 
\begin{align}\label{SWigEck} 
	\begin{split}
		\braket{E', \vec{P}'\,;\, \ell', \m'|\widehat{S}|E, \vec{P}\,;\, \ell, \m}&=\d(E-E')\d^{(3)}(\vec{P}-\vec{P}')\,d_{\ell \ell'}\d_{\m \m'}\, S_\ell(E),\\
		\braket{E', \vec{P}'\,;\, \ell', \m'|\widehat{\mt}|E, \vec{P}\,;\, \ell, \m}&=\d(E-E')\d^{(3)}(\vec{P}-\vec{P}')\,d_{\ell \ell'}\d_{\m \m'}\, \mt_\ell(E)
	\end{split}
\end{align} Making use of this diagonalization we obtain what is called the \emph{partial wave expansion} of the scattering amplitude $\mt(s,t)$. Let us delineate below how one can obtain this expansion. We will consider to work in the centre of mass frame. Thus we consider the amplitude 
\be \label{CoMamp}
\braket{p_3,\,p_4|\widehat{\mt}|p_1,\,p_2}=(2\p)^4\,\d(E-E')\d^{(3)}(0)\,\mt(s,t).
\ee In writing this, we have used  $\vec{p}_1+\vec{p}_2=\vec{p}_3+\vec{p}_4=0$, and $E=p_1^0+p_2^0,\, E'=p_3^0+ p_4^0$. 

Next we will use spherical polar coordinates for the momenta. Next, we will work with the momentum configuration as following 
\begin{align}
	\begin{split}
		\vec{p}_1=-\vec{p}_2= (0, 0,\mathbf{p}),\qquad \vec{p}_3=-\vec{p}_4=(\mathbf{p}\sin\th\cos\f,\mathbf{p}\sin\th\sin\f,\mathbf{p}\cos\th),
	\end{split}
\end{align} with 	$\th\in [0,\p]$ and $\f\in(0,2\p),\, \f=0\sim\f=2\p$. The configuration is chosen so that the polar angle $\th$ coincides with the scattering angle introduced in  \eqref{scatangl}. We will denote the $2-$particle states as 
$ 
\ket{p_1,p_2}\equiv \ket{\mathbf{p},0,0},\,\ket{p_3,p_4}\equiv \ket{\mathbf{p},\th,\f}.
$ In this momentum configuration, one can evaluate the CG coefficient \eqref{CG1} to obtain for the corresponding CG decomposition \eqref{CGdcm} (see \cite{Joaospin} for a recent review and derivation of these results)
\be 
\Ket{\mathbf{p},\th,\f}=\sum_{\ell, \m}\, C_\ell(E)\,e^{-i\m\f}\, d_{\m0}^{(\ell)}(\th)\,\ket{E,0\,;\,\ell, \m},\quad E=2\sqrt{\mathbf{p}^2+m^2},
\ee where $d_{\m_1 \m_2}^{(\ell)}(\th)$ are the Wigner small $d-$matrices, and 
\be 
C_\ell(E)^2=8\p\,(2\ell+1)\, \frac{E}{\mathbf{p}}.
\ee 
In particular
\be 
d_{\m 0}^{(\ell)}(\th)=\d_{\m0}\, P_\ell(\cos\th),
\ee $P_\ell(x)$ being the ususal Legendre polynomial.  We can use this to obtain the simple expression 
\be 
\Ket{\mathbf{p},\th,\f}=\sum_{\ell}\, C_\ell(E)\,\, P_\ell(\cos\th)\,\ket{E,0\,;\,\ell},\qquad \ket{E,0\,;\,\ell}\equiv \ket{E,0\,;\,\ell, \m=0}.
\ee 
Inserting this decomposition into the centre of mass amplitude of  \eqref{CoMamp}, we obtain 
\begin{align}
	\braket{p_3,\,p_4|\widehat{\mt}|p_1,\,p_2}&= \sum_{\ell,\ell'}C_{\ell'}^*(E')\,C_\ell(E)\, P_{\ell'}(\cos\th)\, \braket{E',0\,;\,\ell'|\widehat{\mathcal{T}}|E,0\,;\,\ell}\nn\\
	&=\d(E-E')\,\d^{(3)}(0)\, \sum_{\ell,\ell'}C_{\ell'}^*(E')\,C_\ell(E)\, P_{\ell'}(\cos\th),\d_{\ell\ell'}\mt_\ell(E),\quad [\text{using} \eqref{SWigEck}]\nn\\
	&=(2\p)^4\,\d(E-E')\,\d^{(3)}(0)\,\left[16\p\,\frac{\sqrt{s}}{\sqrt{s-4m^2}}\,\sum_{\ell\,  \text{even}}\,(2\ell+1)\,f_\ell(s)\,P_\ell\left(\cos\th\right)\right],
\end{align} where $E=\sqrt{s}$ and 
\be 
f_\ell(s):=\frac{\mt_\ell(\sqrt{s})}{(2\p)^4}.
\ee The restriction of the sum being over only even values of $\ell$ is a consequence of identity of particles which require invariance of the S-Matrix under $\th\to\p-\th$. Thus we have the desired partial wave expansion  

\be \label{partial wave}
\mt(s,t)=16\p\sqrt{\frac{s}{s-4m^2}}\mathlarger{\mathlarger{\sum}}_{\substack{\ell=0\\ \ell\,\text{even}}}(2\ell+1)\,f_\ell(s)\,P_\ell(\cos\th).
\ee

\section{Causality and Axiomatic Analyticity}
\emph{Causality} is another fundamental principle of physics. Therefore, arises a natural question: how is causality encoded in the structure of the S-Matrix? This question does not have a complete and clear answer yet. Nevertheless, the clear consensus is that causality manifests itself into \emph{analyticity} properties of the S-Matrix. However, a complete understanding of S-Matrix analyticity is still lacking. In this section, we will describe some of the analyticity results that have been proven over time. Before getting into that, we would like to delineate the consequence of causality in a simple signal model, which can be regarded as a $0+1$ dimensional S-Matrix theory. The analysis of this connection between causality and analyticity in such simple system has been  is quite well known. See \cite{TollCaus} for an older discussion and  \cite{CEMZ} for a recent exposition. 
\subsection{Causality in signal model}\label{signal} We consider an initial signal which is a function of of time $F_{i}(t)$ and an outgoing signal $F_{o}(t)$. We assume that these two are related via a  linear response relation of the form
\be \label{response}
F_o(t)=\int_{-\infty}^{\infty}dt'\,S(t-t')F_i(t'),
\ee where $S(t-t')$ is the response function. Using Fourier integral representations 
\be
g(t)=\int_{-\infty}^{\infty}\frac{d\w}{2\p}\,e^{-i\w t} \,\widetilde{g}(\w),
\ee the response equation \eqref{response}  becomes in Fourier space 
\begin{align}
	\int_{-\infty}^{\infty}\frac{d\w}{2\p}\,e^{-i\w t} \widetilde{F}_o(\w)&=\int_{-\infty}^{\infty}dt'\int_{-\infty}^{\infty}\frac{d\w}{2\p}\,e^{-i\w t}\int_{-\infty}^{\infty}\frac{d\w'}{2\p}\,e^{-i\w' (t-t')}\widetilde{S}(\w')\widetilde{F}_i(\w)\nn\\
	&=\int_{-\infty}^{\infty}\frac{d\w}{2\p}\int_{-\infty}^{\infty}\frac{d\w'}{2\p}\int_{-\infty}^{\infty}dt' e^{-i (\w-\w')t'}\,e^{-i\w't'}\,\widetilde{S}(\w')\widetilde{F}_i(\w)\nn\\
	&=\int_{-\infty}^{\infty}\frac{d\w}{2\p}\int_{-\infty}^{\infty}d\w'\, \d(\w-\w') \,e^{-i\w't'}\,\widetilde{S}(\w')\widetilde{F}_i(\w)\nn\\
	&=\int_{-\infty}^{\infty}\frac{d\w}{2\p} \,\widetilde{S}(\w)\widetilde{F}_i(\w),
\end{align} thus yielding the scattering equation
\be 
\widetilde{F}_o(\w)=\widetilde{S}(\w)\widetilde{F}_i(\w).
\ee $\widetilde{S}(\w)$ can be considered the associated S-Matrix for the system.

If $F_o(t)$ is \emph{causally related} to $F_o(t')$ then $F_o(t)$ depends upon $F_i(t')$ only for $t'\le t$. Quantitatively, this can be encoded by: \be F_i(t')=0,\,t'<0\implies F_o(t)=0,\,t<0. \ee The response equation \eqref{response} then implies that $S(t)=0$ for $t<0$. Let's now analyze what this means for $\widetilde{S}(\w)$. Start with the inverse Fouirer representation
\begin{align} 
	\widetilde{S}(\w)&=\int_{-\infty}^{\infty} dt\,e^{i\w t}\, S(t)\nn\\
	&=\int_{0}^{\infty} dt\,e^{i\w t}\, S(t)
\end{align} 
We see that the integral converges for $\text{Im.}[\w]>0$. For $\text{Im.}[\w]>0$, we have an exponentially damping factor which improves the convergence. Thus $\widetilde{S}(\w)$ is \emph{analytic in the upper half plane}. The physical S-Matrix can then be \emph{defined} by considering it to be \emph{boundary value} of the analytic function $\widetilde{S}(\w)$:
\be 
\widetilde{S}_{\text{phys}}(\w)=\lim_{\e\to o}\widetilde{S}(\w+i\e).
\ee 

This shows that causality manifests itself in the analyticity properties of the S-Matrix. However, we would like to emphasize that analyticity in the upper half-plane is \emph{only necessary} for causality, \emph{not suffcient}. To understand the sufficiency condition, let us start with  
\be 
S(t)=\int_{-\infty}^{\infty}\frac{d\w}{2\p}\,e^{-i\w t} \,\widetilde{S}(\w).
\ee For causality, one needs to show that $S(t)=0$ for $t<0$. We can do the integral by closing the contour in the upper half-plane.  Then
\be 
S(t)=\oint_\mc \frac{d\w}{2\p}\,e^{-i\w t} \,\widetilde{S}(\w)-\int_{\mc_\infty^+} \frac{d\w}{2\p}\,e^{-i\w t} \,\widetilde{S}(\w),
\ee where $\mc_\infty^+$ is a semicircular arc in the at infinity centered at origin in the upper half-plane, and $\mc$ is the closed contour formed by union of this arc and the real line. The integral is analytic in the first integral because of analyticity of $\widetilde{S}(\w)$ in the upper-half plane. Therefore by Cauchy's theorem, the integral vanishes. Let us now look at the second integral.
\begin{align} 
	\mi_\infty:=&\int_{\mc_\infty^+} \frac{d\w}{2\p}\,e^{-i\w t} \,\widetilde{S}(\w)\nn\\
	=&\lim_{R\to\infty}\int_{\mc_R^+} \frac{d\w}{2\p}\,e^{-i\w t} \,\widetilde{S}(\w)
\end{align} 
Let us now look at the modulus of the integrand
\be 
\left|e^{-i\w t} \,\widetilde{S}(\w)\right|_{\mc_R^+}=\left[e^{\text{Im.}[\w]t}\left|\widetilde{S}(\w)\right|\right]_{\mc_R^+}\le e^{Rt}\,\left|\widetilde{S}(\w)\right|_{\mc_R^+}
\ee Now if we assume that $\left|\widetilde{S}(\w)\right|$ grows slower than exponential as $\text{Im}.[\w]\to\infty$, then $\mi_\infty$ vanishes for $t<0$. Thus we reach the conclusion that $\widetilde{S}(\w)$ \emph{describes a causal evolution of signals if and only if it is analytic in upper half-plane and has at most subexponential growth for in the upper half plane.} 

Thus we see that causality requirements put strong constraints on the analytic structure of the S-Matrix. One can reverse engineer this to diagnose possible non-causal behaviour in a particular model by examining the analyticity properties of the corresponding S-Matrix. We now proceed to a similar analysis for the relativistic S-Matrix. 
\subsection{Lehmann-Symanzik-Zimmermann (LSZ) axiomatics }
In order to explore the possible connection between causality and analytic structure of the scattering amplitude, we will follow the axiomatic framework for S-Matrix due to Lehmann-Symanzik-Zimmermann (LSZ) \cite{LSZ1, LSZ2}. The LSZ axiomatic framework associates a quantum field operator (below we will use the terms field to mean a quantum field operator )  $\vf(x)$ corresponding to each \emph{stable} particles in the theory. These quantum field operators transform in some irreducible representation of the Poinc{a}r\'e group with the transformation law 
\be 
\mathcal{U}(\Lambda,a)\vf(x)\mathcal{U}(\Lambda,a)^{-1}=\sum (\Lambda^{-1})\, \vf(\Lambda x+a),
\ee  where $\mathcal{U}(\Lambda,a)$ is continuous unitary representation of the Poincar\'{e} group and $\sum (\Lambda^{-1})$ is some linear operator acting only on the internal Poincar\'{e} indices which label the representation $\vf(x)$. For simplicity of description, we will restrict the present discussion to a single type of massive scalar particle with mass $m>0$. Thus the corresponding fields Poincar\'{e} scalars. Now we list out the axioms of the LSZ framework.
\begin{itemize} 
\item  {\bf Asymptotic and interpolating fields:} In order to accommodate the particle picture, LSZ framework assumes existence of asymptotic free fields:  $\vf^{+}(x)$, the \emph{in} field, and $\vf^-(x)$, the out field. These fields satisfy Klein-Gordon equations 
\be 
(\Box_x-m^2)\vf^{\pm}(x)=0.
\ee

An intuitive expectation about the relation between the general field $\vf(x)$ and the asymptotic fields is that $\vf(x)$ tends to $\vf^+(x)$ as $x^0\to-\infty$ and to $\vf^-(x)$ as $x^0\to+\infty$. LSZ axiomatics gives a precise meaning of this intuition as follows. Since the asymptotic fields satisfy the free wave equation, we can write them in terms of Fourier modes:
\be
\vf^{{\pm}}(x)=\frac{1}{(2\p)^{3/2}}\, \int\frac{d^3\mathbf{k}}{2|k^0|}\, \left[e^{ik\cdot x}\, a_{{\pm}}(\mathbf{k})+e^{-ik\cdot x}\, a_{{\pm}}^\dagger(\mathbf{k})\right].
\ee The \emph{interpolating field} $\vf(x)$ (which interpolates between $\vf^+$ and $\vf^-$) is an interacting field and does not satisfy the Klein-Gordon equations, rather it satisfy an equation of generic form 
\be \label{currentdef}
(\Box_x-m^2)\vf(x)=j(x).
\ee This can be taken as a formal definition of the current operator  $j(x)$. A  Fourier decomposition for the interpolating field $\vf(x)$ can be written as  
\be 
\vf(x)=\frac{1}{(2\p)^{3/2}}\, \int\frac{d^3\mathbf{k}}{2|k^0|}\, \left[e^{ik\cdot x}\, a(\mathbf{k}, x^0)+e^{-ik\cdot x}\, a^\dagger(\mathbf{k}, x^0)\right].
\ee The \emph{asymptotic condition} is then assumed in the form 
\be 
\lim_{x^0\to\mp\infty}\Braket{\psi_1|a(\mathbf{k}, x^0)|\psi_2}=\Braket{\psi_1|a_{\pm}(\mathbf{k})|\psi_2}
\ee 

for all states $\ket{\psi_1}, \, \ket{\psi_2}$ belonging to some \emph{dense} domain of the underlying Hilbert space. This is often paraphrased by the statement that $a(\mathbf{k}, x^0)$ \emph{converges weakly} to $a_{\pm}(\mathbf{k})$ as $x^0\to\mp\infty$. Similar statements apply for $a^\dagger$. It is to be emphasized that within LSZ axiomatics this is an \emph{assumption}.

\item {\bf State space: }Next there are assumptions regarding the state space. 
Let us recall that according to the general discussion of quantum scattering as presented in section \ref{Smatbasic} the asymptotic state spaces are $\mh_m^+,\,\mh_m^-$, and the associated Fock state is $\mf_m$. The subscript $m$ expresses the fact that we are considering state spaces corresponding to a single kind of relativistic particle with mass $m>0$. The corresponding Fock vacuum $\ket{0}$ is the unique Poincar\'{e} invariant state $\ket{0}$ defines the in and out vacuua $\ket{0^+}$ and $\ket{0^-}$ respectively. Asymptotic completeness, \eqref{asympcompl}, implies that 
\be 
\ket{0}=\ket{0^+}=\ket{0^-}.
\ee 
The vacuum satisfies 
\begin{align}
	\begin{split}
		\mathcal{U}(\Lambda, a)\ket{0}&=0,\\
		a_{\pm}\ket{0}&=0. 	\end{split}
\end{align}The incoming and the outgoing states are constructed by 
\begin{align} 
	\ket{k_1,\dots, k_n\,;\,+}&=a_+^\dagger(\mathbf{k}_1)\dots a_+^\dagger(\mathbf{k}_n)\,\ket{0},\\
	\ket{k_1,\dots, k_n\,;\,-}&=a_-^\dagger(\mathbf{k}_1)\dots a_-^\dagger(\mathbf{k}_n)\,\ket{0},\qquad -(k_i^0)^2+\mathbf{k}_i^2=-m^2
\end{align} 

\item {\bf Spectral condition:} Finally we need the  \emph{spectral condition} 
\be 
\text{Spec.}[\mmp^\m]\subset \overline{V}^+:=\{p\in\mathbb{R}^{3,1}\,|\,p^2\ge 0, p^0\ge 0\},
\ee where $\mmp^\m$ is the \emph{total} $4-$momentum operator. $\overline{V}^+$ is known as \emph{forward causal cone}. The eigenvalue $p_0^2$ of $\mmp^\m$ satisfying $p^2=0$ should only correspond to the vacuum state, i.e.
\be 
\mmp^\m\ket{0}=0.
\ee The eigenvalues $p=k$ belonging to the $1-$particle states $\ket{k}$ must satisfy $k^2=m^2>0$. 
\end{itemize}
\subsubsection*{Some mathematical pointers}
We would like provide some mathematical pointers regarding LSZ axiomatics. LSZ axiomatics can be accommodated in more general axiomatic treatment of quantum field theory. In particularly, the asymptotic structures that have been taken as assumptions in LSZ framework can be proved rigorously in Wightman's axiomatic treatment of quantum field theory \cite{LSZrig1, LSZrig2}. 

While we will not go into details, we would like to state one particular aspect of this connection with Wightman field theory. In Wightman's framework, quantum fields are actually operator valued distributions on the Minkowski space $\mathbb{M}^{3,1}$. Thus if $\mo_W(x)$ is a Wightman field then 
\be 
\mo(f):=\int d^4 x\, f(x)\, \mo(x),
\ee  $f(x)$ being a \emph{Schwartz function},   is an operator on the Hilbert space, $\mh_W$ of the theory. Let us describe briefly what a Schwartz function is. A Schwartz function is an infinitely differentiable complex-valued function such that it along with all of its derivatives falls off to zero at infinity sufficiently fast. The rate of fall can be encoded as follows.A function $f:\mathbb{M}^{3,1}\to\mathbb{C}$ is a Schwartz function if the derivative of the function $D_x^{|\a|}f(x):=\pd_{x^0}^{\a_0}\dots \pd_{x^3}^{\a_3}\, f(x),\, |\a|=\sum_{i=0}^3 \a_i$ satisfies ( This requirement can be made more mathematically precise in terms of semi-norm. However we will not need them for demonstrating the basic idea of Schwartz function.) 
\be 
\sup_{x\in \mathbb{M}^{3,1}}\, \left|D_x^{|\a|}f(x)\, x^{|\b|} \right|<\infty, \quad x^{|\b|}:=\prod_{i=0}^3\, (x^i)^{\b_i} ,\,|\b|:=\sum_{i=0}^3\b_i,\quad \forall\, |\a|,\,|\b|\in \mathbb{Z}^{\ge}. 
\ee 
Schwartz functions form a normed space, often denoted by $\ms$. The elements of the dual space, denoted by $\ms^{\vee}$, are  called \emph{tempered distributions}. In Wightman formalism, the expectation value $\braket{\psi|\mo_W(f)|\phi},\, \ket{\psi},\,\ket{\phi}\in \mh_W$, is a tempered distributions regarded as \emph{linear functional} of $f\in \ms$.  An important property of tempered distributions is that they are \emph{polynomially bounded}. A tempered distribution $\tilde{f}$ considered as a function on $\mathbb{M}^{3,1}$ can grow at most as fast as some polynomial. We can write this requirement heuristically as 
\be \label{polybound}
|\tilde{f}(x)|<||x||^n,
\ee  where $||x||:=+\sqrt{|x^2|}$ for $x$ spacelike and timelike, $||x||=|x^0|$ for $x$ lightlike, and  $n$ is some positive integer.

\subsection{Microcausality} 
The setup of quantum field theory  we can now introduce the notion of causality. Relativity requires that no signal can propagate at a speed more than that of light, which translates into the fact that space-like separated events cannot be causally connected. In quantum theory, this means that space-like separated microscopic events like field measurements should not influence each other. This can be expressed by commutativity (we are considering Bosonic fields) of space-like separated field operators:
\be \label{microcaus}
\left[\vf(x),\vf(y)\right]=0\quad \text{for}\, (x-y)^2>0.
\ee This condition is often dubbed in literature as \emph{microcausality}.
\subsection{Reduction formulae} 
Now we come to the principle result of LSZ formalism, the reduction formulae. These formulae express scattering amplitudes in terms as Fourier transforms of vacuum expectation values of field operators, and, in general, other local operators constructed of the field operators. The central ingredient of these reduction formula is the \emph{retarded product} (or $R-$product) of field operators as defined by 
\begin{align}\label{Rdef}
	\begin{split}
		&R\vf_0(x_0)\vf_1(x_1)\dots\vf_n(x_n)=(-1)^n\sum_{\p\in S_n}\th(x_0^0-x_{i_{\p(1)}}^0)\th(x_{i_{\p(1)}}^0-x_{i_{\p(2)}}^0)\dots \th(x_{i_{\p(n-1)}}^0-x_{i_{\p(n)}}^0)\\
		&\hspace{5 cm}\times \left[\left[\dots\left[\left[\vf_0(x_0),\vf_{i_{\p(1)}}(x_{i_{\p(1)}})\right],\vf_{i_{\p(2)}}(x_{i_{\p(2)}})\right],\dots\right],\vf_{i_{\p(n)}}(x_{i_{\p(n)}})\right],\\
		&R\vf_0(x_0)=\vf_0(x_0).
	\end{split}
\end{align} Here $S_n$ is the group of permutations of the set $(1,2,\dots,n)$. The $R-$product is hermitian if the fields are hermitian, and is symmetric under exchange of the fields $\vf_1(x_1),\dots,\vf_n(x_n)$ ( the field $\vf_0(x_0)$ kept fixed in its position.)  The Poincar\'{e} transformation properties of the $R-$product dictates that the vacuum expectation value $\Braket{0|R\vf_0(x_0)\vf_1(x_1)\dots\vf_n(x_n)|0}$ depends only on $\{\x_i:=x_i-x_{i-1},\,1\le i\le n\}$. The $R-$product has an extremely important property. From the definition \eqref{Rdef}, it is clear that
\be 
R\vf_0(x_0)\vf_1(x_1)\dots\vf_n(x_n)\ne 0\qquad \text{only if}\quad x_0^0>\text{max.}\{x_{1}^0,\dots, x_n^0\}.
\ee   

Now we will quote the two reduction formulae for $2\to2$ scattering, with the kinematical configuration as described in section \ref{kinem}, which are used for investigating analyticity properties of the scattering amplitude:
\begin{align}\label{reduction1}
	\mt(\{p_i\})&=-\frac{1}{2\p}\int d^4z\, e^{-i\left(\frac{p_2+p_3}{2}\right)\cdot z_\a}\,\Braket{p_4|R j_3\left(\frac{z}{2}\right)j_2^\dagger\left(\frac{z}{2}\right)|p_1},\\
	\label{reduction2}\mt(\{p_i\})&=-\frac{1}{2\p}\int d^4z\, e^{-i\left(\frac{p_3-p_4}{2}\right)\cdot z}\,\Braket{0|R j_3\left(\frac{z}{2}\right)j_4\left(-\frac{z}{2}\right)|p_1,p_2\,;\,+},
\end{align} where $j_\a(x)$ is the \emph{source current} operator, introduced in \eqref{currentdef}, corresponding to $\a^{\rm th}$ particle.  Both expressions are special cases of what are called \emph{retarded function} (upto some trivial numerical factors)
\be 
F_R(q):=\int d^Dz \, e^{-iq\cdot z}\,\th(z_0)\,\Braket{Q_f|\left[ j_\a\left(\frac{z}{2}\right),j_\b\left(-\frac{z}{2}\right)\right]|Q_i},
\ee with $\Ket{Q_i},\,\ket{Q_f}$ being any states of total $4-$momentum $Q_i, \,Q_f$, which are considered as fixed parameters and are not written explicitly as parameters of $F_R$. Similarly, one can define an \emph{advanced function} $F_A(q)$ and \emph{causal function} $F_C(q)$ 
\begin{align}
	\label{advdef} F_A(q)&:=\int d^4z \, e^{-iq\cdot z}\,\th(-z_0)\,\Braket{Q_f|\left[ j_\a\left(\frac{z}{2}\right),j_\b\left(-\frac{z}{2}\right)\right]|Q_i},\\
	\label{causdef} F_C(q)&:=\int d^4z \, e^{-iq\cdot z}\,\Braket{Q_f|\left[ j_\a\left(\frac{z}{2}\right),j_\b\left(-\frac{z}{2}\right)\right]|Q_i}.
\end{align} These functions are related by 
\be \label{causdef2}
F_C(q)=F_R(q)-F_A(q).
\ee 
Let us discuss some analytic properties of the causal function $F_c(q)$ which will be crucial subsequent discussions. Inserting two complete system of \emph{physical} states $\ket{Q_n,\,\a_n},\, \ket{Q_n',\,\a_n'},\,\sum _n \ket{Q_n,\,\a_n}\bra{Q_n,\,\a_n}$ $=\mathbb{1_2},\,\sum _n \ket{Q_n',\,\a_n^{'}}\bra{Q_n',\,\a_n'}=\mathbb{1_2} $ ($\a_n,\,\a_n'$ collectively represent any other quantum numbers.), $F_C(q)$ can be written as 
\begin{align}
	\begin{split} 
		F_C(q)=\sum_n \int d^4\, Q_n\int d^4\xi e^{-iq\cdot\xi}\,&\left[\Braket{Q_f|j_\a\left(\frac{\x}{2}\right)|Q_n,\,\a_n}\Braket{\a_n,\, Q_n|j_\b\left(-\frac{\x}{2}\right)|Q_i}\right.\\
		&\hspace{1 cm}- \left.\Braket{Q_f|j_\a\left(\frac{\x}{2}\right)|Q_n',\,\a_n'}\Braket{\a_n',\, Q_n'|j_\b\left(-\frac{\x}{2}\right)|Q_i}	\right].
	\end{split}
\end{align}
After space-time translation, this becomes 
\begin{align}
	\begin{split} 
		F_C(q)=\sum_n \,&\left[\Braket{Q_f|j_\a\left(0\right)|Q_n=\frac{Q_i+Q_f}{2}-q,\,\a_n}\Braket{\a_n,\, Q_n=\frac{Q_i+Q_f}{2}-q|j_\b\left(0\right)|Q_i}\right.\\
		&\hspace{0.5 cm}- \left.\Braket{Q_f|j_\a\left(0\right)|Q_n'=\frac{Q_i+Q_f}{2}+q,\,\a_n'}\Braket{\a_n',\, Q_n'=\frac{Q_i+Q_f}{2}0q|j_\b\left(0\right)|Q_i}	\right].
	\end{split}
\end{align}
This implies that $F_C(q)$ vanishes whenever simultaneously 
\begin{align}
	\label{causcond1}	&\sum_n \Braket{Q_f|j_\a\left(0\right)|Q_n=\frac{Q_i+Q_f}{2}-q,\,\a_n}\Braket{\a_n,\, Q_n=\frac{Q_i+Q_f}{2}-q|j_\b\left(0\right)|Q_i}=0,\\
	\label{causcond2}	&\sum_n \Braket{Q_f|j_\a\left(0\right)|Q_n'=\frac{Q_i+Q_f}{2}+q,\,\a_n'}\Braket{\a_n',\, Q_n'=\frac{Q_i+Q_f}{2}+q|j_\b\left(0\right)|Q_i}=0.
\end{align} The states $\ket{(Q_i+Q_f)/2\pm q,\a_n}$ are supposed to be physical states, i.e. $[(Q_i+Q_f)/2\pm q]^2\ge 0,\,(Q_i^0+Q_f^0)/2\pm q^0\ge 0$. Then, from the spectral condition , it follows that there will be minimal masses $\mm_-,\,\mm_+\,>0$ such that \eqref{causcond1} is \emph{not} satisfied for 
\be \label{Mplus}
\left(\frac{Q_i+Q_f}{2}- q\right)^2\le -\mm_{-}^2,
\ee and \eqref{causcond2} is \emph{not} satisfied for 
\be \label{Mminus}
-\left(\frac{Q_i+Q_f}{2}+ q\right)^2\le -\mm_{+}^2.
\ee Thus we have that 
\begin{align}\label{suppFc}
	\text{Supp.}[F_C(q)]= &\left\{q\in\mathbb{R}^{1, \,D-1}\,\Bigg|\,\left(\frac{Q_i+Q_f}{2}+ q\right)^2\le -\mm_{+}^2,\quad\left(\frac{Q_i^0+Q_f^0}{2}+ q^0\right)\ge 0\right.\nn\\
	&\left.\hspace{3 cm}\bigcup\,\quad\left(\frac{Q_i+Q_f}{2}- q\right)^2\le -\mm_{-}^2,\,\,\left(\frac{Q_i^0+Q_f^0}{2}- q^0\right)\ge 0\right\}.
\end{align} 	

We have now all the ingredients needed to deduce the analyticity properties of the scattering amplitudes, namely the retardedness property of $F_R(q)$, the advanced character of $F_A(q)$ and the support of $F_C(q)$. This last piece of information is crucial for a complete understanding of the analyticity properties of the scattering amplitude. 
\subsection{Primitive Analyticity Domain}\label{prim}
We will now work on obtaining the analyticity domains of the scattering amplitude. We will see that micro-causality, \eqref{microcaus}, controls the analytic structure of $\mt(p_1,\,p_2,\,p_3,\,p_4)$ following arguments similar to those used for the simple signal model considered in section \ref{signal}. The discussion in this section largely follows \cite{Sommer1}. 

Let us start with the retarded function 
\be 
F_R(q)=\int d^D\xi\, e^{-i q\cdot\x}\,\th(\x^0)\, \Braket{Q_f|\left[j_\a\left(\frac{\x}{2}\right),j_\b\left(-\frac{\x}{2}\right)\right]|Q_i}.
\ee 
We will find the domain of analyticity of this function in terms of complex $D-$momentum $q$. First, we observe that due to micro-causality, the matrix element of the commutator vanishes for space-like $\x$, i.e. $\x^2>0$. Because of the $\th(\x^0)$ factor, the integrand is supported on the open forward causal cone $V^+$ in $\x$ space 
\be 
V^+=\{\x\in\mathbb{R}^{1,3}\,|\x^2<0,\,\x^0>0\}.
\ee Now, let us consider complex $4-$momentum $q\in \mathbb{R}^{1,3}+i \mathbb{R}^{1,3}$. With this, the exponential factor in the integral becomes 
\be 
F_R(q)=\int d^D\xi\, e^{-i \text{Re.}[q]\cdot\x+\text{Im.}[q]\cdot\x}\,\th(\x_0)\, \Braket{Q_f|\left[j_\a\left(\frac{\x}{2}\right),j_\b\left(-\frac{\x}{2}\right)\right]|Q_i}.
\ee Due to temperedness of the field operators, the matrix element of the commutator is polynomially bounded at $|\x|\to\infty$. Then the integral will converge, and therefore $F_R(q)$ will be an analytic function if $\x\cdot\text{Im.}[q]<0$ on the support of the integrand. The integrand is supported on $\x\in V^+$. Now using the result that $\a\cdot\b<0$ for \emph{arbitrary} $\a\in V^+$  only if $\b\in V^+$, one finds that $F_R(q)$ is analytic for all $q\in \mathbb{R}^{1,3}+i\mathbb{R}^{1,3}$ with $\text{Im.}[q]\in V^+$. We define 
\be\label{ftube} 
T^+:=\{ q\in \mathbb{R}^{1,3}+i\mathbb{R}^{1,3}\,:\, (\text{Im.}[q])^2<0,\, \text{Im.}[q^0]>0,\, \text{Re.}[q]\in\md^+\subset\mathbb{R}^{1,3}\},
\ee $\md^+$ being some domain in the Minkowski spacetime $\mathbb{R}^{1,3}$ Following the terminology of many-variable complex analysis, this domain will be called \emph{forward tube}.  

Next, we investigate the advanced function $F_A(q)$ by its integral representation
\be 
F_A(q)=\int d^D\xi\, e^{-i q\cdot\x}\,\th(-\x^0)\, \Braket{Q_f|\left[j_\a\left(\frac{\x}{2}\right),j_\b\left(-\frac{\x}{2}\right)\right]|Q_i}.
\ee We observe that, because of the $\th(-\x^0)$ factor and micro-causality, the inetgrand is supported on the \emph{open backward lightcone}
\be 
V^-:=-V^+=\{\x\in\mathbb{R}^{1,3}\,|\x^2<0,\,\x^0<0\}.
\ee Following the same logic as for $F_R(q)$, we infer that 
$F_A(q)$ is analytic in the \emph{backward tube}, \be\label{btube}  T^-:=\{q\in \mathbb{R}^{1,3}+i\mathbb{R}^{1,3}\,:\,(\text{Im.}[q])^2<0,\, \text{Im.}[q^0]<0,\, \text{Re.}[q]\in\md^-\subset\mathbb{R}^{1,3}\}.\ee
Further, according to \eqref{causdef2} and \eqref{suppFc} , both the functions coincide  for \emph{certain} real vectors 
\be\label{edge} 
q\in \me\subset (\text{Supp.}[F_C(q)])^c\subset\mathbb{R}^{1,3}.
\ee 

We can use the above information to extract the analyticity domain of the scattering amplitude. To illustrate the basic logic, let us analyze a simple case. Suppose we replace $\mathbb{R}^{1,3}\to \mathbb{R}$ in above, i.e. $q\in\mathbb{C}\equiv \mathbb{R}+i\mathbb{R}$. Then we have that $F_R(q)$ is analytic in the upper half-plane, $F_A(q)$ is analytic in the lower half-plane, and they coincide on some domain $ \mg$ over the real axis. Then according to Schwarz' reflection theorem $F_R(q)$ and $F_A(q)$ can be analytically continued to a \emph{single} function $F(q)$ holomorphic on the entire complex plane. There is a generalization of this result that we can apply to our case. The corresponding theorem is known as \emph{edge of the wedge theorem} and was first discovered by Bogoliubov \cite{BogoQFT, BoGoShir}. We will state the theorem in its simplest form.  
\begin{theorem}[{\bf Bogoliubov's Edge of the wedge theorem}]
	Let $C$ be an open subset of $\mathbb{R}^{n}$ satisfying the property that $y\in C\,\implies\, \l y\in C$ for all positive $\l$. Such a set is called an open cone with the vertex at the origin. Let $E$ be an open subset of $\mathbb{R}^{n}$ called the edge. Define an wedge $W^+:=E\times iC$, and the opposite wedge $W^-:=E\times -iC$ in the complex vector space $\mathbb{C}^n\equiv \mathbb{R}^n+i\mathbb{R}^n$. The product domains are topological products. Then the two wedges $W^+$ and $W^-$ meet at the edge $E$, where we identify $E$ with the product of $E$ with the tip of the cone. 
	
	Now, suppose $f_+$ and $f_-$ are holomorphic respectively in the wedges $W^+$ and $W^-$, and they coincide on $E$, i.e. have the same boundary values on $E$
	\be 
	\lim_{\substack{y\to 0\\ y\in C }}f_+(x+iy)=\lim_{\substack{y\to 0\\ y\in -C }}f_-(x+iy).
	\ee Then $f_+$ and $f_-$ can be analytically continued to a function $f$ holomorphic in an open neighborhood of $W^+\cup E\cup W^-$. 
\end{theorem}  
Now we can directly apply this theorem to our case at hand. Although the theorem is formulated for Euclidean space $\mathbb{R}^n$, the result also holds in the present case  where we have Minkwoski space. The respective wedges are the forward tube $T^+$, \eqref{ftube} and the backward tube $T^-$, \eqref{btube} and the edge is the real domain $\me$, \eqref{edge}.  Applying edge of the wedge theorem, we can infer that the retarded function $F_R(q)$ and the advanced function $F_A(q)$ can be analytically continued to a single function $F(q)$ analytic in a domain $\mo\subset \mathbb{R}^{1,3}+i\mathbb{R}^{1,3}$ such that $T^+\cup\me\cup T^-\subset \mo$. 

\paragraph{} We are not done yet! In order to determine the analyticity domain of the scattering amplitude in the way delineated above, one needs to find a complete set of retarded  ( and corresponding advanced) functions, and one needs to find the set of corresponding tubes of holomorphy.    This analysis gives rise to the following $32$ tubes of holomorphy.

\begin{eqnarray}\label{primitive}
	T_j^+\,=\,&\{p\,|\,\text{Im.}[p_k]\in V^+,\qquad& \text{Im.}[p_m]\in V^+,\,\,\text{Im.}[p_n]\in V^+\}\nn\\
	T_j^-\,=\,&\{p\,|\,\text{Im.}[p_k]\in V^-,\qquad& \text{Im.}[p_m]\in V^-,\,\,\text{Im.}[p_n]\in V^-\}\nn\\
	T_{jk}^+\,=\,&\{p\,|\,\text{Im.}[p_k]\in V^-,\qquad& \text{Im.}[p_k+p_m]\in V^+,\,\,\text{Im.}[p_k+p_n]\in V^+\}\nn\\
	T_{jk}^-\,=\,&\{p\,|\,\text{Im.}[p_k]\in V^+,\qquad& \text{Im.}[p_k+p_m]\in V^-,\,\,\text{Im.}[p_k+p_n]\in V^-\}
\end{eqnarray}  where $(jkmn)$ is any permutation of $(1234)$. Observe that we have not specified the restriction on $\text{Im.}[p_j]$ in the definition of $T_j^{\pm},\,T_{jk}^{\pm}$ because that is automatically obtained by the others due to the momentum conservation condition $\sum_i p_i =0$, $\text{Im.}[p_j]\in V_j^-$ in $T_j^{+},\,T_{jk}^+$ and $\text{Im.}[p_j]\in V_j^+$ in $T_j^{-},\,T_{jk}^-$. The union of the $32$ tubes is called \enquote{primitive analyticity domain}. 

The boundary values of $\mt(p)$ for $p$ in the tubes $T_j^+, \, T_j^-,\, T_{jk}^+,\, T_{jk}^-$ and $\text{Im.}[p]\to 0$ are usually called $r_j(p),\, a_j(p),\, r_{jk}(p),\, a_{jk}(p),$ respectively. These boundary values coincide in pairs on certain real domains of the faces of the  tubes. The coincidence relations have the form 
\begin{eqnarray}\label{coincidence}
	&r_j(p)=r_{jk}(p)\qquad\qquad &\text{for}\qquad p\in \{\text{Im.}[p]=0,\,p_k^2<\mathcal{M}_k^2\}\nn\\
	&a_j(p)=a_{jk}(p)\qquad\qquad &\text{for}\qquad p\in \{\text{Im.}[p]=0,\,p_k^2<\mathcal{M}_k^2\}\nn\\
	&r_{jk}(p)=a_{mn}(p)\qquad\qquad &\text{for}\qquad p\in \{\text{Im.}[p]=0,\,(p_j+p_m)^2<\mathcal{M}_{jm}^2\},\nn
\end{eqnarray}
where $\mathcal{M}_k$ and $\mathcal{M}_{jm}$ are certain threshold masses analogous to the thresholds $\mathcal{M}_+,\, \mathcal{M}_-$ of the function $F_C(q)$ defined in \eqref{Mplus} and \eqref{Mminus} respectively.

Finally,  by the edge-of-the-wedge theorem, the analytic functions in any pair of adjacent tubes and thus in all the $32$ tubes can be analytically continued to a single analytic function $\mt(p)$. Equivalently, we have defined a global analytic function $\mt(p_1,\,p_2,\,p_3,\,p_4)$ on the primitive domain of analyticity. 
\subsection{Polynomial boundedness} 
The scattering amplitude is assumed to be polynomially bounded in all the channels. In $s-$channel, this implies that for \emph{fixed physical} $t$,
\begin{align}\label{polybnd}
	|\mt(s,t)|<|s|^n,\qquad n\in \mathbb{Z}^{+}.
\end{align} It was proved by Jin and Martin \cite{JinM} that $n=2$.  Similar bounds exist in other channels. Polynomial boundedness is taken as an \emph{assumption} in LSZ axiomatic framework. However, if one considers LSZ axiomatics within the framework of Wightman field theory, then polynomial boundedness follows from the fact that correlation functions are tempered distributions, which then implies that the  scattering amplitudes are tempered distributions because they are ultimately Fourier transform of correlation functions, and Fourier transform of a tempered distribution itself is a tempered distribution.   
\subsection{Fixed-transfer Dispersion Relation }
Perhaps the most important result that follows from the causality is dispersion relation. Dispersion relation in S-Matrix theory has a long history. \cite{GGT} obtained the first dispersion relation in energy for forward scattering amplitude in a perturbative setting. \cite{GoldB1} then provided a non-perturbative proof. However, the mathematically rigorous proof was first given by \cite{BMP}. Such a proof is obtained using the primitive domain of analyticity discussed above. Before discussing such a proof, let us give a short account of what underlies a dispersion relation.

  In a fixed transfer dispersion relation, one aims to write a dispersion relation in one of the two \emph{independent} Mandelstam variables, \eqref{Mandeldef}, \eqref{onshell}, while keeping the other fixed.  For example, the most common dispersion relation is the fixed $t$ dispersion relation in $s$. Proving such a dispersion relation is \emph{equivalent to analyzing the analyticity property of $\mt(s,t)$ in complex $s$ for fixed $t$}. Let us explain briefly how that is so. The first step towards writing the fixed $t$ dispersion relation is the Cauchy integral formula
\be 
\mt(s,t)=\frac{1}{2\p i}\oint_\mc \frac{\mt(s',t)}{s-s'},
\ee where $\mc$ is some Jordan curve, oriented counter-clockwise, around $s'=s$ such that $\mt(s't)$ is analytic in the interior of $\mc$. Now, we can deform the the contour $\mc$ in accordance with the analyticity properties of $\mt(s',t)$ in complex $s'$. Now, let's assume that $\mt(s,t)$ has a pole at $s=s_0$ and has a branch point at $s=\bar{s}\in \mathbb{R}^+$ with a branch cut along $[\bar{s},\infty)$. Also assume that $\mt(s,t)$ falls to zero as $|s|\to \infty$. Then with these analyticity properties, we can deform the contour $\mc$ suitably to reach 
\be 
\mt(s,t)=-\frac{\l}{s-s_0}+\frac{1}{\p}\int_{\bar{s}}^\infty \frac{\ma_s(s',t)}{s-s'},
\ee where $ \l\equiv \text{Res.}[\mt(s,t)]_{s=s_0}$ and $\ma_s(s,t)$ is the $s-$channel absorptive part as defined in \eqref{absdef}. We have reached the simplest fixed-$t$ dispersion relation. Usually, such a dispersion relation will be valid for some \emph{physical range of} $t$, for some $t\in[-t_M,0],\, t_M>0$. 

We will give a brief account of the key rationale behind obtaining the  fixed transfer dispersion relation exploiting primitive domain of analyticity. See \cite{Sommer1} for a detailed exposition. The primitive analyticity domain is a domain in the space of four complex $4-$momenta, $\mathbb{C}^{12}$. We introduce the mass variables $\zeta_i$ defined as $k_i^2:=-\zeta_i$. $\zeta_i$ are complex in general. Only when $k_i$ are on-shell, $\zeta_i=m^2$ with $m>0$. The main argument is to first derive a dispersion relation for fixed real negative value of $t$ (the momentum transfer variable) and sufficiently large negative values of some of the mass variables $\zeta_i$. Next one shows that the amplitude can be analytically continued to mass-shell $\zeta_i=m^2$ maintaining the analyticity properties in $s$. This works for a real range $t\in I:=[-t_M,0]$. 

As mentioned in the previous section, we have $|\mt(s,t)|<|s|^2,\, |s|\to\infty$.  This  implies that we can write a twice subtracted dispersion relation in order to drop the contribution from arc at infinity. Thus the final fixed-$t$ dispersion relation  is given by 
\begin{align} 
\begin{split}
	\mt(s,t)=\mt(s_0,t)+(s-s_0)\,\frac{\pd\mt}{\pd s}(s_0, t)&+(s-s_0)^2\, \left[ \frac{1}{\p}\,\int_{4m^2}^\infty \,\frac{\ma_s(s',t)}{(s'-s_0)^2\,(s'-s)}\right.\\
	&\hspace{2 cm}\left. +\frac{1}{\p}\,\int_{4m^2}^\infty\, \frac{\ma_u(u',t)}{(4m^2-t-s_0-u')^2\,(u'-u)}
	\right].
\end{split}
\end{align} 
 Here $\ma_u$ is the $u-$channel absorptive part across the $u-$channel cut, and $s_0$ is a point on the cut complex $s-$plane.

\subsection{Lehmann analyticity} 
The reduction formula \eqref{reduction2} allows determination of analyticity properties of the amplitude in complex $t$ plane for fixed $s$. Lehmann\cite{Lehmann1} showed that $\mt(s,t)$ is analytic inside an open elliptical disc in complex $t-$plane for \emph{physical} $s>4m^2$. 
Can one  anticipate that the domain of analyticity will be elliptical without going into details of Lehmann's proof? The answer is yes. A rather simple observation about partial wave expansion provides the answer.  Recall the partial wave expansion from \eqref{partial wave}
\begin{align}
	\mt(s,t)=16\p \sqrt{\frac{s}{s-4m^2}}\mathlarger{\mathlarger{\sum}}_{\substack{\ell=0\\ \ell\,\text{even}}}(2\ell+1)\,f_\ell(s)\,P_\ell(z),
\end{align}
with \be 
z=+1+\frac{2 t}{s-\m}.
\ee and $P_\ell$ being the usual Legendre polynomials. For fixed physical $s\ge \m$, the analyticity domain in $t$ can be obtained by obtaining the analyticity domain in the complex $z$ plane. Now an argument due to Neumann \cite{WhitaWat} shows that a Legendre series, i.e. a series of the form $\sum_n \a_n\, P_n(w)$, converges in some open elliptical disc  with the boundary ellipse having foci at $z=\pm 1$. The exact shape of the ellipse, i.e. the major and minor axes, depends upon the magnitude of coefficients $\{\a_n\}$. In fact, one has the following theorem due to  a generalization of Abel's theorem for power series to Legendre series.

\begin{theorem}\label{LegThm1}
	The series 
	$ 
	\sum_n \a_n\, P_n(w)
	$ is absolutely convergent in the interior of an ellipse with foci at $\pm 1$ and semimajor axis $w_0>1$ iff 
	\be 
	{\limsup_{n\to \infty}} \,|\a_n|^{\frac{1}{n}}= \frac{1}{w_0+\sqrt{w_0^2-1}}.
	\ee 
\end{theorem}

Lehmann, starting from the reduction formula \eqref{reduction2}, showed  that  for physical  $s\ge 4m^2$, the scattering amplitude $\mm(s,z)$ is analytic in complex $z-$plane inside  an ellipse with foci at $\pm 1$ and semi-major axis 
\be \label{LehmannS}
z_s=\left[1+\frac{36m^4}{s(s-4m^2)}\right]^{\frac{1}{2}}.
\ee 
This ellipse is called \textit{small Lehmann ellipse}, $\me_\ml^{(S)}(s)$.  It is called small ellipse because Lehmann also showed the existence of a larger confocal ellipse inside which the absorptive part $\ma_s$ is analytic. This latter ellipse is called \textit{large Lehmann ellipse}, $\me_\ml^{(L)}(s)$. The semi-major axis of this larger ellipse is related to that of the small Lehmann ellipse and is given by 
\be \label{LehmannL}
z_l=2z_s^2-1=1+\frac{72 m^4 }{s(s-4m^2)}.
\ee 


\section{Hermitian analyticity}
Another important property of the scattering amplitude $\mt(s,t)$ is the \emph{hermitian analyticity}  property, which reads
\be \label{hermitian}
\mt(s,t)^{*}=\mt(s^*,\, t^*).
\ee This property follows as a consequence of TCP theorem \cite{Olive}. The basic idea is that TCP theorem effectively converts advanced functions into retarded functions. \section{Crossing properties}\label{crossingdef}
An extremely important property of scattering amplitude is its crossing property. Crossing will be central to the analysis presented in the main body of this thesis. Crossing refers to the fact that if we analytically continue the scattering amplitude in one channel to the physical domain of another channel, we will get the scattering amplitude in that channel. Let us make this concrete. 

We recall from \eqref{physamp} that corresponding to three different domains of Mandelstam variables, there are three different \emph{physical amplitudes} with disjoints support in Mandelstam variables,  $\mt^s(s,t),\, \mt^t(s,t),\, \mt^u(s,t)$. Crossing then posits existence of a \emph{global} analytic function $\mt:\mathbb{C}^2\to\mathbb{C}^2$  with a domain of analyticity $\widetilde{\md}\subset \mathbb{C}^2$ satisfying 
\begin{align}
	\widetilde{\md}\cap \text{Supp.}\mt^i\neq \emptyset, \qquad i=s,\,t,\,u.
\end{align}

This is also paraphrased by the statement that the \emph{physical amplitudes} $\{\mt^i(s,t)\}$ for scattering in different channels are actaully different \emph{boundary values} of a global analytic function, the mathematical scattering amplitude $\mt(s,t)$. Crossing of $2\to 2$ scattering amplitudes was proved rigorously in \cite{BEGcrossing}. 

For elastic scattering of identical neutral scalar particles (the case we are considering), crossing translates to stronger constraint on  $\mt(s,t)$,
\be 
\mt(\p(s),\, \p(t),\, \p(u))=\mt(s,\,t,\,u),\qquad \forall\,\, \p\in S_3,\quad s+t+u=4m^2.
\ee Here $S_3$ is the group of permutations of three objects. This invariance under the permutation of the Mandelstam variables is dubbed as \emph{crossing symmetry}. We want to emphasize that this straightforward symmetry structure is only there for elastic scattering of identical neutral scalar particles. In general, crossing does not result in such strong constraints.   
\section{Unitarity}
The S-Matrix is unitary. To see it in details recall the definition of the S-Matrix
\be 
\ms_{\a\b}=\Braket{\Psi_\a^-|\Psi_\b^+}.
\ee Now applying completeness property of the in states \eqref{Asympcomp}
\be 
\int d\g\, \ms_{\a\g}\,\ms_{\b\g}^*=\int d\g \Braket{\Psi_\a^-|\Psi_\g^+}\Braket{\Psi_\g^+|\Psi_\b^-}=\Braket{\Psi_\a^-|\Psi_\b^-}=\d(\a-\b)
\ee In terms of the operator $\widehat{\ms}$, this translates to $\widehat{\ms}\widehat{\ms}^\dagger=\widehat{\mathbb{1}}$. Similarly, using the completeness property of in states we can prove 
\be 
\int d\g\, \ms_{\a\g}*\,\ms_{\b\g}=\d(\a-\b),
\ee equivalently $\widehat{S}^\dagger \widehat{\ms}=\widehat{\mathbb{1}}$.
\be 
\hat{\ms}\widehat{\ms}^\dagger=\widehat{\ms}^\dagger \widehat{\ms}=\widehat{\mathbb{1}}.
\ee 

Using the decomposition 
\be 
\widehat{\ms}=\widehat{\mathbb{1}}+i\widehat{\mt}, 
\ee the unitarity becomes
\be 
\frac{1}{i}\left[\widehat{\mt}-\widehat{\mt}^\dagger\right]=\widehat{\mt}\widehat{\mt}^\dagger.
\ee For $2\to 2$ scattering this translates to 
\be \label{unitexp1}
\frac{1}{i}\Braket{\Psi|\widehat{\mt}-\widehat{\mt}^\dagger|\Psi}=\Braket{\Psi|\widehat{\mt}\widehat{\mt}^\dagger|\Psi},
\ee where $\ket{\Psi}$ is an arbitrary two-particle state, \emph{not} necessarily a state of definite individual momenta.  Next, due to the fact that  $\widehat{\mt}\widehat{\mt}^\dagger$ is a positive semi-definite operator, we have for any physical states $\ket{\Psi},\, \ket{\Phi}$
\be \label{Unineq1}
\Braket{\Psi\Big|\,\widehat{\mt}\,\Big|\Phi}\Braket{\Phi|\widehat{\mt}^\dagger|\Psi}=\left|\Braket{\Psi\Big|\,\widehat{\mt}\,\Big|\Phi}\right|^2\ge 0.
\ee Now can write for the right hand side of  \eqref{unitexp1}
\begin{align}
	\Braket{\Psi|\widehat{\mt}\widehat{\mt}^\dagger|\Psi}=\sum_{n=2}^{\infty}\sum_{{\Phi_n}} \left|\Braket{\Psi|\widehat{\mt}|\Phi_n}\right|^2
	\ge  \sum_{{\Phi_2}} \left|\Braket{\Psi|\widehat{\mt}|\Phi_2}\right|^2,
\end{align}	where we have used a resolution of identity on the entire state space ($\ket{\Phi_n}$ is $n-$particle state) in writing the first equality, and the sum in second inequality is over any \emph{complete} set of $2-$particle states. We have finally then 
\be 
\frac{1}{i}\Braket{\Psi|\widehat{\mt}-\widehat{\mt}^\dagger|\Psi}=\sum_{n=2}^{\infty}\sum_{{\Phi_n}} \left|\Braket{\Psi|\widehat{\mt}|\Phi_n}\right|^2
\ge  \sum_{{\Phi_2}} \left|\Braket{\Psi|\widehat{\mt}|\Phi_2}\right|^2
\ee  

We will discuss two important consequence of unitarity that will fare prominently in our subsequent analysis.
\begin{itemize}
	\item {\bf  Optical theorem}\\
	Optical theorem is the relation between the total scattering cross-section $\s_{\text{tot}}$ and absorptive part of forward scattering amplitude, $\ma_s(s,t=0)$. To obtain this relation we put $\ket{\Psi}=\ket{p,q}$ into the unitarity equation \eqref{unitexp1}. The matrix element 
	\be 
	\braket{p,q|\widehat{\mt}|p,q}
	\ee  then describes forward scattering amplitude for $A+A\to A+A$ 
	\be 
	\braket{p,q|\widehat{\mt}|p,q}=(2\p)^4\d^{(4)}(0)\, \mt(s, t=0)
	\ee so that the left hand side of the unitarity equation becomes
	\be 
	\frac{1}{i}\Braket{p,q|\widehat{\mt}-\widehat{\mt}^\dagger|p,q}=(2\p)^4\, \d^{(4)}(0)\, \ma_s(s, t=0).
	\ee 
	The right hand side of the unitarity equation becomes
	\begin{align}
		\Braket{p, q|\widehat{\mt}\widehat{\mt}^\dagger|p, q}&=\sum_{n}^{\infty}\sum_{{\Phi_n}} \left|\Braket{p, q|\widehat{\mt}|\Phi_n}\right|^2\nn \\
		=&(2\p)^8\, \d^{(4)}(0)\,\sum_n\sum_{\Phi_n}\,\d^{(4)}(p+q-P_n)\, |\mt_{2\to n}|^2.
	\end{align} Here $P_n$ is the total $4-$momentum for the $n$ particle state $\ket{\Phi_n}$ and $\mt_{2\to n}$ is the amplitude for $2\to n$ scattering. Now the cross-section for $2\to n$ process is given by 
	\be 
	\s_{2\to n}=\frac{1}{2\sqrt{s(s-4m^2)}}\, \left[(2\p)^4\sum_{\Phi_n}\,\d^{(4)}(p+q-P_n)\, |\mt_{2\to n}|^2\right]
	\ee and the total scattering cross-section is given by 
	\be 
	\s_{\text{tot}}^{A+A\to A+A}=\sum_n\, \s_{2\to n}.
	\ee Collecting everything together , we get the optical theorem 
	
	\be 
	\s_{\text{tot}}(s)=\frac{1}{2\sqrt{s(s-4m^2)}}\ma_s(s,t=0).
	\ee Using the partial wave expansion of the absorptive part $\ma_s$ we obtain 
	\be\label{opthm} 
	\s_{\text{tot}}(s)=\frac{8\p}{s-4m^2}\mathlarger{\mathlarger{\sum}}_{\substack{\ell=0\\ \ell\,\text{even}}}(2\ell+1)\,a_\ell(s).
	\ee 
	
	\item {\bf Partial wave unitarity bound} \\ 
	We start with the inequality 
	\be \label{Unineq2}
	\frac{1}{i}\Braket{\Psi|\widehat{\mt}-\widehat{\mt}^\dagger|\Psi}
	\ge  \sum_{{\Phi_2}} \left|\Braket{\Psi|\widehat{\mt}|\Phi_2}\right|^2.
	\ee Next choosing $\ket{\Psi}= \ket{E,0\,;\,\ell,0},\, \ket{\Phi_2}=\ket{\varepsilon,\vec{\rho}\,;\,\l,\m}$, and using the  completeness relation \eqref{2completeness} one gets 
	\begin{align} 
		\sum_{{\Phi_2}} \left|\Braket{\Psi|\widehat{\mt}|\Phi_2}\right|^2&=\sum_{\l}\sum_{\m=-\l}^\l\int \frac{d\varepsilon\,d^3\vec{\rho}}{(2\p)^4}\,\left|\Braket{E,0\,;\,\ell,0|\widehat{\mt}|\varepsilon,\vec{\rho}\,;\,\l,\m}\right|^2,\nn\\
		&=(2\p)^4\,\d^{(4)}(0)\, \left|f_\ell(s)\right|^2
	\end{align} 
	Next we look at the left hand side of the inequality \eqref{Unineq2}. Using 
	\be 
	\frac{1}{i}\Braket{E,0\,;\,\ell,0|\widehat{\mt}-\widehat{\mt}^\dagger|E,0\,;\,\ell,0}=\frac{(2\p)^4}{i}\, \d^{(4)}(0)\,\left[f_\ell(s)-f_\ell(s)^*\right].
	\ee Next Hermitian analyticity of scattering amplitude \eqref{hermitian} implies 
	\be 
	f_\ell(s)^*=f_\ell(s^*),
	\ee using which we infer that 
	\be 
	\frac{1}{i}\Braket{E,0\,;\,\ell,0|\widehat{\mt}-\widehat{\mt}^\dagger|E,0\,;\,\ell,0}=(2\p)^4\, \d^{(4)}(0)\, \left[2 a_\ell(s)\right].
	\ee Finally we obtain the partial wave unitarity bound 
	\be 
	0\le |f_\ell(s)|^2\le 2\, a_\ell(s), \qquad \forall\,\ell\ge 0, \,\,\, s\ge 4m^2.
	\ee An immediate consequence of this unitarity bound is 
	\be \label{apos}
	0\le a_\ell(s)\le 2, \qquad \forall\,\ell\ge 0, \,\,\, s\ge 4m^2.
	\ee 
	When $4m^2<s<s_2$, $s_2$ being the second threshold which equals occurs at $16m^2$ if the theory is $\mathbb{Z}_2$ invariant, else at $9m^2$, we have equality: 
	\be \label{eluni}
	|f_\ell(s)|^2\le 2\, a_\ell(s).
	\ee This is known as \emph{elastic unitarity}.

	In all of our subsequent analysis we will repeatedly use the unitarity in the form of \emph{positivity} of the partial wave coefficients $a_\ell(s)$. The upper bound plays central role in deriving high energy bounds on total scattering cross-section, like the Froissart-Martin bound, which we will discuss soon. 
\end{itemize}
\section{More Analyticity: Martin's Extension of Analyticity Domain}
The key ingredient in Martin's \cite{MartinExt1} extension of analyticity domains for $\mm$ is \emph{unitarity}. It is to be stressed that so far unitarity has not been used in deriving the  Lehmann ellipses.  Lehmann only used  analytic consequences of micro-causality.  Let us first state Martin's theorem on the extension of domain, and then we will discuss the important consequences of this theorem.
\begin{theorem}[{\bf Martin's extension theorem}]\label{MartinTh}
	Consider an \textbf{unitary}  elastic scattering amplitude $\mt(s,t)$, and the corresponding absorptive part $\ma_s(s,t)$, for the process $A+A\to A+A$ of massive particles $A$ satisfying analyticity properties:
	\begin{itemize}
		\item [I.] {\bf Dispersive representation: }$\mt(s,t)$ satisfies fixed $t$ dispertion relations for $-t_{M}<t\le0,\,t_M>0$ where the amplitude is \textbf{analytic in complex $s$ plane} outside of the unitarity cuts;
		\item[II.] {\bf Lehmann analyticity: }For \textbf{fixed physical} $s$, the amplitude $\mt(s,t)$ and the absorptive part $\ma(s,t)$ are analytic in the corresponding Lehmann ellipses;
		\item[III.] {\bf Bros-Epstein-Lehmann-Glaser  (BELG) analyticity:} From the results of Bros, Epstein and Glaser, and Lehmann \cite{BEG1, Lehmann2}, in the neighbourhood of any point $s_0, t_0$, $-t_M<t\le0$, $s_0$ outside the cuts, there is analyticity in both $s$ and $t$ in 
		\be 
		|s-s_0|<\eta(s_0,t_0),\quad |t-t_0|<\eta(s_0,t_0).
		\ee  
		A priori the size of the neighbourhood can vary with $s_0,t_0$.
	\end{itemize} 
	Then the amplitude $\mt(s,t)$ is analytic (except for possible fixed-$t$ poles) in the topological product domain 
	\be 
	\mathcal{D}_M:=\{(s,t):s\not\in \left[[4m^2,\infty)\cup(-\infty,-t]\right]\,\times\, |t|<R\}, 
	\ee $R$ being independent of $(s,t)$. 
\end{theorem}
This dispersive representation is valid in the complex $s-$plane barring the $s-$channel threshold cut starting at $s=4m^2$ and the $u-$channel threshold cut starting at $u=4m^2$, or equivalently, at $s=-t$, for fixed $t\le0$. Actually what we have is the analyticity of $\mm(s,t)$ in the topological product domain $\md_i\subset\mathbb{C}^2$
\be 
\md_i:=\{(s,t)\in \mathbb{C}^2|s\in\mathbb{C}\,\backslash\, \left[[4m^2,\infty)\cup(-\infty,-t]\right],\, t\in [-t_M,0] \}.
\ee 
Martin's theorem proves the existence of a domain of the form 
\be 
\md_e:=\{(s,t)\in \mathbb{C}^2|s\in\mathbb{C}\,\backslash\,\left[[4m^2,\infty)\cup(-\infty,-t]\right],\, t\in \D \},
\ee  where $\D\subset\mathbb{C}$ is a complex domain.  The real interval $[-t_M,0]$ need not lie \emph{entirely} within $\D$ . However it is (naturally) \emph{expected} that $[-t_M,0]\cap\D\ne \emptyset$.

We will now present a proof of the theorem following \cite{martin-II}.

The greatest significance of Martin's extension theorem is that it shows existence of analyticity domain in $t$ for both $\mm$ and $\ma_s$ which \emph{does not} vary with $s$. In particular, as evident from \eqref{LehmannS}, \eqref{LehmannL}, the Lehmann domains in $t-$plane shrink to zero as $|s|\to\infty$.  Martin's theorem establishes that even when $|s|\to \infty$ there exist finite, \emph{non-vanishing} domains of analyticity in $t$ for \emph{both} $\mm(s,t)$ and $\ma_s(s,t)$.  As a consequence of this theorem we can find out an enlarged domain in $t-$plane where the fixed-$t$ dispersion relations hold. For the rest of the discussion we will consider elastic scattering of pions. In this case Martin proved that $R=4m^2$. Then we have that $\ma_s(s,t)$ is analytic in the open disc $|t|<4m^2$. However, $\ma_s$ is also expendable in terms of Legendre polynomials with \emph{positive coefficients} (due to unitarity). Now a Legendre polynomial expansion with positive coefficients \emph{must} have a singularity at the extreme right of the largest ellipse of convergence. Hence $\ma_s(s,t)$ having no singularities for $0\le t<R$ must be analytic in an open elliptical disc bounded by an ellipse  with foci at $t=0, \,4m^2-s$ and right extremity at $\m$. We will call this ellipse  \emph{the Martin ellipse}, $\me_M(s)$. 

But, $\ma_s(s,t)$ is also analytic in the large Lehmann ellipse $\me_L(s)$ with  foci at $t=0, 4m^2-s$ and right extremity at $256m^2s^{-1}$. 
Then the required enlarged domain is given by 
\be
\md=\bigcap_{\m\le s<\infty} \left(\me_L(s)\cup\me_U(s)\right).
\ee
From the dimensions of the ellipse, it is quite evident that $\me_U(s)\subset\me_L(s)$ for $s< 64m^2$, and $\me_L(s)\subset\me_U(s)$ for $s> 64m^2$. Also $\me_U(64m^2)\subset\me_U(s)$ for $s\ge 64m^2$, and $\me_L(64m^2)\subset \me_L(s)$ for $s\in [4m^2,16m^2]$. Thus finally
\be 
\md=\bigcap_{16m^2\le s\le64m^2}\me_L(s).
\ee 

This is not the whole story. The analyticity domain can further be extended using elastic unitarity, \eqref{eluni}, see \cite{martin-II} for more details. All these works are aimed at establishing the conjectured \emph{Mandelstam analyticity}, also known as \emph{maximal analyticity} (to be discussed below). 

\section{Froissart-Martin Bound}\label{Yndurain}
One of the important applications of Martin's extension theorem is  a derivation of the Froissart-Martin bound for the total scattering cross-section in usual 3+1 dimensional Minkowski spacetime. It is straightforward to generalize this derivation to any spacetime dimension. We consider $2\to2$ scattering of identical massive (mass being $m$) scalar spinless particles. Following \cite{Yndurain}, we present with the derivation of the Froissart-Martin bound for the averaged total scattering cross-section given by, 
\be 
\overline{\s}(s)=\frac{1}{s-4m^2}\int_{4m^2}^{s}ds'(s'-4m^2)\s_{\text{tot}}(s').
\ee  
\paragraph{}
The principal ingredients for obtaining the bound are as following
\begin{itemize}
	\item {\bf Unitarity:} Unitarity will be employed in two avatars. One is the optical theorem which enables to write a partial wave expansion $\overline{\s}(s)$: 
	\be 
	\overline{\s}(s)=\frac{8\p}{s-4m^2}\int_{4m^2}^{s}ds'\,\left[\sum_{\substack{\ell=0\\ \ell \text{even}}}^{\infty} (2\ell+1)\,a_\ell(s')\right].
	\ee 
	The most important bit is the partial wave unitarity bound, which ultimately bounds the scattering cross-section, 
	$$ 
	0\le a_\ell(s)\le 2,\qquad,\forall\,\ell\ge 0,\,s\ge 4m^2. 
	$$ 
	\item {\bf Martin analyticity:} The absorptive part $\ma_s (s,t)$ is analytic for $|t|<4m^2$ when $s$ lies in the cut plane. 
	\item {\bf Polynomial boundedness:} There exists a positive integer $n$ such that the integral quantity, 
	\be 
	a_{n}:=\int_{4m^2}^\infty \frac{d\overline{s}}{\overline{s}^{n+1}}\ma_s(\overline{s},t=t_0),\qquad 0<t_0<4m^2.
	\ee 
	The $t_0$ is chosen to lie within the Martin ellipse $\me_M$. This ensures that the partial wave expansion of $\ma_s$ converges.
\end{itemize}

\paragraph{} Next, using the positivity of the Legendre polynomials $P_\ell(x)$ for $x>1$ and the property that $P_\ell(x)$ is a strictly monotonically increasing fucntion of $\ell$ for $x>1$, we reach the following inequality, after some straightforward algebra,
\be \label{fineq1} 
a_n\ge\sqrt{\frac{s}{s-4m^2}}\,s^{-n-1}\, P_{L+2}\left(1+\frac{2t_0}{s-4m^2}\right)\sum_{\ell=L+2}^{\infty}(2\ell+1)\,\int_{4m^2}^{s}d\overline{s}\,a_{\ell}(\overline{s})
\ee for some $L>0$. Next we shift our attention to $\overline{\s}(s)$. We write the the same as 
\be 
\overline{\s}(s)
=
\frac{8\p}{s-4m^2}\,\sum_{\substack{\ell=0\\\ell \text{even}}}^L(2\ell+1)\int_{4m^2}^{s}d\overline{s}\,a_\ell(\overline{s})
+
\frac{8\p}{s-4m^2}\,\sum_{\ell=L+2}^{\infty}(2\ell+1)\int_{4m^2}^{s}d\overline{s}\,a_\ell(\overline{s}).
\ee Now using the partial wave unitarity bound  and \eqref{fineq1} one obtains quite straightforwardly
\be \label{sbarineq1}
\overline{\s}(s)\le \frac{16\p}{s-4m^2}\left[(s-4m^2)\frac{1}{2} (L+1) (L+2)+\sqrt{\frac{s-4m^2}{s}}\,\frac{s^{n+1}}{P_{L+2}\left(1+\frac{2t_0}{s-4m^2}\right)}\,a_n\right].
\ee 
\paragraph{} Let us now focus on the second term above. Using the following inequality satisfied by the Legendre polynomial
\be 
P_\ell(x)\ge \frac{\varphi_0}{\p}\,\left(x+\cos\varphi_0\sqrt{x^2-1}\right)^{\ell}\,\,\,\, x>1,\,\,0<\varphi_0<\p,
\ee 
we have, in the limit $L\gg1$ and $s\gg 4m^2$,
\be 
\sqrt{\frac{s-4m^2}{s}}\,\frac{s^{n+1}}{P_{L+2}\left(1+\frac{2t_0}{s-4m^2}\right)}\,a_n\approx \frac{s^{n+1}}{P_{L}\left(1+\frac{2t_0}{s}\right)}\,a_n
\lesssim \frac{\p a_n}{\varphi_0}\,s^{n+1}
\left[
1+2\cos\varphi_0\sqrt{\frac{{t_0}}{{s}}}
\right]^{-L}.
\ee 
Using the above inequality into \eqref{sbarineq1}, we obtain the following inequality,
\be 
\overline{\s}(s)\lesssim 16\p 
\left[
\frac{L^2}{2}
+
\frac{\p a_n}{\varphi_0}\,s^{n} \left(
1+2\cos\varphi_0\sqrt{\frac{{t_0}}{{s}}}
\right)^{-L}
\right] =
8\p L^2
\left[
1
+
\frac{2\p a_n}{\varphi_0}\,\frac{s^{n}}{L^2} \left(
1+2\cos\varphi_0\sqrt{\frac{{t_0}}{{s}}}
\right)^{-L}
\right].
\ee The $\ell$ sum gets truncated effectively if the second term inside the parenthesis above is very small compared to $1$. Let us consider the marginal case when this is equal to 1. The $L$ determined by this marginal case is the one where we will decide to truncate the $\ell$-sum effectively. 
\paragraph{} Considering an ansatz for the $L$ of the form 
\be 
L=A_0~s^a\ln^bs\ee one readily obtains, in the limit $S\gg 4m^2$, that in the leading order in $S$, the marginal $L$ is given by, 
\be 
a=1/2,\,\,b=1\,\,,
A_0=\frac{n-1}{2\sqrt{t_0}\cos\varphi_0}.
\ee Thus the marginal $L-$cutoff is given by,
\be 
L=\frac{(n-1)}{\cos\varphi_0}\sqrt{\frac{s}{4 t_0}}\ln s.
\ee Now, truncating the $\ell$-sum contributing to $\overline{\s}(S)$ to $L$ we obtain 
\be 
\overline{\s}(s)\le \frac{2\p}{t_0}(n-1)^2\, s\ln^2\frac{s}{s_0},
\ee where $s_0$ is some scale to make the argument of the $\ln$ diemsnionless. Now we can take $t_0=4m^2$ to get the strongest possible bound, so that we get 
\be 
\overline{\s}(s)\le \frac{\p}{2m^2}(n-1)^2\, s\ln^2\frac{s}{s_0},\qquad |s|\to \infty.
\ee A similar bound also follows for $s\to u$, i.e.
\be 
\overline{\s}(u)\le \frac{\p}{2m^2}(n-1)^2\, u\ln^2\frac{u}{u_0}, \qquad  |u|\to \infty.
\ee 
Further, $n$ can be fixed to $2$ using \emph{Phragm\'{e}n-Lindel\"{o}f}  theorem giving finally 
\be 
\overline{\s}(s)\le \frac{\p}{2m^2}\, s\ln^2\frac{s}{s_0}.
\ee 
\section{Mandelstam analyticity}
So far we have discussed analyticity properties of $\mt(s,t)$ in both complex $s$ and $t$ to varied degrees. These various explorations of analyticity are ultimately directed towards proving what is known as \emph{Mandelstam analyticity} or \emph{maximal analyticity}. Mandelstam conjectured \cite{Mandel} that the scattering amplitude $\mt(s,t)$ is analytic in both complex $s$ and complex $t$ except for the cuts due to the physical  thresholds, i.e. $\mt(s,t)$ is analytic in 
\begin{align}
	\begin{split}
&\mathbb{C}^2\backslash \left(D_s\cup D_t\cup\widetilde{D}_{s+t}\right),\\
    D_s:=\{s\in [4m^2,\infty)\},\, &D_t:=\{t\in [4m^2,\infty)\},\,\widetilde{D}_{s+t}:=\{s+t\le 0\}.
    \end{split} 
\end{align} Mandelstam analyticity enables one to write a double dispersion relation for the scattering amplitude. If we assume that there is no poles due to bound states and the amplitude falls to zero at infinity, so that there is no subtraction, then a scattering amplitude which is Mandelstam analytic admits the following dispersive representation 
\begin{align}\label{Mandel2} 
	\begin{split}
\mt(s,t)=\frac{1}{\p^2}\int_{4m^2}^\infty ds'\int_{4m^2}^\infty dt' \frac{\ma_{st}(s,t)}{(s-s')(t-t')}\,
&+\, \frac{1}{\p^2}\int_{4m^2}^\infty dt'\int_{4m^2}^\infty du' \frac{\ma_{tu}(t,u)}{(t-t')(u-u')}\,\\
&+\,\frac{1}{\p^2}\int_{4m^2}^\infty ds'\int_{4m^2}^\infty du' \frac{\ma_{su}(s,u)}{(s-s')(u-u')};\quad s+t+u=4m^2. 
\end{split}
\end{align} Here $\ma_{st}(s,t)$ is known as  \emph{double spectral density}  in $s$ and $t$ defined as 
\begin{align}
	\begin{split} 
\ma_{st}(s,t)&:=\text{Im.}_s\text{Im.}_t\,\,\mt(s,t)\,=\, \text{Im.}_t\text{Im.}_s\,\,\mt(s,t)\\
&=\lim_{\e\to 0}\frac{1}{4}\left[\mt(s+i\e,\,t+i\e)-\mt(s-i\e,\, t+i\e)-\mt(si\e,\,t-i\e)+\mt(s-i\e,\, t-i\e)\right].
\end{split}
\end{align} Observe that $\ma_{st}(s,t)=\ma_{ts}(t,s)$. The other double spectral densities are also defined similarly. Also observe that the MAndelstam double dispersive representation \eqref{Mandel2} is manifestly crossing symmetric.  We have written the dispersion relation without any subtraction. However, since $\mt(s,t)$ is in general polynomially bounded,  subtractions will be required in order to get rid of the contribution from infinity. 

Proving Mandelstam analyticity is still an open problem. Even a perturbative proof will be extremely important.

\begin{subappendices}
\section*{Appendix}

\section{Properties of Legendre polynomials}
In this appendix, we collect some properties of Legendre polynomials, $P_n(x)$ which we have used at various places. The Legendre polynomials $P_n(x)$ are solutions of the differential equation 
\be 
(1-x^2)\,\frac{d^2 P_n(x)}{dx^2}\,-\,2x\,\frac{d P_n(x)}{dx}\,+\,n(n+1)\,P_n(x)=0.
\ee  with the boundary condition 
\be 
P_n(\pm 1)=(\pm 1)^n.
\ee These polynomials satisfy definite parity property
\be 
P_n(-x)=(-1)^n\,P_n(x).
\ee 
\begin{enumerate}
	\item[1.] {\bf Laplace integral representation:} For integer $n\ge 0$, $P_n(x)$ admits the following Laplace integral representations
	\begin{align}
		\label{Laplace1}P_n(x)&=\frac{1}{\p}\,\int_{0}^{\p}d\f\,\left(x+i\cos\f \,\sqrt{1-x^2}\right)^n,\quad\,-1\le x\le 1\\
		\label{Laplace2}P_n(x)&=\frac{1}{\p}\,\int_{0}^{\p}d\f\,\left(x+\cos\f\,\sqrt{x^2-1}\right)^n,\quad\,x\ge 1
	\end{align}
	  We will prove two useful results using these integral representations. 
	  \begin{enumerate}
	  	\item \be P_n(1)\ge |P_n(x)|, \quad -1<x<1.\ee
	  	\begin{proof}
	  		Using \eqref{Laplace1} we have 
	  		\begin{align*}
	  			|P_n(x)|&=\frac{1}{\p}\left|\int_{0}^{\p}d\f\,\left(x+i\cos\f\,\sqrt{1-x^2}\right)^n\right|\nn\\
	  			&\le\frac{1}{\p}\,\int_{0}^{\p}d\f\,\left|\left(x+i\cos\f\,\sqrt{1-x^2}\right)\right|^n,\nn\\
	  			&=\frac{1}{\p}\,\int_{0}^{\p}d\f\,\left(x^2+\cos^2\f\,(1-x^2)\right)^{\frac{n}{2}}\nn\\
	  			&\le\frac{1}{\p}\,\int_{0}^{\p}d\f\,\left(x^2+\,(1-x^2)\right)^{\frac{n}{2}}\qquad \left[\cos^2\f<1,\,\f\in\mathbb{R}\right]\\
	  			&=P_n(1).
	  		\end{align*}
	  	\end{proof}
  	\item \be 
  	P_n(x)\ge \frac{\f_0}{\p} \,\left[x+\sqrt{x^2-1}\cos\f_0\right]^n,\qquad x>1,\quad 0<\f_0<\p.
  	\ee 
  	\begin{proof}
  		First we define  
  		\be 
  		y:=\frac{\sqrt{x^2-1}}{x},\quad x>1.
  		\ee Next we define the functions 
  		\begin{align} 
  		H_n(y,\f_0,\f)&:=\frac{(1+y\cos\f)^n}{(1+y\cos\f_0)^n},\\
  		G_\n(y,\f_0)&:=\int_{0}^\p H_n(y,\f_0,\f)\,d\f.
  		\end{align} 
  	Since $H_n(y,\phi_0,\phi)$ is an increasing function of $y$ for $0<\phi<\phi_0<\p$ and decreasing function of $y$ for $0<\phi_0<\phi<\p$, we can obtain the following inequality quite easily, 
  	\end{proof}
	  \end{enumerate}
\end{enumerate}
\section{A proof of simplified Martin's extension theorem}
In this appendix, we review the proof of Martin's extension theorem \ref{MartinTh} for  a simpler case when the amplitude $\mm(s,t)$ satisfies fixed-$t$ dispersion relation with no subtraction and without the left hand cut . First we state the theorem for this simple case.  
\begin{theorem}[{\bf Simplified Martin's extension theorem}]
	Consider an \textbf{unitary}  elastic scattering amplitude $\mm(s,t)$, and the corresponding absorptive part $\ma_s(s,t)$, for the process $A+A\to A+A$ of massive particle $A$ with mass $m>0$ satisfying analyticity properties:
	\begin{itemize}
		\item [I.] {\bf Dispersive representation: }$\mm(s,t)$ satisfies fixed $t$ dispertion relations for $-t_{M}<t\le0,\,t_M>0$, where the amplitude is \textbf{analytic in complex $s$ plane} outside of the unitarity cut, $s\in\mathbb{C}\,\backslash\, [4m^2,\infty)$,
		\be \label{dispsimp}
		\mm(s,t)=\frac{1}{\p}\int_{4m^2}^\infty ds'\,\frac{\ma_s(s',t)}{s'-s};
		\ee
		\item[II.] {\bf Lehmann analyticity: }For \textbf{fixed physical} $s$, the amplitude $\mm(s,t)$ and the absorptive part $\ma(s,t)$ are analytic in the corresponding Lehmann ellipses;
		\item[III.] {\bf Bros-Epstein-Lehmann-Glaser  (BELG) analyticity:} From the results of Bros, Epstein and Glaser, and Lehmann \cite{BEG1, Lehmann2}, in the neighbourhood of any point $s_0, t_0$, $-t_M<t_0\le0$, $s_0$ outside the cuts, there is analyticity in both $s$ and $t$ in 
		\be 
		|s-s_0|<\eta(s_0,t_0),\quad |t-t_0|<\eta(s_0,t_0).
		\ee  
		A priori the size of the neighbourhood can vary with $s_0,t_0$.
	\end{itemize} 
	Then the amplitude $\mm(s,t)$ is analytic (except for possible fixed-$t$ poles) in the topological product domain 
	\be 
	\mathcal{D}_M:=\{(s,t)\in\mathbb{C}^2\,:\, s\in \mathbb{C}\,\backslash\,[4m^2,\infty),\,|t|<R\}, 
	\ee where $R$ is independent of $(s,t)$. The absorptive part $\ma_{s}(s,t)$ is also analytic in $|t|<R$ for all $s$. \end{theorem}
	\begin{proof}
		The central mathematical result that is used in the proof is what is known as Hartog's fundamental theorem \cite{Hor}.
		\begin{theorem}[{\bf Hartog's fundamental theorem}]
			Let $f(z_1,\dots,z_n)$ be a complex valued function defined in the complex domain $\Omega\subset \mathbb{C}^n$. If $f$is analytic in each variable $z_j$ when the other variables are given arbitrary fixed values , then $f$ is analytic in $\Omega$.  
		\end{theorem} 
		To use this theorem, we consider a Taylor series of $\mm(s,t)$ about $t=0$
		\be \label{mtaylor}
		\mm(s,t)=\sum_{n=0}^\infty \frac{t^n}{n!}\, \left(\frac{\partial}{\partial t}\right)^n \mm(s,0).
		\ee Now $\mm(s,t)$ will be analytic in the product domain $\{(s,t)\,|\, s\in\mathcal{K},\,t\in\D\}$ if all $\left(\frac{\partial}{\pd t}\right)^n\mm(s,0)$ are analytic in $\mathcal{K}$ and if the series is \emph{absolutely convergent} for $t\in \D$ uniformly for $s\in \mathcal{K}$.
		
		Now, unitarity manifests itself in the form of the following positivity properties for the absorptive part $\ma_s$
		\begin{align}
			\label{pos1}	\frac{\pd^n}{\pd t^n}\,\ma_s(s,t=0)&\ge 0\\
			\label{pos2}	\left|\frac{\pd^n}{\pd t^n}\,\ma_s(s,t)\right|&\le \frac{d^n}{dz^n}\,\ma_s(s,t=0),\qquad t\in[-t_0,0].
		\end{align} One important assumption that goes into proving these positivity properties is the convergence of the partial wave expansion in the physical region of $t$ which is guaranteed by the existence of the large Lehmann ellipse. In particular, we will always consider the interval $[-t_0,0]$  to be lying inside the large Lehmann ellipse. In order to apply these properties to the dispersion relation , we need to interchange the integral and the derivatives in the dispersion relation \eqref{dispsimp} above. In particular, we would like to examine whether we can write 
		\be 
		\left(\frac{\pd}{\pd t}\right)^n\mm(s,0)=\frac{1}{\p}\int_{4m^2}^\infty ds'\,\frac{1}{s'-s}\, \left(\frac{\pd}{\pd t}\right)^n \ma_s(s',0).
		\ee \emph{A priori}, this is not guaranteed. In particular, this change of order may require \emph{additional subtractions}, which is equivalent to the possibility that $\ma_s(s,0)$ and $\pd^n\ma_s(s,0)/\pd t^n$ are bounded by the different powers of $s$ as $|s|\to \infty$ . What we will show below is that this is not the case, unitarity makes it sure that both $\ma_s(s,0)$ and $\pd^n\ma_s(s,0)/\pd t^n$ are bounded by the same power of $s$ as $|s|\to \infty$. 
		
		First consider in the $s$ cut plane a fixed point $s_0\in\mathbb{R}$ with $s_0<4m^2$. Then by BELG analyticity, there exists $R(s_0)$ such that $\mm(s_0,t)$ is analytic in the neighborhood $|t|<R(s_0)$.  Therefore the derivatives $\pd^n\mm(s,0)/\pd t^n$ exists and are bounded by the Cauchy inequalities 
		\be \label{cauchyineq}
		\left|\left(\frac{\pd}{\pd t}\right)^n \mm(s_0,0)\right|\le \frac{Mn!}{R(s_0)^n},
		\ee where 
		\be 
		M:=\text{max.}_{|t|=R}\,|\mm(s_0,t)|
		\ee is \emph{finite}. If $M$ is not finite on $|t|=R(s_0)$ then we can choose a smaller $R$. 
		Since the first derivative $\pd\mm(s,t)/\pd t$ exists, we are free to define it as 
		\be 
		\lim_{\substack{\t\to0\\ \t>0}}\,\frac{\mm(s_0,0)-\mm(s_0,-\t)}{\t}=\frac{1}{\p}\lim_{\substack{\t\to0\\ \t>0}}\int_{4m^2}^\infty \,  ds'\,\frac{[\ma_s(s',0)-\ma_s(s',-\t)]/\t}{s'-s_0}.
		\ee It is to be emphasized that $\t>0$ \emph{can be} chosen sufficiently small such that the interval $[-\t,0]$ always lies entirely within the large Lehmann ellipse so that we can freely work with convergent partial series expansions of $\ma_s$ . Next we observe that for $s'>4m^2$ the integrand above is positive as a consequence of \eqref{pos2}. Therefore 
		\be 
		\lim_{\substack{\t\to0\\ \t>0}}\,\frac{\mm(s_0,0)-\mm(s_0,-\t)}{\t}\,\ge\,\frac{1}{\p}\lim_{\substack{\t\to0\\ \t>0}}\int_{4m^2}^X \,  ds'\,\frac{[\ma_s(s',0)-\ma_s(s',-\t)]/\t}{s'-s_0}
		\ee for some \emph{finite} $X>4m^2$. Now the right-hand integral is an integral over a compact interval. If the integral is taken in Lebesgue sense then by \emph{dominated convergence theroem} the limit of the integral for $\t\to 0$ exists if the the same limit exists for the integrand and the integrand is bounded  by some continuous function. Since $\t>0$ can be chosen sufficiently small so that for all $s'\in[4m^2,X]$ the interval $[-\t,0]$ lies within the large Lehmann ellipse, we have that the limit of the integrand exists because
		\be 
		\lim_{\substack{\t\to0\\ \t>0}}\frac{\ma_s(s',0)-\ma_s(s',-\t)}{\t}=\frac{\pd }{\pd t}\,\ma_s(s',0)
		\ee exists  and is bounded by some continuous  function  $\ml(s')$ inside the large Lehmann ellipse. Thus we have 
		\be 
		\frac{\pd}{\pd t}\mm(s_0,0)\ge \frac{1}{\p}\int_{4m^2}^X \frac{ds'}{s'-s_0}\,\frac{\pd\ma_s(s',0)}{\pd t}.
		\ee Since by \eqref{cauchyineq} $\frac{\pd}{\pd t}\mm(s_0,0)$ is bounded by $M/R(s_0)$ which is independent of $X$ and $\frac{\pd\ma_s(s',0)}{\pd t}\ge0$ by \eqref{pos1}, the integral on the right hand side converges for $X\to\infty$. Therefore we have 
		\be 
		\frac{\pd}{\pd t}\mm(s_0,0)\ge \frac{1}{\p}\int_{4m^2}^\infty \frac{ds'}{s'-s_0}\,\frac{\pd\ma_s(s',0)}{\pd t}.
		\ee 
		
		On the other hand, applying the \emph{mean-value theorem}  we can write 
		\be 
		\frac{\ma_s(s',0)-\ma_s(s',-\t)}{\t}=\frac{\pd \ma_s(s',-\tilde{\t}(s'))}{\pd t},\qquad -\t<-\tilde{\t}(s')<0,\,\,s'\in[4m^2,\infty).
		\ee  Once again choosing $\t>0$ to be sufficiently small so that it lies within the large Lehmann ellipse, we can use \eqref{pos2} to write 
		\begin{align} 
			\frac{\pd}{\pd t}\mm(s_0,0)&=\frac{1}{\p}\lim_{\substack{\t\to0\\ \t>0}}\int_{4m^2}^\infty \,  ds'\,\frac{[\ma_s(s',0)-\ma_s(s',-\t)]/\t}{s'-s_0}\nn\\
			&=\frac{1}{\p}\lim_{\substack{\t\to0\\ \t>0}}\int_{4m^2}^\infty\,\frac{ds'}{s'-s_0}\,\frac{\pd \ma_s(s',-\tilde{\t}(s'))}{\pd t} \nn\\
			&\le \frac{1}{\p}\int_{4m^2}^\infty \frac{ds'}{s'-s_0}\,\frac{\pd\ma_s(s',0)}{\pd t}
		\end{align} Therefore, we have finally
		\be 
		\frac{\pd}{\pd t}\mm(s_0,0)= \frac{1}{\p}\int_{4m^2}^\infty \frac{ds'}{s'-s_0}\,\frac{\pd\ma_s(s',0)}{\pd t}.
		\ee 
		Now this argument can be generalized to any $n$th order $t-$derivative by induction principle so that we have 
		\be \label{mder}
		\frac{\pd^n}{\pd t^n}\mm(s_0,0)= \frac{1}{\p}\int_{4m^2}^\infty \frac{ds'}{s'-s_0}\,\frac{\pd^n}{\pd t^n}\ma_s(s',0)\,.
		\ee 
		
		Now we will show that the above representation can be extended to an arbitrary compact subdomain  of the cut $s$ plane. First we observe that the function defined as 
		\be 
		\frac{\pd^n}{\pd t^n}\mm(s,0)= \frac{1}{\p}\int_{4m^2}^\infty \frac{ds'}{s'-s}\,\frac{\pd^n}{\pd t^n}\ma_s(s',0)
		\ee exists for some \emph{finite} interval of the real $s-$axis below the cut due to BELG analyticity, and due to the fact that we can run the same   argument for some shifted $s_0<4m^2$. $\frac{\pd^n}{\pd t^n}\mm(s,0)$ can be continued to complex $s$ because of the existence of analyticity neighborhood. The integral on the right hand-side can be continued to complex $s$ so long it converges uniformly. Both will coincide whenever the integral converges. Therefore we will investigate the integral
		
		We observe that, using \eqref{pos1}, 
		\be \label{intineq}
		\left| \frac{\pd_t^n\ma_s(s',0)}{s'-s}\right|< \m(s,s_0)\, \frac{\pd_t^n\ma_s(s',0)}{s'-s_0},
		\ee  with 
		\be 
		\m(s,s_0):=\sup_{4m^2<s'<\infty}\left|\frac{s'-s_0}{s'-s}\right|.
		\ee Now $\m(s,s_0)$ is \emph{finite} so long $s$ is outside the cut. Let us examine this point a bit more. We start by observing that 
		\be 
		\left|\frac{s'-s_0}{s'-s}\right|=\frac{|s'-s_0|}{\left[(s'-\text{Re.}[s])^2+(\text{Im.}[s])^2\right]^{1/2}}.
		\ee Now, let us write $s'=4m^2+\D,\, \text{Re.}[s]=4m^2+\tilde{\D}$. Also recall that $s_0=4m^2-\ve$ with $\ve$ is sufficiently small positive number. We can consider without l.o.g that $\D>\ve$. With these substitutions we have 
		\be 
		\left|\frac{s'-s_0}{s'-s}\right|=\frac{\D-\ve}{\left[(\D-\tilde{\D})^2+(\text{Im.}[s])^2\right]^{1/2}}.
		\ee Now the maximum occurs at 
		\be \label{max}
		\D=\tilde{\D}+\frac{(\text{Im.}[s])^2}{\tilde{\D}-\ve},
		\ee which shows that in the limit $\text{Im.}[s]\to 0$, the maximal contribution comes from $\D\to\tilde{\D}$. Thus we see that 
		\be 
		\m(s,s_0)\sim\frac{\text{Re.}[s]-s_0}{\Big|\text{Im.}[s]\Big|},\quad \text{Im.}[s]\to 0.
		\ee   
		
		Using \eqref{intineq}, we have 
		\begin{align}
			\int_{4m^2}^\infty ds'\,\left|\frac{\pd_t^n\,\ma_s(s',0)}{s'-s}\right|\,<\,\m(s,s_0)\,\int_{4m^2}^\infty ds'\,\frac{\pd_t^n\,\ma_s(s',0)}{s'-s_0}\,\le\,\m(s,s_0)\, \frac{M n!}{R(s_0)^n},
		\end{align} where we have used \eqref{mder} and \eqref{cauchyineq} to write the last inequality. Thus so long $\m(s,s_0)$ is finite, the integral over $\left|\frac{\pd_t^n\,\ma_s(s',0)}{s'-s}\right|$ is convergent, and therefore, the integral over $\frac{\pd_t^n\,\ma_s(s',0)}{s'-s}$ converges in Lebesgue sense. This then implies that $\pd_t^n\mm(s,0)$ is analytic in the cut $s$-plane. Further
		\be 
		\left| \frac{\pd^n}{\pd t^n}\mm(s,0)\right|=\left|\int_{4m^2}^\infty ds'\,\frac{\pd_t^n\,\ma_s(s',0)}{s'-s}\right|\,\le\, \int_{4m^2}^\infty ds'\,\left|\frac{\pd_t^n\,\ma_s(s',0)}{s'-s}\right|\,<\,\m(s,s_0)\, \frac{M n!}{R(s_0)^n}.
		\ee 
		Now consider the  series 
		\be 
		\sum_{n=0}^\infty \frac{|t|^n}{n!}\,\left| \frac{\pd^n}{\pd t^n}\mm(s,0)\right|\,<\,\m(s,s_0)\,M\,\sum_{n=0}^\infty\, \frac{|t|^n}{R(s_0)^n}.
		\ee This implies that for fixed $s$ away from the cut and for $|t|<R(s_0)$, the Taylor series representation of $\mm(s,t)$, \eqref{mtaylor}, 
		\begin{equation*}
			\mm(s,t)=\sum_{n=0}^\infty \,\frac{t^n}{n!}\,\pd_t^n \mm(s,0) \tag{\ref{mtaylor}}
		\end{equation*} 
		converges absolutely, and therefore is convergent in  $|t|<R(s_0)$ (we would like to emphasize here that $s_0$ is a fixed value) \emph{for fixed $s$ away from the cut}. This then defines an analytic function of $t$ for fixed $s$ in the cut $s-$plane. 
		
		Next we show that for fixed $t$, $|t|<R(s_0)$, $\mm(s,t)$ is analytic in the cut $s-$plane. To prove this,  we will use the following result
		
		\emph{ Consider the infinite series 
			$$
			\sum_n \, f_n(z),
			$$ where $\{f_n(z)\}$ are analytic functions in the domain $\Omega'\subset \mathbb{C}$. If the series converges uniformly on arbitrary compact subsets of the domain $\Omega'$ then the function to which it converges, i.e. $F(z)\equiv\sum_n \, f_n(z)$ is an analytic function on the domain  $\Omega'$. }

		To use the theorem we first show that the   series \eqref{mtaylor} converges \emph{uniformly} in an arbitrary \emph{compact} subset $\md$ of the cut $s$ plane. Since $\md$ is compact, it is both closed and bounded, i.e. there exists a \emph{finite} $S_0\in \mathbb{R}^+$ such that 
		\be 
		|s|\le \Sigma_0,\quad s\in\md.\ee
		Next we have that $\m(s,s_0)$ is finite so long the domain is away from the cut. Further we have, for $s\in\md$,  
		\be 
		\m(s,s_0)= \sup_{4m^2<s'<\infty}\left|\frac{s'-s_0}{s'-s}\right|\le \sup_{4m^2<s'<\infty}\left|\frac{s'-s_0}{s'-|s|}\right|. 
		\ee 
		Since $|s|\le \Sigma_0$, $\m(s,s_0)$ should be \emph{uniformly bounded} over $\md$, i.e. there must exist a finite quantity $\widetilde{\m}(\Sigma_0, s_0)$ such that 
		\be
		\m(s,s_0)\le \widetilde{\m}(\Sigma_0, s_0).
		\ee The exact form of $\widetilde{\m}(\Sigma_0, s_0)$ is not important This then proves the uniform convergence of the series \eqref{mtaylor}, for fixed $t,\,|t|<R(s_0)$ , over \emph{arbitrary} compact subset of the cut $s-$plane, $\md$. Further, each term of the series is analytic in the cut-plane. This implies that $\mm(s,t)$ is analytic in the cut $s$-plane for fixed $t$, $|t|<R(s_0)$.
		
		Collecting everything, we have that $\mm(s,t)$ is analyitc  for $t\in R(s_0)$ and \emph{fixed} $s\in \mathbb{C}\backslash\, [4m^2,\infty)$, and is analytic for $s\in\mathbb{C}\backslash\, [4m^2,\infty)$ and \emph{fixed} $t\in R(s_0)$.  Therefore, applying Hartog's fundamental theorem,  we infer that $\mm(s,t)$ is analytic in the topological product domain 
		\be 
		{\md}_M(R(s_0))=\{(s,t)\in\mathbb{C}^2\,:\, s\in\subset \mathbb{C}\,\backslash\,[4m^2,\infty),\,|t|<R(s_0)\},
		\ee  $s_0$ is some \emph{fixed real} value of $s$ below the cut, $s_0<4m^2$.
	\end{proof}

\end{subappendices}

%% file: GFTrev.tex
\chapter{  Tools from Geometric Function Theory : A walk-through }\label{GFTRev}
\section{Introduction}
In this chapter we will take a small mathematical detour into certain tools from the mathematical field of geometric function theory. We will employ these tools to analyze certain aspects of scattering amplitudes in following chapter.  One way to understand the analytic functions is to consider them as mapping of the complex plane. A complex function $w=f(z)$ can be geometrically viewed to be mapping region in complex $z-$plane to some region in complex $w-$plane, defined by $u=u(x,y)$ and $v=v(x,y)$, where $z=x+iy$ and $w=u+iv$. A natural question then arises as to what are the various geometric properties of this mapping. Geometric function theory addresses this question. In particular, geometric considerations give rise to various interesting bounds on analytic functions. We will discuss two kinds of functions, univalent and typically real functions. 

\section{Univalent and schlicht functions}
A function $f$ is defined to be \emph{univalent} on a domain\footnote{A domain is defined to be a connected open subset of the complex plane $\mathbb{C}$.} $D\subset\mathbb{C}$ if it is holomorphic\footnote{The definition of univalent function can be extended to consider  meromorphic functions as well. A meromorphic univalent function on a domain $D$ can have at most one simple pole there. See \cite{goluzin} for a discussion. Since we are only concerned with holomorphic function in the present work, we decide to include the requirement of holomorphicity in the definition of the univalent function. } and injective.   
The function is said to be \emph{locally univalent} at a point $z_0\in D$ if it is univalent in some neighbourhood of $z_0$. The function $f$ is locally univalent at some $z_0\in D$ if and only if $f'(z_0)\neq0$. It is to be emphasized that even if a function is locally univalent at \emph{each point} of $D$, it may fail to be univalent \emph{globally} on $D $. For example, consider the function $f(z)=e^{z}$ in the domain $\mathbb{D}_r:=\{z: |z|<r\}$ with $r>\pi$. Since $f'(z)=e^z\ne0$ at each point of the open disc above, it is locally univalent at each point of the disc. However, one has that $f(i\p)=f(-i\p)$. Thus $f(z)=e^z$ fails to be univalent \emph{globally} on the domain $D_r$ above.   From now on, by univalent functions, we will always mean globally univalent functions.\par 

We will be primarily concerned with the class $\mathcal{S}$ of univalent functions on the unit disc $\Delta\equiv\mathbb{D}_1=\{z:|z|<1\}$, normalized so that $f(0)=0$ and $f'(0)=1$. These functions are also called \emph{schlicht}\footnote{{In literature often, the terms univalent and schlicht are used interchangeably. Conway \cite{conway} reserves the term schlicht for describing univalent functions with the specific normalization introduced, thus defining the class of schlicht functions, $\ms$, as a  subclass of all the univalent functions. We follow this custom in the present work.}}$^{,}$\footnote{The word \enquote{schlicht} is German and means ``simple"! Simple here is more like plain rather than being easy.} functions. Thus each $f\in\mathcal{S}$ has a Taylor series representation of the form
\be \label{TaylorS}
f(z)=z+\sum_{p=2}^{\infty}b_p z^p,\qquad\qquad |z|<1.
\ee  
Note that a schlicht function $f(z)$ can always be obtained from an arbitrary univalent function defined on $\D$, $g(z)$, by an affine transformation with the definition
\be 
f(z):=\frac{g(z)-g(0)}{g'(0)}.
\ee The usefulness of a schlicht function is that its characteristic normalization makes various numerical estimates pertaining to it simpler compared to an arbitrary univalent function.\par 
The class $\mathcal{S}$ is preserved under a number of transformations. We will mention two of these transformations.
\begin{enumerate}
	\item[(i)] \textbf{Conjugation:} If $f(z)$ belongs to $\mathcal{S}$, so does 
	\be 
	g(z)=g(z^{*})^{*}=z+\sum_{p=2}^{\infty}b_p^{*} z^p.
	\ee 
	\item[(ii)]\textbf{Rotation:} The \emph{rotation of a function $f$} is defined by 
	\be 
	f_{\theta}(z):=e^{-i\th}f(e^{i\th}z), \qquad \qquad \th\in\mathbb{R}
	\ee If $f\in\mathcal{S}$ then $f_\th\in\mathcal{S}$ as well for every $\th\in\mathbb{R}$.
\end{enumerate}
\paragraph{Koebe Function:} 
The leading example of a \emph{schlicht} function is the \emph{Koebe function}
\be \label{Koebe}
k(z):=\frac{z}{(1-z)^2}=z+\sum_{p=2}^{\infty}p\, z^p.
\ee  
Koebe function and its rotations are often solutions to various extremal problems pertaining to \emph{schlicht} functions. Koebe function will play a central role in our subsequent analysis of scattering amplitudes.\par 
Now that we have given a brief overview of univalent functions and the subclass of schlicht functions thereof, we will discuss a few crucial theorems and results on the schlicht functions, which will play critical roles in our analysis of the crossing-symmetric dispersive representation of scattering amplitudes.
\subsection{Conditions for univalence of a function}
In the previous section, we laid down basic notions of univalent and schlicht functions.  We saw that the condition of non-vanishing first derivative of the function over a domain is not sufficient for the function to be globally univalent on the domain. However, the condition is a necessary one. In this section, we will discuss two important sufficient conditions and two necessary conditions for univalence, or equivalently schlichtness (with the specific normalization), of a function on the unit disc. 
\subsubsection{Grunsky inequalities}
Grunsky inequalities \cite{goluzin} are necessary {\it and} sufficient inequalities satisfied by a function $f$ to be a schlicht on the unit disc. 

Consider a holomorphic function $f:\D\to\mathbb{C}$  with the power series representation given by \eqref{TaylorS}, and let
\be 
\ln\frac{f(t)-f(z)}{t-z}=\sum_{j,k=0}^{\infty}\omega_{j,k}\,t^jz^k,
\ee 
with constant coefficients $\{\omega_{j,k}\}$. These are called Grunsky coefficients. It is straightforward to observe that $\w_{j,k}=\w_{k,j}$. Let us note the following simple lemma on the Grunsky coefficients which we will use in a later chapter. 
\begin{lemma}\label{mobius}
	Let $f :\D\to\mathbb{C}$ be a holomorphic function $h:\D\to\mathbb{C}$  with the power series representation given by \eqref{TaylorS}. Let $h:\D\to\mathbb{C}$ be a composition of a Mobius transformation with $f$, i.e.
	\be 
	h(z):=\frac{af(z)+b}{cf(z)+d},\quad ad-bc\ne 0,
	\ee
	and  $\{\widetilde{\w}_{j,k}\}$ be  Grunsky coefficients for $h(z)$ then
	\be \label{grunskymob}
	\widetilde{\w}_{j,k}=\w_{j,k},\quad\forall\,\, j,k\ge1.
	\ee 
\end{lemma}
\begin{proof}
	
\end{proof}

\begin{theorem}
	$f\in\ms $  if and only if the corresponding Grunsky coefficients satisfy the inequalities
	\be\label{eq:grunslyomega}
	\left|\sum_{j,k=1}^N \omega_{j,k}\lambda_j \lambda_k\right|\le\sum_{k=1}^N\frac{1}{k}\left|\lambda_k\right|^2
	\ee for every positive integer $N$ and all $\lambda_k$, $k=1,\dots,N$.\end{theorem}

As an example, the Grunsky coeffiecients of the Koebe function $k(z)$ are given by $\omega_{j,0}=\omega_{0,j}=2/j,\,\w_{j,k}=-\d_{j,k}/j$, with $\d_{j,k}$ being the usual Kronecker delta.

  From lemma \ref{mobius} above, it then readily follows that Mobius composition of a univalent function is again univalent. This is expected because Mobius transformations, being \emph{bi-holomorphisms} (bijective holomorphisms: both the Mobius transformation and its inverse are holomorphic mappings),  \emph{preserves the analytic structure} of any function on composition. 
{\paragraph{Logarithmic coefficients:} Note that the Grunsky coefficients $\{\w_{j,0}\}$ do not enter the inequality above. One wonders whether there exists any bounding relations satisfied by these Grunsky coefficients. Indeed they satisfy very important and interesting bounds. In the literature, these Grunsky coefficients are studied as \emph{logarithmic coefficients} because of the simple observation
	\be 
	\ln\frac{f(z)}{z}=\sum_{n=0}^\infty \w_{n,0}\,z^n,\quad f\in\mathcal{S}.
	\ee The \emph{logarithmic coefficients } $\{\g_n\}$ are defined by 
	\be 
	\g_n=\w_{n,0}/2.
	\ee These logarithmic coefficients satisfy interesting inequalities. They satisfy the celebrated de Branges's inequalities (previously Milin conjecture ) \cite{branges} which state that for   $f\in \mathcal{S}$ the corresponding logarithmic coefficients satisfy
	\be 
	\sum_{k=1}^n k(n-k+1)\,|\g_k|^2\le \sum_{k=1}^n \frac{n+1-k}{k},\qquad n=1, 2,\dots,
	\ee the equality is satisfied if and only if $f$ is a rotation of Koebe function. de Branges used this inequality in his  proof of the famed Bieberbach conjecture [see section \ref{debrangestheo} ]. There have been attempts  at obtaining sharp bounds on the individual coefficients $|\g_n|$. While the following sharp estimates have been obtained \cite{branges} 
	\begin{align}
		|\g_1|<1,\qquad |\g_2|\le \frac{1}{2}\left(1+2e^{-2}\right),
	\end{align} the problem of finding sharp upper bounds for $|\g_n|$ for $n\ge 3$ in general is still an open one. However, there are some sharp estimates for modulus of logarithmic coefficients in  some subclasses of $\ms$.}

\subsubsection{Nehari conditions}
In a  seminal work, Nehari \cite{Nehari} provided a necessary and a sufficient condition for the univalence of a function on the unit disc $\mathbb{D}$. The conditions are expressed in terms of Schwarzian derivative of the function. Schwarzian derivative of a function $f(z)$ w.r.t $z$ is defined by
\be\label{eq:sch}
\{f(z),z\}:=\left(\frac{f''(z)}{f'(z)}\right)'-\frac{1}{2}\left(\frac{f''(z)}{f'(z)}\right)^2.
\ee 
One advantage of these conditions are that these are independent of the normalization corresponding to the schlicht functions. Thus, these conditions work with univalent functions with generic power series 
\be 
g(z)=b_0+b_1z+b_2z^2+\dots\,.
\ee  This is to be contrasted with the Grunsky inequalities whose precise form requires the normalization of schlicht functions.\par 
The conditions can be stated as following theorems.
\begin{theorem}[{\bf Sufficient condition}]\label{NehariS}
	A function $g(z)$ holomorphic on the open disc $\mathbb{D}$ will be univalent if its Schwarzian derivative satisfies the inequality
	\be 
	\Abs{\{g(z);z\}}\le\frac{2}{(1-|z|^2)^2}.
	\ee 
\end{theorem} 	
\begin{theorem}[{\bf Necessary condition}]\label{NehariN}
	If a holomorphic function $g(z)$ is univalent in the open disc $\mathbb{D}$ then
	\be 
	\Abs{\{g(z);z\}}\le\frac{6}{(1-|z|^2)^2}.
	\ee 
\end{theorem}
\subsection{Koebe growth theorem}
A very important theorem that shines in our analysis of scattering amplitude is the \emph{Koebe Growth Theorem}. This theorem, essentially, puts upper and lower bound on the magnitude of a schlicht function $f\in\mathcal{S}$. 
\begin{theorem}\label{Kgrowth}
	If $f\in\mathcal{S}$ and $|z|<1$, then
	
	\be 
	\frac{|z|}{(1+|z|)^2}\le |f(z)|\le \frac{|z|}{(1-|z|)^2}.
	\ee 
	
	One of the equalities holds at some point $z\ne0$ if and only if $f$ is a rotation of the Koebe function.\end{theorem}
We want to emphasize that the bounds are the consequence of $f$ being univalent. Thus, the converse of the theorem need not be true, i.e. a function defined on the unit disc $\mathbb{D}$ with the normalization same as that of a schlicht function satisfying any one of the four bounding relations above need \emph{not} be univalent. 

\subsection{Partial sums of univalent functions: Szeg\"{o} theorem}
Consider a schlicht function $f(z)$ on the unit disc with the power series representation \eqref{TaylorS}. The $n$th partial-sum, or $n$th section, of the function $f$, denoted by $f_n$, is defined by 
\be \label{section}
f_n(z):=z+\sum_{k=2}^{n}b_k z^k. 
\ee  Now, the important question is what is the domain of univalence for the partial sum $f_n$. While that is, in general, a difficult question to answer, one can still ask as to what is the largest domain over which \emph{any section of an arbitrary} $f\in\ms$ is univalent? Szeg\"{o} \cite{szego} proved the following theorem in this aspect. See \cite[\S 8.2, pp. 241-246]{duren} for a proof.	
\begin{theorem}[{\bf Szeg\"{o} theorem}]\label{szegothm}
	Define the numbers $\left\{r_n\in\rb^{+}\right\}$ such that the $m$th section of a schlicht function $f\in\ms$, $f_m$, is univalent in the disc $\mathbb{D}_{r_n}$ for all $m\ge n$. Then,
	\be 
	r_1=\frac{1}{4},
	\ee i.e., each section remains univalent in the disc $\mathbb{D}_{1/4}$, and the number $1/4$ can't be replaced by a higher one.
\end{theorem}
{The statement of the number $1/4$ not being replaceable by a higher number needs some explanation. Consider an \emph{arbitrary} $f\in \ms$, and let the domain over which the $n$th section $f_n$ is univalent be $\mathcal{D}_{f_n}$. Then the above theorem tells that 
	\be 
	\mathcal{D}_{f_n}\supseteq \mathbb{D}_{\frac{1}{4}}\qquad \forall n\ge 2.
	\ee  Equivalently, this can be expressed as 
	\be 
	\bigcap_{\substack{f\in\ms\\n\in\mathbb{Z}^+,\,n\ge2}}\mathcal{D}_{f_n}=\mathbb{D}_{\frac{1}{4}}.
	\ee  
	The number $1/4$ is the best estimate because the domain of univalence for the second section of the Koebe function $k(z)$ is exactly equal to the disc of radius $1/4$, i.e.
	\be 
	\mathcal{D}_{k_2}=\mathbb{D}_{\frac{1}{4}}.
	\ee  Note that this theorem does not tell anything about the exact domains of univalence of sections of arbitrary schlicht functions. All this theorem tells us that whatever  the domain of univalence of a section of a schlicht function be, it is at least large enough to contain the disc $\mathbb{D}_{\frac{1}{4}}$, or equivalently, every section of any schlicht function is univalent on $\mathbb{D}_{\frac{1}{4}}$.}

{The evaluation of exact domains of univalence of sections of arbitrary schlicht function is still an open problem.}

\subsection{de Branges's theorem }\label{debrangestheo}
One of the most important properties of univalent functions is that its Taylor coefficients are bounded. For a schlicht function $f\in\mathcal{S}$, Bieberbach proved in 1916 \cite{bieberbach} that the second coefficient $b_2$ in the Taylor series representation \eqref{TaylorS} is bounded as 
\be 
|b_2|\le 2,
\ee with equality holding \emph{if and only if} $f$ is a rotation of the Koebe function.  
In the same work, Bieberbach conjectured the following bound for the general coefficient $b_n$:
\be \label{Bieber}
|b_n|\le n,\qquad \forall \,n\ge 2;
\ee  with the equality holding \emph{if and only if} $f$ is a rotation of the Koebe function. This conjecture came to be known as the famed \emph{Bieberbach Conjecture} and resisted a rigorous proof for about seven decades until Louis de Branges proved it in 1985 \cite{branges}, and the result came to be known as \emph{de Branges's Theorem.} For completeness, let us note down the full statement of de Branges's theorem.\par 
\begin{theorem}\label{bieber}
	Let $f$ be an arbitrary schlicht function,  $f\in \mathcal{S}$, with the power series representation defined by \eqref{TaylorS}. Then the Bieberbach conjecture holds true, i.e.
	\be 
	|b_n|\le n, \qquad\qquad\forall\,n\ge2;
	\ee with the equality holding if and only if $f$ is a rotation of the Koebe function $k(z)$ defined in \eqref{Koebe}, i.e. if and only if
	\be 
	f(z)=e^{-i\th}k(e^{i\th}z),\qquad\qquad \forall\,\th\in\mathbb{R}.
	\ee \end{theorem}
While de Branges proved the Bieberbach conjecture in its full generality only in 1985, various special cases have been proved earlier. During 1931-1933 by Dieudonn\'{e} \cite{dieudonne}, Rogosinski \cite{rogosinski}, and Sz\'{a}sz \cite{szasz} proved Bieberbach conjectures for univalent functions with \emph{all real} Taylor coefficients, also called typically real univalent functions, which we will give a brief description of in the next section.
\section{Typically real functions}\label{TRintro}
Another class of functions which fairs prominently in the field of geometric function theory, and will makes contact with scattering amplitudes as will be delineated in later chapter, is the so called \textbf{typically real functions}. These are also known as \textbf{Herglotz functions}. 

A function $f:\mathbb{C}\to\mathbb{C}$ is defined to be \emph{typically real} on a domain $\md\in\mathbb{C}$ containing segments of real axis, if it  is real on these segments and satisfies
\be 
\text{Im}\,[f(z)]\,\text{Im}[z]> 0,\quad \text{Im}[z]\ne 0,\,\,z\in\md.
\ee From the definition, it is clear that a typically real function can only be real on the real axis.

 In this thesis, we will be interested in a  particular class of typically real functions which we denote by $TR$. A function $f(z)\in TR$ is analytic and typically real in an open disc $\D$ and admits a Taylor series of the form 
 \be \label{TaylorTR}
 f(z)=z+\sum_{n=2}^\infty c_n \,z^n.
 \ee  It follows from the definition that the  coefficients $\{c_n\}$ are real.

The Taylor series representation bears some resemblance with that of the schlicht functions, \eqref{TaylorS}. However a $TR$ function is not univalent in general. As an example, consider the function 
\be 
g(z)=\frac{(1+z^2)z}{(1-z^2)^2}=\sum_{n=0}^\infty (2n+1) z^{2n+1},\qquad z\in \D.
\ee While the function is holomorphic in $\D$, it is not univalent there. This can be seen by simply noting that $f'(z)=0$ at $z=\pm i\left(-1+\sqrt{2}\right)\in \D$. However, a function $F\in \ms$ can be typically real. 
\begin{lemma}
	A function $f\in \ms$, holomorphic and injective in $\D$, with all \emph{real} Taylor coefficients is as TR function. We shall call the class of such functions $TR_U$.
\end{lemma} \begin{proof}
Consider a function $f\in\ms$ with Taylor series representation \eqref{TaylorS}
\be 
f(z)=z+\sum_{p=2}^\infty b_p\,z^p,\qquad b_p\in\mathbb{R}\,\forall p\ge2.
\ee 

Then, clearly, 
\be 
f(z)^*=z^*+\sum_{p=2}^\infty b_p^*\,(z^*)^p=z^*+\sum_{p=2}^\infty b_p\,(z^*)^p=f(z^*).
\ee Due to injectivity of $f$, $\text{Im.}[f(z)]:=(f(z)-f(z)^*)/2i$ can only vanish on the real axis, i.e. $f(z)$ can be real on $\mathbb{R}$. Further  $f(z)\sim z$ as $z\to 0$,\footnote{The statement $f(z)\sim g(z)$ as $z\to z_0$ is defined as 
\be \lim_{z\to z_0}\frac{f(z)}{g(z)}=1.\ee} and therefore there exists a neighborhood of $z=0$ where $\text{Im.}[f(z)]\text{Im.}[z]>0$. Coupled with the fact that $f(z)$ can be real only on real axis and holomorphicity of $f$ guarantees $\text{Im.}[f(z)]\text{Im.}[z]>0$ in $\D$.
\end{proof}

\subsection{Characterization of $TR$ functions} One of the crucial question is how can we identify a $TR$ function? Following is a characterization theorem which is quite useful in identifying a $TR$ function.
\begin{theorem}
	Let $f:\D\to \mathbb{C}$ be a holomorphic function with a convergent Taylor series expansion 
	\be 
	f(z)= z+\sum_{n=2}^{\infty} c_n z^n,\qquad z\in\D,\quad c_n\in\mathbb{R},\,\,\forall n\ge 2.
	\ee The following statements are equivalent
	\begin{enumerate} 
	\item[(i)] $f(z)\in TR$.
	\item [(ii)] \underline{ Carath\'{e}odory class:} The function 
	\be 
	\f(z):=(1-z^2)\frac{f(z)}{z}=1+a_2z+\sum _{n=2}^\infty (a_{n+1}-a_{n-1})
	\ee has positive real part in $\D$. The class of functions with positive real part in the open unit disc is called Carath\'{e}odory class.
	\item[(iii)] \underline{Robertson integral:} There exists a non-decreasing function $\m(x)$ on $\x\in[-1,1]$ satisfying $\m(1)-\m(-1)=1$ such that $f(z)$ admits the following Stieltjes integral representation 
	\be \label{RObertson1}
	f(z)=\int_{-1}^{1} d\m(\x)\,\frac{z}{1-2\x z+z^2}.
	\ee This integral representation is also known as Robertson integral .
	\end{enumerate}
\end{theorem}
 The Carath\'{e}odory functions are particularly useful in identifying whether a function is $TR$ or not. Therefore we state the following results for the Carath\'eodory class. Let $\f(z)$ be a Caratheodory function, i.e. $\f(z)$ is holomorphic in open unit disc $\D$ with $\text{Re.}[\f(z)]>0$ in $\D$, and it has a Taylor series representation 
 \be 
 \f(z)=1+\sum_{n=1}^{\infty}\nu_n\,z^n.
 \ee 
Then the following are true:
\begin{itemize}
	\item [(i)] $|\n_n|<2$ for all $n\ge 1$.
	\item[(ii)] $\f(z)$ has the integral representation called the \emph{Herglotz} representation
	\be 
	\f(z)=\int_0^{2\p} \frac{e^{it}+z}{e^{it}-z}\,d\m(t),\qquad z\in\D,
	\ee where $d\m(t)$ is a positive measure on $t\in[0,2\p]$.
\end{itemize}
The condition (i) above helps one to do a quick check if a function can be $TR$. For example, $f(z)=\frac{z}{(1-z^2)^3}\not\in TR$  because $\f(z)=\frac{(1-z^2)}{z}f(z)=1+2z^2+3z^4+\dots$ is not in the Carath\'{e}odory class.

\subsection{Biberbach-Rogosinski bounds}
The Bieberbach bounds were proved for $TR_U$ functions as noted in section \ref{debrangestheo}. One wonders whether such bounds exists for general typically real functions. The answer to this question is yes! Rogosinski proved such bounds for $TR$ functions. We will call these bounds as Bieberbach-Rogosinski bounds. 
\begin{theorem}
	Let $f(z)$ be a $TR$ function with Taylor series representation of  \eqref{TaylorTR} then 
	\be \label{BR bound}
	-\k_n\le c_n\le n,
	\ee with 
	\be 
	\kappa_n = n\quad \text{for even}~ n, \qquad \kappa_n = \frac{\sin n~\vartheta_n}{\sin \vartheta_n} \quad\text{for odd}~ n \,,
	\ee 
	where $\vartheta_n$ is the smallest solution of $\tan n \vartheta = n \tan \vartheta$ located in $( \frac{\pi}{n},~\frac{3 \pi}{2n} )$ for $n>3$ and $\kappa_3=1$.
\end{theorem}


\newpage
\vspace{2 cm}

%% file: bieberbach.tex
\chapter{  Quantum field theory and the Bieberbach conjecture }\label{biberbach}
\section{Introduction}
In the previous chapter, while introducing the study of tools from geometric function theory, we had teased that these tools will be used to study scattering amplitudes. We take up the task in this chapter. A crossing symmetric representation of the 2-2 scattering of identical massive scalars, first used in the long-forgotten work by Auberson and Khuri \cite{AK} and resurrected in \cite{ASAZ, RGASAZ}, will be central in our venture. If we consider $\mt(s,t)$ and assume that there are no massless exchanges, then we expect a low energy expansion of the form
\be
\mt(s,t)=\sum_{p\geq 0,q\geq0} {\mathcal W}_{pq} x^p y^q\,,
\ee
where $x$ and $y$ are the quadratic and cubic crossing symmetric combinations of $s,t,u$ to be made precise below. 
Normally, the term dispersion relation for scattering amplitude $\mt(s,t)$ refers to the fixed-$t$ dispersion relation, lacking manifest crossing symmetry. As we will review below, in order to exhibit three-channel crossing symmetry in the dispersion relation, we should work with a different set of variables, $z$ and $a\equiv y/x$. For now, we note that both $x,y \propto z^3/(z^3-1)^2$. As such, the appropriate variable is $\tilde z=z^3$. We write a crossing symmetric dispersion relation in the variable $\zt$ keeping $a$ fixed. This dispersion relation, together with unitarity, leads to bounds on the Taylor coefficient of the $\tilde{z}$ expansion of the amplitude, akin to, and using, the Bieberbach bound \eqref{Bieber} discussed in the previous chapter. The expansion of $\mt(\zt,a)$ around $\zt=0$ is similar to a low energy expansion, and the mentioned bound on the Taylor coefficients of the amplitude relates to bounds on the Wilson coefficients $\mW_{pq}$. Further, we also obtain $2-$sided bounds on the amplitude itself, which follows from the Koebe growth theorem \ref{Kgrowth}.

Recently, the existence of upper and lower bounds on ratios of Wilson coefficients have been discovered \cite{tolley, Caron-Huot:2020cmc, ASAZ}. These bounds are remarkable since they convey that Wilson coefficients cannot be arbitrarily big or small and, in a sense, corroborate the efficacy of effective field theories. 
One of the interesting outcomes of our analysis is that 
\be
-\frac{9}{4\m+6\d_0}<\frac{\mathcal{W}_{0,1}}{{\mW_{1,0}}}< \frac{9}{2\m+3\d_0}\,,
\ee
where $\mu=4m^2$ with $m$ being the mass of the external scalar and $\d_0$ is some cutoff scale in the theory. Expanding $\d_0\gg \mu$ leads to the 2-sided, space-time dimension independent bound $-\frac{3}{2\delta_0}<\frac{\mathcal{W}_{0,1}}{W_{1,0}}< \frac{3}{\d_0}$. Compared to \cite{tolley, Caron-Huot:2020cmc}, the lower bound is identical but the upper bound we quote above is stronger. We have checked this inequality for several known examples. 

We further explore the Grunsky inequalities in the context of the scattering amplitudes  If these were to hold in QFT, they would imply non-linear constraints on $\mW_{pq}$. Unfortunately, proving univalence is a tough problem. The $\tilde z$ dependence in the crossing symmetric dispersion relation arises entirely from the crossing symmetric kernel. One can show that this kernel, for a range of real $a$ values, is indeed univalent! Therefore, one concludes that for unitary theories, the amplitude is a convex sum of univalent functions. However, a complete classification of circumstances as to when a convex sum of univalent functions leads to a univalent function does not appear to be known in the mathematics literature.
We will show that as an expansion around $a\sim 0$, the Grunsky inequalities hold as the resulting inequality on $\mW$ is known to hold using either fixed-$t$ or crossing symmetric dispersion relation. Therefore, at least around $a\sim 0$, it is indeed true that the amplitude, and not just the kernel in the crossing symmetric dispersion relation, is univalent. Our numerical checks for known S-matrices, such as 1-loop $\phi^4$, $\pi^0\pi^0\rightarrow \pi^0\pi^0$ arising from the S-matrix bootstrap, tree-level string theory, suggest that there is always a finite region near $a\sim 0$ where univalence holds. Thus we conjecture that we can impose univalence on the amplitude even beyond the leading order in $a$. This gives rise to a non-linear inequality for the Wilson coefficients; as a sanity check, this inequality is satisfied for all the cases studied in this paper.

The discussion above may seem to suggest that we may need to know the full amplitude in QFT to check for this seemingly magical property of univalence. Fortunately, this is not the case. In QFT, we would like to work in an effective field theory framework where we have access to certain derivative order in the low energy expansion. This can be tackled using Szego's theorem, \ref{szegothm}.  We will use this theorem to rule out situations where univalence fails.

Let us now lay out the organization of the chapter. We start with an introduction to the crossing symmetric dispersion relation and associated structures in section \ref{crossreview}. Following this, we discuss the bounds on the Taylor coefficients of the scattering amplitude in section \ref{abound} the physical implications of which for Wilson coefficients is discussed in section \ref{wilsonbound}. Next, we derive two-sided bounds on the scattering amplitude in section \ref{ampbound}. In section \ref{EFTunival}, we look for hints of univalence in EFT amplitudes with the aid of Szeg\"{o}'s theorem followed by an exploration of Grunsky inequalities for amplitudes in section \ref{Grunsky}. Finally, we summarize our findings in \ref{discuss}. Various explorations associated with the main text providing the analysis with wholesomeness have been placed in the appendices.
\section{An introduction to crossing symmetric dispersion relation (CSDR)}\label{crossreview}
We will begin our QFT discussion by reviewing key aspects of crossing symmetric dispersion relations.
We are considering scattering amplitudes of $2-2$ identical scalars. These are functions of Mandelstam invariants $s,~t,~u$, constrained by $s+t+u=4m^2=\m$. For our convenience, we will work with shifted variables $s_1=s-\frac{\m}{3},~s_2=t-\frac{\m}{3},~s_3=u-\frac{\m}{3}$. As delineated in \ref{crossingdef}, $\mt(s_1,\, s_2,\, s_3)$ is $S_3$ invariant which means that  $\mathcal{T}(s_1,s_2)=\mathcal{T}(s_2,s_3)=\mathcal{T}(s_3,s_1)\,.$ Scattering amplitudes have physical {branch} cuts for $s_k\geq \frac{2\m}{3}$. To write down a crossing symmetric dispersion relation the most useful trick is to parametrize the $s_1,s_2$ as \cite{AK,ASAZ}
\be\label{eq:skdef}
s_{k}=a-\frac{a\left(z-z_{k}\right)^{3}}{z^{3}-1}, \quad k=1,2,3
\ee
with $a$ being real,  and $z_k$'s are cube roots of unity. The $z,a$ are crossing symmetric variables. They are related to the crossing symmetric combinations of $s_1,s_2$, namely
$x=-\left(s_{1} s_{2}+s_{2} s_{3}+s_{3} s_{1}\right)=\frac{-27 a^2z^3}{(z^3-1)^2},~~ y=-s_{1} s_{2} s_{3}=\frac{-27 a^3z^3}{(z^3-1)^2},~~a=y/x\,.$ 

\paragraph{} Fully crossing symmetric amplitude can be expanded like 
\be\label{eq:Wpqdef}
\mathcal{T}(s_1,s_2)=\sum_{p=0,q=0}^{\infty}\mathcal{W}_{p,q} x^p y^q\,.
\ee
The parametrization in \eqref{eq:skdef}, maps the physical cuts $s_k\geq \frac{2\m}{3}$ in a unit circle in $z$-plane, see figure \eqref{fig:Fig1z} for $-\frac{2\m}{9}< a<0$. { If $a<-2\mu/9$, then as \cite{AK} show, there will be branch cuts on the real $z$ axis.} In the transformed variables the amplitude becomes a function of $z,a$, namely $\mathcal{M}(z,a)$. The most usefulness of \eqref{eq:skdef} is that it enables us to write a {dispersion relation} which is  manifestly crossing symmetric, see \cite{AK,ASAZ}.
\be\label{crossdisp}
\mathcal{T}(\tilde{z},a)=\alpha_{0}+\frac{1}{\pi} \int_{\frac{2 \mu}{3}}^{\infty} \frac{d s_{1}^{\prime}}{s_{1}^{\prime}} \mathcal{A}\left(s_{1}^{\prime} ; s_{2}^{(+)}\left(s_{1}^{\prime}, a\right)\right) H\left(s_{1}^{\prime} , \tilde{z}\right)
\ee
where $\mathcal{A}(s_1;s_2)$ is the s-channel discontinuity (discontinuity of the amplitude cross $s_1\geq \frac{2\m}{3}$), $\a_0=\mathcal{T}(z=0,a)$ is the subtraction constant independent of $a$, and 
\be\label{Hs2def}
\begin{split}
	&H(s_1',\tilde{z})=\frac{27 a^2 \tilde{z} \left(2 s_1'-3 a\right)}{27 a^3 \tilde{z}-27 a^2 \tilde{z} s_1'-\left(1-\zt\right)^2 \left(s_1'\right)^3},\\
	&s_{2}^{(+)}\left(s_{1}^{\prime}, a\right)=-\frac{s_{1}^{\prime}}{2}\left[1-\left(\frac{s_{1}^{\prime}+3 a}{s_{1}^{\prime}-a}\right)^{1 / 2}\right],
\end{split}
\ee where we have introduced the new variable $\zt:=z^3$. This is because all the manifestly crossing symmetric functions are functions of $\zt=z^3$. 

Let us expound a bit on the analyticity structure of the amplitude on the complex $\zt-$plane. The figure \eqref{fig:Fig1zcube} below shows the image of the physical cuts in $\tilde{z}=z^3$-plane (for $-\frac{2\m}{9}<a<0$).
Notice that in $\tilde{z}=z^3$ plane, the images of the physical cuts in all three channels are same.  We will focus on the situation where $a$ is real and note that $|\tilde z|=1$ if $s_1, s_2$ are real {(we will set $\m=4m^2=4$ here)}. In the $\tilde z$ plane, the forward limit ($s_2=-4/3, s_1\geq 8/3$) corresponds to arcs that start at $\tilde z=-1$ and approaching $\tilde z=1$ along $|\tilde z|=1$. If $s_2>-4/3$ then the full boundary of the disc is not traversed while if $s_2\leq -4/3$ then the full boundary is traversed\footnote{There are two trajectories corresponding to the two roots of $\tilde z$ which are obtained on starting with $x,y$ in terms of $\tilde z, a$ and solving for the latter in terms of $s_1, s_2$. If $s_2>-4/3$ then the starting point is on the circle away from $\tilde z=-1$. As $s_1$ increases from $8/3$ the trajectory reaches $\tilde z=-1$ and then retraces along the boundary till it reaches $\tilde z=1$.}. A further important point to keep in mind is that since real $s_1,s_2$ correspond to $|\tilde z|=1$, to access the inside of the disc we need to consider complex $s_1,s_2$. Since later on, we will keep $a$ real, a complex $s_1$ will give us a complex $s_2$ since $a=s_1 s_2(s_1+s_2)/(s_1^2+s_1 s_2 +s_2^2)$. Plugging back into $\tilde z$, we get two values, one which lies inside the disc and the other which lies outside. 

\begin{figure}[!htb]
	\centering
	\begin{subfigure}[b]{0.45\textwidth}
		\centering
		\includegraphics[width=\textwidth]{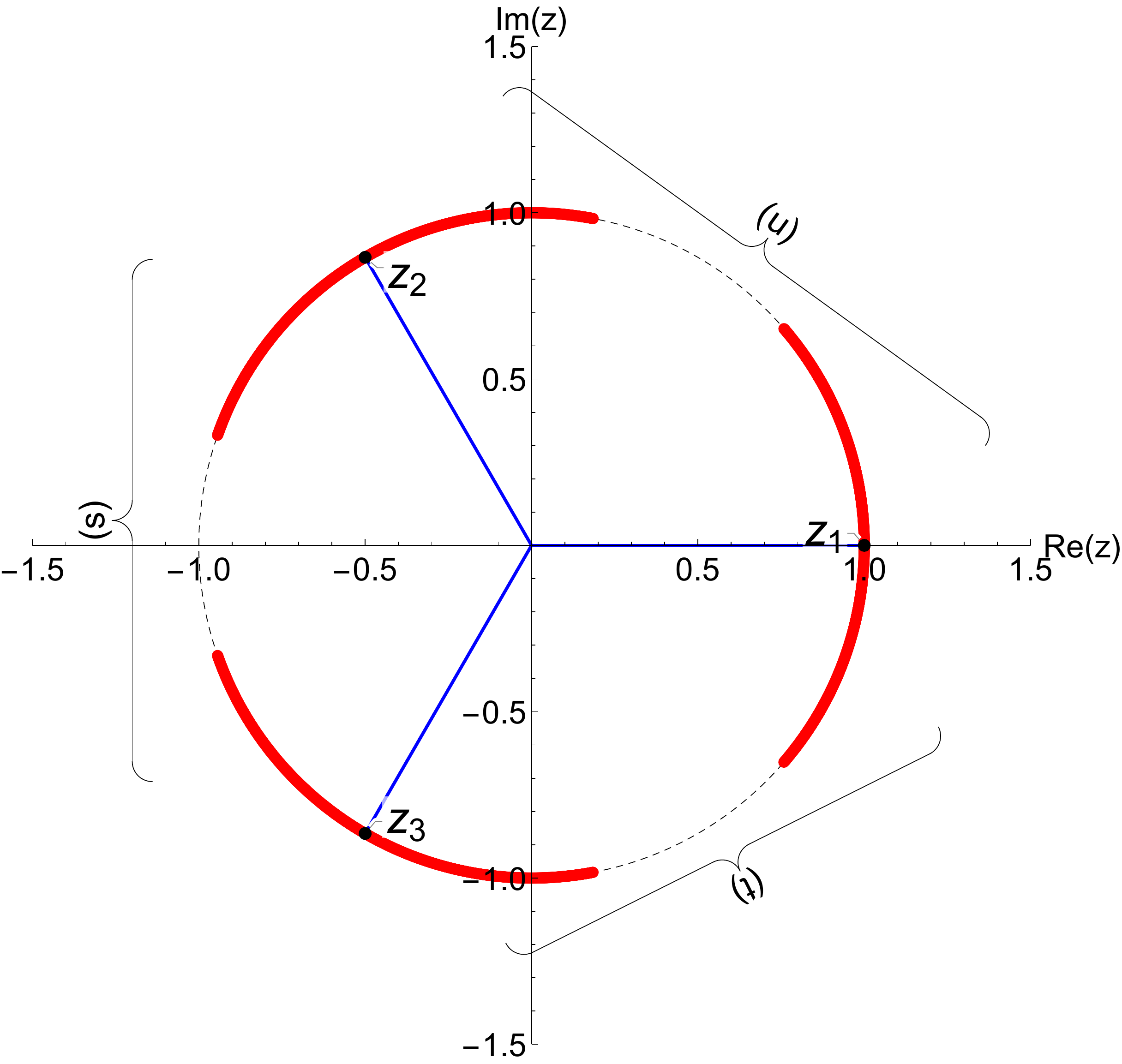}
		\caption{$z$ plane}
		\label{fig:Fig1z}
	\end{subfigure}
	\hfill
	\begin{subfigure}[b]{0.45\textwidth}
		\centering
		\includegraphics[width=\textwidth]{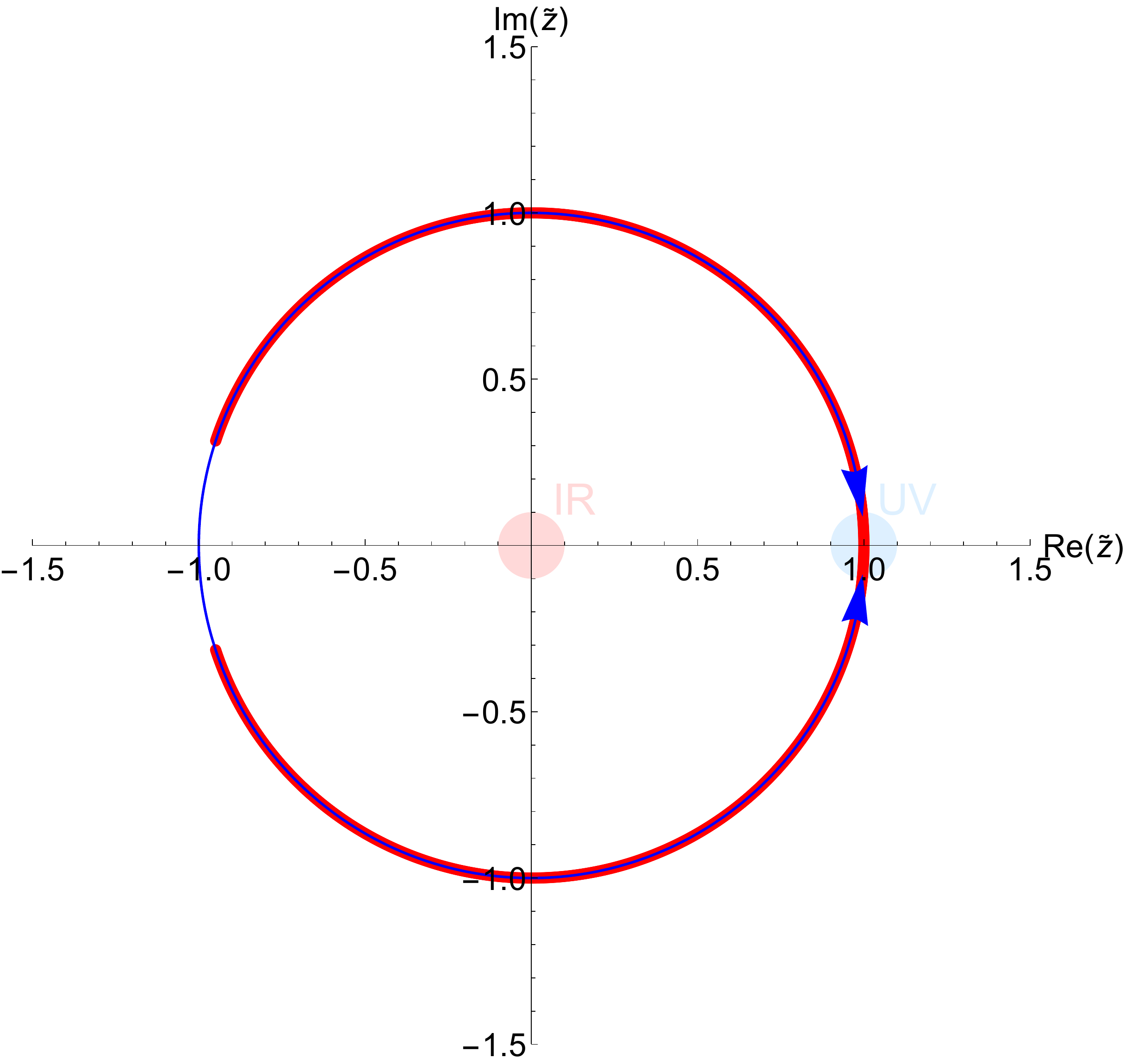}
		\caption{$\tilde{z}=z^3$ plane}
		\label{fig:Fig1zcube}
	\end{subfigure}
	\caption{Image of the physical cuts. The blue line on $\tilde z=1$ indicates the forward limit $s_2=-4/3, s_1\geq 8/3$. The two trajectories start from $\tilde z=-1$ and as $s_1$ increases they approach $\tilde z=1$.}
	\label{fig:Imageofcuts}
\end{figure}

The scattering amplitude $\mm$ admits a power series expansion about $\zt=0$ converging in the unit disc $|\zt|<1$,
\be\label{Mpower}
\mathcal{T}(\tilde{z},a)=\sum_{n=0}^{\infty} \a_n(a) a^{2n}\, \tilde{z}^n\,.
\ee
For a local theory\footnote{This follows from the expansion in \eqref{eq:Wpqdef}, see \cite{AK,ASAZ}
	\be
	\mathcal{T}(z,a)=\sum_{n=0}^{\infty} \bar{w}_n(a) x^n\,.
	\ee
}, $\a_n(a) a^{2n}$ can be  a polynomial in $a$ of order at most $3n$. It can be seen from the expression 
\be\label{eq:alphatoW}
\a_p(a)a^{2p}=\sum_{n=0}^{p}\sum_{m=0}^{n} \mathcal{W}_{n-m,m}a^m (-1)^{p-n}(-27)^n~a^{2n}\binom{-2 n}{p-n}.
\ee This expression also implies that $\a_n(a)$ is in general a Laurent polynomial\footnote{A Laurent polynomial $\ell(x)$ over a field $\mathbb{F}$ is an expression of the form
	\be 
	\ell(x)=\sum_{k\in\mathbb{Z}}\d_k\,x^k,\qquad \d_k\in\mathbb{F}
	\ee   where now $k$ need not be necessarily positive and only finitely many coefficients $\d_k$ are non-zero.  }.
Similarly, the crossing symmetric kernel $H(s_1',\zt)$ admits a power series expansion abut $\zt=0$:
\be\label{Hexp}
H(s_1',\tilde{z})=\sum_{n=0}^{\infty}\beta _n\left(a,s_1'\right) \tilde{z}^n\,,
\ee
with
\be
\begin{split}
	\beta _n\left(a,s_1'\right)&=\frac{3 \sqrt{3} a 2^{-n} \left(s_1'\right)^{-3 n} }{\sqrt{a-s_1'} \sqrt{3 a+s_1'}}\Bigg[\left(27 a^3-27 a^2 s_1'+3 \sqrt{3} a \sqrt{a-s_1'} \sqrt{3 a+s_1'} \left(3 a-2 s_1'\right)+2 \left(s_1'\right){}^3\right)^n\\
	&-\left(27 a^3-27 a^2 s_1'+3 \sqrt{3} a \sqrt{a-s_1'} \sqrt{3 a+s_1'} \left(2 s_1'-3 a\right)+2 \left(s_1'\right){}^3\right)^n\Bigg].
\end{split}
\ee
We can make two immediate and important observations from the above expression:
\begin{enumerate}
	\item[(I)] \be \b_0(a,s_1)=0\quad \text{identically}.\ee
	\item[(II)]
	\be 
	\b_1(a,s_1)=\frac{27a^2}{s_1^3}\left(3a-2s_1\right).
	\ee Now, recall that, the {analytic}\footnote{ {We are calling the above domain of $a$ to be ``analytic'' since for this domain of $a$, the branch cuts in the complex $\tilde{z}$ plane do not lie along the real line. }  } domains for $a$ and $s_1$ are given by $[-2\m/9,2\m/3)$ and $[2\m/3,\infty)$. Then, one can readily infer that $\b_1(a,s_1)$ non-vanishing for the entire physical range of $s_1$ \emph{if and only if}\footnote{The $a=0$ point is trivial since both $x,y=0$ and the amplitude is a constant. In what follows, if on occasion we imply $a\ne0$, this is to be kept in mind.}
	\be 
	a\in\left(-\frac{2\m}{9},0\right)\cup\left(0,\frac{4\m}{9}\right).
	\ee Further, in this domain of $a$, $\b_1<0$ for the entire physical domain of $s_1$. This sign of $\b_1$ will play a crucial role for various proofs in the following analysis. For the string amplitude, that we will frequently consider, $\mu=0$. We have subtracted the massless pole, and the lower limit of the dispersion integral starts at $s_1'=1$, which is the location of the first massive string pole. This effectively leads to the replacement $\mu\rightarrow 3/2$ in the above discussion given $a\in \left[-\frac{1}{3},0\right)\cup \left(0,\frac{2}{3}\right)$.
\end{enumerate}
The coefficients $\{\b_n(a,s_1)\}$ are of extreme importance because using the crossing symmetric dispersion  relation \eqref{crossdisp}  along with \eqref{Hexp}, we can write an inversion formula 
\be\label{eq:alphadisper}
a^{2n}\a_n(a)=\frac{1}{\pi }\int _{\frac{2 \mu }{3}}^{\infty }\frac{ds_1'}{s_1'}\,\mathcal{A}\left(s_1';s_2^{(+)}(s_1',a)\right)\,\beta _n\left(a,s_1'\right),\quad n>0.
\ee Thus, we see that the $\a_n$s are essentially integral transforms of $\b_n$s convoluted with the $s-$channel absorptive part $\imm$. 

Let us conclude this section with a significant result on the absorptive part, which will be crucial for our subsequent analysis in light of the inversion formula above.  
\begin{lemma}[{\bf Positivity lemma}]\label{Apos}
	For a \emph{unitary} theory, if $a\in\left(-\frac{2\m}{9},\frac{2\m}{3}\right)$ then the absorptive part of the amplitude, $\imm$, is {\emph{non-negative}} for $s_1\in\left[\frac{2\m}{3},\infty\right)$.
\end{lemma}
\begin{proof}[{\bf Proof.}]
	The $s$-channel discontinuity has a partial wave expansion
	\be	
	\begin{aligned}
		\imm &=\Phi\left(s_{1} ; \alpha\right) \sum_{\ell=0}^{\infty}(2 \ell+2 \alpha) a_{\ell}\left(s_{1}\right) C_{\ell}^{(\alpha)}\left(\sqrt{\xi\left(s_{1}, a\right)}\right) \\
		\xi\left(s_{1}, a\right) &=\cos ^{2} \theta_{s}=\left(1+\frac{2 s_{2}^{+}\left(s_{1}, a\right)+\frac{2 \mu}{3}}{s_{1}-\frac{2 \mu}{3}}\right)^{2}=\xi_{0}+4 \xi_{0}\left(\frac{a}{s_{1}-a}\right)
	\end{aligned}
	\ee
	with $\xi_{0}=\frac{s_{1}^{2}}{\left(s_{1}-2 \mu / 3\right)^{2}}$ {and $\a=\frac{d-3}{2}$}. Over the domain of $s_1\in\left[\frac{2\m}{3},\infty\right)$, we find that, $\sqrt{\xi\left(s_{1}, a\right)}\geq 1$, which implies $C_{\ell}^{(\alpha)}\left(\sqrt{\xi\left(s_{1}, a\right)}\right)> 0$ if $a\in\left(-\frac{2\m}{9},\frac{2\m}{3}\right)$.
	{Next, for the given domains of $s_1$ and $a$ one has  
		$
		\Re\left[s_2^+(s_1,a)\right]\in\left[-\frac{\m}{3},\frac{2\m}{3}\right].
		$ 
		Now the analyticity domain $E(s_1)$ of $\ma(s_1,s_2^+)$ in $t$ has been determined \cite{martin-II} to be
		\be
		E(s_1)=\begin{cases}
			E\left(0,\frac{2\m}{3}-s_1\,\Bigg|\,4\m+\frac{48\m}{3s_1-2\m}\right),\qquad \frac{2\m}{3}< s_1< \frac{11\m}{3}\\
			\\
			E\left(0,\frac{2\m}{3}-s_1\,\Bigg|\,\frac{192\m}{3s_1+\m}\right),\qquad \frac{11\m}{3}< s_1< \frac{23\m}{3}\\
			\\
			E\left(0,\frac{2\m}{3}-s_1\,\Bigg|\,\m+\frac{48\m}{3s_1-11\m}\right),\qquad  s_1> \frac{23\m}{3}\\
		\end{cases}
		\ee 	where $E(f_1,f_2|d)$ stands for an ellipse with foci at $s_2^+=f_1,\,s_2^+=f_2$ and right extremity at $s_2^+=d$. It is straightforward to see that our $s_2^+$ values always lie in the interior of $E(s_1)$, i.e. the partial wave expansion for $\imm$ above converges for the given domains of $a$ and $s_1$.
	}  Next,  $ a_{\ell}(s_1)\geq0$ for $s_1\in\left[\frac{2\m}{3},\infty\right)$ as a consequence of unitarity. Therefore, if $a\in\left(-\frac{2\m}{9},\frac{2\m}{3}\right)$, $\imm$, is \emph{non-negative} for $s_1\in\left[\frac{2\m}{3},\infty\right)\,.$ If $\mu=0$,  where $\xi=1+4\frac{a}{a-s_1}$, for $\xi>1$, $a>0$ must hold\footnote{For the string case, however, we will find that for $a<0$ the bounds we will consider will still hold. We do not have a general explanation for this apart from observing that $\alpha_1<0$ for certain $-1/3<a<2/3$, which is the range of $a$ we will be interested in.}.
\end{proof}
\section{Bounds on $\{\a_n(a)\}$}\label{abound} 

We can bound the Taylor coefficients $\{\a_n(a)\}$ appearing in the power-series representation of the scattering amplitude $\mt$, \eqref{Mpower}. Towards that end, let us first prove a lemma that will be used repeatedly in our analysis that follows.
\begin{lemma}\label{Hschlicht}
	Consider the kernel $H(\tilde{z};s_1,a)$ of the dispersion relation given by \eqref{Hs2def},
	\be 
	H(\tilde{z};s_1,a)=\frac{27 a^2 \tilde{z} \left(2 s_1-3 a\right)}{27 a^3 \tilde{z}-27 a^2 \tilde{z} s_1-\left(\tilde{z}-1\right)^2 \left(s_1\right)^3}.
	\ee  Define the function
	\be \label{Fdef}
	F(\tilde{z};s_1,a):=\frac{H(\tilde{z};s_1,a)}{\b_1(a,s_1)}.
	\ee For $a\in\left(-\frac{2\m}{9},0\right)\cup\left(0,\frac{4\m}{9}\right)$ and $s_1\in\left[\frac{2\m}{3},\infty\right)$, $F(\tilde{z};s_1,a)$ is a schlicht function, or equivalently, $H(\tilde{z};s_1,a)$ is a \emph{univalent function on the unit disc $|\zt|<1$}.
\end{lemma}
\begin{proof}[{\bf Proof }]
	For $a\in\left(-\frac{2\m}{9},0\right)\cup\left(0,\frac{4\m}{9}\right)$ and $s_1\in\left[\frac{2\m}{3},\infty\right)$, we have already proved that  $\b_1(a,s_1)\ne 0$. Thus, the function $F$ is well-defined in these domains of $a$ and $s_1$. Further, since $\b_0=0$ identically, $F(\tilde{z})$ admits a power series expansion about $\tilde{z}=0 $:
	\be \label{Fpower}
	F(\tilde{z};s_1,a)=\tilde{z}+\sum_{n=2}^{\infty}\frac{\b_n(a,s_1)}{\b_1(a,s_1)}\tilde{z}^n.
	\ee 
	First note that $$F(\tilde z; s_1,a)=\frac{\tilde z}{1+\gamma \tilde z+\tilde z^2}$$ with $\gamma=27(\frac{a}{s_1})^2(1-\frac{a}{s_1})-2$. To avoid a singularity inside the unit disc we need $$|\gamma|<2\,,$$ which translates to $a\in\left(-\frac{2\m}{9},0\right)\cup\left(0,\frac{4\m}{9}\right)$ for real $a$, which is the same condition mentioned above\footnote{Of course we could consider complex $a$ as well at this stage. However, since we want to make use of the positivity of the absorptive part of the amplitude later on, we will restrict our attention to real $a$.}. 
	
	Next, observe that we can write
	\be 
	F(\zt;s_1,a)=k(\zt)\left[1-\frac{27a^2(a-s_1)}{s_1^3}k(\zt)\right]^{-1},
	\ee where $k(\zt)$ is the Koebe function defined in \eqref{Koebe}. It is straightforward to see that $F$ can be considered as a composition of a Moebius transformation with the Koebe function. Then, by \eqref{grunskymob}, $F$ has the Grunsky coefficients
	\be 
	\w_{p,q}=-\frac{\d_{p,q}}{p},\quad p,q\ge1,
	\ee and these satisfy the Grunsky inequalities of \eqref{eq:grunslyomega} for all $N\ge1$. Since Grunsky inequalities are necessary and sufficient for an analytic function inside the unit disc to be univalent, this completes the proof.
\end{proof}
Observe that, $F(\tilde{z};s_1,a)$ and $H(\tilde{z};s_1,a)$ are related by an affine transformation. Thus, the schlichtness of $F$ implies that $H$ \emph{is an univalent function on the unit disc $|\tilde{z}|<1$} for the same domains of $a$ and $s_1$. 
\begin{corollary}\label{bnbound}
	For $a\in\left(-\frac{2\m}{9},0\right)\cup\left(0,\frac{4\m}{9}\right)$ and $s_1\in\left[\frac{2\m}{3},\infty\right)$, the Taylor coefficients $\{\b_n(a,s_1)\}$ in the power-series expansion of $H$ are bounded by
	\be 
	\left|\frac{\beta_n(a,s_1)}{\beta_1(a,s_1)}\right|\le n,\quad n\ge 2
	\ee 
\end{corollary}	
\begin{proof} Since $ F(\tilde{z};s_1,a)$ is a Schllicht function on the unit disc $|\tilde{z}|<1$ in the given domains of $a$ and $s_1$, we can apply \emph{de Branges's theorem} to the same to obtain the bound.\end{proof}
Let us note down another corollary of the lemma \ref{Hschlicht} for future reference.
\begin{corollary}
	For $a\in\left(-\frac{2\m}{9},0\right)\cup\left(0,\frac{4\m}{9}\right)$ and $s_1\in\left[\frac{2\m}{3},\infty\right)$, the Taylor coefficients $\{\b_n(a,s_1)\}$ in the power-series expansion of $H$ are bounded by
	\be \label{FKoebe}
	\frac{|\tilde{z}|}{(1+|\tilde{z}|)^2}\le\left|F(\tilde{z};s_1,a)\right|\le \frac{|\tilde{z}|}{(1-|\tilde{z}|)^2},\quad |\tilde{z}|<1.
	\ee 
\end{corollary}
\begin{proof}
	Applying Koebe growth theorem \ref{Kgrowth} to $F(\tilde{z};s_1,a)$, one obtains the bounds.
\end{proof}
Now that we have collected the necessary results, let us now turn to prove the following theorem. 
\begin{theorem}
	\label{alphabound}
	For {non-zero $\mt(\tilde{z},a)$} and $a\in\left(-\frac{2\m}{9},0\right)\cup\left(0,\frac{4\m}{9}\right)$, with $\m>0$,
	\be 
	\left|\frac{\a_n(a)a^{2n}}{\a_1(a)a^2}\right|\le n,\quad\forall\,n\ge2.
	\ee 
\end{theorem}
\begin{proof}[{\bf Proof }]
	First, we make sure that the ratio is well-defined in $a\in\left(-\frac{2\m}{9},0\right)\cup\left(0,\frac{4\m}{9}\right)$. In particular, we need to make sure that $\a_1(a)\ne 0$ for $a\in\left(-\frac{2\m}{9},0\right)\cup\left(0,\frac{4\m}{9}\right)$ since $a^2$ is never zero due to $\m\ne 0$. To do so, let us start with the inversion formula, \eqref{eq:alphadisper}, to write 
	\be \label{a1inv}
	-\a_1(a)a^2=\frac{1}{\pi}\int _{\frac{2 \mu }{3}}^{\infty }\frac{ds_1'}{s_1'}\,\mathcal{A}\left(s_1';s_2^{(+)}(s_1',a)\right)\,\left[-\beta _1\left(a,s_1'\right)\right]\,.
	\ee 
	From the \emph{positivity lemma} \ref{Apos}, we have that $\mathcal{A}\left(s_1';s_2^{(+)}(s_1',a)\right)\ge 0$ for $a\in\left(-\frac{2\m}{9},0\right)\cup\left(0,\frac{4\m}{9}\right)$ and $s_1'\in\left[\frac{2\m}{3},\infty\right)$. Further, we have already seen that $\beta _1\left(a,s_1'\right)<0$ for the same domains of $a$ and $s_1'$. Thus, the integrand is \emph{non-negative}. {Further, since $\mt(\tilde{z},a)$ is non-zero, $\mathcal{A}\left(s_1';s_2^{(+)}(s_1',a)\right)$ does not vanish identically over $s_1'\in\left[\frac{2\m}{3},\infty\right)$.} Hence, the integral will be \emph{positive}   implying 
	\be \label{a1neg}
	\a_1(a)<0,\quad a\in 	\left(-\frac{2\mu}{9},\frac{4\m}{9}\right).
	\ee   
	Therefore, the ratio under consideration,
	\be 
	a^{2n-2}\,\frac{\a_n(a)}{\a_1(a)}
	\ee is well-defined for  $a\in\left(-\frac{2\m}{9},0\right)\cup\left(0,\frac{4\m}{9}\right)$.\par 
	Now, let us use the inversion formula \eqref{eq:alphadisper} once again, and taking the absolute values on both sides of the formula one obtains
	\begin{align} 
		\left|\a_n(a)a^{2n}\right|&=\frac{1}{\p}\left|\int _{\frac{2 \mu }{3}}^{\infty }\frac{ds_1'}{s_1'}\,\mathcal{A}\left(s_1';s_2^{(+)}(s_1',a)\right)\,\beta _n\left(a,s_1'\right)\right|,\nn\\
		&\le \frac{1}{\p}\int _{\frac{2 \mu }{3}}^{\infty }\frac{ds_1'}{s_1'}\,\left|\mathcal{A}\left(s_1';s_2^{(+)}(s_1',a)\right)\,\beta _n\left(a,s_1'\right)\right|\quad\left[\text{Applying triangle inequality}\right],\nn\\
		&=\frac{1}{\p}\int _{\frac{2 \mu }{3}}^{\infty }\frac{ds_1'}{s_1'}\,\mathcal{A}\left(s_1';s_2^{(+)}(s_1',a)\right)\,\left|\beta _n\left(a,s_1'\right)\right|\quad\left[\ma\left(s_1';s_2^{(+)}(s_1',a)\right)\ge0\,\, \text{by lemma \ref{Apos} }\right],\nn
	\end{align}
	\begin{align}
		&\le \frac{1}{\p}\int _{\frac{2 \mu }{3}}^{\infty }\frac{ds_1'}{s_1'}\,\mathcal{A}\left(s_1';s_2^{(+)}(s_1',a)\right)\,n\left|\beta _1\left(a,s_1'\right)\right|\quad\left[\text{Applying corollary}\,\, \ref{bnbound}\right],\nn\\
		&=\frac{n}{\p}\int _{\frac{2 \mu }{3}}^{\infty }\frac{ds_1'}{s_1'}\mathcal{A}\left(s_1';s_2^{(+)}(s_1',a)\right)\,\left[-\beta _1\left(a,s_1'\right)\right]\quad\left[\b_1<0\iff|\b_1|=-\b_1\right],\nn\\
		&=n\,\left(-\a_1(a)a^2\right)=n\left|\a_1(a)a^2\right|\quad \left[\a_1(a)<0\,\,\text{from}\,\eqref{a1neg}\right].
	\end{align}
	$\therefore\,$ For $a\in\left(-\frac{2\m}{9},0\right)\cup\left(0,\frac{4\m}{9}\right)$, 
	\be 
	\left|\frac{\a_n(a)a^{2n}}{\a_1(a)a^2}\right|\le n,\qquad\forall\,n\ge2.
	\ee 
\end{proof}	
For illustration, we show $ \left|\frac{\a_n(a)a^{2n}}{\a_1(a)a^2}\right|$ as a function of $a$ for 1-loop $\phi^4$-amplitude and tree level type II string amplitude in figure \eqref{fig:coef}. The presented proof assumes the range $a\in\left(-\frac{2\m}{9},0\right)\cup\left(0,\frac{4\m}{9}\right)$, even though the plots suggest that bound is still valid till $a\in\left(-\frac{2\m}{9},0\right)\cup\left(0,\frac{2\m}{3}\right)$, for some cases. For the string amplitude $a\in\left(-\frac{s_1^{(0)}}{3},0\right)\cup\left(0,\frac{2s_1^{(0)}}{3}\right)$, where $s_1^{(0)}$ is the starting point of the physical cuts in s-channel. We use the fact that first massive state is at $s_1^{(0)}=1$.

\begin{figure}[hbt!]
	\centering
	\begin{subfigure}[b]{0.48\textwidth}
		\centering
		\includegraphics[width=\textwidth]{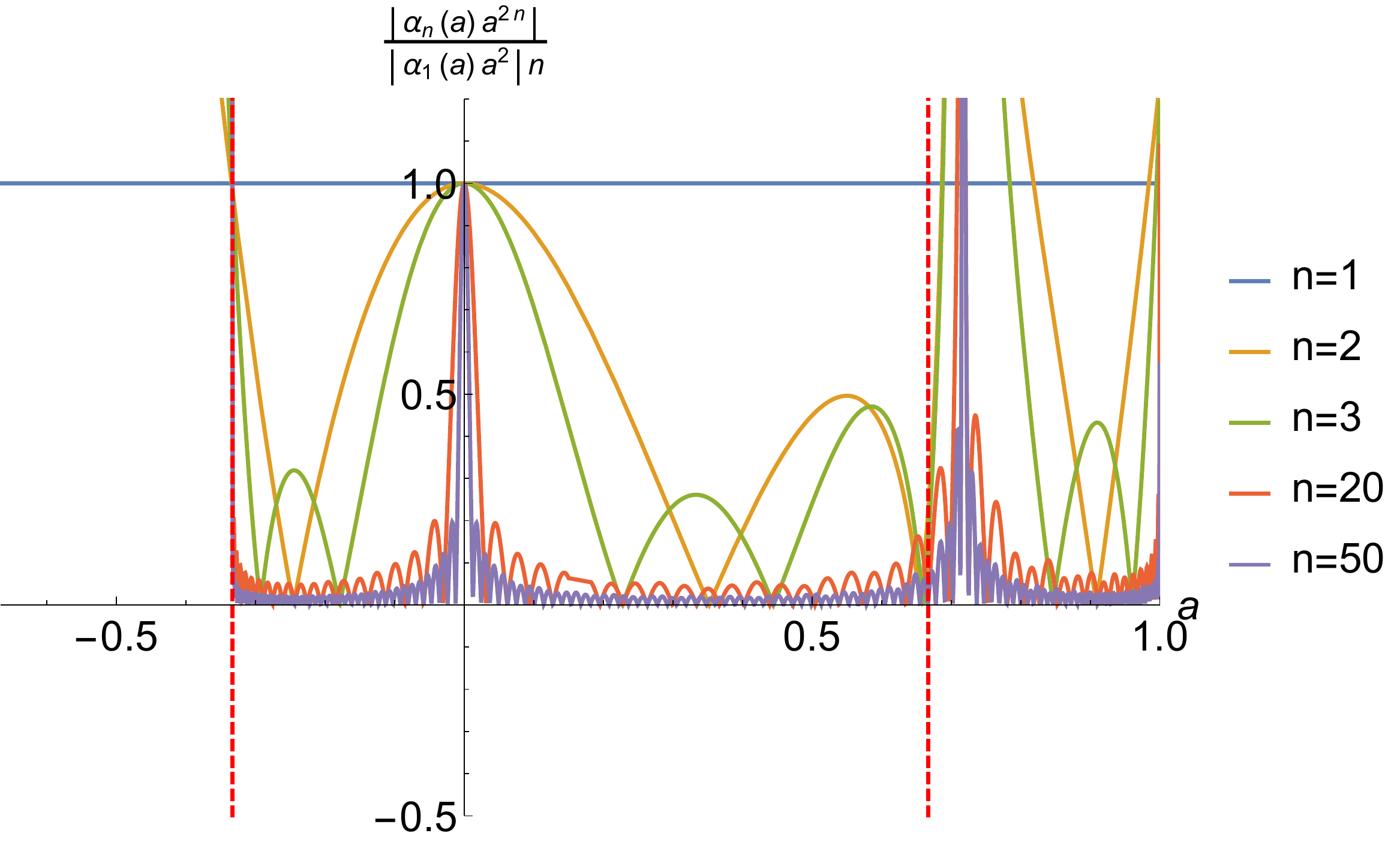}
		\caption{Tree level type II string amplitude. Red lines are $a=-\frac{1}{3}, \frac{2}{3}$}
		\label{fig:coef_string}
	\end{subfigure}
	\hfill
	\begin{subfigure}[b]{0.48\textwidth}
		\centering
		\includegraphics[width=\textwidth]{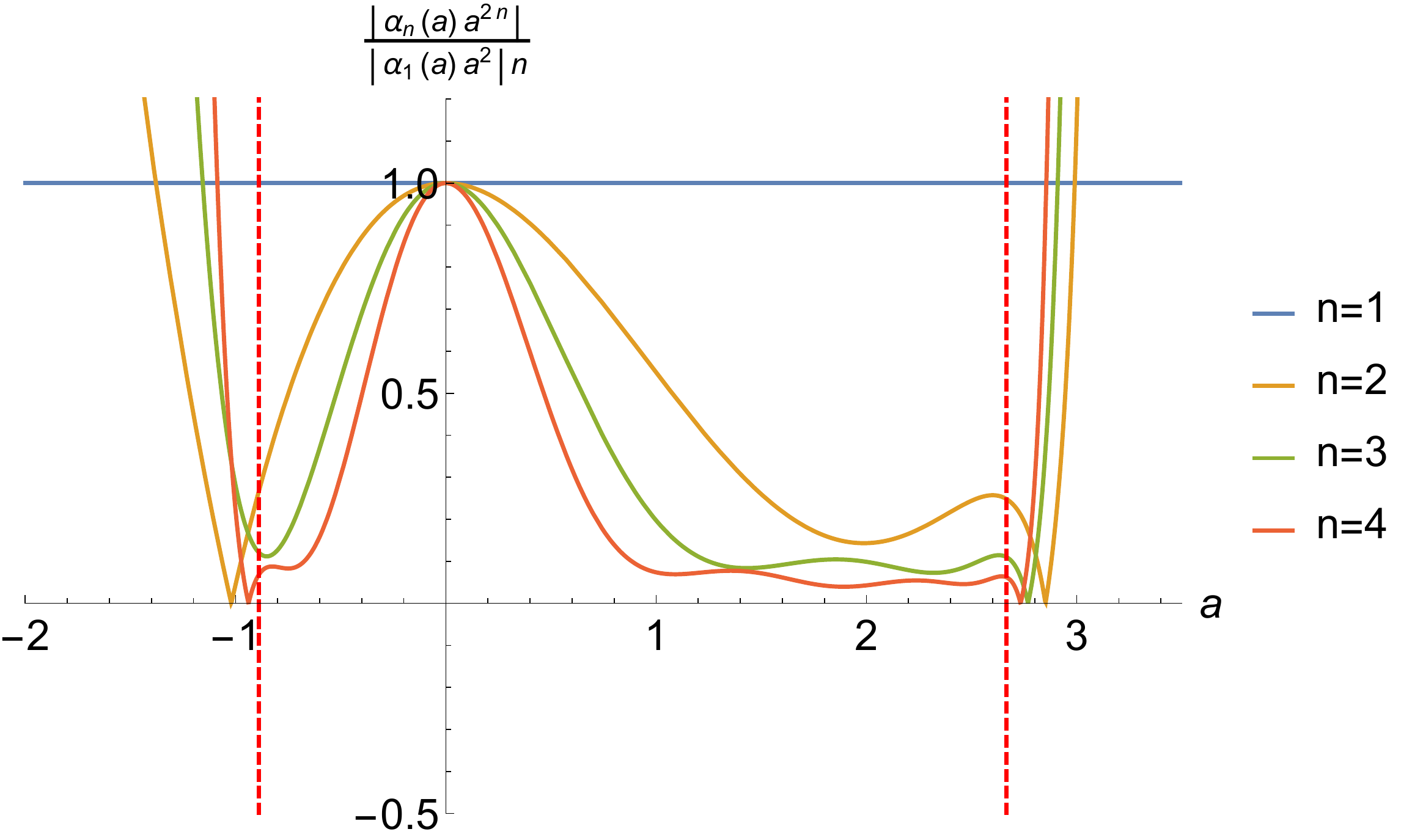}
		\caption{1-loop $\phi^4$-amplitude.  Red lines are $a=-\frac{8}{9}, \frac{8}{3}$}
		\label{fig:coef_phi4}
	\end{subfigure}
	\caption{Bounds on $ \left|\frac{\a_n(a)a^{2n}}{\a_1(a)a^2}\right|$ as a function of $a$ }
	\label{fig:coef}
\end{figure}

\section{Stronger bounds on the Wilson coefficients ${\mathcal W}_{p,q}$}\label{wilsonbound}
In order to derive bounds on ${\mathcal W}_{p,q}$, we first recall the formula of \eqref{eq:alphatoW}. 
{ We have already proved in equation (4.10) that $\alpha_1 < 0$ for $a\in \left(-\frac{2\mu}{9},0\right) \cup\left(0,\frac{4\m}{9}\right)$. Now, $\alpha_1=-\mathcal{W}_{1,0}\left(a\,\frac{\mathcal{W}_{0,1}}{\mathcal{W}_{1,0}}+1\right)$ which follows from  \eqref{eq:alphatoW}. Further, $\mathcal{W}_{1,0}$ is strictly positive, which was shown in \cite{ASAZ, tolley, nima}. This immediately implies that $\left( a\, \frac{\mathcal{W}_{0,1}}{\mathcal{W}_{10}}+1\right)>0$ for the given range of $a$ quoted above. From here, the bound in \eqref{w0110} follows. 
	\be\label{w0110}
	-\frac{9}{4\m}<\frac{\mathcal{W}_{0,1}}{\mathcal{W}_{1,0}}< \frac{9}{2\m}\,.
	\ee
	It is to be emphasized that if \eqref{w0110} is not satisfied, i.e., if $\mathcal{W}_{0,1}/\mathcal{W}_{1,0}$ is outside the range given by \eqref{w0110}, then dialing $a$ within the range $-\frac{2\m}{9}< a <\frac{4\m}{9}, \,a\ne0$, one can make the factor $\left( a \,\frac{\mathcal{W}_{0,1}}{\mathcal{W}_{1,0}}+1\right)$ \emph{change sign} contradicting $\alpha_1(a)<0$  in the said range of $a$.\footnote{{Note that if we were to find stronger bounds on the ratio, we would need a {\it wider} allowed range for $a$ and vice versa, contrary to the naive expectation.}}
	Let us present this derivation a little differently as well.   Using theorem \ref{alphabound} for $n=2$ and \eqref{eq:alphatoW}, we find
	\be\label{eq:n2alphaW}
	-2\leq 2-\frac{27 a^2 \left(a \left(a \mathcal{W}_{0,2}+\mathcal{W}_{1,1}\right)+\mathcal{W}_{2,0}\right)}{a \mathcal{W}_{0,1}+\mathcal{W}_{1,0}} \leq 2\,.
	\ee
	Now if the denominator $a \mW_{0,1}+\mW_{1,0}$ vanishes, then unless at the same point the numerator vanishes, we will contradict the inequality. Suppose for $a=a_0$, the denominator vanishes. Then we must have $a_0 (a_0 \mW_{0,2}+\mW_{1,1})+\mW_{2,0})=0$ giving a relation between three apparently independent Wilson coefficients. This appears unnatural to us and if we were to avoid this possibility we would again get eq.(\ref{w0110}). Using eq.(\ref{eq:n2alphaW}), it is also possible to deduce bounds \cite{RS} on the individual $\mW_{0,2}/\mW_{1,0}, \mW_{1,1}/\mW_{1,0}, \mW_{2,0}/\mW_{1,0}$ and these appear comparable to \cite{Caron-Huot:2020cmc}.}

The condition $-\frac{9}{4\m}<\frac{\mathcal{W}_{0,1}}{\mathcal{W}_{1,0}}$ was derived in \cite[eq (5)]{ASAZ}, while $\frac{\mathcal{W}_{0,1}}{\mathcal{W}_{1,0}}< \frac{9}{2\m}$ is a new finding. For illustration purpose, we show in figure \eqref{fig:W01W10pion} that pion S-matrices satisfy them in a very non-trivial way. We will comment more on the behaviour exhibited in figure (\ref{fig:W01W10pion}) below. 
\begin{figure}[hbt!]
	\centering
	\includegraphics[width=0.6\textwidth]{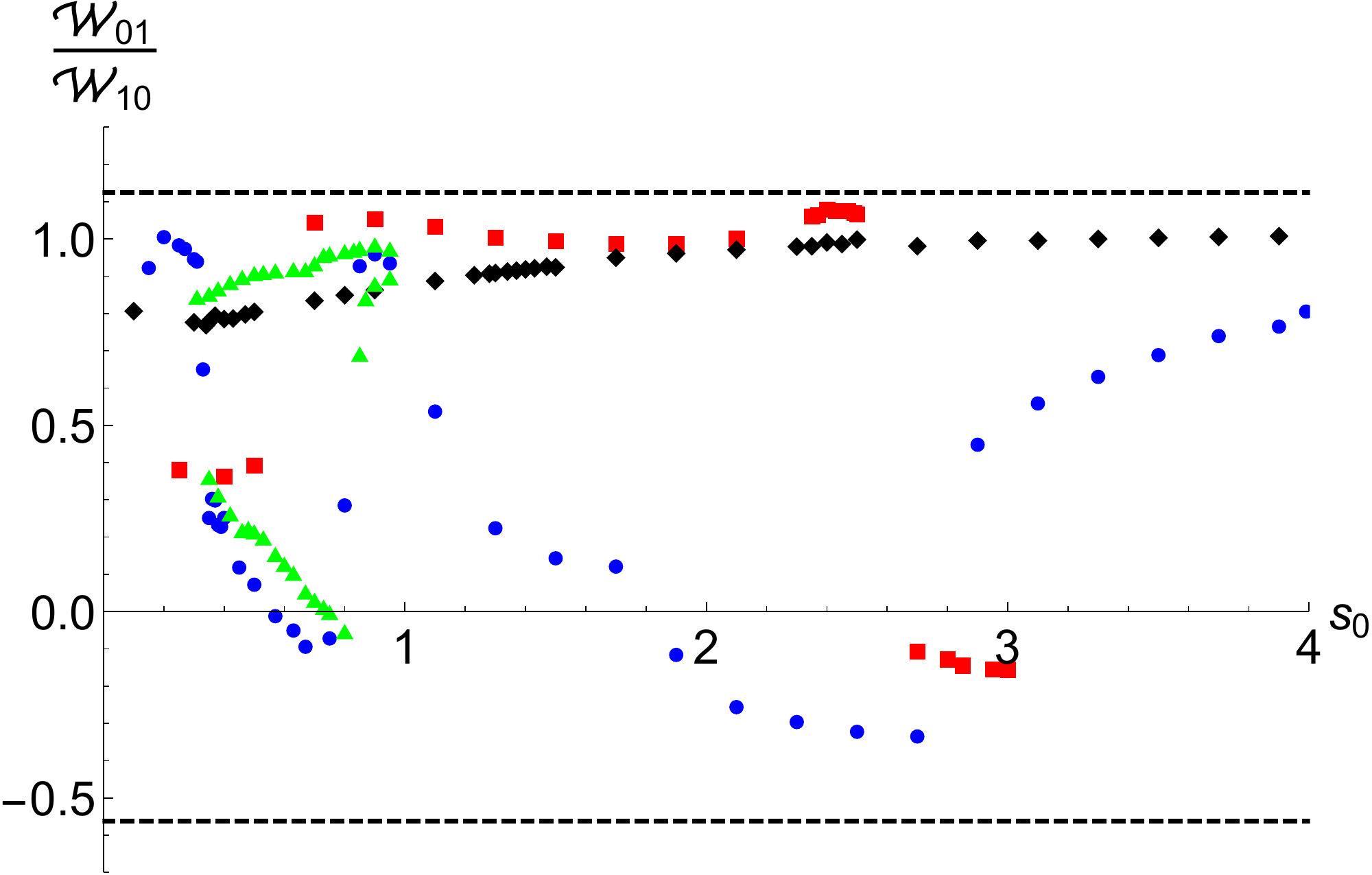}
	\caption{Ratio of $\frac{\mathcal{W}_{0,1}}{W_{1,0}}$ obtained from the S-matrix bootstrap. The horizontal axis is the Adler zero $s_0$. The green points are for the pion lake \cite{gpv}. The blue and red points are for the upper and lower river boundaries \cite{bhsst,bst} while the black points are for the line of minimum averaged total cross section S-matrices \cite{bst}. }
	\label{fig:W01W10pion}
\end{figure}

For the string example, we can ask the following question: Given ${\mathcal W}_{0,1},{\mathcal W}_{1,0},{\mathcal W}_{1,1}, {\mathcal W}_{2,0}$, in other words the Wilson coefficients till the eight-derivative order term $x^2$, how constraining is \eqref{eq:n2alphaW}? The situation is shown in fig(\ref{fig:a2string}). Quite remarkably, the range of $\mW_{02}$ is very limited; the inequality is very constraining indeed! 

\begin{figure}[hbt!]
	\centering
	\begin{subfigure}[b]{0.48\textwidth}
		\centering
		\includegraphics[width=\textwidth]{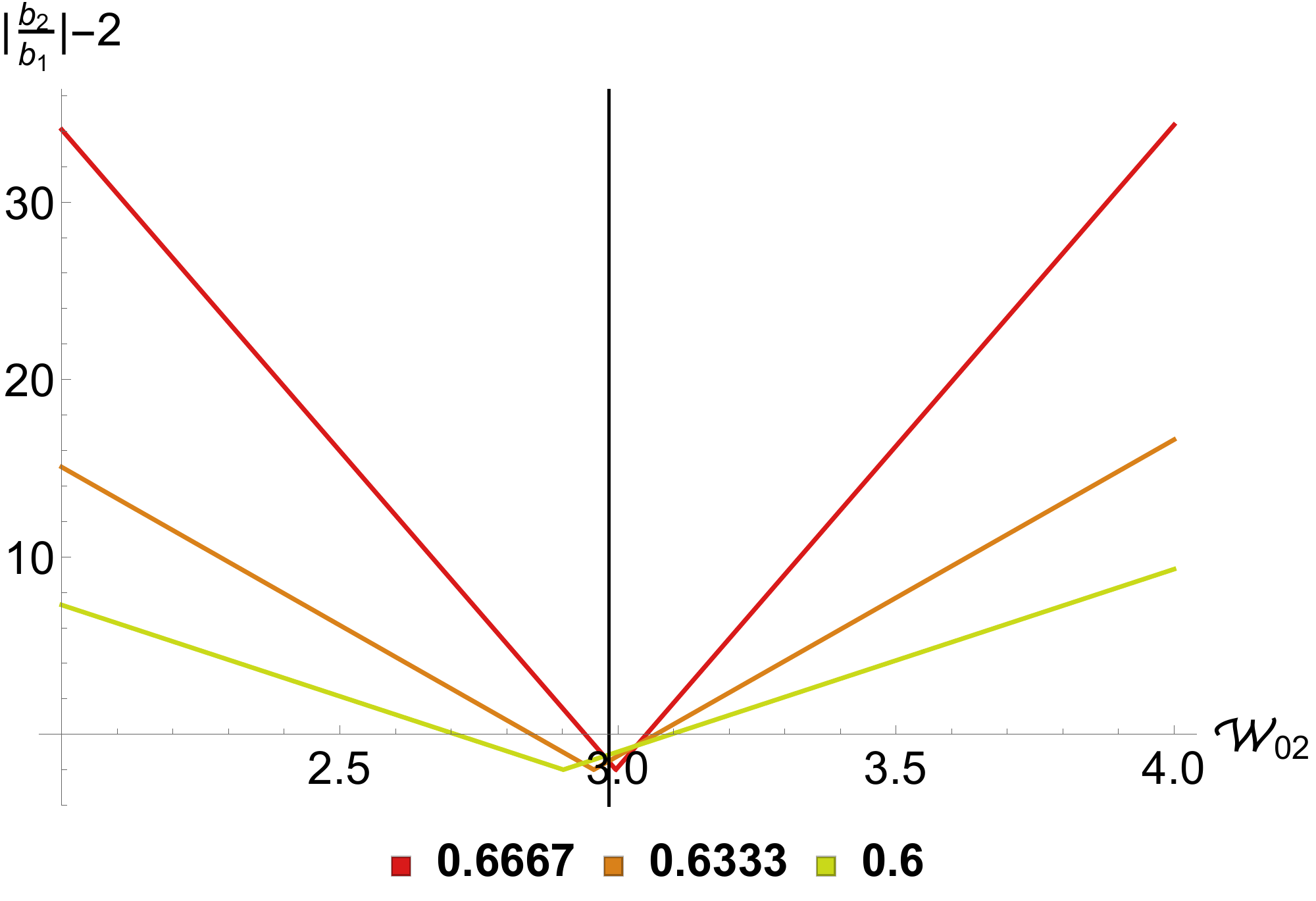}
		\caption{Tree level type II string amplitude}
		\label{fig:a2string}
	\end{subfigure}
	\hfill
	\begin{subfigure}[b]{0.48\textwidth}
		\centering
		\includegraphics[width=\textwidth]{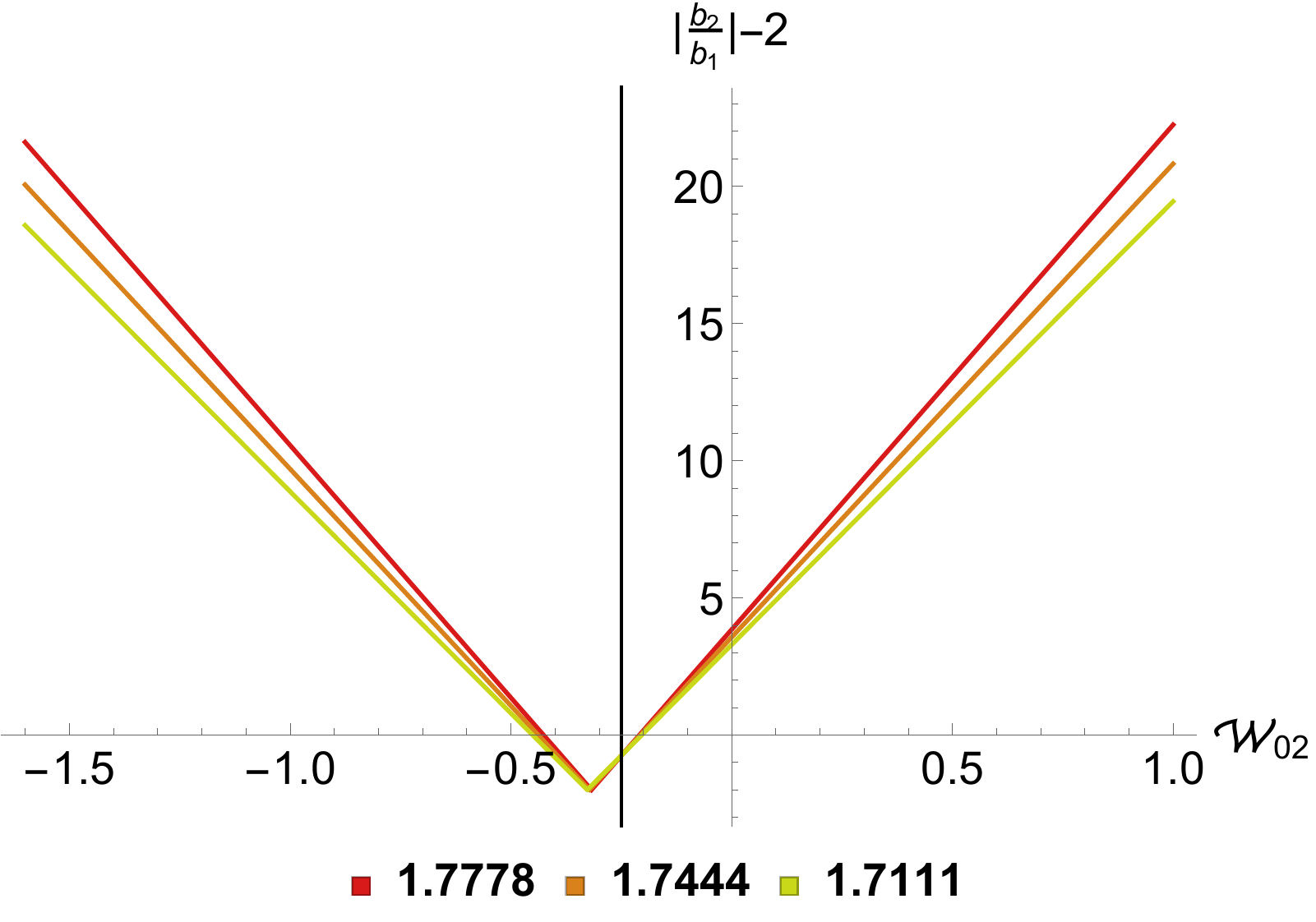}
		\caption{Pion scattering amplitude, $s_0=0.35$}
		\label{fig:a2upper4}
	\end{subfigure}
	\caption{Constraints on Wilson coefficients using \eqref{eq:n2alphaW}.  Given ${\mathcal W}_{0,1},{\mathcal W}_{1,0},{\mathcal W}_{1,1}, {\mathcal W}_{2,0}$ figure shows that bound on the $\mathcal{W}_{0,2}$. Since $\left|\frac{b_2}{b_1}\right|-2$ should be less than zero, $\mathcal{W}_{0,2}$ must lie inside the triangle. {Black} line is the exact answer. Different values of $a$ are indicated with different colours.}
	\label{fig:a2a3bound}
\end{figure}
We can further investigate the situation for $n=3$. We get
\be\label{eq:n3alphaW}
\begin{split}
	&-3\leq 3+\frac{27 a^2 \left(a \left(-4 a \mathcal{W}_{0,2}+27 a \left(a \left(a \left(a \mathcal{W}_{0,3}+\mathcal{W}_{1,2}\right)+\mathcal{W}_{2,1}\right)+\mathcal{W}_{3,0}\right)-4 \mathcal{W}_{1,1}\right)-4 \mathcal{W}_{2,0}\right)}{a \mathcal{W}_{0,1}+\mathcal{W}_{1,0}}\leq 3
\end{split}
\ee

We give in the appendix a demonstration of how to constrain $\mathcal{W}_{0,3}$ using this inequality.

\subsection{Bounds in case of EFTs:} 
In an EFT, usually, the Lagrangian is known up to some energy scale. From that information, one can calculate the amplitude up to that scale.  In such cases, we can subtract off the known part of the amplitude. These steps result in a shift in  the lower limit of the dispersion integral \eqref{crossdisp} by the scale $\d_0$, namely $\m\to \m+3\d_0/2$ (see \cite{ASAZ}). Therefore, making this replacement in \eqref{w0110}, we have
\be \label{mudelta}
-\frac{9}{4\m+6\d_0}<\frac{\mathcal{W}_{0,1}}{\mathcal{W}_{1,0}}< \frac{9}{2\m+3\d_0}\,.
\ee
Here $\mW$ are the Wilson coefficients of the amplitude with subtractions. Now notice that if we consider $\delta_0\gg \mu$ then we have
\be\label{strong}
-\frac{3}{2\d_0}<\frac{\mathcal{W}_{0,1}}{\mathcal{W}_{1,0}}< \frac{3}{\d_0}\,.
\ee
Let us compare this to \cite{Caron-Huot:2020cmc}. Converting their results to our conventions, we find that the lower bound above is identical to their findings--this is corroborated by the results of \cite{tolley} as well as what arises from crossing symmetric dispersion relations \cite{ASAZ}. The other side of the bound is more interesting. The strongest result in $d=4$ in \cite{Caron-Huot:2020cmc} places the upper bound at $\approx 5.35/\delta_0$ in our conventions. Their approach also makes the bound spacetime dimension dependent. Now remarkably, the bound we quote above and the $d\gg 1$ limit of \cite{Caron-Huot:2020cmc} are identical! In EFT approaches, one takes the so-called null constraints or locality constraints and expands in the limit $\delta_0\gg\mu$.  It is possible that a more exact approach building on \cite{Caron-Huot:2020cmc} will lead to a stronger bound as in \eqref{strong}.

Let us now comment on the behaviour in the figure (\ref{fig:W01W10pion}). First, notice that all S-matrices appear to respect the upper bound we have found above; for the S-matrix bootstrap results, we set $\delta_0=0$. For comparison, note that for 1-loop $\phi^4$, and 2-loop chiral perturbation theory, we have
\be
\left(\frac{\mW_{0,1}}{\mW_{1,0}}\right)_{\phi^4}\approx-0.315 \,,\quad \left(\frac{\mW_{0,1}}{\mW_{1,0}}\right)_{\chi-PT}\approx-0.135\,,
\ee
both in units where $m=1$. These numbers would be closer to the lower black dashed line in the figure (\ref{fig:W01W10pion}), which is the bound that is common in all approaches so far. The upper black dashed line is what we find in the current paper. For future work, it will be interesting to search for an interpolating bound as in \eqref{mudelta} which enables us to interpolate between $\delta_0=0$ and $\delta_0\gg \mu$.

\section{Two-sided bounds on the amplitudes}\label{ampbound}
The crossing symmetric dispersion relation can enable us to derive $2-$sided bounds on the scattering amplitude $\mt$. In this section, we will derive such bounds. The result that we will first prove is the following:
\begin{theorem}\label{th:ampbound}
	Let $\mt(\tilde{z},a)$ be a unitary and crossing symmetric scattering amplitude admitting the dispersive representation \eqref{crossdisp} and admits the power series expansion \eqref{Mpower} about $\tilde{z}=0$ which converges in the open disc $|\tilde{z}|<1$. Define the function 
	\be\label{eq:fdef} 
	f(\tilde{z},a):=\frac{\mt(\tilde{z},a)-\a_0}{\a_1(a)a^2},\qquad\a_0=\mt(\tilde{z}=0,a).
	\ee 
	Then for $a\in\left(-\frac{2\m}{9},0\right)\cup\left(0,\frac{4\m}{9}\right)$,
	\begin{enumerate}
		\item [(1)] 
		\be \label{ubound}
		\left|f(\tilde{z},a)\right|\le \frac{|\tilde{z}|}{(1-|\tilde{z}|)^2},\quad |\tilde{z}|<1
		\ee 
		\item[(2)] 
		\be 
		|f(\tilde{z},a)|\ge \frac{|\tilde{z}|}{(1+|\tilde{z}|)^2},\quad\tilde{z}\in\mathbb{R}\,\land\, |\tilde{z}|<1.
		\ee 
	\end{enumerate} 
\end{theorem}
\begin{proof}[{\bf Proof.}]
	Let us first prove the upper bound. Starting with the dispersion relation \eqref{crossdisp}, we can obtains
	\begin{align}
		|\mt(\zt,a)-\alpha_{0}|&\le\imeas\left|\immp\right|\,\left|H(s_1',\zt)\right|\quad\left[\text{Applying triangle inequality}\right],\nn\\
		&=\frac{1}{\p}\imeas\immp\left|\b_1(a,s_1')\right|\left|F(\zt;s_1',a)\right|\quad\left[\text{Using lemma \ref{Apos} and \eqref{Fdef}}\right],\nn\\
		&\le \frac{1}{\p}\imeas\immp	\left|\b_1(a,s_1')\right| \frac{|\zt|}{(1-|\zt|)^2}\quad\left[\text{Using corollary \ref{FKoebe}}\right],\nn\\
		&=\frac{|\zt|}{(1-|\zt|)^2}\frac{1}{\p}\imeas\immp	\left[-\b_1(a,s_1')\right]\quad\left[\b_1<0\iff|\b_1|=-\b_1\right],\nn\\
		&=\frac{|\zt|}{(1-|\zt|)^2}\,\left[-\a_1(a)a^2\right]=\frac{|\zt|}{(1-|\zt|)^2}\,\left|\a_1(a)a^2\right|\quad\left[\a_1(a)<0\,\,\text{from}\,\eqref{a1neg}\right].
	\end{align}
	Thus, we have finally, 
	\be 
	\left|\frac{\mt(\zt,a)-\a_0}{\a_1(a)a^2}\right|\le \frac{|\zt|}{(1-|\zt|)^2},\quad |\zt|<1;
	\ee which proves part $(1)$ of the theorem.\par 
	Next, let us prove the lower bound. First, we write 
	\be \label{rintro}
	F(\zt;s_1,a)\equiv\frac{H(\zt;s_1,a)}{\b_1(a,s_1)}=\frac{\zt}{(1+\zt)^2}\times\rho(\zt;s_1,a),
	\ee where
	\be 
	\rho(\zt;s_1,a):=\left(\frac{1+\zt}{1-\zt}\right)^2\times \left[1-\frac{27a^2(a-s_1)}{s_1^3}k(\zt)\right]^{-1},
	\ee  $k(\zt)$ being the Koebe function of \eqref{Koebe}. Now it turns out that $\r>1$ for $\zt\in\rb^+$ and $\r<-1$ for $\zt\in\rb^-$. This immediately implies
	\begin{align} 
		\label{rhomod}	|\r|&>1,\\
		\label{rhosign}	\zt\r&=|\zt||\r|.
	\end{align} Next, we use the dispersion relation  and \eqref{Fdef} to get 
	\begin{align}
		\left|\mt(\zt,a)-\a_0\right|&=\frac{1}{\p}\left|\imeas\immp \b_1(a,s_1')\,F(\zt;s_1,a)\right|,\nn\\
		&=\frac{1}{\p}\left|\imeas\immp \b_1(a,s_1')\,\frac{\zt}{(1+\zt)^2}\times\rho(\zt;s_1,a)\right|\quad\left[\text{Using \eqref{rintro}}\right],\nn\\
		&=\frac{1}{\p}\left|\imeas\immp \b_1(a,s_1')\,\frac{|\zt|}{(1+\zt)^2}\times\left|\rho(\zt;s_1,a)\right|\right|\quad\left[\text{Using \eqref{rhosign}}\right],\nn\\
		&\ge \frac{1}{\p}\frac{|\zt|}{(1+|\zt|)^2}\left|\imeas\immp\b_1(a,s_1')\right|\quad\left[\text{Using triangle inequality and \eqref{rhomod}}\right],\nn\\
		&=\frac{1}{\p}\frac{|\zt|}{(1+|\zt|)^2}\left|\a_1(a)a^2\right|.
	\end{align}
	Therefore, we have finally,
	\be 
	\left|\frac{\mt(\zt,a)-\a_0}{\a_1(a)a^2}\right|\ge \frac{|\zt|}{(1+|\zt|)^2},\quad \zt\in\rb\,\land\,|\zt|<1;
	\ee proving the second part and, hence, the theorem.
\end{proof}
%
We emphasize that our derivation for upper bound considers $\tilde{z}$ as complex numbers, while we present the derivation for lower bound considering only $\tilde{z}$ real numbers. We worked with a range of $-\frac{2\m}{9}<a<\frac{4\m}{9}$. Nevertheless, both of the bounds are valid for complex $\tilde{z}$ as far as we have observed. These bounds are satisfied by the $1-$loop $\phi^4$ amplitude and close string amplitude even for complex $\tilde z$, demonstrated in figure \eqref{fig:bound}, \eqref{fig:boundcomplex}. 

Even though, we have presented our proof for the range $a\in\left(-\frac{2\m}{9},0\right)\cup\left(0,\frac{4\m}{9}\right)$, for certain cases, we observed that bound is still valid till $a\in\left(-\frac{2\m}{9},0\right)\cup\left(0,\frac{2\m}{3}\right)$. For massless amplitude its $a\in\left(-\frac{s_1^{(0)}}{3},0\right)\cup\left(0,\frac{2s_1^{(0)}}{3}\right)$, where $s_1^{(0)}$ is the starting point of the physical cuts in s-channel.
Let us now rewrite the two sided bounds of the familiar Mandelstam variables. 
Now first note that, for real $\tilde{z}>0$, in terms of $x,a$ variables, the bounds become
\be
\frac{a^2 |\a_1(a)||x|}{4|x|+27 a^2}\leq \left|\mathcal{T}(z,a)-\a_0\right| \leq\frac{ |\a_1(a)||x|}{27}\,.
\ee
Here since $x=-27 a^2 \tilde z/(\tilde z-1)^2$, $x$ is real and negative. Now let us examine these bounds in the Regge limit where $|s_1|\rightarrow \infty$ with $s_2$ fixed.  From \eqref{eq:skdef}, it is clear that in this case in the $z$-variable,  $z\rightarrow e^{2\pi i/3}$ which also takes $s_2\rightarrow a$.  Now writing $s_1=|s_1|e^{i\theta/2}$, so that $x\sim |s_1|^2 e^{i\theta}$ when $|s_1|\rightarrow\infty$, we find
\be
\frac{1}{4}-\frac{27 a^2 \sin^2\frac{\theta}{2}}{16 \left| x\right| }+O\left(\frac{1}{|x|^{3/2}}\right)\leq \left|\frac{\mt(\zt,a)-\a_0}{\a_1(a)a^2}\right| \leq\frac{\left| x\right|  }{27 a^2 \sin ^2 \frac{\theta }{2}}-\frac{1}{4}\cot^2 {\theta}+O\left(\frac{1}{|x|^{1/2}}\right)\,.
\ee
Let us comment on this form. First, since $|x|\sim |s_1|^2$, the upper bound for fixed $\theta$ is the $|s_1|^2$ bound on the amplitude. Next, note the the important $\sin^2\frac{\theta}{2}$ factor. If we took $s_1$ to be real and positive, then the upper bound would be trivial. However, the real $s_1$-axis gets mapped to boundary of the unit disc, which is not a part of the open disc. Thus to use these bounds, we have to keep $\theta\in(0,2\pi)$, with the end-points excluded. Next, note the more interesting lower bound which begins with a constant! Finally, and importantly, the bound involves $\alpha_1(a)a^2$. In terms of Wilson coefficients, this involves only $\mathcal{W}_{01}, \mathcal{W}_{10}$. Thus, in the Regge limit, we have the following interesting bound
\be
\frac{27a^2}{4} \left|a\mathcal{W}_{01}+\mW_{10}\right|\lesssim \left|\mt(\tilde z,a)-\mt(0,a)\right| \lesssim\frac{\left| s_1\right|^2  }{\sin ^2 \frac{\theta }{2}}\left|a\mathcal{W}_{01}+\mW_{10}\right|\,.
\ee


\begin{figure}[hbt!]
	\centering
	\begin{subfigure}[b]{0.48\textwidth}
		\centering
		\includegraphics[width=\textwidth]{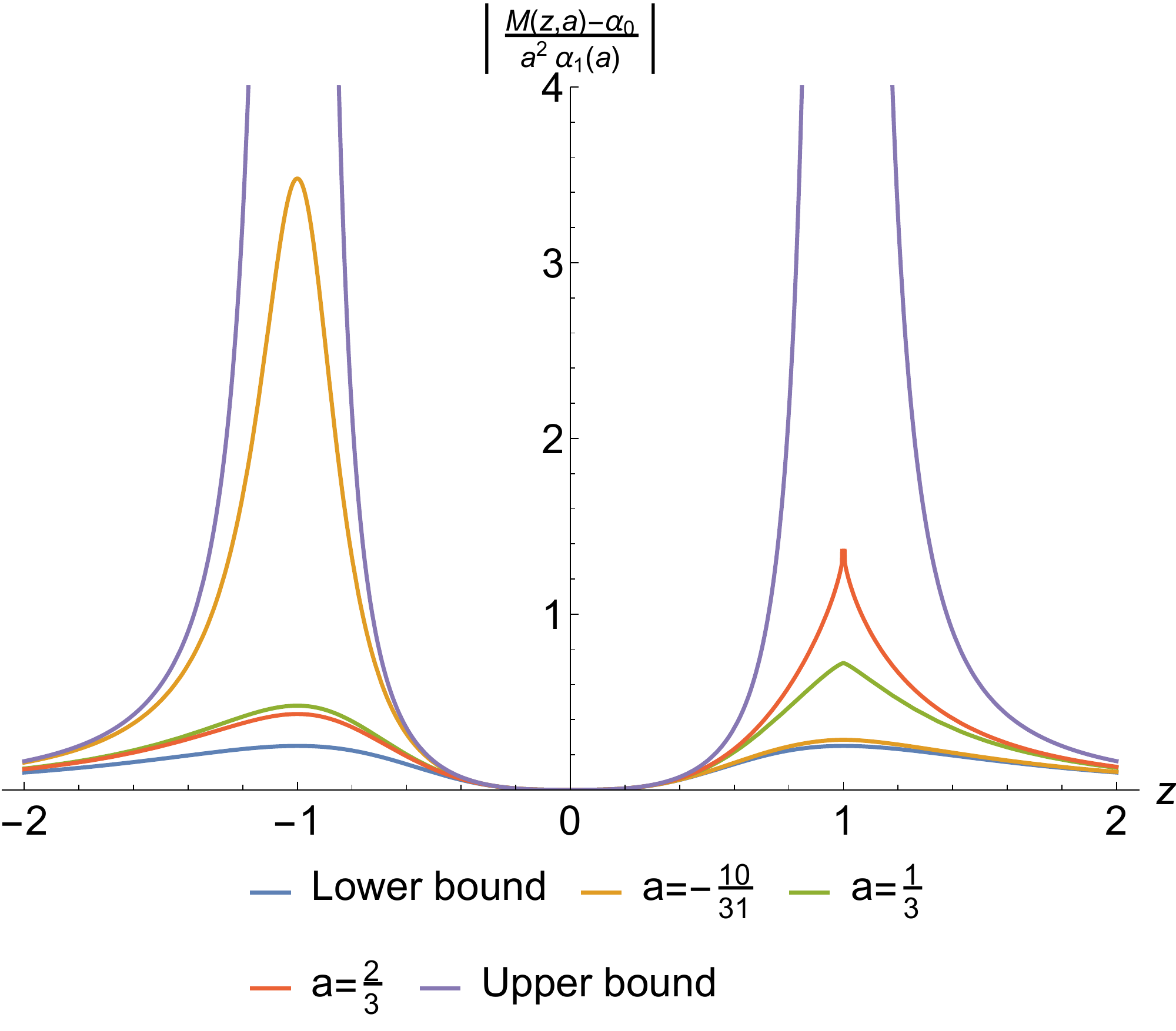}
		\caption{Tree level type II string amplitude}
		\label{fig:bound_string}
	\end{subfigure}
	\hfill
	\begin{subfigure}[b]{0.48\textwidth}
		\centering
		\includegraphics[width=\textwidth]{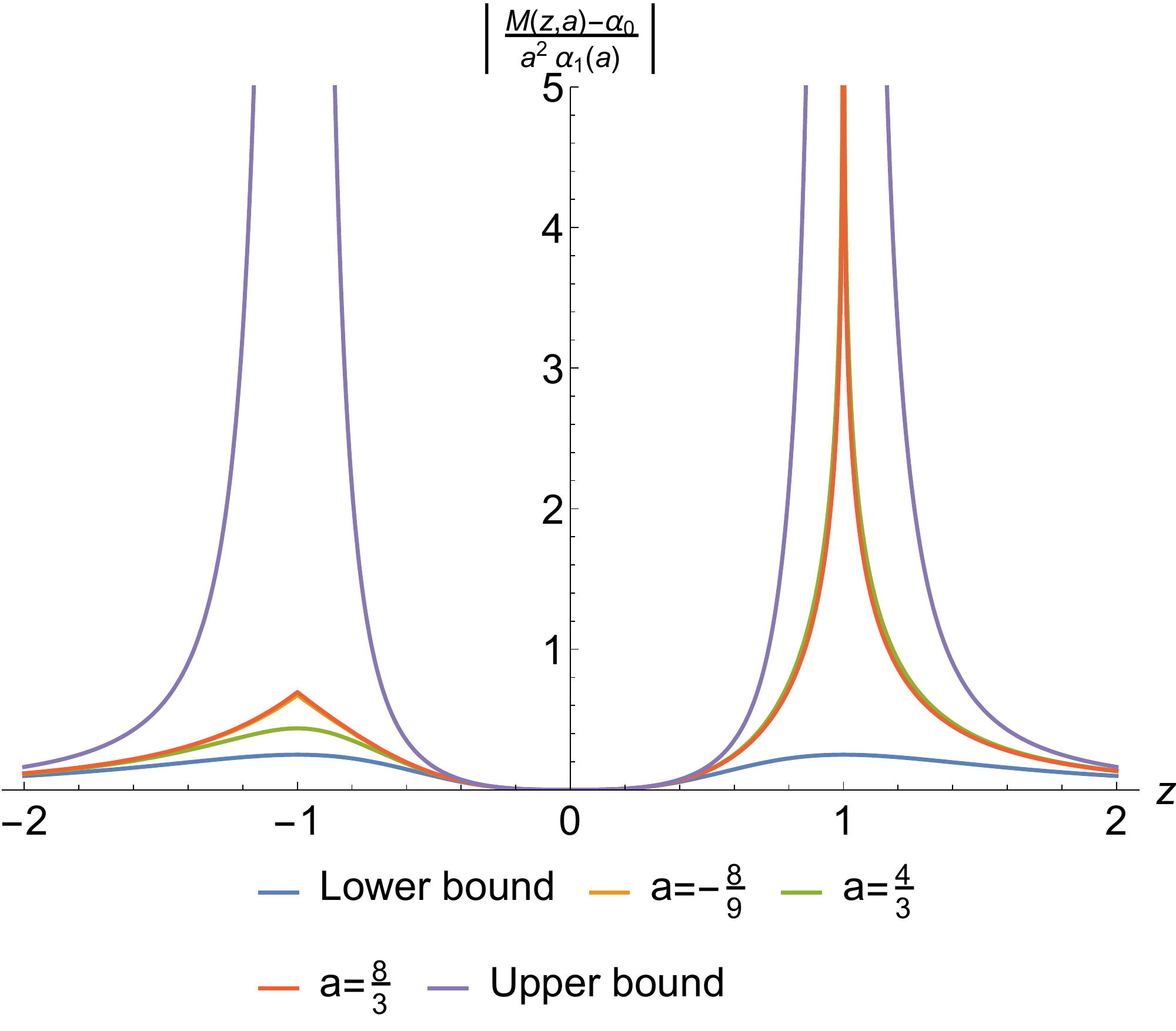}
		\caption{1-loop $\phi^4$-amplitude}
		\label{fig:bound_phi4}
	\end{subfigure}
	\caption{{Bounds on amplitude, as in theorem \eqref{th:ampbound}, are satisfied by Tree level type II string amplitude and 1-loop $\phi^4$-amplitude.}}
	\label{fig:bound}
\end{figure}

\begin{figure}[hbt!]
	\centering
	\begin{subfigure}[b]{0.48\textwidth}
		\centering
		\includegraphics[width=\textwidth]{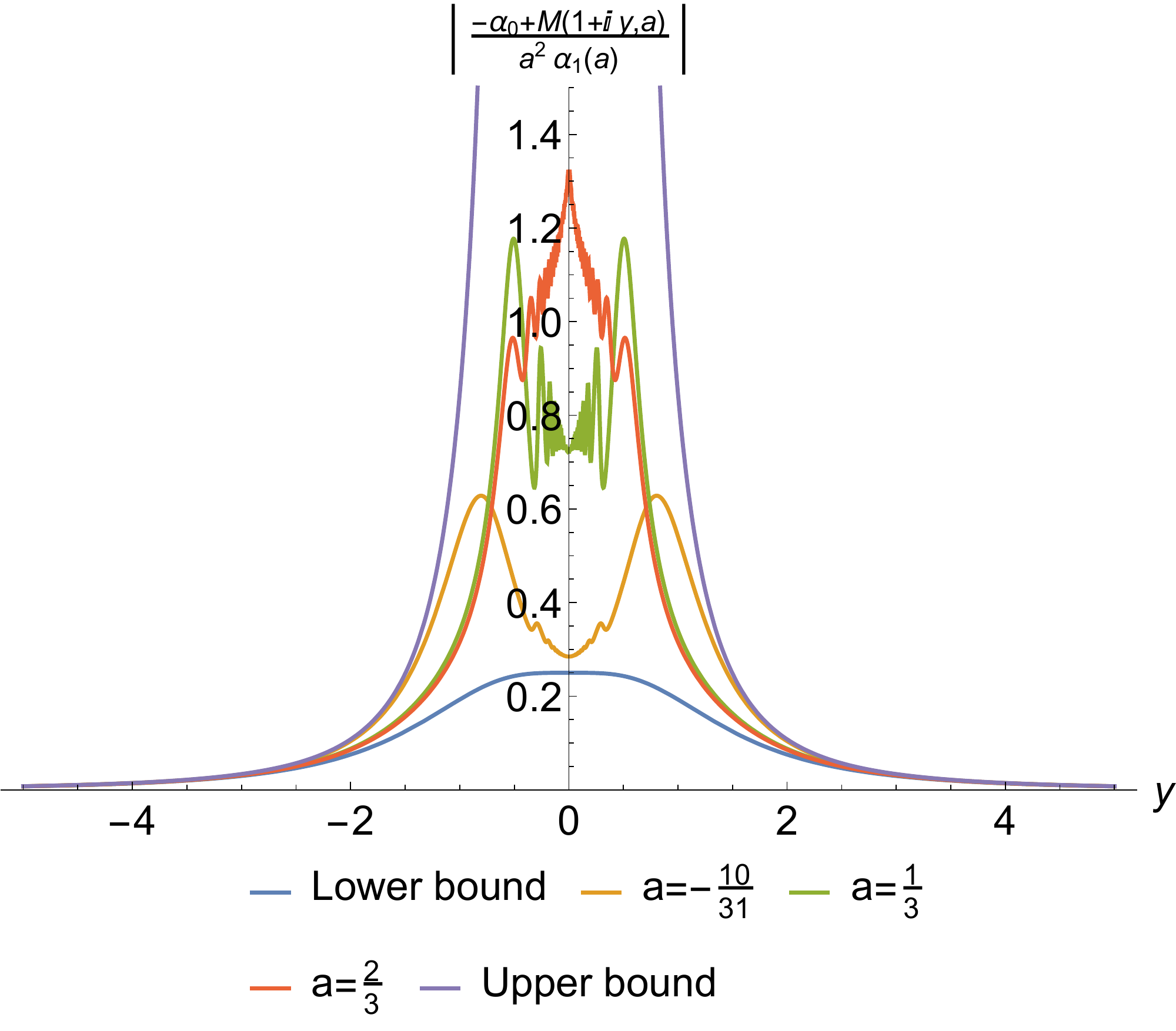}
		\caption{Tree level type II string amplitude}
		\label{fig:bound_stringc}
	\end{subfigure}
	\hfill
	\begin{subfigure}[b]{0.48\textwidth}
		\centering
		\includegraphics[width=\textwidth]{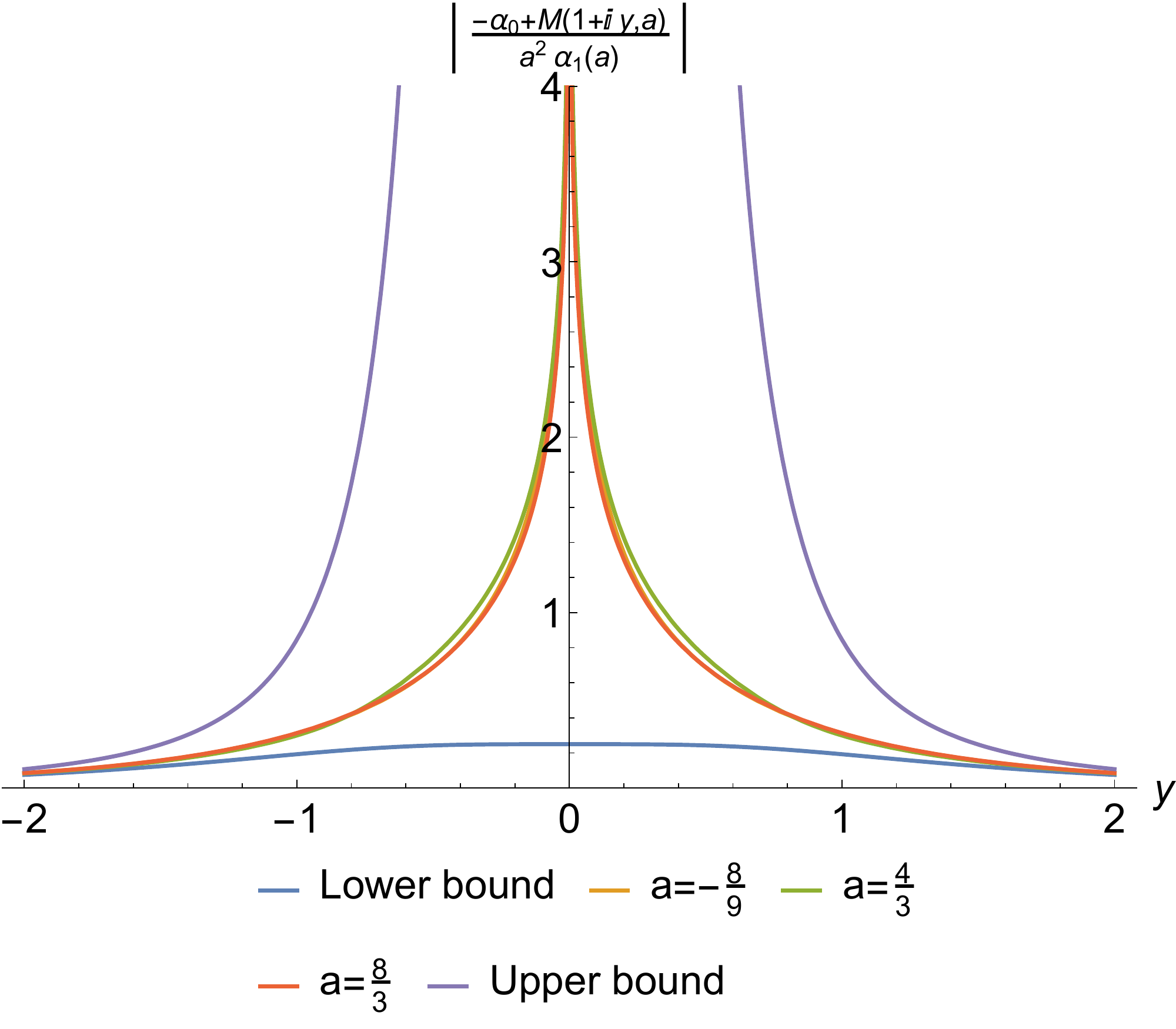}
		\caption{1-loop $\phi^4$-amplitude}
		\label{fig:bound_phi4c}
	\end{subfigure}
	\caption{{Bounds on amplitude, as in theorem \eqref{th:ampbound}, are satisfied by Tree level type II string amplitude and 1-loop $\phi^4$-amplitude. These bounds on amplitude valid for complex $\tilde{z}$.}}
	\label{fig:boundcomplex}
\end{figure}

\section{Univalence of the EFT expansion}\label{EFTunival}
So far, on the QFT side, we have focused on results motivated by the Bieberbach conjecture but which one can derive by using the crossing symmetric dispersion relation. All the QFT conditions we have derived so far hint at univalence. Establishing univalence in generality is a hard question and beyond the scope of our present work. Nevertheless, we can investigate scenarios where univalence is guaranteed to {\underline{not}} hold. We will begin our investigations using Szeg\"{o}'s theorem. For definiteness, consider the low energy expansion of the 2-2 dilaton scattering in type II string theory with the massless pole subtracted. We will consider
\be
f(\tilde z,a)= \frac{\mt(\tilde z,a)-\mt(0,a)}{\partial_{\tilde z}\mt(\tilde z,a)|_{\tilde z=0}}\,,
\ee
expanded around $\tilde z=0$. This corresponds to a low energy expansion of the amplitude since $\tilde z\sim 0$ corresponds to $x\sim 0, y \sim 0$. If $\mt(\tilde z,a)$ was univalent inside a disc $D$ of radius $R$, for a certain range of $a$, then $f(\tilde z,a)$ should be locally univalent inside $D$. This means that the absolute value of the smallest root ($\tilde z=\zeta_{min}$) of $\partial_{\tilde z} f(\tilde z,a)=0$ should be greater than $1/4$ to satisfy Szeg\"{o}'s theorem. This can be easily checked using Mathematica to some high power in the expansion in $\tilde z$.  The plot in fig.(\ref{fig:szego}) shows that univalence of the full amplitude is only possible in the range $a\in(-1/3,2/3)$ as one may have anticipated from our previous discussion. In fig.(\ref{fig:zetaminvsn}) we show monotonicity of $\zeta_{min}$ as a function of $n$ for small $a$. This is intuitive in the sense that for higher $n$'s, we are putting in more number of terms in the EFT expansion, as a result of which the radius of the disk, within which there is potential univalence, increases. We should point out, however, that for slightly larger values of $a$, for instance, $a>0.05$, there are other features that arise in the plot, which do not respect monotonicity---the physical implication of this finding is unclear to us. For 1-loop $\phi^4$, our findings are qualitatively similar. In appendix D, we comment on the Nehari conditions. 

\begin{figure}[hbt!]
	\centering
	\begin{subfigure}[b]{0.46\textwidth}
		\centering
		\includegraphics[width=\textwidth]{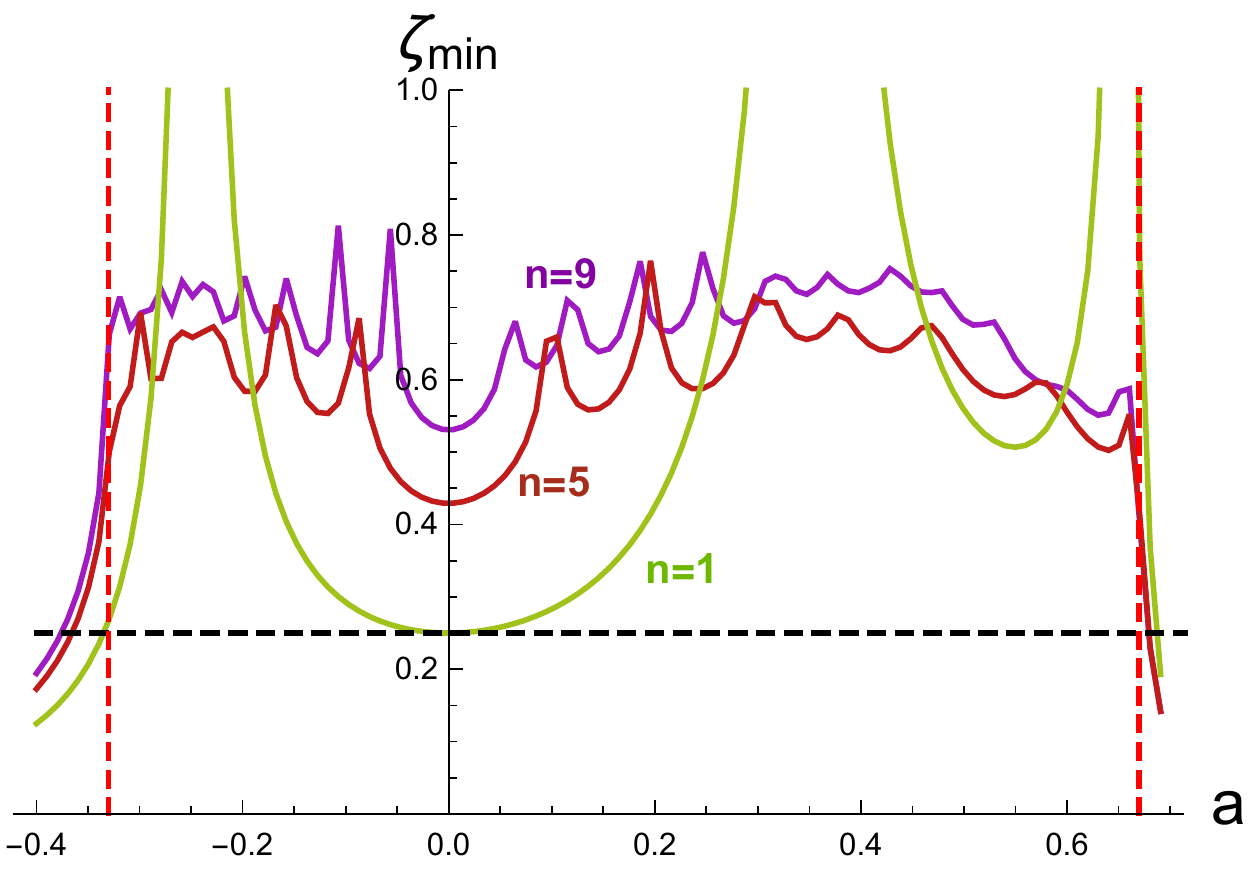}
		\caption{}
		\label{fig:szego}
	\end{subfigure}
	\hfill
	\begin{subfigure}[b]{0.46\textwidth}
		\centering
		\includegraphics[width=\textwidth]{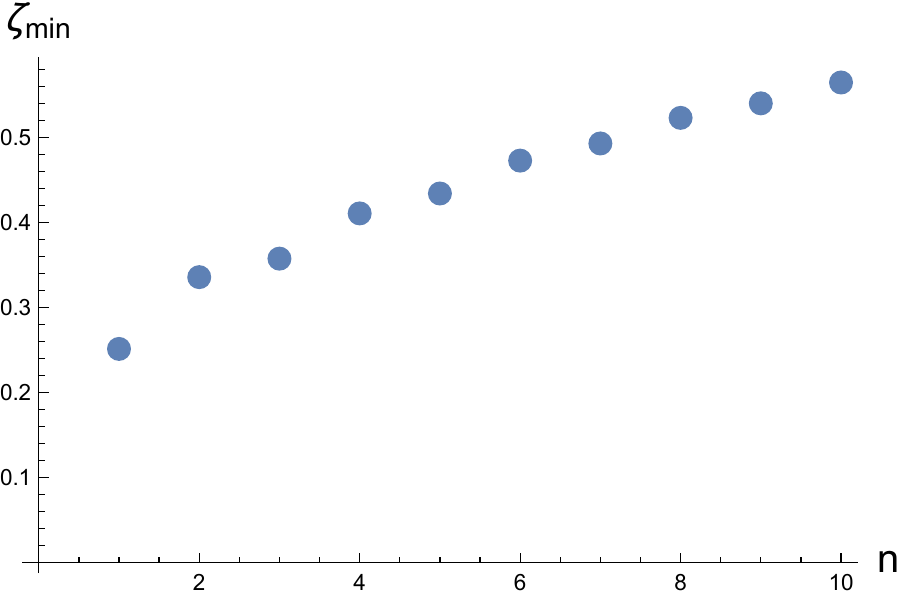}
		\caption{}
		\label{fig:zetaminvsn}
	\end{subfigure}
	\caption{Testing univalence using Szeg\"{o}'s theorem. (a) Plot of $\zeta_{min}$ vs $a$. This demonstrates that for the string amplitude to be potentially univalent, $-1/3<a<2/3$ must hold. (b) Plot of $\zeta_{min}$ vs $n$ for $a=0.02$.}
	\label{fig:szegotest}
\end{figure}

\subsection{Grunsky inequalities and EFT expansion}\label{Grunsky}
From the expansion in \eqref{eq:Wpqdef}, one can relate the $\mathcal{W}_{p,q}$ with the $\omega_{j,k}$ in \eqref{eq:grunslyomega}. One can easily check that for each $N$ in \eqref{eq:grunslyomega}, $p+q\leq 2N+1$ numbers of $W_{p,q}$ appears in \eqref{eq:grunslyomega}. Therefore, in order to hold the \eqref{eq:grunslyomega} till $N$,  for \eqref{eq:fdef}, one has to retain terms in EFT expansion \eqref{eq:Wpqdef} up to
\be
p+q\geq 2N+1\,.
\ee
Before delving into generalities, we begin by considering a toy problem.
\subsubsection{A toy example: scalar EFT approximation}
Since the general case of the univalence of the convex sum of univalent functions is not entirely clear (or more appropriately known to us), let us consider a toy problem below. This problem is enlightening for several reasons. For starters, the amplitudes we consider below are in two ``standard'' forms. Both of these were considered in \cite{Caron-Huot:2020cmc} to study scalar EFTs, and it was found that scalar EFTs could be approximated as a convex sum of the amplitudes below.
Thus consider the sum 
\be
\mt^{(toy)}(s_1,s_2)=\mu_1 \mt_1(s_1,s_2)+\mu_2 \mt_0(s_1,s_2)
\ee
where
\be
\begin{split}
	&\mt_0(s_1,s_2)=\frac{1}{M_1^2-s_1}+\frac{1}{M_1^2-s_2}+\frac{1}{M_1^2+s_1+s_2}-\frac{3}{M_1^2}=\frac{27 a^2 \tilde{z} \left(2 M_1^2-3 a\right)}{27 a^3 \tilde{z}-27 a^2 M_1^2 \tilde{z}-M_1^6 \left(\tilde{z}-1\right)^2}\\
	&\mt_1(s_1,s_2)=-\frac{1}{\left(M_2^2-s_1\right) \left(M_2^2-s_2\right) \left(M_2^2+s_1+s_2\right)}=\frac{\left(\tilde{z}-1\right)^2}{27 a^3 \tilde{z}-27 a^2 M_2^2 \tilde{z}-M_2^6 \left(\tilde{z}-1\right)^2}\,.
\end{split}
\ee
In the range (we are considering real $a$ only), $-\frac{M_{1,2}^2}{3}<a< \frac{2M_{1,2}^2}{3}$, individually $\mt_0(s_1,s_2),~ \mt_1(s_1,s_2)$ do not have any singularity inside the unit disk. A straightforward calculation shows  for both of them $\omega_{j,k}=-\frac{\d_{j,k}}{k}$ for $j>0, k>0$. Therefore, $\mt_0(s_1,s_2),~ \mt_1(s_1,s_2)$ are individually univalent inside the unit disk when the restriction on $a$ holds. Now the sum of $\mt_0(s_1,s_2),~ \mt_1(s_1,s_2)$, we denoted $\mt^{(toy)}(s_1,s_2)$ does not have any singularity inside the unit disk. We can check the univalence of the combination using the Grunsky inequalities again. Quite remarkably, if $M_1=M_2$,  we again find that $\omega_{j,k}=-\frac{\d_{j,k}}{k}$ for $j>0, k>0$. Therefore, for any $\l_1,\l_2$ for the range of $a$ above, the combination is univalent inside the disk! When $M_1\neq M_2$ we get nontrivial constraints for univalence. Already at $N=1$,  the Grunsky inequality \eqref{eq:grunslyomega} leads to
\be\label{bigdeal}
{\bigg |}1-\frac{729 a^4 \mu _1 \mu _2 \left(M_2^2-M_1^2\right){}^2 \left(a-M_1^2\right) \left(2 M_1^2-3a\right){}^3 }{M_1^{10} \left(\mu_1(a-M_1^2)+\mu_2 M_1^6(2M_1^2-3a)\right){}^2}{\bigg |}\leq 1\,,
\ee
where we have expanded around $M_1^2=M_2^2$ and retained only the leading term. This leads to a constraint on $\mu_1,\mu_2$ in terms of $M_1, M_2, a$. One interesting point to make note of is the following: for $a\sim 0$, for the above condition to hold, we will need $\mu_1\mu_2<0$. This is a consequence of unitarity and conforms with the signs in \cite{Caron-Huot:2020cmc}. We leave a detailed investigation of such constraints for the future\footnote{In light of such complications, it seems to us that it will be more useful in physics to think about an approximate notion of univalence, may be saying that a function is approximately univalent if the first few $N$'s in the Grunsky inequalities hold.}. Next, we will consider expanding the amplitude around $a\sim 0$ and will find that to leading order the Grunsky inequalities hold. 

\subsubsection{Proof of univalence of $f(\zt,a)$ for $|a|$ small}
Using  the expansion \eqref{eq:Wpqdef}, one can calculate the Grunsky coefficients, $\{\omega_{j,k}\}$, for $f(\zt,a)$ in leading order of small $a$ to obtain 
\be
\omega_{j,k}=-\frac{\delta _{j,k}}{k}+\frac{729 a^4 j k \left(\mathcal{W}_{1,0} \mathcal{W}_{3,0}-\mathcal{W}_{2,0}^2\right)}{\mathcal{W}_{1,0}^2}+O(a^5).
\ee
From positivity \cite[eq (6.4)]{tolley}, $\mathcal{W}_{1,0} \mathcal{W}_{3,0}\geq \mathcal{W}_{2,0}^2$. Therefore second term in the above equation is positive. 
Therefore, we get
\be
|\omega_{j,k}|\leq \frac{\delta _{j,k}}{k}\text{ , for j=k}\,.
\ee
The off-diagonal terms ($j\neq k$) starts at $O(a^4)$. From \eqref{eq:grunslyomega}, we can say that for small $a$, Grunsky inequalities are satisfied by $f(\tilde{z},a)$ in \eqref{eq:fdef}.


We can push the Grunsky inequalities further. If we assume that perturbatively around $a=0$, univalence should hold, then we can derive nonlinear inequalities by making clever choice for the complex parameters $\lambda_k$ in \eqref{eq:grunslyomega}. For instance, for $N=2$ in  \eqref{eq:grunslyomega}, choosing $\lambda_2=-\lambda_1/2$ and $\lambda_1,\lambda_2$ as real, we easily find the following complicated nonlinear inequality
\be\label{complic}
-3\mW_{2,0}^4+8 \mW_{1,0}\mW_{2,0}^2\mW_{3,0}-4\mW_{1,0}^2\mW_{2,0}\mW_{4,0}-3\mW_{1,0}^2\mW_{3,0}^2+2 \mW_{1,0}^3\mW_{5,0}\geq 0\,.
\ee
We have verified that the string amplitude, as well as the pion S-matrices, satisfy this inequality. \subsection{Bounds in case of EFTs:} 
In an EFT, usually, the Lagrangian is known up to some energy scale. From that information, one can calculate the amplitude up to that scale.  In such cases, we can subtract off the known part of the amplitude. These steps result in a shift in  the lower limit of the dispersion integral \eqref{crossdisp} by the scale $\d_0$, namely $\m\to \m+3\d_0/2$ (see \cite{ASAZ}). Therefore, making this replacement in \eqref{w0110}, we have
\be \label{mudelta}
-\frac{9}{4\m+6\d_0}<\frac{\mathcal{W}_{0,1}}{\mathcal{W}_{1,0}}< \frac{9}{2\m+3\d_0}\,.
\ee
Here $\mW$ are the Wilson coefficients of the amplitude with subtractions. Now notice that if we consider $\delta_0\gg \mu$ then we have
\be\label{strong}
-\frac{3}{2\d_0}<\frac{\mathcal{W}_{0,1}}{\mathcal{W}_{1,0}}< \frac{3}{\d_0}\,.
\ee
Let us compare this to \cite{Caron-Huot:2020cmc}. Converting their results to our conventions, we find that the lower bound above is identical to their findings--this is corroborated by the results of \cite{tolley} as well as what arises from crossing symmetric dispersion relations \cite{ASAZ}. The other side of the bound is more interesting. The strongest result in $d=4$ in \cite{Caron-Huot:2020cmc} places the upper bound at $\approx 5.35/\delta_0$ in our conventions. Their approach also makes the bound spacetime dimension dependent. Now remarkably, the bound we quote above and the $d\gg 1$ limit of \cite{Caron-Huot:2020cmc} are identical! In EFT approaches, one takes the so-called null constraints or locality constraints and expands in the limit $\delta_0\gg\mu$.  It is possible that a more exact approach building on \cite{Caron-Huot:2020cmc} will lead to a stronger bound as in \eqref{strong}.

Let us now comment on the behaviour in the figure (\ref{fig:W01W10pion}). First, notice that all S-matrices appear to respect the upper bound we have found above; for the S-matrix bootstrap results, we set $\delta_0=0$. For comparison, note that for 1-loop $\phi^4$, and 2-loop chiral perturbation theory, we have
\be
\left(\frac{\mW_{0,1}}{\mW_{1,0}}\right)_{\phi^4}\approx-0.315 \,,\quad \left(\frac{\mW_{0,1}}{\mW_{1,0}}\right)_{\chi-PT}\approx-0.135\,,
\ee
both in units where $m=1$. These numbers would be closer to the lower black dashed line in the figure (\ref{fig:W01W10pion}), which is the bound that is common in all approaches so far. The upper black dashed line is what we find in the current paper. For future work, it will be interesting to search for an interpolating bound as in \eqref{mudelta} which enables us to interpolate between $\delta_0=0$ and $\delta_0\gg \mu$. 

\section{Summary}\label{discuss}

This chapter has explored a remarkable correspondence between aspects of geometric function theory and quantum field theory. Let us summarize our accomplishments in this endeavour. 
\begin{itemize}
	\item We proved that the kernel of the crossing symmetric dispersion relation is univalent  inside the open disc $\D$ for a range of $a$. 
	
	\item The univalence of the kernel enabled us to derive upper bounds on the Taylor coefficients of the scattering amplitude via application of the de Branges's theorem (Bieberbach conjecture) to the kernel. These bounds can be used to obtain inequalities concerning the Wilson coefficients. We found strong bounds which are respected by all theories considered in this paper, which include 1-loop $\phi^4$, pion S-matrix bootstrap (which included a plethora of examples that respect unitarity, crossing symmetry and has information about the standard model $\rho$-meson mass) and even the massless pole subtracted string tree-level dilaton scattering.
	\item We further derived two-sided bounds on the scattering amplitude. In deriving the upper bound, we used the univalence of the kernel in the form of the Koebe growth theorem. The upper bound, expressed in terms of the usual Mandelstam variables, translates to, for large $|s|$, fixed $t$, $|\mm(s,t)|\lesssim s^2$.
	
	\item We proved that to leading order in $a$, around $a\sim 0$, the scattering amplitude is univalent as it respects the Grunsky inequalities.
\end{itemize}

\begin{subappendices} 
\section*{Appendix}
\section{Various Amplitudes}

\subsection{Tree level type II superstring theory amplitude }
%
The low energy expansion of the type II superstring amplitude is well known, see for example \cite{green} for a recent discussion. The amplitude after stripping off a kinematic factor and subtracting off the massless pole is given below. This is what we will use. In order to facilitate expansion, it is also useful to recast the Gamma function in terms of an exponential of sum of Zeta functions as in \cite{green}.
\be
\mt^{(cl)}(s_ 1, s _2)=-\frac{\Gamma\left(1-s_{1}\right) \Gamma\left(1-s_{2}\right) \Gamma\left(s_{1}+s_{2}+1\right)}{s_{1} s_{2}\left(s_{1}+s_{2}\right) \Gamma\left(s_{1}+1\right) \Gamma\left(-s_{1}-s_{2}+1\right) \Gamma\left(s_{2}+1\right)}+\frac{1}{s_{1} s_{2}\left(s_{1}+s_{2}\right)}
\ee

\begin{table}[hbt!]
	\centering
	\begin{tabular}{|c|c|c|c|c|c|c|}
		\hline
		$\mathcal{W}_{p,q}$& \text{q=0} & \text{q=1} & \text{q=2} & \text{q=3} & \text{q=4} & \text{q=5} \\ \hline
		p=0& 2.40411 & -2.88988 & 2.98387 & -2.99786 & 2.99973 & -2.99997 \\ \hline
		p=1& 2.07386 & -4.98578 & 7.99419 & -10.9987 & 13.9998 & -17. \\ \hline
		p=2& 2.0167 & -6.99881 & 14.9984 & -25.9995 & 39.9999 & -57. \\ \hline
		p=3& 2.00402 & -9.00023 & 23.9996 & -49.9998 & 89.9999 & -147. \\ \hline
		p=4 & 2.00099 & -11.0002 & 34.9999 & -84.9999 & 175. & -322. \\ \hline
		p=5& 2.00025 & -13.0001 & 48. & -133. & 308. & -630. \\ \hline
	\end{tabular}
	\caption{$\mathcal{W}_{p,q}$ for tree level type II superstring theory amplitude}
	\label{tab:stringWpq}
\end{table}
Note that we have stripped off the kinematic factor $x^2=(s_1 s_2 +s_2 s_3+s_1 s_3)^2$. Had we retained it then the graviton pole subtracted amplitude in the Regge limit would have behaved like $|s_1|^{2}/t$ so that the dispersion relation would need three subtractions.  Therefore, it is important that we remove this kinematic factor in what we do. The $a_\ell$'s with this factor removed continue to be positive--which is the main thing we used in our derivation.
\subsection{1-loop $\phi^4$ amplitude}
We just note the well-known standard result for the 1-loop $\phi^4$ amplitude. 
\begin{align}\label{eq:phi4amp}
	\begin{split}
\mt^{(\phi^4)}(s_ 1, s_ 2)&=-\left[\mathcal{F}(s_1)+\mathcal{F}(s_2)+\mathcal{F}(-s_1-s_2)\right],\\
\mathcal{F}(x)&:=\frac{1}{\sqrt{x+\frac{4}{3}}}\left[2 \sqrt{x-\frac{8}{3}} \tanh ^{-1}\left(\frac{\sqrt{x+\frac{4}{3}}}{\sqrt{x-\frac{8}{3}}}\right)\right].
\end{split}
\end{align}

\begin{table}[H]
	\centering
	\scalebox{0.80}{
	\begin{tabular}{|c|c|c|c|c|c|c|}
		\hline
		$\mathcal{W}_{p,q}$&  \text{q=0} & \text{q=1} & \text{q=2} & \text{q=3} & \text{q=4} & \text{q=5} \\\hline
		p=0& -5.22252 & -0.0209238 & 0.000401094 & -0.0000116118 & \text{3.9934$\times$ $10^{-7}$} & -\text{1.5104$\times$ $10^{-8}$} \\\hline
		p=1& 0.0663542 & -0.0023309 & 0.0000983248 & -\text{4.4442$\times$ $10^{-6}$} & \text{2.0832$\times$ $10^{-7}$} & - \\\hline
		p=2& 0.00344623 & -0.00027954 & 0.00001862 & -\text{1.1521$\times$ $10^{-6}$} &- & - \\\hline
		p=3& 0.000267396 & -0.0000348355 & \text{3.1948$\times$ $10^{-6}$} & -\text{2.5174$\times$ $10^{-7}$} &- & - \\\hline
		p=4& 0.0000245812 & -\text{4.4442$\times$ $10^{-6}$} & \text{5.2081$\times$ $10^{-7}$} & - &- & - \\\hline
		p=5& \text{2.4827$\times$ $10^{-6}$} & -\text{5.7605$\times$ $10^{-7}$} & - &- & - & - \\\hline
	\end{tabular}}
	\caption{$\mathcal{W}_{p,q}$ for 1-loop $\phi^4$ amplitude}
	\label{tab:stringWpq}
\end{table}

\subsection{ Amplitude for pion scattering from S-matrix bootstrap}
The S-matrix bootstrap puts constraints on pion scattering using unitarity and crossing symmetry. Some additional phenomenological inputs like $\rho$-meson mass or certain theoretical constraints like S/D wave scattering length inequalities are used. For more details, the reader is referred to \cite{gpv, bhsst, bst}. The allowed S-matrices are displayed as regions on the Adler-zeros $(s_0,s_2)$ plane. In \cite{bhsst}, a river like region of S-matrices on this plane was identified. The chiral perturbation theory appeared to lie close to a kink-like feature near $s_0=0.35$. As such this particular S-matrix is of interest to us. In the main text, we have considered a plethora of S-matrices like the lake in \cite{gpv}, the upper and lower boundaries of the river in \cite{bhsst} as well as the more interesting line of minimization (where the total scattering cross-section is minimized for a given $s_0$) in \cite{bst}. 
The amplitude for pion scattering from S-matrix bootstrap with $s_0=0.35$ is given\footnote{One can write these kind of expansion for a general $a$ upto desired order in $\tilde{z}$. To minimize numerical errors in our calculations, we had to  rationalize upto 20 decimal place.} below for $a=1/2$
\be
\begin{split}
	&\mt(\tilde{z},a)=-1.90562-55.586 \tilde{z}-75.7314 \tilde{z}^2-49.2812 \tilde{z}^3+3.43872 \tilde{z}^4+45.4445 \tilde{z}^5+O\left(\tilde{z}^6\right)
\end{split}
\ee
In table \eqref{tab:pionWpq}, we have listed various $\mathcal{W}_{p,q}$.
\begin{table}[hbt!]
	\centering
	\scalebox{0.85}{
	\begin{tabular}{|c|c|c|c|c|c|c|}
		\hline
		$\mathcal{W}_{p,q}$&\text{q=0} & \text{q=1} & \text{q=2} & \text{q=3} & \text{q=4} & \text{q=5} \\
		\hline
		\text{p=0}& -1.90562 & 5.02671 & -0.249527 & 0.0118008 & -0.000555517 & 0.0000262344 \\
		\hline
		\text{p=1}&  5.72161 & 0.395863 & -0.0520982 & 0.00402939 & -0.000264317 &- \\
		\hline
		\text{p=2}&  0.642298 & 0.0217519 & -0.00787377 & 0.000904172 & - & - \\
		\hline
		\text{p=3}&  0.0796397 & -0.000836409 & -0.000995454 & 0.000166504 &- &- \\
		\hline
		\text{p=4}&  0.0101505 & -0.000579411 & -0.000103708 &-& -&- \\
		\hline
		\text{p=5}&  0.0013093 & -0.000136893 &- &- & - & - \\
		\hline
	\end{tabular}}
	\caption{$\mathcal{W}_{p,q}$ for pion scattering from S-matrix bootstrap with $s_0=0.35$}
	\label{tab:pionWpq}
\end{table}

\section{Grunsky inequalities \eqref{eq:grunslyomega} and $s_0=0.35$ pion amplitude}
Using the table \eqref{tab:pionWpq}, one can check the Grunsky inequalities \eqref{eq:grunslyomega} for $N=2$, with some random $\l_1,\l_2$. Since we are truncating the sum over Wilson coefficient expansion, if this truncated sum comes from a univalent function (in the range of $-\frac{2\m}{9}<a<\frac{4\m}{9}$) and the truncated sum is itself univalent. Therefore, it may be expected that the radius of the disc where univalence holds should be smaller. This translates into the range of $a$, which should be now $-\frac{\m}{9}<a<\frac{2\m}{9}$ (or maybe a smaller range of $a$). This can be realized from  Szeg\"{o}'s theorem, since $a^{2n}$ always comes with $\tilde{z}^n$, reducing the radius to $1/4$ means reducing the range of $a$  by $1/2$ for unit disk in $\tilde{z}$-plane. One can see in figure \eqref{fig:Grunsky_pion_random} that our expectation matches\footnote{There can be some random $\lambda_1,\l_2$ for which the curves can be slightly below the line $a=8/9$} exactly. The main point of the above discussion is that, \textit{there exists a finite range of $a$ for which the $f(\tilde{z},a)$ is univalent.}
\begin{figure}[hbt!]
	\centering
	\includegraphics[width=0.50\textwidth]{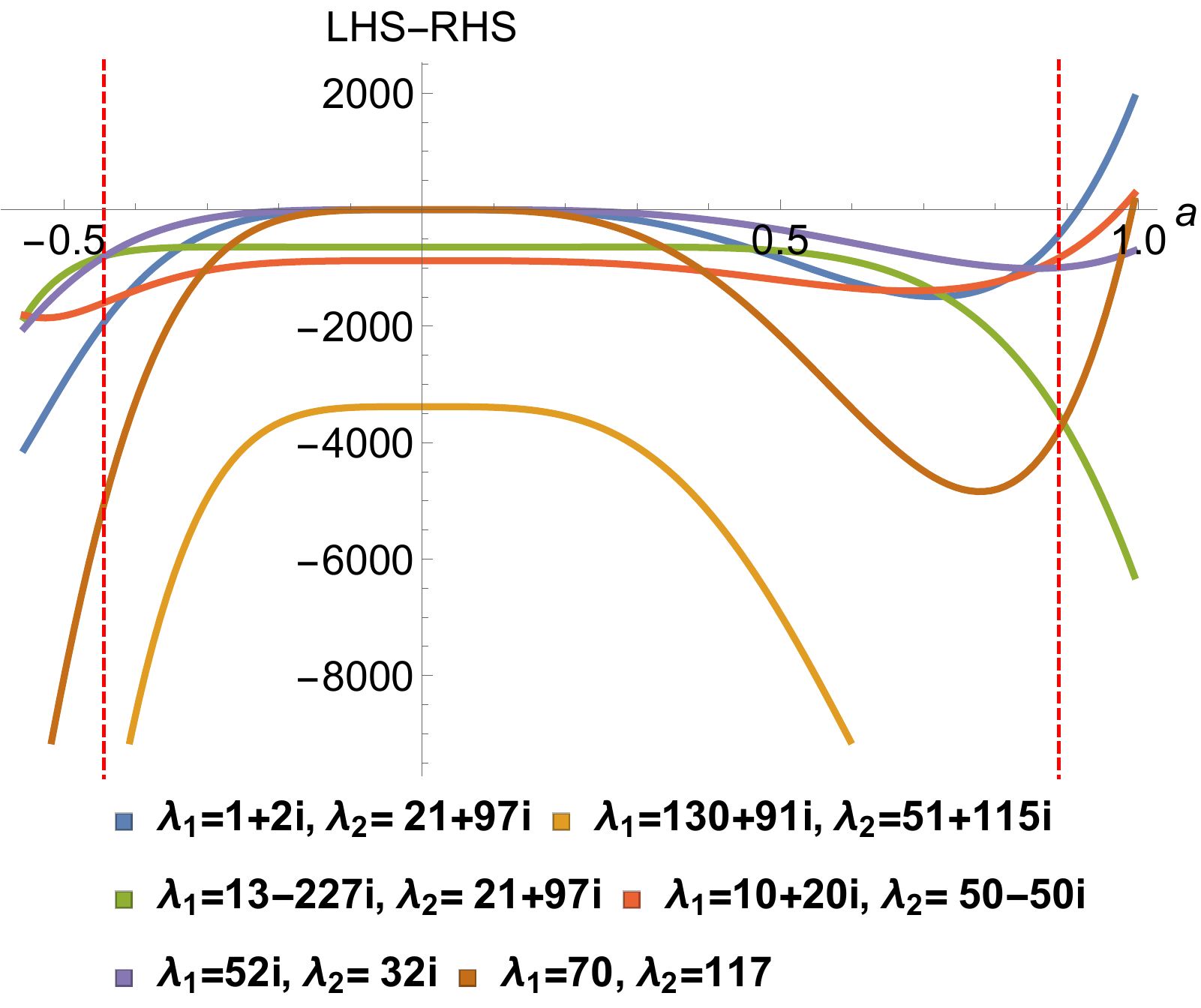}
	\caption{Red lines are the $a=-4/9,8/9$}
	\label{fig:Grunsky_pion_random}
\end{figure} 

\section{Constraints from \eqref{eq:n3alphaW}}
Suppose we consider $\mathcal{W}_{0,1},\mathcal{W}_{1,1},\mathcal{W}_{1,0},\mathcal{W}_{2,1},\mathcal{W}_{1,2}\mathcal{W}_{2,0},\mathcal{W}_{3,0}$ as given. We can constrain $\mathcal{W}_{0,3}$, using \eqref{eq:n3alphaW}. See figure \eqref{fig:a3bound}.
\begin{figure}[H]
	\centering
	\begin{subfigure}[b]{0.45\textwidth}
		\centering
		\includegraphics[width=\textwidth]{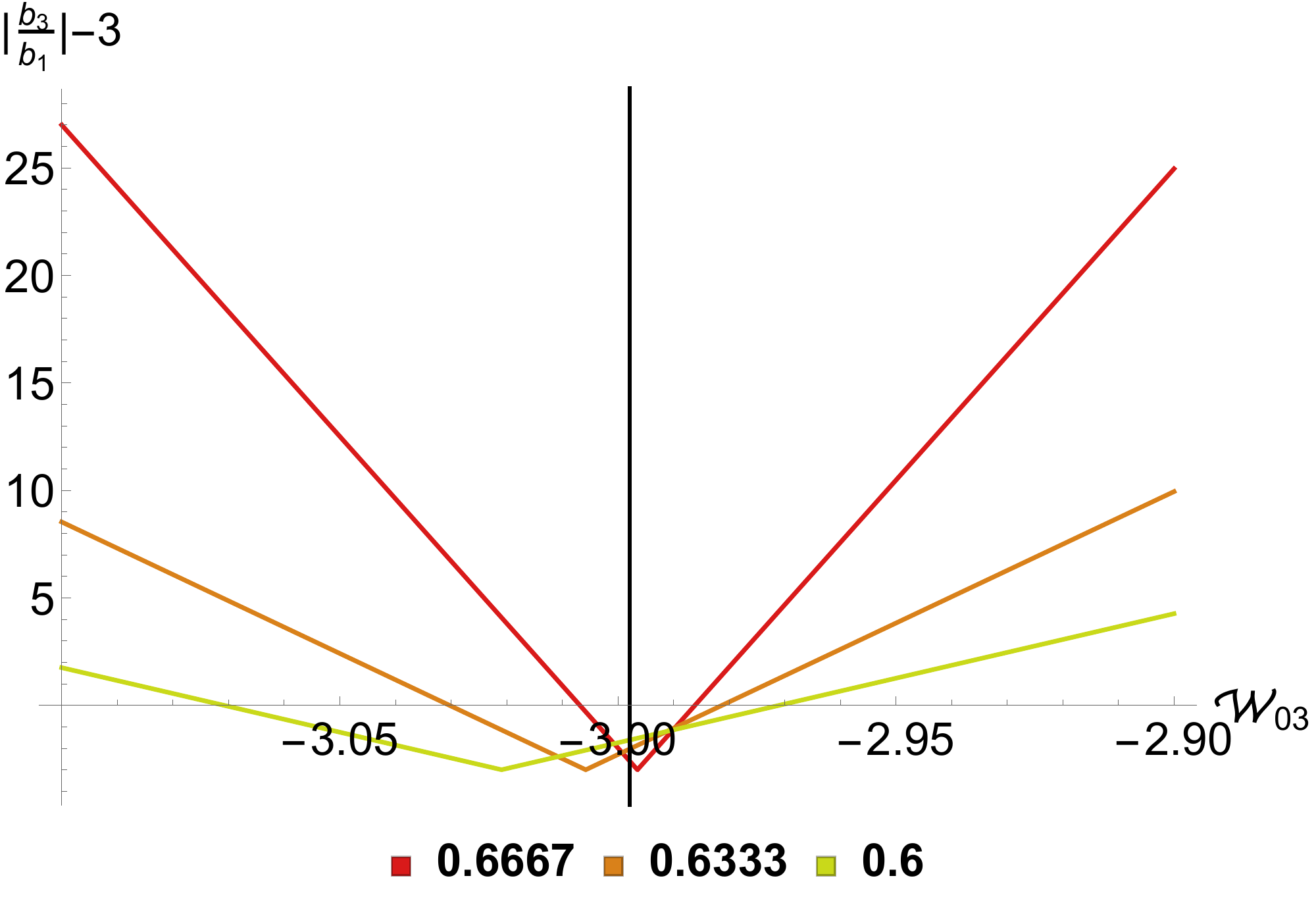}
		\caption{Tree level type II string amplitude}
		\label{fig:a3string}
	\end{subfigure}
	\hfill
	\begin{subfigure}[b]{0.45\textwidth}
		\centering
		\includegraphics[width=\textwidth]{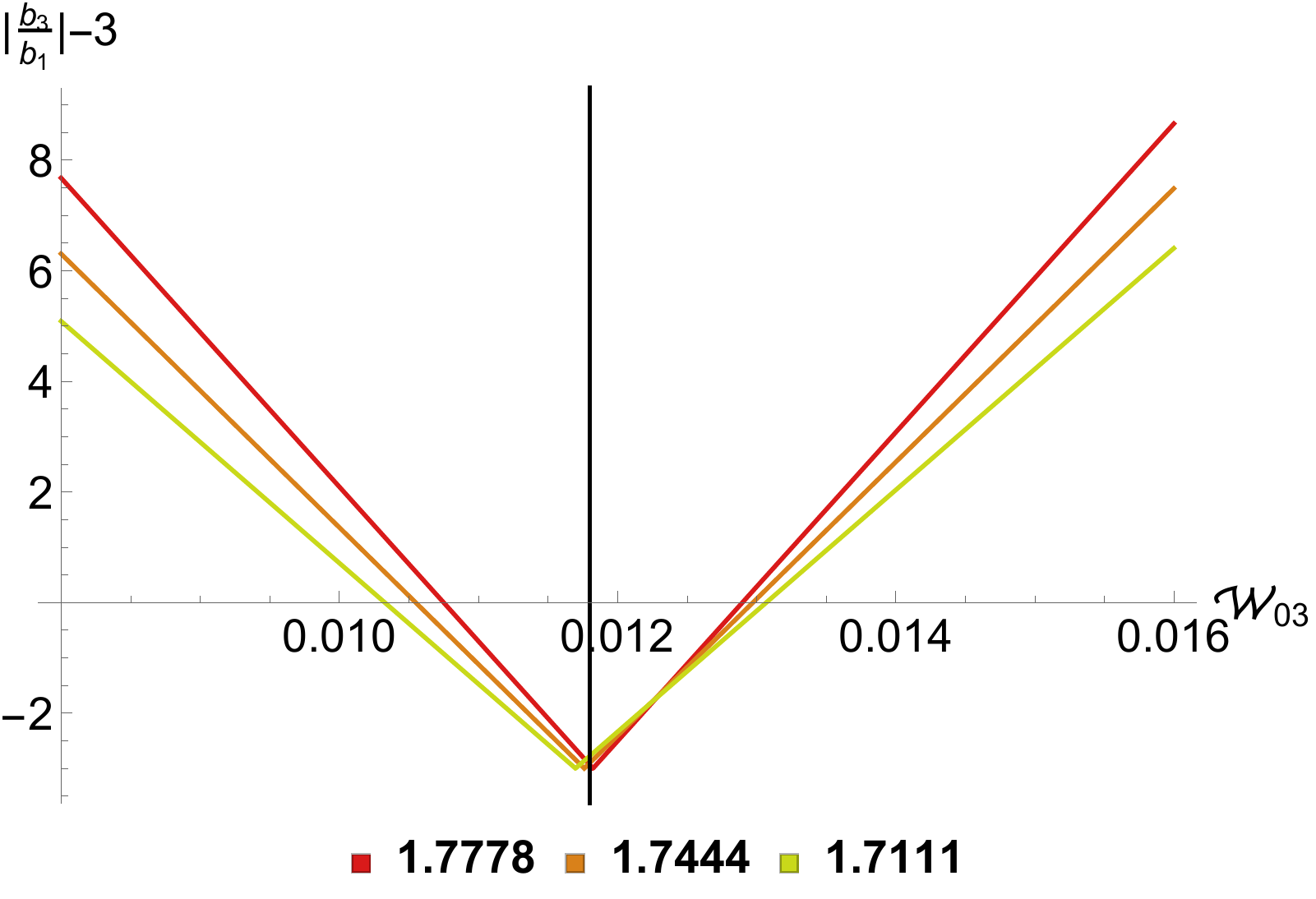}
		\caption{Pion scattering amplitude, $s_0=0.35$}
		\label{fig:a3upper4}
	\end{subfigure}
	\caption{Constraints on Wilson coefficient $W_{0,3}$ using \eqref{eq:n3alphaW}, where $\mathcal{W}_{0,1},\mathcal{W}_{1,1},\mathcal{W}_{1,0},\mathcal{W}_{2,1},\mathcal{W}_{1,2}\mathcal{W}_{2,0},\mathcal{W}_{3,0}$ are given. Figure shows that bound on the $\mathcal{W}_{0,3}$. Since $\left|\frac{b_3}{b_1}\right|-3$ should be less than zero, $\mathcal{W}_{0,3}$ must lie inside the triangle. Black dashed line is the exact answer. Different values of $a$ are indicated with different colours.}
	\label{fig:a3bound}
\end{figure}
\section{Nehari conditions in 1-loop $\phi^4$-theory. }
Using the 1-loop $\phi^4$-amplitude, we can check the Nehari conditions. For the range $-4/9<a<16/9$, we find that Nehari necessary condition \eqref{NehariN} always holds. Further, we find that Nehari sufficient condition \eqref{NehariS} does not always hold within the unit circle. Nevertheless, there are regions where 1-loop $\phi^4$-amplitude respects Nehari sufficient condition \eqref{NehariS}. For example within the radius\footnote{These can be realized replacing $z\to 2z/3$, and check the conditions for the given ranges of $a$.} of $\frac{2}{3}$ for the range $-4/9<a<16/9\,,$  the Nehari sufficient condition \eqref{NehariS} holds.

We can also check the Nehari conditions in $a\sim 0$ region. We can expand the amplitude around $a=0$, then calculate the Schwarzian derivative \eqref{eq:sch}. For example upto $a^4$, we find
\be\label{eq:nehariphi4a4}
\{f(z),z\}=-\frac{6 }{\left(z^2-1\right)^2}\left(1-\frac{0.971 a^4 (z+1)^4}{(z-1)^4}\right)
\ee
Of course, for the full range of $a$, the above \eqref{eq:nehariphi4a4} need not to satisfy the Nehari necessary condition \eqref{NehariN}, since this is an EFT type expansion, and by Szeg\"{o} theorem, we don't expect the univalent to hold in the same range of $-4/9<a<16/9$. Further, there always exists a smaller range of $a$, where it is univalent. For example if we consider the radius $1/2$, the above \eqref{eq:nehariphi4a4} satisfies Nehari necessary condition \eqref{NehariN} for $-0.301<a<0.301$. Qualitatively similar features hold for the string amplitude as well.
\end{subappendices}

%% file: spinbootnew.tex
\chapter{  Crossing Symmetric Spinning S-Matrix Bootstrap : EFT Bounds}\label{spinboot}
\section{Introduction}

In the previous chapter, we explored the connection between crossing-symmetric dispersion relation (CSDR) and univalent functions to obtain two-sided bounds on Wilson coefficients of the low energy expansion of scattering amplitude. The dispersion variable $z$ and the parameter $a$ of CSDR, rather than the usual  Mandelstam variables $(s,\,t,\,u)$, led us to this fascinating connection. The amplitude for identical scalars then has nontrivial properties in terms of univalent function in the complex $z$ plane. The authors of \cite{Raman:2021pkf} unearthed yet another connection between scattering amplitudes and geometric function theory in the form of typically real functions. They showed that for a particular range of $a$, the amplitude is typically real, i.e. it satisfies the condition
\be
\text{Im.} [f(z)] \text{Im. }[z]>0\,
\ee
inside the unit disk $|z|<1$, which in turn imposes the Bieberbach-Rogosinski (BR) two-sided bounds on the Taylor expansion coefficients of $f(z)$, \eqref{BR bound}. In terms of Wilson coefficients, an argument based on the Markov brothers' inequality, as shown in \cite{Raman:2021pkf}, leads to two-sided bounds on the ratios of Wilson coefficients. 

The analysis in this chapter will extend the CSDR for identical, neutral external particles carrying spin. Our formalism is general, although, for concreteness, we will focus on the 2-2 scattering of photons and gravitons and neutral Majorana fermions. In the photon and graviton cases, we will be able to identify linear combinations of helicity amplitudes whose Taylor coefficients are two-sided bounded using GFT arguments. We will be able to write down a general expression for the locality constraints. Our formalism paves the way for a future systematic study of the S-matrix bootstrap for the 2-2 scattering of identical particles with spin. 
Let us recall that the crossing symmetric dispersion relation of a scalar amplitude $\mt_0(s_1,s_2)$ takes the following form,
\begin{equation}
	\mt_0(s_1,s_2)=\alpha_0+\frac{1}{\pi}\int_{M^2}^{\infty} \frac{d s_1'}{s_1'}\mathcal{A}(s_1',s_2^+(s_1',a))H(s_1';s_1,s_2,s_3),
\end{equation}
where $\mathcal{A}(s_1',s_2^+(s_1',a))$, called the absorptive part, is the $s$- channel discontinuity and $H(s_1';s_1,s_2,s_3)$ is a manifestly crossing symmetric kernel. The parameter $a=(s_1 s_2 s_3)/(s_1 s_2+s_2s_3+s_3s_1)\equiv y/x$ is kept fixed writing this dispersion relation and $s_2^+$ is one of the two roots obtained from this equation on using $s_1+s_2+s_3=0$. For a massive theory with a gap, as for pion scattering, the dispersive integral starts at $8 m^2/3$, where $m$ is the mass of the pion. In this case $s_1=s-4m^2/3, s_2=t-4m^2/3, s_3=u-4m^2/3$ with $s,t,u$ being the usual Mandelstam variables. For EFTs, the lower limit starts at some cut-off $M^2$ and all external particles are considered massless. The absorptive part can be expanded in partial waves involving Gegenbauer polynomials. Then Taylor expanding around $a=0$ leads to the conclusion that for each partial wave, there are in principle any arbitrary power of $a$, and hence of $x=s_1 s_2+s_2 s_3+s_3 s_1$ which are absent for a local theory. On demanding that such powers responsible for non-local terms cancel,  leads to what we call the ``locality" constraints.

The range of the parameter $a$ is crucial in this story. For theories with a gap, as for pion scattering, axiomatic arguments can be used to find this range of $a$ \cite{Haldar:2021rri}. However, when describing EFTs \cite{tolley1, tolley3,Caron-Huot:2020cmc, rastelli}, these axiomatic arguments do not work. In this chapter, taking a leaf out of \cite{Caron-Huot:2020cmc}, we will use the locality constraints and linear programming and establish the range of $a$ where the absorptive part of the amplitude is positive. This is crucial since, in our approach, it is vital that this range of $a$ satisfies $-a_{min}<a<a_{max}$ with both ends non-zero to have two-sided bounds. One surprising conclusion that emerges from our analysis is that this range is related to the weak Low Spin Dominance (wLSD) conjecture made in \cite{Bern:2002kj}. Our findings lead to the conclusion that a few low lying spins control the sign of the absorptive part. For non-unitary theories, generically, the 
imaginary part of the partial wave coefficient (often referred to as the spectral function) is not of a definite sign. However, if it is known that the spectral function for some low lying spins is positive, it is possible then to have two-sided bounds in a local but non-unitary theory using wLSD.

We will focus on light-by-light scattering and graviton scattering in weakly coupled EFTs \cite{Bern:2002kj,Henriksson:2021ymi, rastelli} and derive two-sided bounds. We consider the linearly independent helicity amplitudes, $T^{\lambda_3 \lambda_4}_{\lambda_1 \lambda_2}(s_1, s_2, s_3)$, for 2-2 scattering $(\lambda_1 \lambda_2\rightarrow \lambda_3 \lambda_4)$ of graviton, photon and massive Majorana fermions in four spacetime dimensions. Here $\lambda_i$ are helicity labels and these take values $-j$ and $+j$ for massless particle with spin $j$ while there are $2j+1$ independent helicities for a massive particle with spin $j$. Generically these  helicity amplitudes mix among themselves under crossing,
\begin{equation}
	\mt^{\lambda_3 \lambda_4}_{\lambda_1 \lambda_2}(s_1, s_2, s_3)\rightarrow \sum_{i j k l} C_{ijkl} \mt^{\lambda_{k}\lambda_{l}}_{\lambda_i \lambda_j}(P_{s_1},P_{s_2},P_{s_3}),
\end{equation}
where $(P_{s_1},P_{s_2},P_{s_3})$ represents some permutation of $(s_1,s_2,s_3)$. Using representation theory of $S_3$ (the permutation group relevant for Mandelstam invariants), we construct a basis of crossing symmetric amplitudes, $F(s_1,s_2,s_3)$\footnote{We don't put any helicity labels for  crossing symmetric amplitudes since they are often combinations of amplitudes with different helicity labels \eqref{f1}.}, using the helicity amplitudes $T^{\lambda_3 \lambda_4}_{\lambda_1 \lambda_2}(s_1,s_2,s_3)$. These crossing symmetric helicity amplitudes transform as a singlet under $S_3$,
\begin{equation}
	F(s_1, s_2, s_3)=F(s_2,s_1,s_3)=F(s_3,s_2,s_1).
\end{equation}
We then write down locality constraints associated with the crossing symmetric amplitudes $F^{\lambda_3 \lambda_4}_{\lambda_1 \lambda_2}(s_1,s_2,s_3)$. Explicit formulae for a subclass of the amplitudes in closed form can be found. Our method allows us to write the locality constraints for all the crossing symmetric amplitudes. To be precise, we consider the following crossing symmetric amplitudes for the photon case\footnote{The relevant amplitudes for graviton are a bit subtle and will be dealt with in section \ref{gravitonbound}. We will just quote the results here.}, 
\begin{equation}\label{fintro}
	F^\gamma_1(s_1,s_2,s_3) = \mt_2(s_1,s_2,s_3),~~ F^\gamma_2(s_1,s_2,s_3)= \m\mt_1(s_1,s_2,s_3)+\mt_3(s_1,s_2,s_3)+\mt_4(s_1,s_2,s_3)
\end{equation} 
where the helicity amplitudes $\mt_i$ are defined in \eqref{PTboseph}. These amplitudes have the low energy EFT expansion 
\bea
F^{\g,i} (s_1,s_2)&=& \sum_{p,q}\mw^i_{p,q} x^p y^q  \,.
\eea

The locality constraints for the amplitude $F^\gamma_2 (s_1,s_2,s_3) + x_1 F^\gamma_1(s_1,s_2,s_3)= \sum_{p,q}\mw^{(x_1)}_{p,q} x^p y^q$ for $x_1 \in [-1,1]$, are
$$\mw^{(x_1)}_{p,q}=0,~~\forall~~ p < 0.$$

We present the explicit expressions for $\mw^{(x_1)}_{p,q}$ using CSDR in the main text (eqn \eqref{photonbell}) and from those we obtain positivity conditions called $PB^\g_C$ (eqn \eqref{pbcphot}). We also show that the dispersive part of the amplitude can be written as a Typically Real function leading to bounds on the range of the variable $a$. In general, for massless theories the lower bound on $a=a_{min}$ is zero \cite{Raman:2021pkf}, which only leads to one-sided bounds. We observe that the Wigner-$d$ functions, $d^{\ell}_{m,n}(\sqrt{\xi(s_1, a)})$, are positive for all spins when its argument $\xi(s_1,a)$ is greater than 1. Adding a suitable linear combination of the locality constraints, we can show that the positivity of the absorptive part arises even when $\xi(s_1,a)<1$. This translates to $-a_{min}<a<a_{max}$. This indicates the dominance of low spin partial waves in EFTs and is called Low spin dominance (LSD). This behaviour was observed for gravitons in \cite{sasha, nima2}. We will show how this naturally emerges out of our analysis using the locality constraints. We will show that the lower range of $a$ tells us which spins dominate in the determination of the positivity of the absorptive part for $-a_{min}<a<a_{max}$. We demonstrate the same for the case of type-II string amplitude in appendix \ref{appG}.

After showing that the amplitude  is typically-real for a range of $a\in [-a_{min},a_{max}]$, we can directly find two-sided bounds on the ratio of Wilson coefficients $w_{p,q}=\frac{\mathcal{W}_{p,q}}{\mathcal{W}_{1,0}}$ from GFT. Below we show examples of bounds found for the scattering of scalars, photons and gravitons in Table \ref{tab1}. The detailed list of bounds for photon and graviton scattering are summarized in Table \ref{mintph} and \ref{mintgrav}.
\begin{table}[h!] 
	\centering 
	\scalebox{0.9}{
		\begin{tabular}{ |c| c|c|c|} 
			\hline
			Theory & EFT amplitudes & Range of $a$ and LSD & $w_{01}$ bound\\ 
			\hline
			Scalar & $F(s_1,s_2,s_3)= \mathcal{W}_{1,0}x+\mathcal{W}_{01} y+\cdots$& $-0.1933 M^2<a^{scalar}< \frac{2 M^2}{3}$& $\frac{-3}{2 M^2}< w_{0,1}<\frac{5.1733}{M^2}$\\ 
			& & (Spin-2 dominance) & \\
			\hline
			Photon & $F_2(s_1,s_2,s_3)= 2 g_2 x- 3 g_3 y+\cdots $  & $-0.1355 M^2<a^{\gamma}< \frac{2 M^2}{3}$ & $\frac{-4.902}{M^2}< \frac{g_3+ x_1 \frac{f_3}{3}}{g_2 + x_1 f_2}<\frac{1}{M^2}$\\
			& $F_1(s_1,s_2,s_3)= 2 f_2 x- f_3 y+\cdots$ & (Spin-3 dominance) & where $x_1 \in [-1,1]$\\
			\hline
			Graviton & $\tilde{F}_2^{h}= 2 x f_{0,0}+3 y f_{1,0}+\cdots$ & $-0.1933 M^2<a^{h}< \frac{2 M^2}{3}$ & \\ 
			& & ( Spin-2 dominance) &$-\frac{1}{M^2} < \frac{f_{1,0}}{f_{0,0}}<\frac{3.44}{ M^2}$\\
			\hline
	\end{tabular}}
	\caption{Example of two sided bounds we have found for scalars, photons and gravitons using GFT.}\label{tab1}
\end{table}

Apart from the conceptual clarity that the GFT techniques enable us, are there any technical advantages to using our approach? We wish to point out a couple of obvious ones. First, unlike the fixed-$t$ methods where one uses SDPB techniques and hence needs to worry about convergence in the spin, the dispersive variable as well as the number of null constraints, in our approach, once the range of $a$ has been determined, one only needs to check for convergence in the number of $BR$ inequalities we use. Second, we can write simple codes directly in Mathematica to study bounds. However, there are also some disadvantages. The main one is that while we do obtain two-sided bounds quite easily, these are not necessarily the sharpest ones possible since we do not use all the locality constraints. It is not clear to us if there is a way to get optimum bounds\footnote{A bound will be considered optimum if there is a consistent S-matrix saturating it.} using purely GFT techniques.

The chapter is organized as follows. In section \ref{crampconstr}, we describe the construction of fully crossing symmetric amplitude. Through multiple subsections of section \ref{csdrsec}, we describe the key formulas like CSDR, locality constraints, typical realness of the amplitude and then we introduce BR bounds as well. We also discuss in section \ref{csdrsec}  how low spin dominance emerge out of our analysis taking into account the locality constraints. Through section \ref{scalboundsec}, \ref{photonbound},  \ref{gravitonbound} we describe the bounds obtained for scalars, photons and gravitons respectively. We end our discussion with a summary of our findings in section \ref{summary}. Several technical details are relegated to multiple appendices at the end.

\section{Crossing symmetric amplitudes}\label{crampconstr}
In this section, we present a general construction for crossing symmetric amplitudes following \cite{Roskies}. Let us begin with a short review of scattering amplitudes of identical particles as irreducible representations (irreps) of $S_3$ \cite{SDC, Henning:2017fpj}. Consider the scattering of four identical particles (massive or massless, with or without spin) in $d=4$. The momenta of the particles satisfy, 

\begin{equation}\label{osmc}
	p_i^2=-m^2,\qquad \sum_{i=1}^{4}p_i^\mu=0,
\end{equation}
where $m$ is mass of each particle.
We use the mostly positive convention and define Mandelstam variables,
\begin{equation}\label{stu} \begin{split}
		s&:=-(p_1+p_2)^2=-(p_3+p_4)^2=2m^2-2 p_1.p_2 =2m^2-2 p_3.p_4,\\
		t&:=-(p_1+p_3)^2=-(p_2+p_4)^2=2m^2-2 p_1.p_3=2m^2-2 p_2.p_4,\\
		u&:=-(p_1+p_4)^2=-(p_2+p_3)^2=2m^2-2 p_1.p_4=2m^2-2 p_2.p_3.
	\end{split}
\end{equation}
Due to momentum conservation we have $s+t+u=4m^2$. For identical bosonic particles, the S-matrix is to be thought of as the function of Mandelstam invariants (and polarizations), which is $S_4$ invariant, the symmetry group of permutations of four particles. In the present context, $S_4$ acts on the momenta and the helicities of the particles. We usually impose the $S_4$ invariance in two steps. Recall that that $\Z_2 \times \Z_2$ is the normal subgroup of $S_4$ and the remnant symmetry is $\frac{S_4}{\Z_2 \times \Z_2}=S_3$. Action of $\Z_2 \times \Z_2$ on four objects $(1,2,3,4)$ is the simultaneous exchange of two particles- (12)(34), (13)(24) and (14)(23) while $S_3$ is the permutation of three objects $(1,2,3)$. Since the Mandelstam invariants $s, t$ and $u$ are invariant under the $\Z_2 \times \Z_2$, we first impose $\Z_2 \times \Z_2$ invariance, which leaves the Mandelstam invariants unchanged and we are left with the remnant $S_3$ symmetry which acts on $(s,t,u)$. Note that helicities (or equivalently tensor structures in higher dimensions) may not be $\Z_2 \times \Z_2$ invariant, and we might need to impose $\Z_2 \times \Z_2$ invariance. However, for most of the non-crossing symmetric helicity amplitudes that we consider in this work, the $\Z_2 \times \Z_2$ symmetry has already been taken care of \cite{spinpen}. The S-matrix, which is invariant under the $\Z_2 \times \Z_2$ invariance, is often referred to as ``Quasi-invariant" S-matrix. The ``Quasi-invariant" S-matrix can be decomposed into irreps of $S_3$ and the \emph{crossing equations are relations between the orbits of} $S_3$. 

To simplify the discussion, unless otherwise mentioned, we will work with the following shifted Mandelstam variables,
\begin{equation}
	s_1=s-\frac{4m^2}{3},\,\,\,\,\,\,\,\,\,\,\,\, s_2=t-\frac{4m^2}{3}\,\,\,\,\,\,\\,\,\,\,\,\, s_3=u-\frac{4m^2}{3},
\end{equation}
such that, $s_1+s_2+s_3=0$.
With the aid of the representation theory of $S_3$, which we  review in appendix  \ref{S3rep}, one can write that the most general Quasi-invariant S-matrix takes the form \cite{Roskies}
\begin{eqnarray}\label{gendecomp}
	F(s_1,s_2,s_3)&=&f(s_1,s_2,s_3) + (2s_1-s_2-s_3) g_1(s_1,s_2,s_3) + (s_2-s_3) g_2(s_1,s_2,s_3)\nonumber\\&&+ (2s_1^2-s_2^2-s_3^2) h_1(s_1,s_2,s_3) 
	+ (s_2^2-s_3^2) h_2(s_1,s_2,s_3)\nonumber\\&&+(s_1-s_2)(s_2-s_3)(s_3-s_1)j(s_1,s_2,s_3)\,,
\end{eqnarray}  
where $f(s_1,s_2,s_3), j(s_1,s_2,s_3), g_i(s_1,s_2,s_3)$ and $h_i(s_1,s_2,s_3)$ are crossing symmetric amplitudes. 
We can decompose $F(s_1,s_2,s_3)$ into irreps of $S_3$,
\bea\label{gendecomp1}
F(s_1,s_2,s_3)=f_{\rm Sym}(s_1,s_2,s_3)+f_{\rm Anti-sym}(s_1,s_2,s_3)+f_{\rm Mixed+}(s_1,s_2,s_3)+f_{\rm Mixed-}(s_1,s_2,s_3)\,.
\eea
From eqn \eqref{irrepfstu} and \eqref{gendecomp1}, we have the following set of equations
\begin{eqnarray} \label{relationcrosssym}
	f_{\rm Sym}(s_1,s_2,s_3)&=& f(s_1,s_2,s_3)\,,\nonumber\\
	f_{\rm Anti-sym}(s_1,s_2,s_3)&=& (s_1-s_2)(s_2-s_3)(s_3-s_1)j(s_1,s_2,s_3)\,,\nonumber\\
	f_{\rm Mixed+}(s_1,s_2,s_3)&=&(2s_1-s_2-s_3) g_1(s_1,s_2,s_3)+(2s_1^2-s_2^2-s_3^2) h_1(s_1,s_2,s_3)\,, \nonumber\\
	f_{\rm Mixed-}(s_1,s_2,s_3)&=&(s_2-s_3) g_2(s_1,s_2,s_3)+ (s_2^2-s_3^2) h_2(s_1,s_2,s_3)\,.
\end{eqnarray} 
This can be inverted to give the required crossing symmetric basis \cite{Roskies}.  
\begin{eqnarray}\label{crosssymmbasis}
	f(s_1,s_2,s_3)&=&f_{\rm Sym}(s_1,s_2,s_3)\,,\nonumber\\
	j(s_1,s_2,s_3)&=&\frac{f_{\rm Anti-sym}(s_1,s_2,s_3)}{(s_1-s_2)(s_2-s_3)(s_3-s_1)}\,,\nonumber\\
	g_1(s_1,s_2,s_3)&=&\frac{f_{\rm Mixed+}(s_1,s_2,s_3)(s_1^2+s_2^2-2s_3^2)-f_{\rm Mixed+}(s_3,s_1,s_2)(s_2^2+s_3^2-2s_1^2)}{3 (s_1-s_2)(s_2-s_3)(s_3-s_1)}\,,\nonumber\\
	h_1(s_1,s_2,s_3)&=&\frac{f_{\rm Mixed+}(s_3,s_1,s_2)(s_3+s_2-2 s_1)-f_{\rm Mixed+}(s_1,s_2,s_3)(s_1+s_2-2 s_3)}{3(s_1-s_2)(s_2-s_3)(s_3-s_1)}\,,\nonumber\\
	g_2(s_1,s_2,s_3)&=&\frac{f_{\rm Mixed-}(s_3,s_1,s_2)(s_2^2-s_3^2)-f_{\rm Mixed-}(s_1,s_2,s_3)(s_1^2-s_2^2)}{(s_1-s_2)(s_2-s_3)(s_3-s_1)}\,, \nonumber\\
	h_2(s_1,s_2,s_3)&=&\frac{f_{\rm Mixed-}(s_1,s_2,s_3)(s_1-s_2)-f_{\rm Mixed-}(s_3,s_1,s_2)(s_2-s_3)}{(s_1-s_2)(s_2-s_3)(s_3-s_1)}\,.
\end{eqnarray}
The following additional comments are in order:
\begin{itemize}
	\item Given a Quasi-invariant S-matrix the algorithm to construct the crossing symmetric basis, therefore, is straightforward. We construct the irreps $\{f_{\rm Sym},f_{\rm Anti-sym}, f_{\rm Mixed \pm}\}$ following \eqref{irrepfstu} and use \eqref{crosssymmbasis} to construct the crossing symmetric basis\footnote{See also \cite{Mahoux:1974ej, AZ21} for similar considerations for the massive pion case.}.
	\item The basis elements \emph{do not} have any spurious poles at $s_i=s_j$ and are analytic functions of $s_1,s_2,s_3$, which can be easily checked by plugging \eqref{irrepfstu} into  \eqref{crosssymmbasis}.
	\item The basis in eq.\eqref{crosssymmbasis} is not unique since the last two equations of \eqref{relationcrosssym} are two 2 equations for 4 unknowns $\{g_i(s_1,s_2,s_3),h_i(s_1,s_2,s_3)\}_{i=1,2}$. One can use any permutation of the arguments of $f_{\rm Mixed \pm}$ to get a system of full rank. We have used the permutations $s_i \rightarrow s_{i+1 {\rm {mod}}(3)}$ on $f_{\rm Mixed \pm}$ as these are best suited for our purposes\footnote{In \cite{Roskies} the permutation $s_1 \rightarrow s_3$ was used. We warn the reader referring to \cite{Roskies} that there is a minor typo in the analog of \eqref{gendecomp} where the sign of the $j(s_1,s_2,s_3)$ term is wrong.}.
	\item If $F(s_1,s_2,s_3)$ is symmetric or anti-symmetric then only $f(s_1,s_2,s_3)$ and $j(s_1,s_2,s_3)$ are non zero respectively. Furthermore $F(s_1,s_2,s_3)$ is $t-u$ symmetric then only $\{f(s_1,s_2,s_3),g_1(s_1,s_2,s_3),$ \\
	\noindent$h_1(s_1,s_2,s_3)\}$ are nonzero. 
\end{itemize}


\subsection{Photons and Gravitons}
Now we apply the formalism developed in the previous section to the case of parity even photon and graviton amplitudes. We will work with helicity amplitudes and show that they transform in irreps of $S_3$. Subsequently, we construct the crossing symmetric amplitudes from them using \eqref{crosssymmbasis}. As a consequence of the $CPT$ theorem and the fact that we will be considering particles on which charge conjugation acts trivially, our helicity amplitudes are $PT$ invariant. We will consider the subcases whether parity is preserved or not. We will be following the notations and conventions of \cite{spinpen}. 


\subsubsection{$\mathcal{P}$ invariant theories}

Massless photon and graviton theories in $d=4$ are characterised by their helicities which can take values  $(\pm 1)$ for photons and $(\pm 2)$ for gravitons respectively. This tells us that there are possibly 16 helicity amplitudes. Since the particles are identical, the helicity amplitudes enjoy a $\Z_2\times\Z_2$ symmetry. 
\begin{eqnarray}
	&&\mt_{\lambda_1, \lambda_2}^{\lambda_3, \lambda_4}(s_1,s_2,s_3)=\mt_{\lambda_2, \lambda_1}^{\lambda_4, \lambda_3}(s_1,s_2,s_3), \qquad \mt_{\lambda_1, \lambda_2}^{\lambda_3, \lambda_4}(s_1,s_2,s_3)=\mt_{-\lambda_4, -\lambda_3}^{-\lambda_2, -\lambda_1}(s_1,s_2,s_3),\nonumber\\
	&&\mt_{\lambda_1, \lambda_2}^{\lambda_3, \lambda_4}(s_1,s_2,s_3)=\mt_{-\lambda_1, -\lambda_2}^{-\lambda_3, -\lambda_4}(s_1,s_2,s_3). \nonumber\\
\end{eqnarray}
Additionally, since we are looking at parity invariant theories, we have the following constraints from parity, time-reversal respectively,
\begin{eqnarray}\label{PTboseph}
	\mt^{\lambda_3, \lambda_4}_{\lambda_1,\lambda_2}(s_1,s_2,s_3) &=& \eta_1^* \eta_2^*\eta_3\eta_4 (-1)^{j_1+j_2+j_3+j_4} (-1)^{\lambda_1-\lambda_2-\lambda_3+\lambda_4}\mt^{\lambda_1, \lambda_2}_{\lambda_3,\lambda_4}(s_1,s_2,s_3), \nonumber\\\nonumber\\
	\mt^{\lambda_3, \lambda_4}_{\lambda_1,\lambda_2}(s_1,s_2,s_3) &=& \varepsilon_1^* \varepsilon_2^*\varepsilon_3\varepsilon_4 \mt^{\lambda_1, \lambda_2}_{\lambda_3,\lambda_4}(s_1,s_2,s_3),
\end{eqnarray}
where $|\eta_i|^2=|\epsilon_i|^2=1$. Note that for scattering of four identical photons and gravitons, we have $\eta_1^* \eta_2^*\eta_3\eta_4= (|\eta|^2)^2 =1$ trivially. These conditions reduce the number of independent  parity preserving helicity amplitudes which are given by \cite{spinpen},  
\begin{equation}
	\begin{split}
		& \mt_1(s_1,s_2,s_3)=\mt^{++}_{++}(s_1,s_2,s_3), \qquad \mt_2(s_1,s_2,s_3)=\mt^{--}_{++}(s_1,s_2,s_3),\qquad \mt_3(s_1,s_2,s_3)=\mt^{+-}_{+-}(s_1,s_2,s_3),\\
		& \mt_4(s_1,s_2,s_3)=\mt^{-+}_{+-}(s_1,s_2,s_3),\qquad \mt_5(s_1,s_2,s_3)=\mt^{+-}_{++}(s_1,s_2,s_3).\\
	\end{split}
\end{equation}
\footnote{Note that due to \eqref{PTboseph}, the amplitudes $\mt_2$ and $\mt_5$ enjoy the additional symmetry 
	$$\mt^{--}_{++}(s_1,s_2,s_3) = \mt^{++}_{--}(s_1,s_2,s_3), \qquad \mt^{+-}_{++}(s_1,s_2,s_3) = \mt^{++}_{+-}(s_1,s_2,s_3)$$
}
These linearly independent set of five amplitudes are the basis of Quasi-invariant S-matrices defined in the previous section. They transform in irreps of $S_3$ which we determine from the following crossing equation \cite{spinpen, Henriksson:2021ymi}. 
\begin{equation}\label{crossphgv}
	\begin{split}
		& \m2{\lambda_1,\lambda_2}^{\lambda_3,\lambda_4}(s_1,s_2,s_3)=\epsilon_{23}' \mt_{\lambda_1,-\lambda_3}^{-\lambda_2,\lambda_4}(s_2,s_1,s_3),\\
		& \mt_{\lambda_1,\lambda_2}^{\lambda_3,\lambda_4}(s_1,s_2,s_3)=\epsilon_{24}' \mt_{\lambda_1,-\lambda_4}^{\lambda_3,-\lambda_2}(s_3,s_2,s_1),
	\end{split}
\end{equation}
where $\epsilon_{23}'$ and $\epsilon_{24}'$ are arbitrary phases which were left unfixed from the general considerations of crossing symmetry using which \eqref{crossphgv} were derived. We will fix them in this section using constraints from consistency of crossing equations and comparing against explicit helicity amplitudes in literature. Note that \eqref{crossphgv} differs from the equivalent equation of \cite{spinpen} (eqns 2.81 and 2.82). This is due to the fact that we assume the following assignment for the Wigner-$d$ angles \begin{eqnarray}\label{wignerd}
	\alpha_1=0,~~\alpha_2=\pi,~~\alpha_3=0,~~\alpha_4=\pi, \nonumber\\  
	\beta_1=0,~~\beta_2=\pi,~~\beta_3=0,~~\beta_4=\pi, \nonumber\\
\end{eqnarray}
in contrast with eqns 2.78 and 2.79 of \cite{spinpen}. Using \eqref{crossphgv}, we can determine the crossing matrices to be,  
\begin{equation}\label{cstcsuph}
	C^p_{st}=\epsilon'_{23}\begin{pmatrix}
		0 & 0 & 0 & 1 & 0\\
		0 & 1 & 0 & 0 & 0\\
		0 & 0 & 1 & 0 & 0\\
		1 & 0 & 0 & 0 & 0\\
		0 & 0 & 0 & 0 & 1
	\end{pmatrix},\qquad C^p_{su}=\epsilon'_{24}\begin{pmatrix}
		0 & 0 & 1 & 0 & 0\\
		0 & 1 & 0 & 0 & 0\\
		1 & 0 & 0 & 0 & 0\\
		0 & 0 & 0 & 1 & 0\\
		0 & 0 & 0 & 0 & 1
	\end{pmatrix}.
\end{equation}

At this stage we have two undetermined phases $\epsilon_{23}'$ and $\epsilon_{24}'$. In order to determine the $t-u$ crossing relation we use the following relation for identical scattering particles \cite{spinpen}
\begin{equation}
	\begin{split}
		& \mt_{\lambda_1,\lambda_2}^{\lambda_3,\lambda_4}(s_1,s_2,s_3)=(-1)^{\lambda_2-\lambda_1+\lambda_4-\lambda_3}\mt_{\lambda_2,\lambda_1}^{\lambda_3,\lambda_4}(s_1,s_3,s_2),\\
		& \mt_{\lambda_1,\lambda_2}^{\lambda_3,\lambda_4}(s_1,s_2,s_3)=(-1)^{-\lambda_2+\lambda_1+\lambda_4-\lambda_3}\mt_{\lambda_1,\lambda_2}^{\lambda_4,\lambda_3}(s_1,s_3,s_2).\\
	\end{split}
\end{equation}
We can independently try to derive the $C^p_{tu}$ crossing matrix by using the following composition for the generators of $S_3$. 
\begin{equation}
	C^p_{tu}=C^p_{st} C^p_{su} C^p_{st},\qquad C^p_{tu}=\epsilon_{23}'^2\epsilon_{24}'\begin{pmatrix}
		1 & 0 & 0 & 0 & 0\\
		0 & 1 & 0 & 0 & 0\\
		0 & 0 & 0 & 1 & 0\\
		0 & 0 & 1 & 0 & 0\\
		0 & 0 & 0 & 0 & 1\\
	\end{pmatrix}.
\end{equation}
This lets us to conclude $\epsilon_{24'}=1$ while the phase $\epsilon_{23}'$ is undetermined. We can try to fix $\epsilon_{23}'$ in the following way. We can compare against a known amplitude to check the phase. To be precise let us compare against the explicit helicity amplitudes computed in the Euler-Heisenberg EFT, from the last equality in eqn 2.9 of \cite{Alberte:2020bdz} and tree level graviton amplitude from eqn 17 of \cite{Bern:2002kj}, we see $\epsilon_{23}'=1$ for both photons and graviton amplitudes\footnote{Note the difference in convention in defining helicity amplitudes, $A^{{\rm us},\l_3,\l_4}_{\l_1,\l_2}=A^{{\rm them},\l_3,\l_4,-\l_1,-\l_2}$}. For convenience we write down the crossing matrices finally 
\begin{eqnarray}
	C^p_{st}=\begin{pmatrix}
		0 & 0 & 0 & 1 & 0\\
		0 & 1 & 0 & 0 & 0\\
		0 & 0 & 1 & 0 & 0\\
		1 & 0 & 0 & 0 & 0\\
		0 & 0 & 0 & 0 & 1
	\end{pmatrix},\qquad C^p_{su}=\begin{pmatrix}
		0 & 0 & 1 & 0 & 0\\
		0 & 1 & 0 & 0 & 0\\
		1 & 0 & 0 & 0 & 0\\
		0 & 0 & 0 & 1 & 0\\
		0 & 0 & 0 & 0 & 1
	\end{pmatrix}.
\end{eqnarray} 
One can immediately see from these crossing matrices that  $\mt_2(s_1,s_2,s_3)$ and $\mt_5(s_1,s_2,s_3)$ are crossing symmetric by themselves while it takes a little bit more effort to see that $( \mt_1(s_1,s_2,s_3), \mt_3(s_1,s_2,s_3),$ $\mt_4(s_1,s_2,s_3) )$ transforms in a $\bf{3_S}={\bf 1_S}+ {\bf 2_M}$ (a reducible representation of dimension 3). To see that $\mt_2(s_1,s_2,s_3)$ and $\mt_5(s_1,s_2,s_3)$ are crossing symmetric, note that under $(s_1,s_2)$ and $(s_1,s_3)$ they map to themselves and since the other orbits of $S_3$ are generated by products of this transposition, all the orbits will map to themselves.  However, to systematise the procedure we explain in detail the case of photons. Using the projector \eqref{projector} we see that 
\begin{eqnarray}
	P_{\bf{1_S}}(\mt_2(s_1,s_2,s_3))&=& \mt_2(s_1,s_2,s_3)\,,\nonumber\\
	P_{\bf{1_S}}(\mt_5(s_1,s_2,s_3))&=& \mt_5(s_1,s_2,s_3)\,,\nonumber\\
	P_{\bf{1_S}}(\mt_1(s_1,s_2,s_3))&=&P_{\bf{1_S}}(\mt_3(s_1,s_2,s_3))=P_{\bf{1_S}}(\mt_4(s_1,s_2,s_3))\,,\nonumber\\ &=&\frac{(\mt_1(s_1,s_2,s_3)+\mt_3(s_1,s_2,s_3)+\mt_4(s_1,s_2,s_3))}{3}\,, \nonumber\\
	P_{\bf{1_A}}(\mt_i(s_1,s_2,s_3))&=&0\,. \nonumber\\
\end{eqnarray} 

This tells us that the triplet $(\mt_1(s_1,s_2,s_3),\mt_3(s_1,s_2,s_3),\mt_4(s_1,s_2,s_3))$ has a ${\bf 1_S}$ part while $\mt_2(s_1,s_2,s_3)$ and $\mt_5(s_1,s_2,s_3)$ are crossing symmetric by themselves. We now want to check whether there is a ${\bf 2_M}$ also in $(\mt_1(s_1,s_2,s_3),\mt_3(s_1,s_2,s_3),\mt_4(s_1,s_2,s_3))$. From \eqref{projectorM+}, we get, 
\begin{eqnarray}
	P^{(1)}_{\bf{2_{M+}}}(\mt_1(s_1,s_2,s_3))&=&-2P^{(1)}_{\bf{2_{M+}}}(\mt_3(s_1,s_2,s_3))=-2P^{(1)}_{\bf{2_{M+}}}(\mt_4(s_1,s_2,s_3))\,,\nonumber\\
	&=&\frac{(2\mt_1(s_1,s_2,s_3)-\mt_3(s_1,s_2,s_3)-\mt_4(s_1,s_2,s_3))}{3}\,. \nonumber\\
	P^{(2)}_{\bf{2_{M+}}}(\mt_1(s_1,s_2,s_3))&=&-2P^{(2)}_{\bf{2_{M+}}}(\mt_3(s_1,s_2,s_3))=-2P^{(2)}_{\bf{2_{M+}}}(\mt_4(s_1,s_2,s_3))\,,\nonumber\\
	&=&\frac{(2\mt_3(s_1,s_2,s_3)-\mt_1(s_1,s_2,s_3)-\mt_4(s_1,s_2,s_3))}{3}\,. \nonumber\\
\end{eqnarray}

Therefore, we identify our crossing symmetric matrix by substituting the following sets of solutions in \eqref{crosssymmbasis}, 
\begin{eqnarray}
	f^{\a , 1}_{\rm Sym}(s_1,s_2,s_3)&=& \mt_2(s_1,s_2,s_3)\,,\nonumber\\
	f^{\a , 2}_{\rm Sym}(s_1,s_2,s_3)&=& \mt_5(s_1,s_2,s_3)\,,\nonumber\\
	f^{\a , 3}_{\rm Sym}(s_1,s_2,s_3)&=& \frac{\mt_1(s_1,s_2,s_3)+\mt_3(s_1,s_2,s_3)+\mt_4(s_1,s_2,s_3)}{3}\,,\nonumber\\
	f^{\a}_{\rm Mixed+}(s_1,s_2,s_3)&=&\frac{(2\mt_1(s_1,s_2,s_3)-\mt_3(s_1,s_2,s_3)-\mt_4(s_1,s_2,s_3))}{3}\,. \nonumber\\
\end{eqnarray}

\noindent where $\a \equiv \g, h$ for photons and gravitons, respectively. Explicitly written out, the crossing symmetric photon and graviton amplitudes  are,
\begin{equation}\label{f1}
	F^\alpha_1(s_1,s_2,s_3) = \mt_2(s_1,s_2,s_3)\,, \hspace*{280 pt}
\end{equation} 
\begin{equation}\label{f2}
	F^\alpha_2(s_1,s_2,s_3)= \mt_1(s_1,s_2,s_3)+\mt_3(s_1,s_2,s_3)+\mt_4(s_1,s_2,s_3) \,,\hspace*{130 pt}
\end{equation} 
\begin{equation}\label{f3}
	F^\alpha_3(s_1,s_2,s_3) = \mt_5(s_1,s_2,s_3)\,,\hspace*{270 pt}
\end{equation}
\begin{equation}\label{f4}
	F^\alpha_4(s_1,s_2,s_3) =\frac{f^\a_{\rm Mixed+}(s_3,s_1,s_2)(s_3+s_2-2 s_1)-f^\a_{\rm Mixed+}(s_1,s_2,s_3)(s_1+s_2-2 s_3)}{3(s_1-s_2)(s_2-s_3)(s_3-s_1)}\,,
\end{equation} 
\begin{equation}\label{f5}
	F^\alpha_5(s_1,s_2,s_3) = \frac{f^\a_{\rm Mixed+}(s_1,s_2,s_3)(s_1^2+s_2^2-2s_3^2)-f^\a_{\rm Mixed+}(s_3,s_1,s_2)(s_2^2+s_3^2-2s_1^2)}{3 (s_1-s_2)(s_2-s_3)(s_3-s_1)}\,.
\end{equation}

\subsubsection{$\mathcal{P}$ violating theories}

In this subsubsection we consider Parity violating (and hence Time-reversal violating theories) theories where we do not impose the condition \eqref{PTboseph}. As a result, the independent helicity amplitudes are, 
\begin{equation}\label{parviohelamp}
	\begin{split}
		& \mt_1(s_1,s_2,s_3)=T^{++}_{++}(s_1,s_2,s_3), \qquad \mt_2(s_1,s_2,s_3)=T^{--}_{++}(s_1,s_2,s_3),\qquad \mt_3(s_1,s_2,s_3)=T^{+-}_{+-}(s_1,s_2,s_3),\\
		& \mt_4(s_1,s_2,s_3)=T^{-+}_{+-}(s_1,s_2,s_3),\qquad \mt_5(s_1,s_2,s_3)=T^{+-}_{++}(s_1,s_2,s_3), \qquad T'_2(s_1,s_2,s_3)=T^{++}_{--}(s_1,s_2,s_3),\\
		&T'_5(s_1,s_2,s_3)=T^{++}_{+-}(s_1,s_2,s_3)\,.
	\end{split}
\end{equation}
The crossing matrices are modified to 
\begin{eqnarray}
	C^{pv}_{st}=\begin{pmatrix}
		0 & 0 & 0 & 1 & 0 & 0 & 0\\
		0 & 1 & 0 & 0 & 0 & 0 & 0\\
		0 & 0 & 1 & 0 & 0 & 0 & 0\\
		1 & 0 & 0 & 0 & 0 & 0 & 0\\
		0 & 0 & 0 & 0 & 1 & 0 & 0\\
		0 & 0 & 0 & 0 & 0 & 1 & 0\\
		0 & 0 & 0 & 0 & 0 & 0 & 1\\
	\end{pmatrix},\qquad C^{pv}_{su}=\begin{pmatrix}
		0 & 0 & 1 & 0 & 0 & 0 & 0\\
		0 & 1 & 0 & 0 & 0 & 0 & 0\\
		1 & 0 & 0 & 0 & 0 & 0 & 0\\
		0 & 0 & 0 & 1 & 0 & 0 & 0\\
		0 & 0 & 0 & 0 & 1 & 0 & 0\\
		0 & 0 & 0 & 0 & 0 & 1 & 0\\
		0 & 0 & 0 & 0 & 0 & 0 & 1\\
	\end{pmatrix}.
\end{eqnarray} 
The new objects  $T'_2$ and $T'_5$ are crossing symmetric by themselves. 
\begin{eqnarray}
	P_{\bf{1_S}}(T'_2(s_1,s_2,s_3))&=& T'_2(s_1,s_2,s_3)\,,\nonumber\\
	P_{\bf{1_S}}(T'_5(s_1,s_2,s_3))&=& T'_5(s_1,s_2,s_3)\,.\nonumber\\
\end{eqnarray} 

\noindent Therefore we identify our crossing symmetric matrix by substituting the following sets of solutions in \eqref{crosssymmbasis}, 
\begin{eqnarray}\label{pvcm}
	\tilde{f}^{\a , 1}_{\rm Sym}(s_1,s_2,s_3)&=& \mt_2(s_1,s_2,s_3)\,,\nonumber\\
	\tilde{f}^{\a , 2}_{\rm Sym}(s_1,s_2,s_3)&=& \mt_5(s_1,s_2,s_3)\,,\nonumber\\
	\tilde{f}^{\a , 3}_{\rm Sym}(s_1,s_2,s_3)&=& \frac{\mt_1(s_1,s_2,s_3)+\mt_3(s_1,s_2,s_3)+\mt_4(s_1,s_2,s_3)}{3}\,,\nonumber\\
	\tilde{f}^{\a , 4}_{\rm Sym}(s_1,s_2,s_3)&=& \mt'_2(s_1,s_2,s_3)\,,\nonumber\\
	\tilde{f}^{\a , 5}_{\rm Sym}(s_1,s_2,s_3)&=& \mt'_5(s_1,s_2,s_3)\,,\nonumber\\
	\tilde{f}^{\a}_{\rm Mixed+}(s_1,s_2,s_3)&=&\frac{(2\mt_1(s_1,s_2,s_3)-\mt_3(s_1,s_2,s_3)-\mt_4(s_1,s_2,s_3))}{3}\,. \nonumber\\
\end{eqnarray} 

\noindent where, $\a \equiv \g, g$ for photons and gravitons respectively. Explicitly written out, the crossing symmetric photon and graviton s-matrices  are,
\begin{equation}\label{f1pv}
	\tilde{F}^\alpha_1(s_1,s_2,s_3) = \mt_2(s_1,s_2,s_3)\,,\hspace*{370 pt}
\end{equation} 
\begin{equation}\label{f2pv}
	\tilde{F}^\alpha_2(s_1,s_2,s_3)= \mt_1(s_1,s_2,s_3)+\mt_3(s_1,s_2,s_3)+\mt_4(s_1,s_2,s_3)\,,\hspace*{370 pt}
\end{equation} 
\begin{equation}\label{f3pv}
	\tilde{F}^\alpha_3(s_1,s_2,s_3) = \mt_5(s_1,s_2,s_3)\,,\hspace*{370 pt}
\end{equation}
\begin{equation}\label{f4pv}
	\tilde{F}^\alpha_4(s_1,s_2,s_3) = \frac{f^\a_{\rm Mixed+}(s_3,s_1,s_2)(s_3+s_2-2 s_1)-f^\a_{\rm Mixed+}(s_1,s_2,s_3)(s_1+s_2-2 s_3)}{3(s_1-s_2)(s_2-s_3)(s_3-s_1)}\,,\hspace*{270 pt}
\end{equation} 
\begin{equation}\label{f5pv}
	\tilde{F}^\alpha_5(s_1,s_2,s_3) = \frac{f^\a_{\rm Mixed+}(s_1,s_2,s_3)(s_1^2+s_2^2-2s_3^2)-f^\a_{\rm Mixed+}(s_3,s_1,s_2)(s_2^2+s_3^2-2s_1^2)}{3 (s_1-s_2)(s_2-s_3)(s_3-s_1)}\,,\hspace*{370 pt}
\end{equation} 
\begin{equation}\label{f6pv}
	\tilde{F}^\alpha_6(s_1,s_2,s_3) = \mt'_2(s_1,s_2,s_3)\,,\hspace*{390 pt}
\end{equation} 
\begin{equation}\label{f7pv}
	\tilde{F}^\alpha_7(s_1,s_2,s_3) = \mt'_5(s_1,s_2,s_3)\,.\hspace*{390 pt}
\end{equation} 

We note that the crossing equations are consistent with the photon module classification done in \cite{Chowdhury:2020ddc}. In \cite{Chowdhury:2020ddc}, it was found that there is one parity even module transforming in a ${\bf 3}$, two parity even ${\bf 1_S}$ module and two parity odd ${\bf 1_S}$ module. It is satisfying to see that the degrees of freedom encoded in crossing in the two different approaches nicely match.


\subsection{Massive Majorana fermions}\label{mmf}
Let us now consider the scattering amplitude of four massive Majorana fermions in parity conserving theory. The five independent helicity structures are the following
\begin{equation}\label{MMamp}
	\begin{split}
		& \Phi_1(s_1,s_2,s_3)=\mt^{++}_{++}(s_1,s_2,s_3), \qquad \Phi_2(s_1,s_2,s_3)=\mt^{--}_{++}(s_1,s_2,s_3),\qquad \Phi_3(s_1,s_2,s_3)=\mt^{+-}_{+-}(s_1,s_2,s_3),\\
		& \Phi_4(s_1,s_2,s_3)=\mt^{-+}_{+-}(s_1,s_2,s_3),\qquad \Phi_5(s_1,s_2,s_3)=\mt^{+-}_{++}(s_1,s_2,s_3).\\
	\end{split}
\end{equation}
Further one can separate out the kinematical singularities and branch cuts to define the improved amplitudes $H_{I}(s_1,s_2,s_3)$ such that \cite{spinpen},
\begin{equation}
	\phi_I(s_1,s_2,s_3)=\sum_{J=1}^5 M^{-1}_{IJ} H_{J}(s_1,s_2,s_3),
\end{equation}
where $M$ matrix is defined as follows,
\begin{equation}
	M=\begin{psmallmatrix}
		\frac{4}{s_1-\frac{8m^2}{3}} & \frac{-4}{s_1-\frac{8m^2}{3}} & \frac{2(1-\frac{s_2+4m^2/3}{s_3+4m^2/3})}{s_1-8m^2/3} & \frac{2(1-\frac{s_3+4m^2/3}{s_2+4m^2/3})}{s_1-8m^2/3} & \frac{2(s_1+16m^2/3)(s_2-s_3)}{m(s_1-8m^2/3)\sqrt{(s_1+4m^2/3)(s_2+4m^2/3)(s_3+4m^2/3)}} \\
		0 & 0 & \frac{2}{s_3+4m^2/3} & \frac{-2}{s_2+4m^2/3} & -\frac{8m}{\sqrt{(s_1+4m^2/3)(s_2+4m^2/3)(s_3+4m^2/3)}}\\
		0 & 0 & \frac{2}{s_3+4m^2/3} & \frac{-2}{s_2+4m^2/3} & -\frac{2s}{m\sqrt{(s_1+4m^2/3)(s_2+4m^2/3)(s_3+4m^2/3)}}\\
		0 & 0 & \frac{2}{s_3+4m^2/3} & \frac{2}{s_2+4m^2/3} & 0\\
		-\frac{4}{s_1+4m^2/3} & -\frac{4}{s_1+4m^2/3} & \frac{2}{s_3+4m^2/3}+\frac{4}{s_1+4m^2/3} & \frac{2}{s_2+4m^2/3}+\frac{4}{s_1+4m^2/3} & \frac{2(s_2-s_3)}{m\sqrt{(s_1+4m^2/3)(s_2+4m^2/3)(s_3+4m^2/3)}}
	\end{psmallmatrix}.
\end{equation}
Crossing symmetry is imposed by the following two crossing matrices, 
\begin{equation}
	\tilde{C}^f_{st}=\begin{pmatrix}
		-\frac{1}{4} & -1 & \frac{3}{2} & 1 & -\frac{1}{4}\\
		-\frac{1}{4} & \frac{1}{2} & 0 & \frac{1}{2} & \frac{1}{4}\\
		\frac{1}{4} & 0 & \frac{1}{2} & 0 & \frac{1}{4}\\ 
		\frac{1}{4} & \frac{1}{2} & 0 & \frac{1}{2} & -\frac{1}{4}\\
		-\frac{1}{4} & 1 & \frac{3}{2} & -1 &-\frac{1}{4}
	\end{pmatrix},\qquad 
	\tilde{C}^f_{su}=\begin{pmatrix}
		-\frac{1}{4} & 1 & -\frac{3}{2} & 1 & -\frac{1}{4}\\
		\frac{1}{4} & \frac{1}{2} & 0 & -\frac{1}{2} & -\frac{1}{4}\\
		-\frac{1}{4} & 0 & \frac{1}{2} & 0 & -\frac{1}{4}\\ 
		\frac{1}{4} & -\frac{1}{2} & 0 & \frac{1}{2} & -\frac{1}{4}\\
		-\frac{1}{4} & -1 & -\frac{3}{2} & -1 &-\frac{1}{4}
	\end{pmatrix}.
\end{equation}

The analysis for massive fermions is a bit more involved. Using the projectors defined in \eqref{projector} we find,
\begin{equation}
	\begin{split}
		P_{\bf{1_S}}(H_1(s_1,s_2,s_3))&=P_{\bf{1_S}}(H_4(s_1,s_2,s_3))=P_{\bf{1_S}}(H_5(s_1,s_2,s_3))\,\\
		&=\frac{(H_1(s_1,s_2,s_3)+4H_4(s_1,s_2,s_3)-H_5(s_1,s_2,s_3))}{6}\,,\\
		P_{\bf{1_S}}(H_2(s_1,s_2,s_3))&=P_{\bf{1_S}}(H_3(s_1,s_2,s_3))=0\,,\\
		P_{\bf{1_A}}(H_i(s_1,s_2,s_3))&=0 \,.
	\end{split}
\end{equation}

This implies that the we have an irrep that transforms in an ${\bf 1_S}$ and none in ${\bf 1_A}$. We now use the projector for the mixed symmetry to evaluate
\begin{equation}
	\begin{split}
		& P^{(1)}_{\bf{2_{M+}}}(H_1(s_1,s_2,s_3))=\frac{1}{6} (5 H_1(s_1,s_2,s_3)-4 H_4(s_1,s_2,s_3)+ H_5(s_1,s_2,s_3))\,,\\
		&      P^{(2)}_{\bf{2_{M+}}}(H_1(s_1,s_2,s_3))=\frac{1}{12} (-5 H_1(s_1,s_2,s_3) +12 H_2 (s_1,s_2,s_3)-18 H_3 (s_1,s_2,s_3)+4 H_4 (s_1,s_2,s_3)\\&~~~~~~~~~~~~~~~~~~~~~~~~~~~~~~~~-H_5 (s_1,s_2,s_3))\,,\\
		& P^{(1)}_{\bf{2_{M+}}}(H_4(s_1,s_2,s_3))=\frac{1}{6} (-H_1(s_1,s_2,s_3)+2 H_4(s_1,s_2,s_3)+ H_5(s_1,s_2,s_3))\,,\\
		& P^{(2)}_{\bf{2_{M+}}}(H_4(s_1,s_2,s_3))=\frac{1}{12} ( H_1(s_1,s_2,s_3) -6 H_2 (s_1,s_2,s_3)-2 H_4 (s_1,s_2,s_3)-H_5 (s_1,s_2,s_3))\,.
	\end{split}
\end{equation}

Rest of the ${\bf 2_{M+/-}}$ projections are either zero or a linear combinations of these. Hence the independent data that can be used in \eqref{crosssymmbasis} are
\begin{eqnarray} 
	f^{\psi}_{\rm Sym}(s_1,s_2,s_3)&=& \frac{(H_1(s_1,s_2,s_3)+4H_4(s_1,s_2,s_3)-H_5(s_1,s_2,s_3))}{6},\,\nonumber\\
	f^{\psi, 1}_{\rm Mixed+}(s_1,s_2,s_3)&=&\frac{1}{6} (5 H_1(s_1,s_2,s_3)-4 H_4(s_1,s_2,s_3)+ H_5(s_1,s_2,s_3)),\, \nonumber\\
	f^{\psi, 2}_{\rm Mixed+}(s_1,s_2,s_3)&=&\frac{1}{6} (-H_1(s_1,s_2,s_3)+2 H_4(s_1,s_2,s_3)+ H_5(s_1,s_2,s_3)),\,. \nonumber\\
\end{eqnarray} 

Explicitly written out, the crossing symmetric fermion  S-matrices  are
\begin{equation}\label{psi1}
	\Psi_1(s_1,s_2,s_3) =(H_1(s_1,s_2,s_3)+4H_4(s_1,s_2,s_3)-H_5(s_1,s_2,s_3))\,,
\end{equation} 
\begin{equation}\label{psi2}
	\Psi_2(s_1,s_2,s_3) =\frac{f^{\psi,1}_{\rm Mixed+}(s_3,s_1,s_2)(s_3+s_2-2 s_1)-f^{\psi,1}_{\rm Mixed+}(s_1,s_2,s_3)(s_1+s_2-2 s_3)}{3(s_1-s_2)(s_2-s_3)(s_3-s_1)}\,,
\end{equation} 
\begin{equation}\label{psi3}
	\Psi_3(s_1,s_2,s_3) = \frac{f^{\psi,1}_{\rm Mixed+}(s_1,s_2,s_3)(s_1^2+s_2^2-2s_3^2)-f^{\psi,1}_{\rm Mixed+}(s_3,s_1,s_2)(s_2^2+s_3^2-2s_1^2)}{3 (s_1-s_2)(s_2-s_3)(s_3-s_1)}\,,
\end{equation} 
\begin{equation}\label{psi4}
	\Psi_4(s_1,s_2,s_3) =\frac{f^{\psi,2}_{\rm Mixed+}(s_3,s_1,s_2)(s_3+s_2-2 s_1)-f^{\psi,2}_{\rm Mixed+}(s_1,s_2,s_3)(s_1+s_2-2 s_3)}{3(s_1-s_2)(s_2-s_3)(s_3-s_1)}\,,
\end{equation} 
\begin{equation}\label{psi5}
	\Psi_5(s_1,s_2,s_3) = \frac{f^{\psi,2}_{\rm Mixed+}(s_1,s_2,s_3)(s_1^2+s_2^2-2s_3^2)-f^{\psi,2}_{\rm Mixed+}(s_3,s_1,s_2)(s_2^2+s_3^2-2s_1^2)}{3 (s_1-s_2)(s_2-s_3)(s_3-s_1)}\,.
\end{equation} 
\section{Crossing symmetric dispersion relation: Overview} \label{csdrsec}
In the previous section we constructed various fully crossing symmetric amplitudes, i.e., amplitudes invariant under $S_3$, the group of permutations of $(s_1, s_2, s_3)$. We will employ the crossing symmetric dispersion relation (CSDR) to these amplitudes. We have already expounded the principal rationale behind derivation of CSDR in section \ref{crossreview} of previous chapter.   But the discussion in the previous chapter was done for scattering amplitudes of massive particles. We will discuss how a similar construction can be obtained for EFT amplitudes for scattering of massless particles.    

Let us consider  a $S_3$-invariant amplitude associated\footnote{There can be multiple such amplitudes associated with a given scattering when the particles have extra quantum numbers such as spin, isospin e.t.c}  with scattering of identical particles $\mm(s,t,u)$, with $s+t+u=4m^2=\mu$, $m$ being mass of the scattering particles. The amplitude is known/assumed to satisfy following two crucial properties.
\begin{enumerate} \item[I.] We assume that the amplitude is analytic in some
	domain\footnote{We are considering both $s$ and $t$ as complex variables.} $\md\subset\mathbb{C}^2$, which includes the physical domains of all the three channels. For massive theories such domains (e.g. enlarged Martin domain \cite{Martin:1965jj})  have been established rigorously from axiomatic field theory considerations. For massless theories, even though such domains are not established within rigorous framework of axiomatic field theory, they can be argued physically in general. Thus we will assume the existence of such domains to begin with. 
	\item[II.] The amplitude is Regge bounded in all the three channels. While for massive theories this is established rigorously from axiomatic field theory, for massless theories this is a working assumption which we will make. Thus, for example, fixed $t$ Regge-boundedness reads
	\be 
	\mt(s,t)=o(s^2)\quad \text{for}\,\,\, |s|\to\infty,\quad t\,\,\text{fixed},\quad (s,t)\in\md.
	\ee  \end{enumerate} This is equivalent to the amplitude admitting a twice subtracted fixed-transfer (for example, fixed-$t$) dispersive representation. 

\subsection{Massive amplitudes}
We have already discussed this in the previous chapter. Nevertheless, we write key expressions for completeness. At the heart of CSDR lies the following rational parametrization \cite{ Auberson:1972prg, Sinha:2020win} for the Mandelstam variables $\{s_1=s-\m/3,\, s_2=t-\m/3, \,s_3=u-\m/3\}$:
\be \label{crosspar}
s_k(a,z)= a\left[1-\frac{ (z-z_k)^3}{z^3-1}\right],\quad a\in\mathbb{R},\quad k=1,2,3.
\ee Here $\{z_k\}$ are the cube-roots of unity.   $\tilde{z}:= z^3,a$ are crossing symmetric variables. We also introduce the following crossing symmetric combinations, $x:=-(s_1s_2+s_2s_3+s_3s_1)=\frac{-27a^2z^3}{(z^3-1)^2},\,y:=-(s_1s_2s_3)=\frac{-27a^3z^3}{(z^3-1)^2}$ such that $a=y/x$. With these parametrizations, as shown by \cite{Auberson:1972prg}, one can write the following dispersive representation of the amplitude $\mm$ with a \emph{manifestly crossing-symmetric kernel}: 
\bea 
\mathcal{T}(s_1, s_2)=\alpha_{0}+\frac{1}{\pi} \int_{M^2 }^{\infty} \frac{d s_1^{\prime}}{s_1^{\prime}} \mathcal{A}\left(s_1^{\prime} ; s_2^{(+)}\left(s_1^{\prime} ,a\right)\right) H\left(s_1^{\prime} ;s_1, s_2, s_3\right)\,,
\eea
where $\mathcal{A}\left(s_1; s_2\right)$ is the s-channel discontinuity, $s_{2}^{(+)}\left(s_1^{\prime}, a\right)=-\frac{s_1^{\prime}}{2}\left[1 - \left(\frac{s_1^{\prime}+3 a}{s_1^{\prime}-a}\right)^{1 / 2}\right]\,$ and the kernel is given by
\bea\label{kernel}
H\left(s_1^{\prime} ; s_1, s_2,s_3\right)&=&\left[\frac{s_1}{\left(s_1^{\prime}-s_1\right)}+\frac{s_2}{\left(s_1^{\prime}-s_2\right)}+\frac{s_3}{\left(s_1^{\prime}-s_3\right)}\right]\,\nonumber\\
&=&\frac{27a^2(3a-2s_1')}{(-27a^3+27a^2s_1'+s_1^3\left(\frac{-27a^2}{x}\right))}\,.
\eea
\subsection{Massless theories: EFT amplitudes}
For massless theories, we will consider crossing-symmetric dispersion relation for the amplitudes in the sense of effective field theories (EFT) as detailed below. One can write the following (twice subtracted) fixed $t$ dispersion relation for a massless amplitude \cite{Caron-Huot:2020cmc}
\be \label{Mdisp}
\frac{\mt(s,t)}{s (s+t)}=\int_{-\infty}^{\infty}\frac{ds'}{\pi (s'-s)}\text{Im}\left[\frac{\mt(s',t)}{s (s'+t)}\right],\qquad (t<0,\,\,s\notin\mathbb{R}),
\ee where the subtraction points are chosen to be $s=0$ and $s=-t,\,t<0.$
Now, the amplitude can be divided into two parts, the high energy amplitude $\mm_{\text{high}}$ and the low energy amplitude $\mm_{\text{low}}$. The high energy amplitude $\mm_{\text{high}}$  admits an `effective (fixed-transfer) dispersion relation' with two subtractions. For example, the fixed $t$ effective dispersion relation is of the form 
\be 
\frac{\mt_{\text{high}}(s,t)}{s (s+t)}=\int_{M^2}^{\infty}\frac{ds'}{\pi }\left(\frac{1}{s'-s}+\frac{1}{s'+s+t}\right)\text{Im}\left[\frac{\mt_{\text{high}}(s',t)}{s (s'+t)}\right].
\ee Here $M^2$ is some UV cut-off such that the physics beyond this scale is unknown a-priori to us. The low energy amplitude $\mm_{\text{low}}$ is to be understood in the sense of effective field theory amplitude. The dispersion relation for the full amplitude $\mm$, \eqref{Mdisp}, relates $\mm_{\text{high}}$ to this EFT amplitude.

Now, the amplitudes are all crossing symmetric. Therefore, we can write a crossing-symmetric dispersion relation for $\mm_{high}$ \cite{Sinha:2020win} similar to \eqref{crossdisp}
\be 
\mt_{\text{high}}(s_1,s_2)=\alpha_{0}+\frac{1}{\pi} \int_{M^2 }^{\infty} \frac{d s_1^{\prime}}{s_1^{\prime}} \mathcal{A}_{\text{high}}\left(s_1^{\prime} ; s_2^{(+)}\left(s_1^{\prime} ,a\right)\right) H\left(s_1^{\prime} ;s_1, s_2, s_3\right),
\ee 
with the kernel as in \eqref{kernel} and $\ma_{\text{high}}$ being the absorptive part of the high energy amplitude $\mm_{\text{high}}$.

In summary, we can write the crossing-symmetric dispersive representation for the amplitude ( in the sense of EFT when required ) as 
\be \label{crossdisp}
\mt(s_1,s_2)=\a_0+\frac{1}{\p}\int_{\Lambda_0}^{\infty}\frac{ds_1^{\prime}}{s_1^{\prime}}\ma\left(s_1^{\prime} ; s_2^{(+)}\left(s_1^{\prime} ,a\right)\right) H\left(s_1^{\prime} ;s_1, s_2, s_3\right),
\ee where $\Lambda_0=2\m/3$ for massive theories and $\Lambda_0=M^2$ (the UV cut-off) for massless amplitudes and the partial wave decomposition reads
\be\label{crossdispexp}
\ma\left(s_1^{\prime} ; s_2^{(+)}\left(s_1^{\prime} ,a\right)\right)=\Phi(s_1)\sum_{\ell=0}^\infty (2\ell+2\a)\,a_\ell(s_1)\,C_\ell^{(\frac{d-3}{2})}\left(\sqrt{\xi(s_1, a)}\right)
\ee  
where $\xi(s_1, a)=\xi_0 + 4 \xi_0 \left(\frac{a}{s_1-a}\right)$ and $\xi_0=\frac{s_1^2}{(s_1-\Lambda_0)^2}$ for massive theories while $\xi_0=1$ for massless theories and $a_\ell(s_1)$ is the spectral density which is defined as the imaginary part of the partial wave amplitude.    

\subsection{Wilson coefficients and locality constraints}
In this section we outline two central ingredients needed for  our bounds. Let us first review the case of massive scalar EFTs. The crossing symmetric amplitude, pole subtracted if required, admits a crossing symmetric double power series can be expanded in terms of crossing-symmetric variables $x:=-(s_1s_2+s_2s_3+s_3s_1),\,y:=-(s_1s_2s_3)$:
\be \label{crossexp}
\mt(s_1,s_2)=\sum_{p,q=0}^\infty \mw_{p,q}\,x^p y^q
\ee 
This is equivalent to a low-energy (EFT) expansion 
for the amplitude. The coefficients $\{\mw_{p,q}\}$ are themselves or  related to the Wilson coefficients appearing in the effective Lagrangian of a theory. Thus, these coefficients parametrize the space of EFTs. These coefficients can be obtained from the amplitude via the inversion formula \cite{Sinha:2020win}
\be \label{Winv}
\mw_{n-m,m}=\int_{\Lambda_0}^\infty \frac{ds_1}{2\pi s_1^{2n+m+1}}\Phi(s_1)\sum_{\ell=0}^\infty (2\ell+2\a)\,a_\ell(s_1)\,\mathcal{B}_{n,m}^{(\ell)}(s_1),\ee
with 
\bea \label{bellmassive}
\mathcal{B}_{n,m}^{\ell}(s_1)=2\sum_{j=0}^{m}\frac{(-1)^{1-j+m}p_{\ell}^{(j)}\left(\x_0\right) \left(4\xi_0 \right)^{j}(3 j-m-2n)\Gamma(n-j)}{ j!(m-j)!\Gamma(n-m+1)}\,\quad \x_0:=\frac{s_1^2}{(s_1-2\m/3)^2}
\eea 
Here  $\{a_\ell\}$ are the \emph{spectral functions} which appear as coefficients in partial wave expansion of the absorptive parts $\ma$. The functions $\{p_\ell^{(j)}(\x_0)\}$ are derivatives of Gegenbauer polynomials $C_\ell^{(\a)}$ for scalars
\be 
p_\ell^{(j)}(\x_0):= \left.\frac{\partial^j C_\ell^{(\a)}(\sqrt{z})}{\partial z^j}\right|_{z=\x_0}.
\ee 

The inversion formula leads to two kinds of constraints that play a central role in our subsequent analysis. Let us briefly review these. For more details, the readers are encouraged to look into  \cite{Sinha:2020win,Gopakumar:2021dvg}.

\subsubsection*{Locality constraints}

The first type of constraint that we will consider is related to the locality. In any local theory, a crossing symmetric amplitude admits a low energy expansion of the form \eqref{crossexp}. In particular, there should not be any negative power of $x$. However, this is not manifest at the level of the crossing symmetric dispersion relation since \eqref{Winv} is valid for both $n \le m$ and $n>m$, the latter leading to the negative powers of $x$. As explained in \cite{Sinha:2020win,Gopakumar:2021dvg} this is the price one has to pay for making crossing symmetry manifest. Fixed-transfer dispersion relations are manifestly local, but crossing symmetry has to be imposed as an additional constraint. In the crossing symmetric dispersion relation, by making crossing symmetry manifest, locality is compromised at the outset and has to be imposed. The equivalence of these two approaches was argued in \cite{Gopakumar:2021dvg}. 

Thus one needs to \emph{impose the locality constraint by demanding}

\be \label{Nulldef}
\mw_{n-m,m}=0,\quad\text{for}\quad n<m.
\ee
This gives rise to an infinite number of constraints on the partial wave coefficients $\{a_\ell\}$ known as \emph{locality constraints}, which are, in general, linear combinations of the null constraints \cite{Caron-Huot:2020cmc}. In principle, solving these infinite number of constraints can drastically restrict the space of allowed theories. However, even solving a finite number of such constraints give valuable information that we will see later.

\subsubsection*{Massless spinning particles}
For massless spinning particles, the decomposition of the amplitude into partial waves remains more or less unchanged with the technical difference due to the fact that instead of Gegenbauer polynomials, we have Wigner-$d$ functions in the expansion.
From section \ref{crampconstr}, we know how to construct the crossing symmetric amplitudes given the linearly independent basis of helicity amplitudes for various massless and massive spinning particles. The crossing symmetric decomposition \eqref{crossdisp} is therefore modified to be, 
\begin{eqnarray}\label{crossdipspin}
	\mt^{\text{spin}}(s_1,s_2)=\a_0+\frac{1}{\p}\int_{\Lambda_0}^{\infty}\frac{ds_1^{\prime}}{s_1^{\prime}}\ma^{\text{spin}}\left(s_1^{\prime} ; s_2^{(+)}\left(s_1^{\prime} ,a\right)\right) H\left(s_1^{\prime} ;s_1, s_2, s_3\right).
\end{eqnarray}
If the crossing symmetric amplitude is given by $\sum_i \b_i \mt_i$, where $\mt_i$ denote linearly independent helicity amplitudes and $\beta_i$ are some numbers, 
\be\label{maspin}
\ma^{\text{spin}}\left(s_1; s_2^{(+)}\left(s_1 ,a\right)\right) = \sum_{i} \b_i~\Phi(s_1)\sum_{\ell=0}^\infty (2\ell+2\a)\,a^i_\ell(s_1)\,f_i(d^{\ell}_{m,n}(\sqrt{\xi(s_1, a)}))   
\ee
where $f_i(d^{\ell}_{m,n}(\sqrt{\xi(s_1, a)}))$ denote the particular linear combinations of Wigner-$d$ functions that appear for $i$th helicity  amplitude and $a^i_\ell(s_1)$ now denotes the imaginary part of the partial waves of the particular helicity amplitude $\mt_i=\mt^{\l_3, \l_4}_{\l_1, \l_2}\left(s_1; s_2^{(+)}\left(s_1 ,a\right)\right)$. In our conventions, a helicity amplitude $T_i$ admits the following Wigner-$d$ matrix decomposition, 
\begin{eqnarray}\label{fidef}
	f_i\left(d^{\ell}_{m,n}(\sqrt{\xi(s_1, a)})\right)&=& d^\ell_{\l_1-\l_2, \l_3-\l_4}(\sqrt{\xi(s_1, a)}),\qquad \ell\geq |\l_1-\l_2| \,.
\end{eqnarray}
The restriction over the spin has been explained in Appendix \ref{uniconst}. Consequently, the inversion formula gets modified to,

\be \label{Winvspin}
\mw^{\text{spin}}_{n-m,m}=\int_{\Lambda_0}^\infty \frac{ds_1}{2\pi s_1^{2n+m+1}}\Phi(s_1)\sum_i \beta_i \sum_{\ell=0}^\infty (2\ell+2\a)\,a^i_\ell(s_1)\,\hat{\mathcal{B}}_{n,m}^{(\ell),i}(s_1),
\ee

with 
\bea \label{bellspin}
&&\hat{\mathcal{B}}_{n,m}^{\ell,i}(s_1)=2\sum_{j=0}^{m}\frac{(-1)^{1-j+m}p_{\ell}^{(j),i}\left(1\right) \left(4 \right)^{j}(3 j-m-2n)\Gamma(n-j)}{ j!(m-j)!\Gamma(n-m+1)}\, \qquad p_\ell^{(j),i}(1):=\left.\frac{\partial^j f_i\left(d^{\ell}_{m,n}(\sqrt{z})\right)}{\partial z^j}\right|_{z=1},\nn
\eea 

The locality constraints are now modified to be, 
\be \label{Nulldefspin}
\mw^{\text{spin}}_{n-m,m}=0,\quad\text{for}\quad n<m.
\ee

\subsection{$PB_C$ constraints}\label{positivity}

In order to get bounds, we would like to show that the expression on the RHS of \eqref{Winv} (and also \eqref{Winvspin}) or a linear combination of them must be of a definite sign. We identify the main characters of this analysis. 
\begin{itemize}
	\item Unitarity translates to positivity conditions on the spectral functions $\{a_\ell\}$ as reviewed in appendix \ref{uniconst}.
	
	\item We have $\x_0\ge 1$ for the entire range of $s_1$ integration in the inversion formula while $\x_0= 1$ identically for massless scalars (\eqref{bellmassive}). The functions $p_\ell^{(j)}(\x_0)$ are positive due to the fact that  the Gegenbauer polynomials, and its derivatives,  are positive for arguments larger than unity.
\end{itemize}

From these conditions, explicitly it follows\footnote{These conditions can also be shown to arise from the $TR_U$ inequalities discussed next, on Taylor expanding those conditions around $a=0$ \cite{Sinha:2020win}.} 
\be 
\mw_{n,0}\ge 0.
\ee 

More generally, this positivity property does not hold because  $\mb_{n,m}^{(\ell)}(s_1)$ (see \eqref{Winv}) control the sign of any term in $\ell$-expansion of the inversion formula. In that case one can ask whether taking suitable linear combinations of  $\mb_{n,m}^{(\ell)}(s_1)$s can restore the positivity. The answer turns out to be yes and one obtains \cite{Sinha:2020win}
\bea \label{Positivity}
\sum_{r=0}^m \chi_n^{(r,m)}(\Lambda_0^2)\mw_{n-r,r}&\ge& 0\nonumber\\
0\le \mw_{n,0}&\le& \frac{\mw_{n-1,0}}{(\Lambda_0)^2}\,,  ~~~~~~~~~~~~~~~~~~~~~~~~~~~~~~~n\ge2\,.
\eea The coefficient functions $\{\chi_{n}^{(r,m)}(\Lambda_0^2)\}$ satisfy the recursion relation:

\begin{equation}\label{chi}
	\begin{split}
		\chi_n^{(m,m)}(\Lambda_0^2)=1,\qquad
		\chi_n^{(r,m)}(\Lambda_0^2)=\sum_{j=r+1}^{m}(-1)^{j+r+1}\chi_n^{(j,m)}(\Lambda_0^2)\, \frac{\mathfrak{U}_{n,j,r}^{(\a)} (\Lambda_0^2)}{\mathfrak{U}_{n,r,r}^{(\a)}(\Lambda_0^2)}
	\end{split}
\end{equation} with \be\label{frakU}\mathfrak{U}_{n,m,k}^{(\a)}(s_1)= \sum_{k=0}^m \frac{\sqrt{16\x_0}^k (\a)_k (m+2n-3j)\Gamma(n-j)\Gamma(2j-k)}{s_1^{m+2n} \Gamma(k) j!(m-j)!(j-k)!(n-m)!}.\ee The conditions \eqref{Positivity}  are the so-called \emph{Positivity conditions}. In short we will call them $PB_C$.

\subsubsection*{Massless spinning particles}
For massless spinning particles, the conditions for positivity are modified as follows.

\begin{itemize}
	\item Recall that based on the construction outlined in \ref{crampconstr}, our crossing symmetric amplitudes are linear combinations of helicity amplitudes. Therefore, the spectral functions need to appear in the right combinations amenable to positivity constraints. Unlike the scalar case, spectral functions of all the helicity amplitudes do not obey positivity conditions. Instead, they are constrained by unitarity considerations in a way that certain specific linear combinations are positive (See appendix \ref{uniconst} for details). In subsequent sections, we will see how it is achieved in our crossing symmetric amplitudes.
	
	\item For helicity amplitudes, we obtain Wigner-$d$ functions in the partial wave decomposition with the precise form given by \eqref{maspin} and \eqref{fidef}. One can check that for the relevant helicity amplitude basis we have, and the linear combination in which they appear for our crossing symmetric amplitudes, relevant linear combinations of Wigner-$d$ function and its derivatives are positive for $\xi_0=1$.     
\end{itemize} 

Considering these points, just like the scalar case, we would like to construct a linear combination of $W^{\text{spin}}_{n-m,m}$, which is positive. It will turn out that the scalar ansatz suffices for the cases we consider. The structural reason which allows us to do this will also become apparent as we work out the relevant examples in subsections  \ref{photonboundp} and \ref{gravitonboundp}.

\subsection{Typically-realness and Low Spin Dominance: $TR_U$}\label{tr}
Now we turn our attention to exploring the connection between typically real functions and the scattering amplitudes that were first unearthed in \cite{Raman:2021pkf}. The key observation leading to this connection is that the crossing-symmetric dispersive integral can be cast into Roberston integral representation, \eqref{RObertson1}. This leads to the discovery of the typical realness of $\mt$. 

\subsubsection{Robertson form of dispersion integral} Defining 
\begin{equation}\label{xidef}
	\x:=1+\frac{27a^2}{2 (s_1^\prime)^3}(a-s_1^\prime),\qquad
	d\m(\x):=\frac{\ma(\x,s_2(\x,a))\,d\x}{\int_{-1}^{1}d\x \ma(\x,s_2(\x,a))},
\end{equation}

\noindent one can formally cast the dispersion integral into  
\be 
\widetilde{\mt}(\tilde{z},a):=\frac{\mt(\tilde{z},a)-\a_0}{\frac{2}{\p}\int_{-1}^{1}d\x \ma(\x,s_2(\x,a))}=\int_{-1}^{1}d\m(\x)\,\frac{\tilde z}{1-2\x \tilde z+\tilde z^2}
\ee which is of the Robertson form \eqref{RObertson1}. We also need to establish analyticity property of  $\widetilde{\mm}(\tilde{z},a)$ in $\tilde{z}$ and the desired non-decreasing property of $\m(\x)$ over $\x\in[-1,1]$. The latter is satisfied as long as $\ma$ is non-negative. These properties are explored with respect to free parameter $a$, and there exists an interval of $a$, $[a_{min},a_{max}]$, where the desired properties are satisfied. Then,  we can infer  
\be\widetilde{\mt}(\tilde{z},a)\in TR,\quad \forall a\in [a_{min},a_{max}]. \ee
Now one can exploit various bounding properties of $TM$ functions to derive bounds on $\{\mw_{n,m}\}$. From the expression of the dispersion relation for spinning particles, it is clear that the similar argument goes through for the spinning case as well. This enables us to use Bieberbach-Rogosinski bounds \eqref{BRbounds} \cite{Raman:2021pkf}. In particular, for massless theories, this implies  
\be\label{alimitstru}
a\in \left[-\frac{M^2}{3},0\right) \cup \left(0, \frac{2M^2}{3}\right].
\ee 
Moreover in deriving this, we had assumed the positivity of the amplitude $A(s_1, s_2^{(+)}(s_1,a))$ which also puts constraints on the range of $a$. The resulting range of $a$ is an intersection  
\begin{eqnarray}\label{postruscalar}
	\left(a\in \left[-\frac{M^2}{3},0\right) \cup \left(0, \frac{2M^2}{3}\right]\right)~~ \cap ~~ \left(a ~~\forall~~ A(s_1,; s_2^{(+)}(s_1,a) \geq 0\right).
\end{eqnarray}
We can demand a very strong constraint that each term in the partial wave decomposition \eqref{crossdispexp} be positive.  The Gegenbauer polynomial functions are positive for $\cos \theta = \sqrt{\frac{s_1+3a}{s_1-a}}\geq 1$. This leads to the constraint $a\in [0, M^2]$ with the upper limit coming from the fact that $a<s_1$ for real $\cos\theta$. This combined with \eqref{postruscalar} gives us the bound
$$0\leq a^{scalar}\leq \frac{2M^2}{3}\,.$$

\subsubsection{ Positivity and Low-Spin Dominance(LSD): Massless scalar EFT }
\label{plsdscalar}
 The analysis we have presented so far only requires   positivity of the absorptive part as a whole i.e $\mathcal{A}(s_1',s_2^{+}(s_1',a))\ge 0$. We imposed the positivity of each term in \eqref{crossdispexp}, spin by spin, which of course, guarantees the positivity of the total absorptive part. In particular, in this way of imposing positivity we demanded the positivity of Gegenbauer polynomials $C_J^{(\a)}(\sqrt{\xi(s'_1, \, a)})$, which led to the constraint $\sqrt{\xi(s'_1, \, a)}>1$ in the previous subsection, and the positivity of the partial wave amplitudes $a_J$ following from the unitarity. However, this is a rather weak condition on the partial wave amplitudes. The dynamical consequence of locality captured by the locality constraints \eqref{Nulldef} has not been considered. It is quite natural to expect that locality constraints result in relative magnitudes of the partial amplitudes such that the positivity of $\ma$ can still be satisfied for $ \sqrt{\x(s'_1,a)}<1$.

Let us illustrate this with the case of massless scalar EFT. We begin with the dispersion relation \eqref{crossdisp} and add 
\be N_c :=- \sum_{n<m\atop m\ge 2} c_{n,m} \mathcal{W}_{n-m,m} a^{2n+m-3} y \ee
Here  $\{c_{n,m}\}$ are arbitrary weights.   Using \eqref{Winv}, we obtain

\bea\label{scalarnc}
& &\mathcal{M}(s_i,a)+N_c \nonumber \\
& &=\a_0+\int_{M^2}^{\infty}\frac{d s_1'}{\pi s_1'} \sum_{\substack{J\ge 0\\ J\, \text{even} }}(2J+1)a_{J}(s_1')\left[C_J^{(\frac{d-3}{2})}(\sqrt{\xi(s_1', a)})- \sum_{{n<m}\atop{m\ge 2}} c_{n,m} \mathcal{\hat{B}}_{n,m}^{J} \frac{ a^{2n+m}  H(a;s_i)}{2(s_1^{'})^{2n+m} H(s_1',s_i)}\right]H(s_1',s_i) \,\nonumber\\
\eea
Where $\sqrt{\xi}= 1+\frac{2 s_2^{+}}{s_1^\prime} = \sqrt{\frac{s_1'+3 a}{s_1'-a}}$ and we have used the fact that $H(a;s_i)=-\frac{y}{a^3}$. We also have the following crucial difference from the massive scalar EFT expression in \eqref{bellmassive}
\begin{eqnarray}\label{bellmassless}
	\mathcal{\hat{B}}_{n,m}^{J}=2\sum_{J=0}^{m}\frac{(-1)^{1-J+m}p_{J}^{(J)}\left(1\right) \left(4 \right)^{J}(3 J-m-2n)\Gamma(n-J)}{ J!(m-J)!\Gamma(n-m+1)}\,
\end{eqnarray}
i.e $\xi_0=1$. The locality constraints \eqref{Nulldef} ensure that $N_c=0$. Thus adding this to $\mm(s_i, a)$ does not change the amplitude. But now we can analyze the consequence of the locality constraints inside the dispersive representation. In fact, we now have the equivalent dispersive representation 
\be 
\mm(s_i, a)=\a_0+\frac{1}{\p}\int_{M^2}^\infty \frac{d s'_1}{s'_1}\, {\ma_L}(s'_1, a)\, H(s'_1, \, s_i),
\ee 
with 
\be \label{ALdef}
{\ma_L}(s'_1, a):= \sum_{\substack{J\ge 0\\ J\, \text{even} }}(2J+1)a_{J}(s_1')\left[C_J^{(\frac{d-3}{2})}(\sqrt{\xi(s_1', a)})- \sum_{{n<m}\atop{m\ge 2}} c_{n,m} \mathcal{\hat{B}}_{n,m}^{J} \frac{ a^{2n+m}  H(a;s_i)}{2(s_1^{'})^{2n+m} H(s_1',s_i)}\right].
\ee Let us call $\ma_L$  \emph{local absorptive part}. For the purpose of our analysis, we now impose the \emph{local positivity} condition as 
\be \label{Lpos}
\ma_L(s'_1, a)\ge 0,\qquad \forall~~s'_1\ge M^2. 
\ee  This condition will result into a new range of validity for $a$. Let us analyze below how this range can be found out. It is worth emphasizing that this condition is \emph{equivalent} to the usual positivity condition $\ma(s'_1,\, a)\ge 0$ when restricted to the to the subspace $N_c=0$ in the space of partial wave amplitudes $\{a_J\}$. 

Since we will eventually study the Wilson coefficient expansion as a Laurent series about $x,y=0$, we analyze $\ma_L(s'_1, a)$  in a low energy expansion about $x=0$. In order to do so, we can replace $s_1'\rightarrow a\frac{\xi^2+3}{\xi^2-1}$ in $\ma_L$, \eqref{ALdef},  and write it as an expansion about $x=0$. Using
\be
\frac{H(a;s_i)}{H(s_1';s_i)}= \frac{(\xi^2+3)^3} {(\xi^2-9)(\xi^2-1)^2} + O(x)
\ee into \eqref{ALdef}, to leading order in $x$, the local positivity requirement becomes 
\be \label{locposlead}
{\ma_L}(s'_1, a):= \sum_{\substack{J\ge 0\\ J\, \text{even} }} (2J+1)a_{J}(s_1')\left[C_J^{(\frac{d-3}{2})}(\sqrt{\xi(s_1', a)})
- \sum_{{n<m}\atop{m\ge 2}} c_{n,m} \mathcal{\hat{B}}_{n,m}^{J} \frac{ (\xi^2-1)^{2n+m-2}}{2(\xi^2-9)(\xi^2+3)^{2n+m-3} }\right]\ge 0.
\ee 
Observe that this leading contribution does not depend \emph{explicitly} on $a$ and is purely a function\footnote{{{ We note that it seems like the denominator has a pole at $\x=3$, but this value is never attained since, from the analyticity requirement of $\mm$, $a<\frac{2M^2}{3}$, we have $\x^2<3$. }}} of $\xi$.  
We can then sum over $n,m$ to $2n+m\le k$, and find the smallest solution $\xi_{min}$ such that $0<\xi_{min}<\xi <1$ for which we can find a set of $c_{n,m}$'s such that \eqref{Lpos} is satisfied.  We can in turn use this value to determine an the new range of $a$ corresponding to the local positivity condition, \be \sqrt{\frac{s_1'+3a}{s_1'-a}}> \xi_{min} \, \implies \,  \frac{(\xi_{min})^2-1}{(\xi_{min})^2+3} M^2\le a.\ee Since, $0<\xi_{min}<1$ the lower bound is stronger than $0<a$ but also weaker than $-\frac{M^2}{3}<a$. We can now combine this with \eqref{alimitstru} to have
\be  \frac{(\xi_{min})^2-1}{(\xi_{min})^2+3} M^2\le a \le \frac{2M^2}{3}. \ee 


The above exercise leads us to $\xi_{min}=0.593$ for the scalar EFT when we consider all locality constraints up to $k=21$\footnote{{{In order to obtain this numerical coefficient, we have used linear programming in Mathematica with 1700 digits of precision to find solutions to the system of inequalities \eqref{locposlead} for $k\leq 21$ while varying $\xi$ in steps of $0.048$ from $\xi_{min}=0.59329$ to $\xi_{max}=2.99$ and spin $J$ from $J=0$ to $J_{max}=56$. The value of $\xi_{min}$ is determined by the lowest value of $\xi$ for which the system of inequalities \eqref{locposlead} have a solution such that not all $c_{n,m}= 0~ \&~ c_{n,m}>-\infty$.}}} This gives the following:
\bea \label{rangeofasc}
{\bf Scalar:} -0.1933 M^2 < a^{scalar} < \frac{2 M^2}{3}\, ,
\eea
Note that the lower bound of $a$ has been modified from the case where we demanded the positivity of each partial wave without considering the locality constraints. From a physical perspective, the dynamical constraints of UV consistency on scalar IR EFT are responsible for lowering the bound on $a$ from the previous subsection.

Our findings are also indicative of the well-known phenomenon of {\it low spin dominance}(LSD), i.e. the higher spin partial wave amplitudes are suppressed. To be precise, we can calculate the $\xi_{min}$ in an alternative way. 
\begin{itemize}
	\item Consider \eqref{scalarnc} without the locality constraints, and instead, we truncate the sum over spin to some $J=J_c$ try to find $\xi_{min}$ demanding this finite sum be positive. This means we assume that the sign of the absorptive part in \eqref{scalarnc} does not change beyond a certain critical spin $J=J_c$ because of the smallness of $a_{J>J_c}(s_1)$. Therefore truncating the partial wave sum and doing the positivity analysis is justified. Formally,  the positivity condition for the truncated expression reads
	\be
	\sum_{\substack{J= 0\\ J\, \text{even} }}^{J_c}(2J+1)a_{J}(s_1')\, C_J^{(\frac{d-3}{2})}(\sqrt{\xi(s_1', a)})> 0,\qquad \forall~~\x> \xi_{min},\,\, \, s'_1> M^2. 
	\ee  
	We are also assuming that for $a_J\ne 0$ at least for one $J\in \{0,2,4,\cdots, J_c\}$. 
	\item Therefore, we look for the smallest simultaneous root $\x_{J_c}$ of the set 
	$$\left\{C_J^{(\frac{d-3}{2})}\left(\sqrt{\xi(s_1', a)}\right)\, \Big|\,J\in \{0,2,4,\cdots,J_c\}\right\}$$ such that 
	\be 
	C_J^{(\frac{d-3}{2})}\left(\sqrt{\xi(s_1', a)}\right)>0,\quad \forall~~ \xi>\xi_{J_c}\,,\, J\in \{0,2,4,\cdots,J_c\}.
	\ee 
	For a given $J_c$, this $\x_{J_c}$ is the $\xi_{min}$ that we considered before. In particular, we have that 
	$\x_{J_c}\to 1$ as $J_c\to \infty$, which is expected. Therefore, for the truncated set we consider, $1 \ge \xi > \xi_{J_c}$ ensures that the LHS is positive since we have assumed that the sign of the absorptive part doesn't change after $J_c$. 
	\item Combining this with $PB_C$, this constrains the range of $a$ to \be  \frac{(\xi^{(J_c)})^2-1}{(\xi^{(J_c)})^2+3} M^2\le a \le \frac{2 M^2}{3} .\ee The first few values after rationalising to agree with 2 significant digits are:
	\begin{center}
		\begin{tabular}{ |c| c| }
			\hline
			$J_c$ & Scalar \\ 
			\hline
			2 & $-0.2 M^2 <a< \frac{2 M^2}{3}$ \\  
			\hline
			3 &$-0.69 M^2 <a< \frac{2 M^2}{3}$  \\
			\hline
			4 &$-0.034 M^2 <a< \frac{2 M^2}{3}$ \\ 
			\hline
		\end{tabular}
	\end{center}
	\item We can see that the argument with locality constraints combined with the above analysis clearly indicates Spin-$2$ dominance for the scalar case. More precisely, we input locality constraints in estimating the range of $a$ in first part of this subsection (i.e in the analysis leading up to \eqref{rangeofasc}). Locality constraints can also be interpreted as constraints on allowed $a_J(s)$ for scalar EFTs. Thus the range of $a$, after including the locality constraints (i.e., considering the allowed space of scalar theories), approximately coincides with the range that we get from a completely different analysis without using the null constraints {\it and assuming} that higher scalar partial waves do not change the sign of the absorptive part after spin 2 (see 1st entry of the table above). This implies that UV consistency of scalar EFTs leads to spin 2 dominance.
\end{itemize}    

\subsubsection{Massive scalars}

For a massive theory, we can repeat the analysis of the previous subsection. We shall consider the case of the massive scalar with mass $m$ and $\mu=4 m^2$. This was already considered in \cite{Raman:2021pkf} where it was argued that the range of $a$ was $\frac{-M^2}{3}<a<\frac{2 M^2}{3}$ and bounds were obtained for various Wilson coefficients. We revisit this using our new method using the locality constraints.
The key changes are in the relation between $\xi$ and $s_1^{'},a$ which is given by $\xi = \xi_0 \sqrt{\frac{s_1^{'}+3 a}{s_1^{'}-a}}$ with $\xi_0=\frac{s_1^{'}}{s_1^{'}-\mu}>1$ and the locality constraints \eqref{bellmassive}:
\bea
\mathcal{B}^{\ell,i}_{n,m}(s_1)=2\sum_{j=0}^{m}\frac{(-1)^{1-j+m} p_{\ell}^{(j,i)}\left(\xi_{0}\right) \left(4 \xi_{0}\right)^{j}(3 j-m-2n)\Gamma(n-j)}{ j!(m-j)!\Gamma(n-m+1)}\,.
\eea

\noindent It can be easily checked that this gives 
\bea
\frac{a^{2n+m} H(a;s_i)}{ (s_1^{'})^{2n+m} H(s_1';s_i)}= \frac{(\xi^2-\xi_0^2)} {(\xi^2-9 \xi_0^2)(\xi^2+3 \xi_0^2)} +o(x) \, .\nonumber 
\eea
Proceeding with the analysis, it turns out that there are no solutions for any $\xi<1$. However since $\xi=1$ was used to obtain the previous range of $a$ namely $\frac{-M^2}{3}<a<\frac{2 M^2}{3}$ these do not give us a stronger range of $a$. Thus, we conclude that 

\begin{center}
\emph{there is no low spin dominance for the massive case.}
\end{center}
This justifies the results in \cite{Raman:2021pkf} and highlights a key difference between the massive and massless cases. We have carried out explicit checks using the pion S-matrices from the S-matrix bootstrap \cite{joaopion, aspion1, aspion2} which verifies this claim. 
\subsection{Bieberbach-Rogosinski bounds}\label{BRbounds}
We can expand $\widetilde{\mt}(\tilde z ,a)$ about $\tilde z=0$ by expanding the kernel $H(s_1,\tilde z)$
\bea
H(s_1,\tilde z)= \frac{27 a^2 \tilde z (2 s_1-3a)}{27 a^3 \tilde z-27 a^2 \tilde z s_1-(\tilde z-1)^2 s_1^3}=\sum_{n=0}^\infty \beta_n(a,s_1)\tilde z^n\,,
\eea
Comparing this with the low energy expansion of the amplitude  
\bea
\widetilde{\mathcal{T}}_0(\tilde z ,a)=\sum_{p,q=0}^\infty \mathcal{W}_{p,q} x^p y^q =\sum_{n=0}^\infty  a^{2n} \alpha_n(a) \tilde z^n\nonumber 
\eea
after rewriting in-terms of $\tilde z$ using $x=-\frac{-27 a^3 \tilde z}{(1-\tilde z)^2 }$ and $y=-\frac{-27 a^2 \tilde z}{(1-\tilde z)^2 }$ gives: 
\bea
a^{2n}\alpha_n(a)&=&\frac{1}{\pi}\int_{M^2}^\infty \frac{ds_1'}{s_1'}                                          {\mathcal A}(s_1';s_2^+(s_1',a))\beta_n(a,s_1')\,, \nonumber \\
\text{with} ~ \alpha_p(a)&=&\sum_{n=0}^p \sum_{m=0}^n \mathcal{W}_{n-m,m} a^{2n+m-2p}(-27)^n\frac{\Gamma(n+p)}{\Gamma(2n)(p-n)!}\,,\quad p\geq 1\,.
\eea
In particular we have $\mathcal{W}_{0,0}= \alpha_0$ and 
$ a^{2}\alpha_1(a)=\frac{1}{\pi}\int_{M^2}^\infty \frac{ds_1'}{s_1'} \mathcal {A}(s_1';s_2^+(s_1',a))\beta_1 (a,s_1')\,.$
Note that since $\beta_1(a,s_1)=\frac{27 a^2}{s_1^3}(3a-2 s_1)$ and $a< \frac{2 M^2}{3}< \frac{2 s_1}{3}$ we have $\beta_1 <0$. Thus,
\be \label{w01}
\boxed{
	\alpha_1(a) <0}
\ee
\noindent We can apply the Bieberbach-Rogosinski inequalities on the coefficients of any typically-real function $f(z)=z+ a_2 z^2+a_3 z^3\cdots$ inside the unit disk following \cite{Raman:2021pkf}:
\bea \label{tru}
&& -\kappa_n \le \frac{\alpha_n(a) a^{2n}}{\alpha_1(a)a^2}
\le n 
\eea
with 
\bea
\kappa_n = n~ \text{for even}~ n, \qquad \kappa_n = \frac{\sin n~\vartheta_n}{\sin \vartheta_n} ~\text{for odd}~ n \,,
\eea 
where $\vartheta_n$ is the smallest solution of $\tan n \vartheta = n \tan \vartheta$ located in $( \frac{\pi}{n},~\frac{3 \pi}{2n} )$ for $n>3$ and $\kappa_3=1$, 
to constrain the Wilson coefficients in a low-energy expansion of the amplitude. We call these conditions \eqref{tru} collectively as $TR_U$.


\subsection{Summary of algorithm}
In this section, we summarise our algorithm.  The central characters of the story are the Wilson Coefficients, the  partial wave decomposition of the amplitude and the crossing symmetric kernel. Firstly, unitarity of the partial wave amplitude decomposition and positivity of the spherical harmonics and their derivatives for unphysical region of scattering, translate to positivity relations of the Wilson coefficients ( also known as the $PB_C$ conditions \cite{Raman:2021pkf}). To be more precise, unitarity demands that the imaginary part of the partial wave coefficients is positive. The Gegenbauer polynomials (or the relevant linear combination of the Wigner-$d$ functions for the spinning case) or its derivatives which appear in the partial wave expansion of the amplitude are positive in the unphysical region of scattering ($\cos \theta >1$). The Wilson coefficients themselves however might contain positive or negative sum of both the manifestly positive quantities. The $PB_C$ conditions are then linear combination of the Wilson coefficient expressions such that it is manifestly positive. Secondly, the fact that the amplitude is typically real for a range of the parameter $a$ then allows us to systematically obtain two-sided bounds on the Wilson coefficients (also known as the $TR_U$ conditions \cite{Raman:2021pkf}). This is because the typically real amplitude, as an expansion in $\tilde{z}$, has Bieberbach-Rogosinski bounds on the expansion coefficients \cite{Raman:2021pkf}. Thirdly, we use locality, which modifies the lower range of $a$ as obtained from $\cos \theta>1$ and $TR_U$. In the following sections, we systematically implement this algorithm to first review bounds on the scalar and then obtain the same for graviton and photon EFTs. { These steps are summarised in the flow chart below:}     

\begin{figure}[H]
	\begin{centering}
		\includegraphics[width=0.9\textwidth]{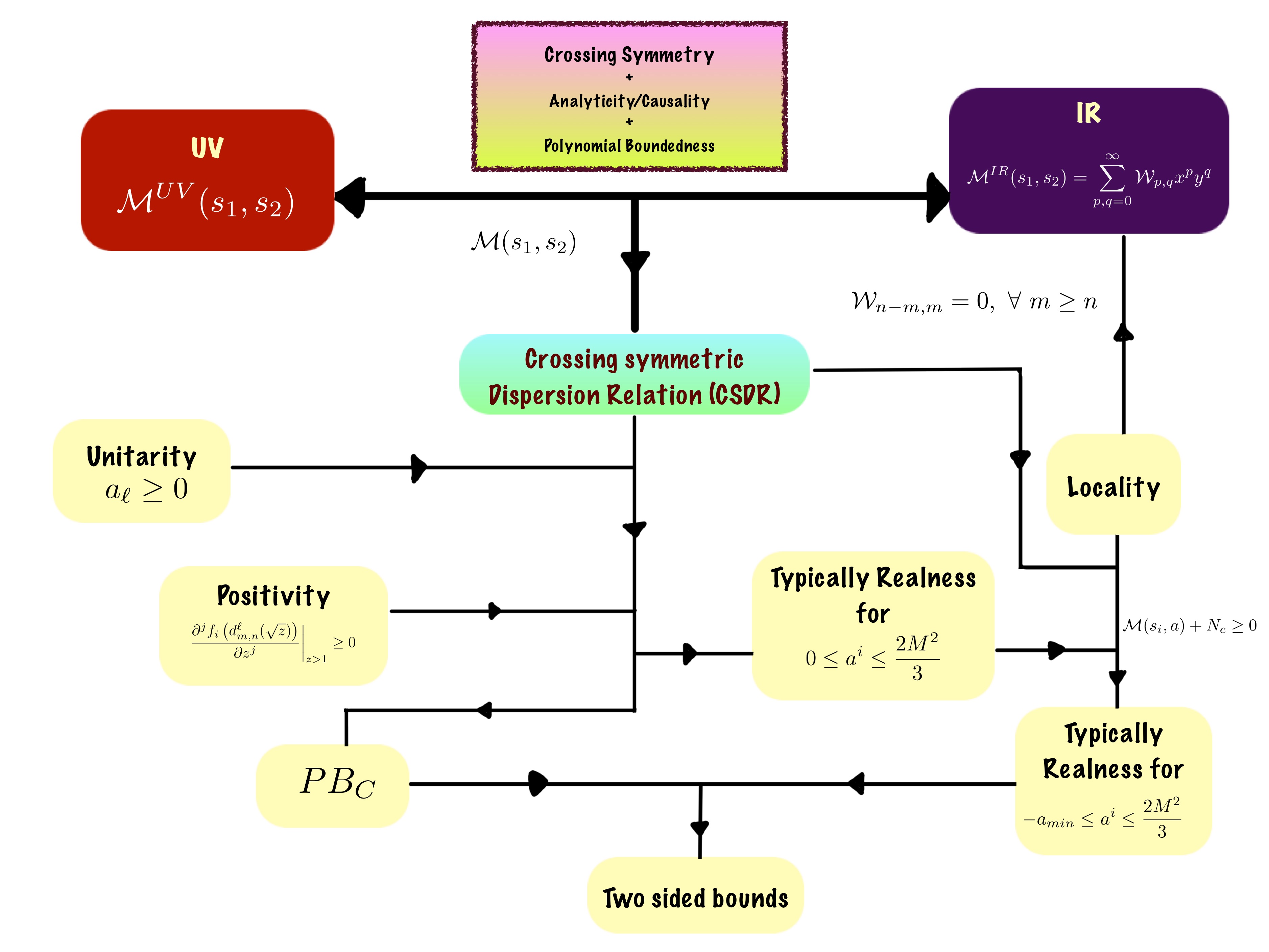}
		\caption{The above flowchart shows the steps involved in the GFT approach.}
	\end{centering}
\end{figure}

\section{Scalar bounds} \label{scalboundsec}
We now present the applications of formalism developed in the previous sections for various EFTs starting with the massless scalar case. The massive scalar was already addressed in \cite{Raman:2021pkf}.  Recall that the low energy EFT expansion of the amplitude takes the following form in terms of crossing-symmetric variables $x,\,y$
\bea \label{leexp}
F(s_1,s_2,s_3)&=& \sum_{p,q=0} \mathcal{W}_{p,q} x^p y^q .\, \nonumber
\eea
Starting with the dispersion relation is given by \eqref{crossdisp} and \eqref{crossdispexp} we can systematically derive the positivity bounds. We linearise the steps in the following manner which will serve as a guideline for us when evaluating EFTs with spinning particles.  
\begin{itemize}
	\item Unitarity implies that in the dispersion relation, the spectral functions $a_\ell$ are positive. 
	
	\item From the positivity of the Gegenbauer polynomials in the dispersion relation and the typical realness of the amplitude, we have the following range for $a$ \eqref{rangeofasc}, 
	
	$$-0.1933 M^2 <a^{scalar}< \frac{2 M^2}{3}\,.$$
\end{itemize}

\noindent From \eqref{Positivity}, \eqref{chi} and \eqref{frakU}, we get positivity conditions on linear combinations \footnote{This is well defined as $\mathcal{W}_{1,0} >0$ as argued in \eqref{wn0}.} of $w_{p,q}=\frac{\mathcal{W}_{p,q}}{\mathcal{W}_{1,0}}$. These set of conditions have been referred to in literature \cite{Raman:2021pkf} as $PB_C$ conditions. Since the amplitude is typically real, we also obtain Bieberbach Rogosinski bounds on $w_{p,q}$ from \eqref{tru} (also known in literature as $TR_U$). The algorithm will generically follow \cite{Raman:2021pkf} and we refer the interested reader to the details there. We begin by first noting that \eqref{w01} implies
\bea
\alpha_1(a)&=&-27 a^2(1+ a~ w_{01})\leq 0 \, \qquad\forall-0.1933 M^2 <a< \frac{2 M^2}{3}, \nonumber \\
&\implies& \frac{-3}{2M^2} \leq w_{01}\leq \frac{5.1733}{M^2}.
\eea
We note that this precisely agrees with the result in \cite{Caron-Huot:2020cmc} when we account for the difference in the definitions of $x$ due to the conventions. We have used $x= -s_1 s_2-s_2 s_3-s_3 s_1$ while \cite{Caron-Huot:2020cmc} used $s_1^2+s_2^2+s_3^2$ which gives $w_{01}=\frac{-\tilde{g}_3}{2}$ (since $s_1+s_2+s_3=0$) which translates to $-10.34<\tilde{g}_3<3$. In the second step, we solve for the set of inequalities derived from \eqref{Positivity}, \eqref{chi}, \eqref{frakU} and \eqref{tru} upto a certain value of $n=n_{max}$. Note that the conditions derived from \eqref{tru} are $a$ dependent. In order to efficiently solve the inequalities, we discretize the variable $a$ over the range specified in \eqref{rangeofasc} in steps of $\delta a$ and then solve for the resulting larger set of inequalities. We present our results for $n_{max}=5$ and $\delta a= \frac{1}{101}$ in the table below. We shall follow the convention of \cite{Raman:2021pkf} namely $M^2=\frac{8}{3}$ and re-write the results of \cite{Caron-Huot:2020cmc} in this convention for ease of comparison.

\begin{table}[h!]
	\centering
	$\begin{array}{|c|c|c|c|c|}
		\hline
		w_{p,q}=\frac{\mathcal{W}_{p,q}}{\mathcal{W}_{1,0}} & TR_U^{min} &SDPB^{min} & TR_U^{max} & SDPB^{max}\\
		\hline
		w_{01} & -0.5625 &-0.5625 & 1.939 &1.939\\
		\hline
		w_{11} & -0.1318 & -0.1318 & 0.219 &0.216\\
		\hline
		w_{02} & -0.1533 &-0.1268  &0.063 & 0.0296\\
		\hline
		w_{20} & 0 &  0 &0.140625 & 0.140625\\
		\hline
		w_{21} &-0.02595 &  -0.02595  &0.0513& 0.023\\
		\hline
		w_{12}& -0.061& -0.02789& 0.0275& 0.0111\\
		\hline
		w_{30} & 0 &  0 & 0.01977 & 0.01977\\
		\hline
		w_{03} & -0.011& -0.00156 & 0.017 & 0.0071 \\
		\hline
		w_{31}& -0.0047 &-0.0047 &0.0022 & 0.0022\\
		\hline
		w_{40}& 0& 0 &0.00278& 0.00278\\
		\hline
		w_{50}& 0&0& 0.00039 &0.00039\\
		\hline
	\end{array} $
	\caption{A comparison of the values obtained using our results $TR_U$ up to $n=5$ and $SDPB$ in \cite{Caron-Huot:2020cmc}}\label{mint}
\end{table}
We note that we get an excellent agreement with \cite{Caron-Huot:2020cmc}. In \cite{Raman:2021pkf} a comparison was done with the massive case and the results of \cite{Caron-Huot:2020cmc}, where it was noted that some of the results $TR_U$ were stronger. However, since \cite{Caron-Huot:2020cmc} considered only the massless case, so the above is more appropriate comparison as the results show. We attribute the discrepancy in the $w_{02},w_{21},w_{12},w_{03}$ values to the following:\\
\noindent We have not completely solved the Locality constraints $N_c$ but we have implicitly assumed that they are zero when we consider a low energy expansion \eqref{leexp}. However, each Wilson coefficient actually involves an infinite sum of locality constraints for instance $$\alpha_1(a) = -27 a^2 \mathcal{W}_{1,0}\left((1+ a w_{01})+\color{red}{ \sum_{n=1}^{\infty}{ w_{-n, n+1} ~a^{n-1}}}\right)\,.$$
\noindent Strictly speaking, what we have are bounds on these combinations and not the $w_{p,q}$'s themselves. In practice however one would have expected that since we obtained the range of $a$ by using some of the locality constraints this should have resolved the issue. However, we remind the reader that to get the bounds listed in the table we used $PB_c$ conditions in addition to the $TR_U$'s. The $PB_c$'s are linear conditions in the $w_{p,q}$'s which are {\it independent} of $a$.

\section{Photon bounds}\label{photonbound}

In this section we will constrain parity preserving photon EFTs. To be precise let us consider the following crossing symmetric helicity amplitudes from \eqref{f1} and \eqref{f2}. This set up applies almost identically to the graviton case also so we present here the general amplitude that we will be considering for later use. 

\begin{eqnarray}\label{photamp}
	M^{\alpha}(s_1,s_2, s_3)&=&F^{\alpha}_2(s_1,s_2,s_3)+ x F^{\alpha}_1(s_1,s_2,s_3) \nonumber\\
	&=& (\mt^{\alpha}_1(s_1,s_2,s_3)+ \mt^{\alpha}_3(s_1,s_2,s_3)+ \mt^{\alpha}_4(s_1,s_2,s_3)) + x_1 \mt^{\alpha}_2(s_1,s_2,s_3),
\end{eqnarray}
where $x_1 \in [-1,1]$ and $\alpha= \gamma, h$ for photons and gravitons respectively. The partial wave expansion of this amplitude is given by, 
\bea \label{abs}
&&(F^\alpha_2(s_1,s_2)+x F^\alpha_1(s_1,s_2)) = \sum_{J=0,2,4,\cdots} 16 \pi (2 J+1) (\rho^{1, \alpha}_J+x \rho^{2, \alpha}_J ) d^{J}_{0,0}(\theta) + \nonumber\\
&&\sum_{J=2,4,\cdots} 16 \pi (2 J+1) \rho^{3,\alpha}_J (d^J_{2,2}(\theta)+d^J_{2,-2}(\theta))+\sum_{J=3,5,\cdots} 16 \pi (2 J+1) \rho^{3,\alpha}_J (d^J_{2,2}(\theta)-d^J_{2,-2}(\theta)), \nonumber\\ 
\eea
$d^J_{m,m'}$ being the Wigner $d$-matrix defined in appendix (\ref{repsofwigd}). From the positivity of the spectral functions in these cases( see appendix (\ref{A.1})), the reader can understand that this combination is positive\footnote{We leave the analysis of $F_4,F_5$ which have denominators involving $s_i$ for later analysis. For $F_3$ since unitarity does not fix the sign of $\r^{\alpha,5}_J$ our methods are not applicable.}- since $\rho^{1, \alpha}_J \pm \rho^{2,  \alpha}_J \ge 0$ we have 
\bea
\rho^{1, \alpha}_J + x_1 \rho^{2, \alpha}_J = \underbrace{\frac{(1+x_1)}{2} (\rho^{1, \alpha}_J +  \rho^{2, \alpha}_J)}_{ \ge 0}+\underbrace{\frac{(1-x_1)}{2}  (\rho^{1, \alpha}_J -\rho^{2, \alpha}_J) }_{\ge 0} \ge 0,
\eea 
and $\rho^{3, \alpha}_J \geq 0$ from the analysis in appendix (\ref{A.1}). The crossing symmetric dispersion relation for the photon amplitude is given by
\bea\label{dispphoton} 
(F^\g_2(s_1,s_2)+x F^\g_1(s_1,s_2)) =\alpha^\g_0+\frac{1}{\pi}\int_{M^2}^{\infty}\frac{d s_1}{s'_1} \ma^\g\left(s_1^{\prime} ; s_2^{(+)}\left(s_1^{\prime} ,a\right)\right) H\left(s_1^{\prime} ;s_1, s_2, s_3\right),\nonumber\\
\eea
where $H\left(s_1^{\prime} ;s_1, s_2, s_3\right)$ is defined in \eqref{kernel} and the partial wave decomposition reads
\bea\label{crossdispexpphoton}
\ma^\g\left(s_1^{\prime} ; s_2^{(+)}\left(s_1^{\prime} ,a\right)\right)&=& \sum_{J=0,2,4,\cdots} 16 \pi (2 J+1) (\rho^{1, \g}_J+x_1 \rho^{2, \g}_J ) d^{J}_{0,0}(\theta) + \nonumber\\
&&\sum_{J=2, 4,\cdots} 16 \pi (2 J+1) \rho^{3,\g}_J (d^J_{2,2}(\theta)+d^J_{2,-2}(\theta)) + \nonumber\\
&&+\sum_{J=3,5,\cdots} 16 \pi (2 J+1) \rho^{3,\g}_J (d^J_{2,2}(\theta)-d^J_{2,-2}(\theta))\,, 
\eea  
where $\cos^2\theta=\xi(s'_1, a)=1+ 4 \left(\frac{a}{s'_1-a}\right)$. Note that due to the fact that we have written down crossing symmetric combination of helicity amplitudes, the crossing symmetric dispersion relation is essentially of the same structure as the scalar one. In writing the dispersion relations \eqref{crossdisp} and \eqref{abs}, we have used \eqref{maspin} and \eqref{fidef}. The low energy EFT expansion of the amplitude reads, 
\bea
F^\g_1(s_1,s_2)&=& \sum_{p,q}\mw^1_{p,q} x^p y^q,\qquad F^\g_2(s_1,s_2)= \sum_{p,q}\mw^2_{p,q} x^p y^q  \,.
\eea
For our analysis, we will be considering the most general Euler-Heisenberg type EFT for the photon
\bea
\mathcal{L}= -\frac{1}{4} F_{\mu \nu} F^{\mu \nu}+a_1 \left(F_{\mu \nu}F^{\mu \nu} \right)^2+a_2 (F_{\mu \nu}\tilde{F}^{\mu \nu})^2+\cdots
\eea
obtained starting with a UV complete theory such as QED and integrating out the other massive particles in the theory such as say the electron. To compare against the corresponding low energy EFT expansion coefficients of \cite{vichi} we can rewrite our EFT expansion in the form (see \eqref{EFTexpappphoton}), 
\bea
F_2(s_1,s_2,s_3)&=& 2 g_2 x- 3 g_3 y+ 2 (g_{4,1}+2g_{4,2})x^2+\cdots \\
F_1(s_1,s_2,s_3)&=& 2 f_2 x- f_3 y+ 4  f_4 x^2+\cdots
\eea
where the Wilson coefficients can be related to the EFT couplings such as $a_1 =\frac{f_2+g_2}{16},~a_2= \frac{f_2-g_2}{16}$ etc. 

\subsection{Wilson coefficients and Locality constraints: $PB^\g_C$}\label{photonboundp}

The local low energy expansion of the amplitude \eqref{photamp} can be written as 
\bea \label{lowexpa}
F^\g_2(s_1,s_2)+x_1 F^\g_1(s_1,s_2) &=& \sum_{p,q=0}^{\infty} \mathcal{W}^{(x)}_{p,q} x^p y^q \qquad= \sum_{p,q=0}^{\infty} \mathcal{W}^{(x)}_{p,q} x^{p+q} a^q, \eea
where we have used $a=y/x$ and $\mathcal{W}^{(x)}_{p,q}= \mathcal{W}^{2}_{p,q}+x_1 \mathcal{W}^{1}_{p,q}$. Just like in the scalar case, we would like to expand both the sides of the dispersion relation \eqref{dispphoton} to derive an expression for the locality constraints- recall that by incorporating crossing symmetry we have compromised on locality which serves as constraints in our formalism. In order to do so, we expand the kernel \eqref{kernel} and the partial wave  Wigner-$d$ functions in \eqref{dispphoton} about $a=0$ and compare powers on both sides. Note that for $a=0$, the Wigner-$d$ functions are Taylor expanded about $\xi_0=1$ (since the argument of the Wigner-$d$ functions are $\xi(s_1, a)=1+ 4 \left(\frac{a}{s_1-a}\right)$). We obtain 
\bea \label{photonbell}
\mathcal{W}^{(x_1)}_{n-m,m}&=&\int_{M^2}^{\infty}\frac{d s_1}{2\pi s_1^{2n+m+1}} \sum_{J=0,2,4,\cdots}(2J+1)a^{(1)}_J(s_1)\mathcal{G}_{n,m}^{J,1}\,\nonumber\\
&+&\int_{M^2}^{\infty}\frac{d s_1}{2\pi s_1^{2n+m+1}} \sum_{J=2,4,\cdots}(2J+1)a^{(2)}_J(s_1)\hat{\mathcal{G}}_{n,m}^{J,2}\,\nonumber\\
&+& \int_{M^2}^{\infty}\frac{d s_1}{2\pi s_1^{2n+m+1}}, \sum_{J=3,5,7,\cdots}(2J+1)a^{(3)}_J(s_1)\hat{\mathcal{G}}_{n,m}^{J,3}\,,\nonumber\\
\hat{\mathcal{G}}_{n,m}^{J,i}&=&2\sum_{j=0}^{m}\frac{(-1)^{1-j+m}q_{J}^{(j,i)}\left(1\right) \left(4 \right)^{j}(3 j-m-2n)\Gamma(n-j)}{ j!(m-j)!\Gamma(n-m+1)}\,.
\eea 
Here $a^{(1)}_J =\rho^{1, \g}_J+x \rho^{2, \g}_J$, $a^{(2)}_J=a^{(3)}_J=\rho^{3, \g}_J$ and $q_{J}^{(j,i)}(1)= \frac{\partial^j f^{(i)}(\sqrt{\xi})}{\partial \xi^j}\bigg|_{\xi=\xi_0=1}$ with $f^{(1)}=d^J_{0,0},f^{(2)}=d^J_{2,2} + d^J_{2,-2},~f^{(3)}=d^J_{2,2} -d^J_{2,-2}$. For convenience, in order to compute the partial derivatives $q_{J}^{(j,i)}(1)$, we use the representation of the Wigner-$d$ functions in terms of Hypergeometric functions, given in \eqref{wignerphoton}. The locality constraints for this case are therefore given by 

\bea\label{nullphoton}
\mw^{(x_1)}_{n-m,m}=0~~~ \forall n<m \,.
\eea

We would also like to construct the spinning equivalent of $PB_C$ as done for scalars. To this end, we note that spectral functions $a_J^{(i)} \ge 0$ by unitarity  and $\{d^J_{0,0}(\theta), d^J_{m,m}(\theta) \pm d^J_{m,-m}(\theta)\}$ are positive for all $J$ whenever cosine of the argument is bigger than or equal to $1$ (see \ref{repsofwigd}) i.e., $q_{J}^{(j,i)}(1)>0$ for all $J,j =0,1,2,\cdots$ and $i=1,2,3$. In particular we have
\bea \label{wn0}
\mathcal{W}^{(x_1)}_{n,0}&=& \int_{M^2}^{\infty}\frac{d s_1}{2\pi s_1^{2n+m+1}} \sum_{J=0,2,4,\cdots}(2J+1)a^{(1)}_J(s_1)q_{J}^{(0,1)}(1)\,\nonumber\\
&+&\int_{M^2}^{\infty}\frac{d s_1}{2\pi s_1^{2n+m+1}} \sum_{J=2,4,\cdots}(2J+1)a^{(2)}_J(s_1) q_{J}^{(0,2)}(1)\,\nonumber\\
&+& \int_{M^2}^{\infty}\frac{d s_1}{2\pi s_1^{2n+m+1}} \sum_{J=3,5,7,\cdots}(2J+1)a^{(3)}_J(s_1) q_{J}^{(0,3)}(1)  \ge 0 \, .
\eea
More generally in \eqref{photonbell} the sign of any term in $J$ expansion is controlled by $\mathcal{G}_{n,m}^{J,i}(s_1)$ alone. We can thus take linear combinations of various $\mathcal{W}^{(x)}_{p,q}$'s which is a positive sum of $\{d^J_{0,0}(\theta), d^J_{m,m}(\theta) \pm d^J_{m,-m}(\theta)\}$ and their derivatives and hence is manifestly positive. This gives us the {\it Positivity conditions}:
\bea \label{pbcphot}
\sum_{r=0}^m \chi_n^{(r,m)} (M^2) \mathcal{W}^{(x)}_{n-r,r} &\ge& 0,\qquad 
0\le \mathcal{W}^{(x)}_{n,0}\le \frac{1}{\left(M^2\right)^2} \mathcal{W}^{x}_{n-1,0}\,,  ~~~~~~~~~~~~~~n\ge2\,.
\eea
The $\chi_n^{(r,m)} (M^2)$ satisfy the recursion relation:
\bea
\chi_n^{(m,m)} (M^2) &=&1,\nonumber\\
\chi_n^{(r,m)} (M^2) &=& \sum_{j=r+1}^m (-1)^{j+r+1} \chi_n^{(j,m)} \frac{{\mathscr U}_{n,j,r} (M^2)}{{\mathscr U}_{n,r,r}(M^2)},
\eea
with 
 \be 
 \scalebox{0.95}{$
 {\mathscr U}_{n,m,k}= -\frac{4^k \Gamma \left(\frac{1}{2} (2 k+1)\right) (3 k-m-2 n) \Gamma (n-k) \text{s1}^{-m-2 n} \, _4F_3\left(\frac{k}{2}+\frac{1}{2},\frac{k}{2},k-m,k-\frac{m}{3}-\frac{2 n}{3}+1;k+1,k-n+1,k-\frac{m}{3}-\frac{2 n}{3};4\right)}{\sqrt{\pi } \Gamma (k+1) \Gamma (-k+m+1) \Gamma (-m+n+1)}.$}
 \ee 
  We call the conditions \eqref{pbcphot} collectively as $PB^\g_C$\footnote{ We have used the closed form  expression for ${\mathscr U}_{n,m,k}$ in \cite{Sinha:2020win} }. We note here that the positivity conditions $PB^\g_C$ in this case are identical to the ones for massive scalar in \cite{Sinha:2020win,Raman:2021pkf}.  This is simply a consequence of the fact that \eqref{photonbell} is the sum of three scalar like terms each of which has an identical structure except for the functions $q_{J}^{(j,i)}(1)$ in $\mathcal{G}_{n,m}^{J,i}(s_1)$. Since we do not use the explicit form of the function $q_{J}^{(j,i)}(1)$ anywhere in the argument above but just the fact that its positive, the result simply follows. Note that these positive combinations are certainly not unique and one can definitely find different linear combinations which may result in a stronger constraint however we will not pursue that here. \\

\subsection{Typical Realness and Low Spin Dominance: $TR^\g_U$}\label{photonboundTRLSD}

In this section we try to get $a$ range of a using positivity of the amplitude coupled with locality constraints and typical realness of the amplitude. 
The analysis for typical realness is straightforward. From, section \ref{tr} and the discussion regarding the Robertson form of the integral (see the discussion around \eqref{alimitstru}), the limit of $a$ is given by, 

\begin{eqnarray}\label{postruphoton}
	\left(a\in \left[-\frac{M^2}{3},0\right) \cup \left(0, \frac{2M^2}{3}\right]\right)~~ \cap ~~ \left(a ~~\forall~~ A(s_1,; s_2^{(+)}(s_1,a) \geq 0\right). \nonumber\\
\end{eqnarray}

We can assume the positivity of the absorptive part as a whole i.e $\mathcal{A}(s_1',s_2^{+}(s_1',a))\ge 0$ in \eqref{dispphoton} which gives us the range of $a$ as $a \in (0, \frac{2M^2}{3}]$ . This is obtained by considering the positivity of each term in the spin sum which of course guarantees the positivity of the full absorptive part though it maybe too strong (similar to the massless scalar case). A more careful analysis requires us to use the locality constraints $N^\g_c =- \sum_{n<m\atop m\ge 2} c_{n,m} W^{x_1}_{n-m,m} a^{2n+m-3} y $ to it with arbitrary weights $c_{n,m}$'s \eqref{nullphoton}. 
\bea\label{ineqphoton}
& &\mathcal{M}(s_i,a)+N_c = \nonumber \\
& &\int_{M^2}^{\infty}\frac{d s_1}{2\pi s_1} \sum_{J=0,2,4,\cdots}(2J+1)a^{(1)}_J(s_1)\left[f^{(1)}_J(\xi)- \sum_{{n<m}\atop{m\ge 2}} c_{n,m} \mathcal{\hat{G}}_{n,m}^{J,1} \frac{ a^{2n+m}  H(a;s_i)}{(s_1^{'})^{2n+m} H(s_1',s_i)}\right]H(s_1',s_i) \,\nonumber\\
&+&\int_{M^2}^{\infty}\frac{d s_1}{2\pi s_1} \sum_{J=2,4,\cdots}(2J+1)a^{(2)}_J(s_1)\left[f^{(2)}_J(\xi)- \sum_{{n<m}\atop{m\ge 2}} c_{n,m} \mathcal{\hat{G}}_{n,m}^{J,2} \frac{ a^{2n+m}  H(a;s_i)}{(s_1^{'})^{2n+m} H(s_1',s_i)}\right]H(s_1',s_i)\,\nonumber\\
&+& \int_{M^2}^{\infty}\frac{d s_1}{2\pi s_1} \sum_{J=3,5,7,\cdots}(2J+1)a^{(3)}_J(s_1)\left[f^{(3)}_J(\xi)- \sum_{{n<m}\atop{m\ge 2}} c_{n,m} \mathcal{\hat{G}}_{n,m}^{J,3} \frac{ a^{2n+m}  H(a;s_i)}{(s_1^{'})^{2n+m} H(s_1',s_i)}\right]H(s_1',s_i)\,\geq 0,\nonumber\\
\eea
where $\mathcal{\hat{G}}_{n,m}^{J,i}$ has been defined in \eqref{photonbell}. Note that this is similar to the equation we had for the scalar case \eqref{ineqscalar} and therefore the analysis is also similar. The algorithm is very similar with the only difference is that the $\xi_{min}$ is determined by the maximum lower bound obtained by considering the positivity of three different classes of inequalities- the coefficients of $a_J^{(i)}$ for $i=1,2,3$. This exercise, outlined in detail in subsubsection \ref{plsdscalar}, leads us to $\xi^\g_{min}=0.723$ for the photon EFT when we consider all locality constraints up to $2n+m\leq 12$ and $J_{max}\leq 20$. Using the relation $\frac{\xi_{min}^2-1}{\xi_{min}^2+3}M^2<a<M^2$ and \eqref{postruphoton}, we obtain,  

\bea \label{rangeofaph}
{\bf Photon:} -0.1355 M^2< a^\g < \frac{2 M^2}{3}\,.\nonumber \\ \eea

\noindent Similar to the scalar case, for the photon also we discover the phenomenon of {\it Low Spin Dominance} (LSD). Consider the set of equations \eqref{ineqphoton} without the locality constraints but with a maximal spin cut-off $J=J_c$. If we assume that the absorptive part is unaffected by the contributions from partial waves after $J>J_c$, the positivity of this finite sum of partial waves leads us to an independent derivation of $\xi^\g_{min}$. It suffices to choose the largest root $\xi^\g(J) \le 1$ of the set of polynomials $\{ d^J_{0,0}, d^J_{2,2}\pm d^J_{2,-2}\}$ for a fixed $J\leq J_c$ to ensure the positivity of the corresponding term in \eqref{abs}. We observe the following table. 

\begin{center}\label{LSDphoton}
	\begin{tabular}{ |c| c|}
		\hline
		$J_c$ & Photon \\ 
		\hline
		2  & $-0.2 M^2 <a^\g< \frac{2 M^2}{3}$ \\  
		\hline
		3   &$-0.143 M^2 <a^\g< \frac{2 M^2}{3}$ \\
		\hline
		4 & $-0.069 M^2 <a^\g< \frac{2 M^2}{3}$\\ 
		\hline
	\end{tabular}
\end{center}
We can see that the argument with Locality constraints combined with the above clearly indicates spin-$3$ dominance for the photon case. Therefore for this range of $a$, we can impose the Bieberbach Rogosinski bounds of subsection \ref{BRbounds} on the Wilson coefficients $\mw^{(x_1)}_{n-m,m}$ -these constraints are called $TR^\g_U$. In the next section we present the bounds obtained from $PB^\g_C$, $TR^\g_U$ and the corresponding range of $a$ \eqref{rangeofaph}.


\subsection{Bounds}
We now apply our formalism to bound Wilson coefficients in the Euler-Heisenberg type EFT for the photon. Recall that the low energy EFT expansion has the form, 
\bea
\mathcal{L}= \frac{-1}{4} F_{\mu \nu} F^{\mu \nu}+a_1 \left(F_{\mu \nu}F^{\mu \nu} \right)^2+a_2 (F_{\mu \nu}\tilde{F}^{\mu \nu})^2+\cdots
\eea
For such an EFT we have the following crossing symmetric S-matrices (see appendix \ref{eecbs}) , 
\begin{align}
F_1(s_1,s_2,s_3)&= 2 f_2 x-f_3 y+4 f_4 x^2-2f_5 xy +f_{6,1} y^2+8f_{6,2}x^3+\cdots \,\nonumber\\
F_2(s_1,s_2,s_3)&= 2 g_2 x -3 g_3 y +2 (g_{41}+2 g_{42})x^2+(-5 g_{5,1}-3 g_{5,2}) xy \\
&\hspace{4 cm} + 3 \left(g_{6,1}-g_{6,2}+g_{6,3}\right) y^2 + 2 g_{6,1}x^3 + \cdots  \,, \nonumber\\
\end{align}
where the Wilson coefficients can be related to the EFT couplings such as $a_1 =\frac{f_2+g_2}{16},~a_2= \frac{f_2-g_2}{16}$ etc. We begin by listing out the $PB^\g_c$ and $TR^\g_U$ conditions for $n=3$ (see \eqref{tru},  \eqref{pbcphot} and \eqref{rangeofaph}). The  $PB^\g_c$ conditions are, 
\bea \label{PBCphn3}
&&\frac{9 w^{(x_1)}_{20}}{4 M^4}+\frac{3 w^{(x_1)}_{11}}{2 M^2}+w^{(x_1)}_{02}\geq 0,~ \frac{5 w^{(x_1)}_{20}}{2 M^2}+w^{(x_1)}_{11}\geq 0,~ 0\leq w^{(x_1)}_{20}\leq \frac{1}{M^4},\nonumber\\
&&\nonumber\\
&&8 w^{(x_1)}_{03}+3 (4 w^{(x_1)}_{12}+6 w^{(x_1)}_{21}+9 w^{(x_1)}_{30})\geq 0,~4 w^{(x_1)}_{12}+14 w^{(x_1)}_{21}+33 w^{(x_1)}_{30}\geq 0, \nonumber\\
&&\nonumber\\
&&2 w^{(x_1)}_{21}+7 w^{(x_1)}_{30}\geq 0,\,0\leq w^{(x_1)}_{30}\leq w^{(x_1)}_{20}
\eea
while the $TR^\g_U$ conditions are,

\bea 
&&-2\leq\frac{a (2 w^{(x_1)}_{01}-27 a (a (a w^{(x_1)}_{02}+w^{(x_1)}_{11})+w^{(x_1)}_{20}))+2 w^{(x_1)}_{10}}{a w^{(x_1)}_{01}+1}\leq 2,\nonumber\\
&&-1\leq \frac{3 (a(9 a (a (a (27 a (a (w^{(x_1)}_{03} a+w^{(x_1)}_{12})+w^{(x_1)}_{21})-4 w^{(x_1)}_{02}+27 w^{(x_1)}_{30}-4 w^{(x_1)}_{11})-4 w^{(x_1)}_{20})+w^{(x_1)}_{01})+1)}{a w^{(x_1)}_{01}+1} \leq 3, \nonumber\\
\eea 
where as before we have used the notation $\frac{\mw^{(x_1)}_{p,q}}{\mw^{(x_1)}_{1,0}}=w^{(x_1)}_{pq}$ and the range of $a$ has been specified in \eqref{rangeofaph}. The coefficients $w^{(x_1)}_{ij}$ are related to the EFT expansion as follows , 
\bea 
w^{(x_1)}_{01}&=&\frac{-3 g_3 -x_1 f_3}{2g_2+2 x_1f_2},\qquad~~~~~~~~~~~~ w^{(x_1)}_{02}=\frac{3(g_{6,1}-g_{6,2}+g_{6,3}) +x_1 f_{6,1}}{2g_2+ 2x_1f_2} \nonumber\\
w^{(x_1)}_{20}&=& \frac{2(g_{4,1}+2g_{4,2})+x_1 4f_4}{2g_2+ 2x_1f_2},\qquad w^{(x_1)}_{11}=\frac{(-5 g_{5,1}-3 g_{5,2})-x_1 2f_5}{2g_2+ 2x_1f_2} \,,
\eea 
where, $\mw^{(x_1)}_{1,0}=2g_2+ 2x_1f_2$. Before solving these constraints and getting bounds, we want to point some salient features of our inequalities. The positivity of $\mathcal{W}_{1,0}$ \eqref{wn0} gives us:
\bea
g_2+x_1 f_2 \ge 0 \, , 
\eea
In particular this translates to $g_2 \pm f_2 \ge 0$ in other words $a_1, a_2 \ge 0$. After expanding $F_2+ x_1 F_1$ in $\tilde{z},a$ we can use relation \eqref{w01} which translates to the following:
\bea \label{w01p}
-27 a^2 (-2 g_2 -2 x_1 f_2+ 3 a g_3 + a x_1 f_3 ) <0 \,. \nonumber
\eea
Firstly we note that if $f_2=\pm g_2$ then, since the above relation has to hold for all $x_1 \in [-1,1]$ and all $-\frac{5 M^2}{37}<a< \frac{2 M^2}{3}$, we get $f_3 =\pm 3 g_3$, the reasoning is as follows. Suppose $f_2=\pm g_2$ then by looking at $x_1=\mp 1$ we get
\bea
a(3 g_3 \mp f_3) >0, ~ \forall -\frac{5 M^2}{37}<a< \frac{2 M^2}{3} \,. \nonumber
\eea
which gives us the result. In particular for the $f_2=g_2$ case we note that if we truncate to 6-derivatives there is no difference between the massless scalar case and this one since $F_1=F_2$. This gives us 
\be
\frac{-3.44}{M^2}< \frac{f_3}{f_2}<\frac{1}{M^2}
\ee
Secondly if $g_2 +x_1 f_2 \neq 0 $ then from \eqref{w01p} we have 
\be
\frac{-4.902}{M^2}< \frac{g_3+ x_1 \frac{f_3}{3}}{g_2 + x_1 f_2}<\frac{1}{M^2}
\ee

These are in good agreement with the results of \cite{vichi}. We can in fact use the above relations to get region plots as shown in the figure below. We have benchmarked where different theories lie in this allowed space of EFT's.
\begin{figure}[H]
	\begin{centering}
		\includegraphics[width=0.9\textwidth]{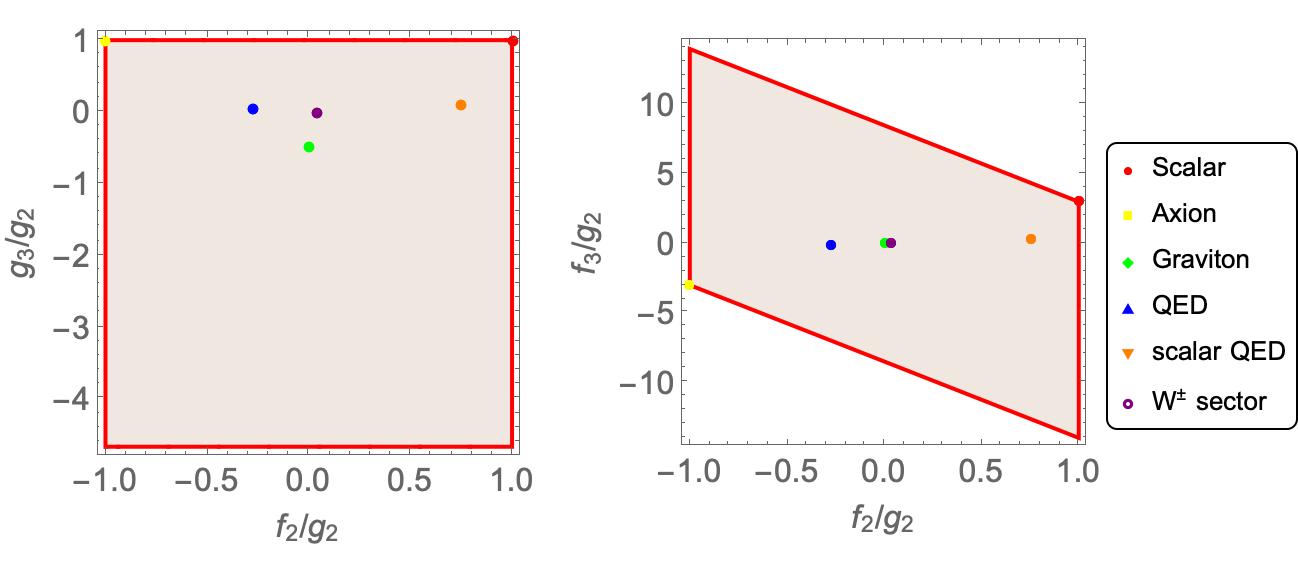}
		\vspace*{-30 pt}
		\caption{The allowed regions in the $\left(\frac{g_3}{g_2},\frac{f_3}{f_2}\right)$ vs $\frac{f_2}{g_2}$ space with scalar,axion,graviton,QED,scalar QED, $W^{\pm}$ sector benchmarked.}
	\end{centering}
\end{figure}
\noindent Furthermore, whenever we have $f_2=k g_2$ with $k \in [0,1]$ we can see the space of allowed theories as in this case by choosing a suitable $x_1$ one can make $g_2+x_1 f_2=0$. The plot is shown below. 
\begin{figure}[h!]
	\begin{centering}
		\includegraphics[width=0.5\textwidth]{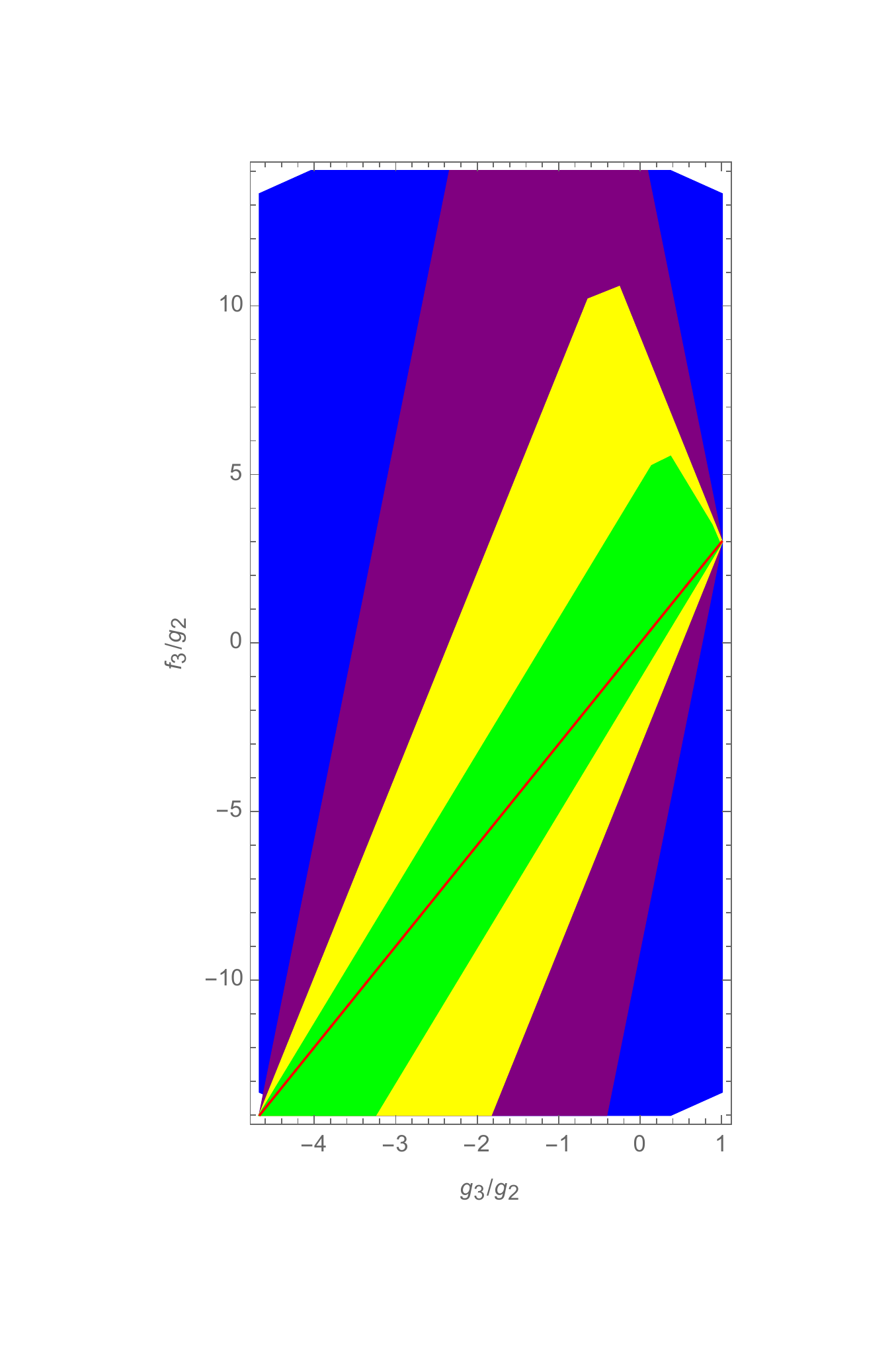}
		\vspace*{-20 pt}
		\caption{The space of allowed regions in the $\left(\frac{f_3}{f_2},\frac{g_3}{f_2}\right)$for $f_2=k g_2$ with $k=0,0.25.0.5,1$ corresponding to blue, purple,yellow,green and red respectively.}
	\end{centering}
\end{figure}
By working out the $n=3$ $PB^\g_C$ and $TR^\g_U$ conditions explicitly we obtain the following values for $w^{(x_1)}_{pq}$ listed in table \ref{mintph}. A comparative plot for the the first few higher derivative coefficients is given in figure \ref{w11w02}. As before in the Wilson coefficients $w^{x_1}_{pq}$, the region $f_2 = \pm g_2$ is special and must be treated with caution. From \eqref{PBCphn3}, it immediately follows that consistency of the equations for all $x_1 \in [-1,1]$, enforces the relations of the form $f_i = k \sum_j g_{i,j}$ for $i,J>2$ whenever $f_2=\pm k g_2$.     

\begin{table}[h!]
	\centering
	$\begin{array}{|c|c|c|}
		\hline
		w^{(x_1)}_{p,q}=\frac{\mathcal{W}^{(x_1)}_{p,q}}{\mathcal{W}^{(x_1)}_{1,0}} & TR_U^{min}  & TR_U^{max}\\
		\hline
		w^{(x_1)}_{01} &-1.5 & 7.353\\
		\hline
		w^{(x_1)}_{20} &0 & 1\\
		\hline
		w^{(x_1)}_{02} & -11.029 & 4.368\\
		\hline
		w^{(x_1)}_{11} &-2.5 & 6.353\\
		\hline
		w^{(x_1)}_{03} &-18.479 & 64.601\\
		\hline
		w^{(x_1)}_{12} &-84.255 & 15.980\\
		\hline
		w^{(x_1)}_{21} &-3.5 & 28.1121\\
		\hline
		w^{(x_1)}_{30} &0 & 1\\
		\hline
	\end{array} $
	\caption{A list of bounds obtained using our results $TR_U$ up to $n=3$ in the normalisation $M^2=1$. }\label{mintph}
\end{table}

\begin{figure}[t]
	\centering
\begin{subfigure} [b]{0.40\textwidth}
	\includegraphics[width=\textwidth]{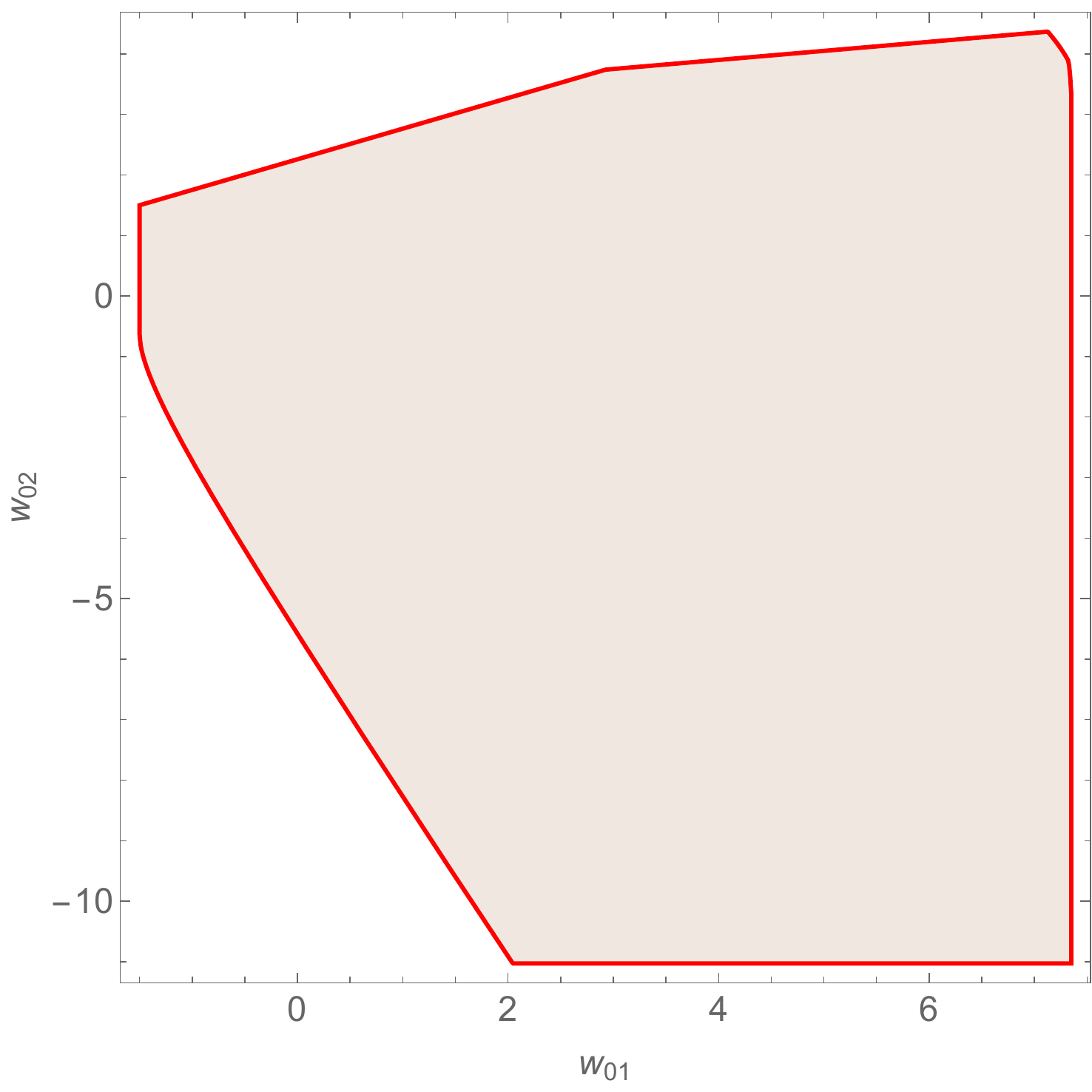}
	\caption{Allowed regions in $(w_{02}, w_{01})$ space.}
\end{subfigure}
\hfill
\begin{subfigure}[b]{0.40\textwidth}
\includegraphics[width=\textwidth]{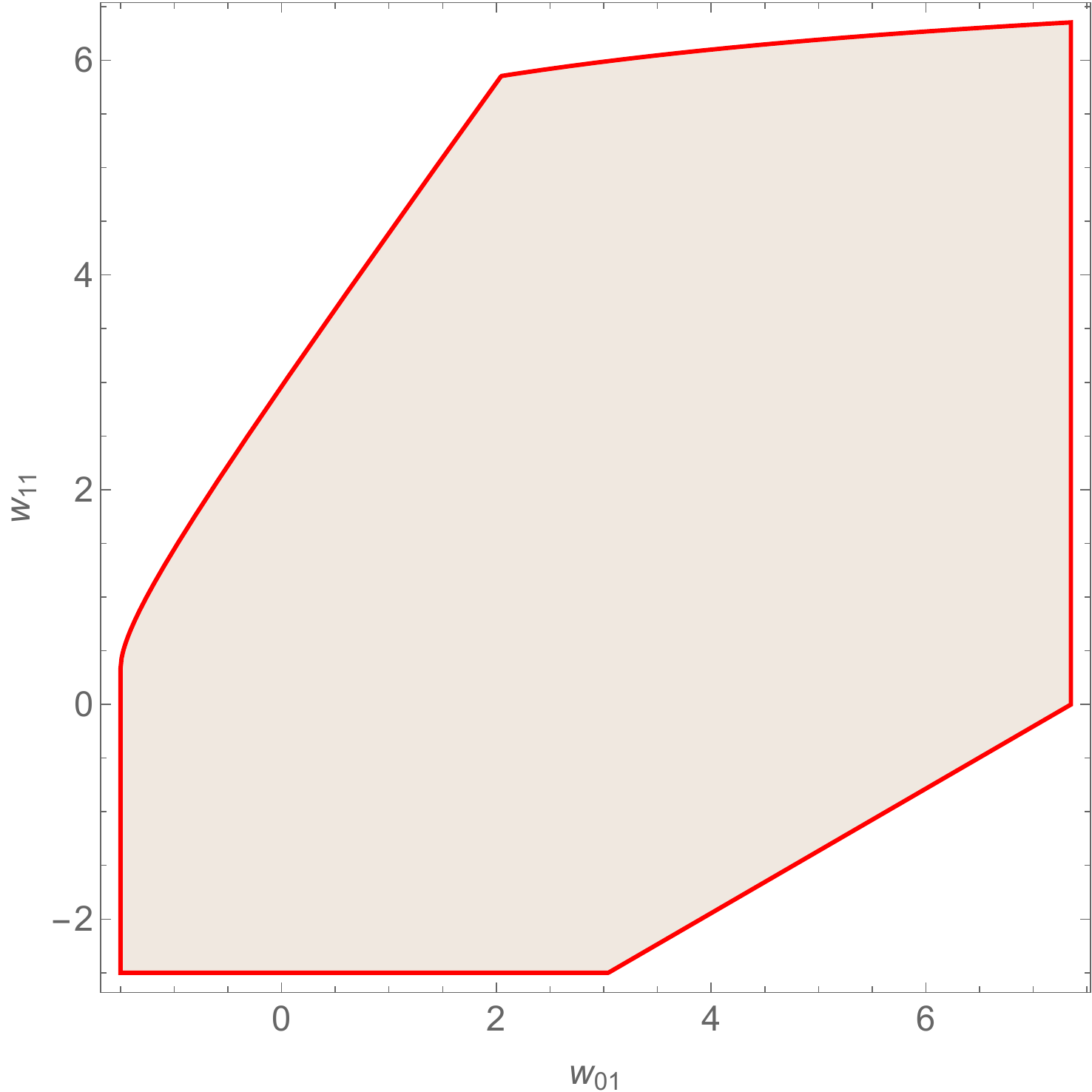}
\caption{ Allowed regions in $(w_{11}, w_{01})$ space.}
\end{subfigure}
\caption{Plots for $(w_{11}$ and $w_{02})$ vs $w_{01}$ for $TR^\g_U$ upto $n=3$ .}%
\label{w11w02}
\end{figure}



\section{Graviton bounds}\label{gravitonbound}
In this section we will be considering parity preserving graviton amplitudes. We would like to consider the same combination of amplitudes as in the photon case since unitarity guarantees the positivity of these combinations. However, the low energy expansion of  $F^{h}_2(s_1,s_2,s_3) =\sum_{i=1,3,4}T_i(s_1,s_2,s_3)$ starts only at 8-derivatives (the first regular term is the one from $R^4$) which translates to the low energy expansion in $\tilde {z}$ starting from $\tilde {z}^2$ order. Such a function cannot be typically real as can be seen using the following simple argument. Suppose we have a typically-real function $f(z)$ which has a Taylor expansion $f(z) =z^2+ a_3 z^3+\cdots$ around the origin. In a small  neighbourhood of $z=0$  the leading term is the dominant one and we have $\Im f(z) \Im z >0 \implies r^2 \sin 2\theta \sin \theta >0$ for all $z= r e^{I \theta}$ and $\theta \in (0,\pi)\cup (\pi,2 \pi)$, this however is not possible as $\sin 2 \theta$ changes sign in the upper/lower half plane but $\sin \theta$ does not. Thus our hypothesis that $f(z)$ is typically real is incorrect. 

Thus our methods will not directly apply to these combinations. For our purposes we will considering the modified combination:
\begin{eqnarray}\label{gravamp}
	M^{h}(s_1,s_2, s_3)&=& \tilde F^{h}_2(s_1,s_2,s_3) \nonumber\\
	&=& \left(\frac{T^{h}_1(s_1,s_2,s_3)}{s_1^2}+ \frac{T^{h}_3(s_1,s_2,s_3)}{s_3^2}+ \frac{T^{h}_4(s_1,s_2,s_3)}{s_2^2}\right)
\end{eqnarray}
As can be readily checked the above combination $\tilde F^{h }_2(s_1,s_2,s_3)$ does not have any additional low energy spurious poles, is fully crossing symmetric and obeys the same $o(s^2)$ Regge growth we demand for $F^{h }_i(s_1,s_2,s_3)$. Thus $\tilde F^{h }_2(s_1,s_2,s_3)$ also satisfies the  crossing symmetric dispersion \eqref{crossdisp} \footnote{We have verified that this is indeed true for all the 4-graviton string amplitudes for details see \ref{appB}.}.

\noindent Furthermore it has the low energy expansion given by
\begin{equation}
	\begin{split}
		& \tilde{F}^{h}_2(s_1,s_2,s_3)=\mathcal{W}^{(f)}_{1,0} x+\mathcal{W}^{(f)}_{0,1} y+\mathcal{W}^{(f)}_{1,1} x y +\mathcal{W}^{(f)}_{2,0} x^2+\cdots \,.
	\end{split}
\end{equation}
We can also consider  $F^{h}_1(s_1,s_2,s_3)$ which\footnote{As for the photon case we could have considered $\tilde F^{h}_2(s_1,s_2,s_3)+ x_1 F^{h }_1(s_1,s_2,s_3)$ for $x_1 \in [-1,1]$however this leads to a spectral coefficient $ \frac{\rho_1}{s1'^2}+\rho_2$ which doesn't seem to have a fixed sign from unitarity alone $\rho_1 \ge 0, \rho_1\pm \rho_2 \ge 0$. We shall use a different method to bound $F^{h }_1(s_1,s_2,s_3)$. } has an expansion 
\begin{equation}\label{F1grav}
	F^{h}_1(s_1,s_2,s_3)=\mathcal{W}^{(g)}_{0,1} y +\mathcal{W}^{(g)}_{1,1} x y +\mathcal{W}^{(g)}_{2,0} x^2+\cdots \,.
\end{equation}
We shall not explore this case in the current work. When we write the above expansions we have a low-energy gravitational EFT in mind 
\be
\mathcal{L}= \frac{-2}{\kappa^2} \sqrt{-g} R+8 \frac{\beta_{R^3}}{\kappa^3} R^3+ 2 \frac{\beta_{R^4}}{\kappa^4} C^2 +  \frac{2 \tilde{\beta}_{R^4}}{\kappa^4} \tilde{C}^2 +\cdots \,,
\ee
where $R$ is the Ricci scalar, $\kappa^2=32 \pi G$ and $C=R^{\mu \nu \kappa \lambda} R_{\mu \nu\kappa\lambda}$, $\tilde{C}=\frac{1}{2}R^{\mu \nu \alpha \beta} \epsilon_{\alpha \beta}^{\gamma \delta} R_{\gamma \delta \mu \nu}$ and the metric $g_{\mu \nu}= \eta_{\mu \nu}+h_{\mu \nu}$ is given in-terms of the gravitational field $h_{\mu \nu}$. We subtract out the poles corresponding to the $R$ and $R^3$ terms and look at the low energy expansion of the rest of the amplitude. The Wilson coefficients of the low-energy expansion of the amplitudes are related to the parameters in the gravitational EFT Lagrangian such as
\bea
\mathcal{W}^{(f)}_{1,0}= \frac{\beta_{R^4}+\tilde{\beta}_{R^4}}{{\kappa^4}} \,.
\eea

\subsection{Wilson coefficients and Locality constraints: $PB_C^h$}\label{gravitonboundp}
The local low energy expansion of the amplitude \eqref{gravamp} can be written as 
\bea 
\tilde F^{h}_2(s_1,s_2) &=& \sum_{p,q=0}^{\infty} \mathcal{W}^{(f)}_{p,q} x^p y^q \qquad= \sum_{p,q=0}^{\infty} \mathcal{W}^{(f)}_{p,q} x^{p+q} a^q \eea
where we have used $a=y/x$ and $\mathcal{W}^{(h)}_{0,0}=0$. We can solve for the $\mathcal{W}^{(h)}_{p,q}$ by expanding around $a=0$ and comparing powers of $x,a$. We obtain, 
\bea \label{gravitonbell}
\mathcal{W}^{(f)}_{n-m,m}&=&\int_{M^2}^{\infty}\frac{d s_1}{2\pi s_1^{2n+m+1}} \sum_{J=0,2,4,\cdots}(2J+1)\tilde a^{(1)}_J(s_1)\mathcal{K}_{n,m}^{J,1}\,\nonumber\\
&+&\int_{M^2}^{\infty}\frac{d s_1}{2\pi s_1^{2n+m+1}} \sum_{J=4,6,\cdots}(2J+1)\tilde a^{(2)}_J(s_1)\hat{\mathcal{K}}_{n,m}^{J,2}\,\nonumber\\
&+& \int_{M^2}^{\infty}\frac{d s_1}{2\pi s_1^{2n+m+1}} \sum_{J=5,7,\cdots}(2J+1)\tilde a^{(3)}_J(s_1)\hat{\mathcal{K}}_{n,m}^{J,3}\,,\nonumber\\
\hat{\mathcal{K}}_{n,m}^{J,i}&=&2\sum_{j=0}^{m}\frac{(-1)^{1-j+m}q_{J}^{(j,i)}\left(1\right) \left(4 \right)^{j}(3 j-m-2n)\Gamma(n-j)}{ j!(m-j)!\Gamma(n-m+1)}\,.
\eea 
where $\tilde a^{(1)}_J =\frac{\rho^{1, h}_J}{s_1^2}$, $\tilde a^{(2)}_J=\tilde a^{(3)}_J=\frac{\rho^{3, h}_J}{s_1^2}$ and $q_{J}^{(j,i)}(1)= \frac{\partial^j f^{(i)}(\sqrt{\xi})}{\partial \xi^j}\bigg|_{\xi=\xi_0=1}$ with $f^{(1)}=d^J_{0,0},f^{(2)}=\frac{d^J_{4,4}\left(\cos^{-1}\left(\sqrt \xi \right)\right)}{(1+\sqrt \xi)^2} + \frac{d^J_{4,-4}\left(\cos^{-1}\left(\sqrt \xi\right) \right)}{(1-\sqrt \xi)^2},~f^{(3)}=\frac{d^J_{4,4}\left(\cos^{-1}\left(\sqrt \xi \right)\right)}{(1+\sqrt \xi)^2} - \frac{d^J_{4,-4}\left(\cos^{-1}\left(\sqrt \xi\right) \right)}{(1-\sqrt \xi)^2}$.\\

\noindent \emph {A key difference between the scalar/photon case and the graviton case we are considering now is that the combinations $f^{(i)}$ are no longer positive even for $\xi >1$}. 

However for $\xi=1$ we can check that $f^{(i)}=1$ and since the spectral functions $\tilde a_J^{(i)} \ge 0$ by unitarity namely $a_J^{(i)} \ge 0$ so this in particular implies 
\bea \label{whn0}
\mathcal{W}^{(h)}_{n,0}&=& \int_{M^2}^{\infty}\frac{d s_1}{2\pi s_1^{2n+m+1}} \sum_{J=0,2,4,\cdots}(2J+1)\tilde a^{(1)}_J(s_1)\,\nonumber\\
&+&\int_{M^2}^{\infty}\frac{d s_1}{2\pi s_1^{2n+m+1}} \sum_{J=4,6,\cdots}(2J+1)\tilde a^{(2)}_J(s_1) \,\nonumber\\
&+& \int_{M^2}^{\infty}\frac{d s_1}{2\pi s_1^{2n+m+1}} \sum_{J=5,7,\cdots}(2J+1) \tilde a^{(3)}_J(s_1)  \ge 0 \, .
\eea
We  can see straightforwardly that the above implies $$ 0\leq \mathcal{W}^{(f)}_{n,0}\leq \frac{1}{M^4} \mathcal{W}^{(f)}_{n-1,0}\,.$$
As alluded to before, the non-positivity of $f^{(i)}$ in \eqref{gravitonbell} implies the sign of any term in $J$ expansion is no longer controlled by $\mathcal{K}_{n,m}^{J,i}(s_1)$ alone. So this makes obtaining a closed form for $PB_C^h$ much harder in this case. We can however do this case by case. For $n=2$ these read:
\bea
\frac{9 w^{(f)}_{20}}{4 M^4}+\frac{3 w^{(f)}_{11}}{2 M^2}+w^{(f)}_{02}\geq 0,~ \frac{5 w^{(f)}_{20}}{2 M^2}+w^{(f)}_{11}\geq 0,~ 0\leq w^{(f)}_{20}\leq \frac{1}{M^4},\nonumber\\
\eea
where $w^{(f)}_{p,q}=\frac{\mathcal{W}^{(f)}_{p,q}}{\mathcal{W}^{(f)}_{1,0}}$.
As before the locality constraints for this case are therefore given by 

\bea\label{nullgraviton}
\mw^{(f)}_{n-m,m}=0~~~ \forall n<m \,.
\eea

\subsection{Typically-Realness and Low spin dominance: $TR_U^h$}\label{gravitonboundTR}
In this section we try to get $a$ range of a using positivity of the amplitude coupled with locality constraints and typically-realness of the amplitude. 
The analysis in this case has key differences due to the non-positivity of the $f^{(i)}\left(\sqrt{\xi}\right)$ even for $\xi>1$. We know the typically-realness of the amplitude followed from two crucial ingredients namely the regularity of the kernel inside the unit disk and the positivity of the absorptive part. The former remains unchanged the latter however crucially needs the locality constraints to justify now, since $\xi>1$ is no longer sufficient to guarantee 
positivity.
\begin{eqnarray}\label{postrugraviton}
	\left(a\in \left[-\frac{M^2}{3},0\right) \cup \left(0, \frac{2M^2}{3}\right]\right)~~ \cap ~~ \left(a ~~\forall~~ A(s_1,; s_2^{(+)}(s_1,a) \geq 0\right)_{LSD}. \nonumber\\
\end{eqnarray}

We can proceed with the LSD analysis as before by including the locality constraints $N^h_c =- \sum_{n<m\atop m\ge 2} c_{n,m} \mw^{(f)}_{n-m,m} a^{2n+m-3} y $ to it with arbitrary weights $c_{n,m}$'s \eqref{nullgraviton}. 
\bea\label{ineqgraviton}
& &\mathcal{M}(s_i,a)+N_c = \nonumber \\
& &\int_{M^2}^{\infty}\frac{d s_1}{2\pi s_1} \sum_{J=0,2,4,\cdots}(2J+1)\tilde a^{(1)}_J(s_1)\left[f^{(1)}_J(\xi)- \sum_{{n<m}\atop{m\ge 2}} c_{n,m} \mathcal{\hat{K}}_{n,m}^{J,1} \frac{ a^{2n+m}  H(a;s_i)}{(s_1^{'})^{2n+m} H(s_1',s_i)}\right]H(s_1',s_i) \,\nonumber\\
&+&\int_{M^2}^{\infty}\frac{d s_1}{2\pi s_1} \sum_{J=4,6,\cdots}(2J+1)\tilde a^{(2)}_J(s_1)\left[f^{(2)}_J(\xi)- \sum_{{n<m}\atop{m\ge 2}} c_{n,m} \mathcal{\hat{K}}_{n,m}^{J,2} \frac{ a^{2n+m}  H(a;s_i)}{(s_1^{'})^{2n+m} H(s_1',s_i)}\right]H(s_1',s_i)\,\nonumber\\
&+& \int_{M^2}^{\infty}\frac{d s_1}{2\pi s_1} \sum_{J=5,7,\cdots}(2J+1)\tilde a^{(3)}_J(s_1)\left[f^{(3)}_J(\xi)- \sum_{{n<m}\atop{m\ge 2}} c_{n,m} \mathcal{\hat{K}}_{n,m}^{J,3} \frac{ a^{2n+m}  H(a;s_i)}{(s_1^{'})^{2n+m} H(s_1',s_i)}\right]H(s_1',s_i)\,\geq 0,\nonumber\\
\eea
where $\mathcal{\hat{K}}_{n,m}^{J,i}$ has been defined in \eqref{gravitonbell}. As before $\xi_{min}$ is determined by the maximum lower bound obtained by considering the positivity of three different classes of inequalities namely corresponding to the coefficients of $\tilde a_J^{(i)}$ for $i=1,2,3$. We can also determine $\xi_{max}$ now which is determined by the minimum upper bound obtained by considering the positivity of the same three classes of inequalities. This exercise leads us to $\xi^h_{min}=0.593$ and $\xi^h_{max}=3$ for the graviton EFT when we consider all locality constraints up to $2n+m\leq 12$ and $J_{max}\leq 20$. Using the relation $\frac{\xi_{min}^2-1}{\xi_{min}^2+3}M^2<a<\frac{\xi_{max}^2-1}{\xi_{max}^2+3}M^2$ and \eqref{postrugraviton}, we obtain,  

\bea \label{rangeofagr}
{\bf Graviton:} -0.1933 M^2< a^h < \frac{2 M^2}{3}\, ,\nonumber \\ \eea

We would now like to show that this is indicative of Spin-2 dominance for the graviton case.  Since in the set of polynomials $\{ d^J_{0,0}, \frac{d^J_{4,4}(cos^{-1} x)}{(1+x)^2}\pm \frac{d^J_{4,-4}(cos^{-1} x)}{(1-x)^2}\}$ the latter two elements do not have straightforward positivity properties for general $J$. The identification of the critical spin $J_c$ is more complicated and needs more detailed consideration in this case. A key difference between the scalar/photon cases and the graviton case we are looking at now is that both the upper and lower bound of $a$ can change. Let us recall how that happens- the condition $TR_U$ tells us that the allowed range of $\xi \in [0,3]$. The overlap of this region with the positive part of the absorptive part gave us the required range of $\xi$ to be used in our analysis.
In the analogous exercise for the photons and scalars, we had truncated the partial wave sum to a finite cut-off in spin and so then the range of $\xi$ was determined by what range for which these polynomials were positive. We had determined the lower range of $\xi$ to be given by the largest root of the Wigner-$d$ combinations that appear with respective spectral coefficients- this was usually such that $\xi_{min}<1$. The upper range of $\xi$ was automatically determined by the $TR_U$ conditions since the relevant Wigner-$d$ matrices were manifestly positive for $\xi>1$--- in other words there were no restrictions on the upper limit of $\xi$ from the Wigner-$d$ polynomials. The story for gravitons remains the same for the lower bound for $\xi$ but we note the following changes for the upper bound.

To illustrate this point, notice  that the Wigner-$d$ combination $f^{(3)}=\frac{d^J_{4,4}\left(\cos^{-1}\left(\sqrt \xi \right)\right)}{(1+\sqrt \xi)^2} - \frac{d^J_{4,-4}\left(\cos^{-1}\left(\sqrt \xi\right) \right)}{(1-\sqrt \xi)^2}$ is not always positive for $\xi>1$ for $J\geq9$. Therefore if we assume $J_c=9$, the upper limit for $\xi$ (and hence $a$) also changes along with the lower limit. As an example we present shortening of the positive regions for $f^{(3)}$ for $J=9, 11$ in figure \ref{f3j911}. We present the allowed range of $a$ as a function of $J_c$ in the form of a table below. 


\begin{center}\label{LSDgraviton}
	\begin{tabular}{ |c| c|}
		\hline
		$J_c$ & Graviton \\ 
		\hline
		2  & $-0.2 M^2  <a^h< \frac{2 M^2}{3}$ \\  
		\hline
		4   &$-0.069M^2 <a^h< \frac{2 M^2}{3}$ \\
		\hline
		6 & $-0.034M^2 <a^h< \frac{2 M^2}{3}$\\ 
		\hline
		9 & $-0.014M^2 <a^h < \frac{19641M^2}{140000}$\\ 
		\hline
	\end{tabular}
\end{center}
We can see the spin-$2$ dominance clearly for the graviton case from the above table. We have also illustrated this for the case of the type-II string amplitude in appendix(\ref{appG}).Note that the locality constraints play an important role in maintaining the positivity of the amplitude despite the Wigner-$d$ combination $f^{(3)}$ not having nice positivity properties. This is not a surprise since the locality constraints encode the details of the theory and put constraints on allowed spectral densities that appear in each sector. It would be interesting to explore the detailed implications of locality constraints in future.  


\begin{figure}[hbt!]
	\centering
	\begin{subfigure}[b]{0.46\textwidth}
		\centering
		\includegraphics[width=\textwidth]{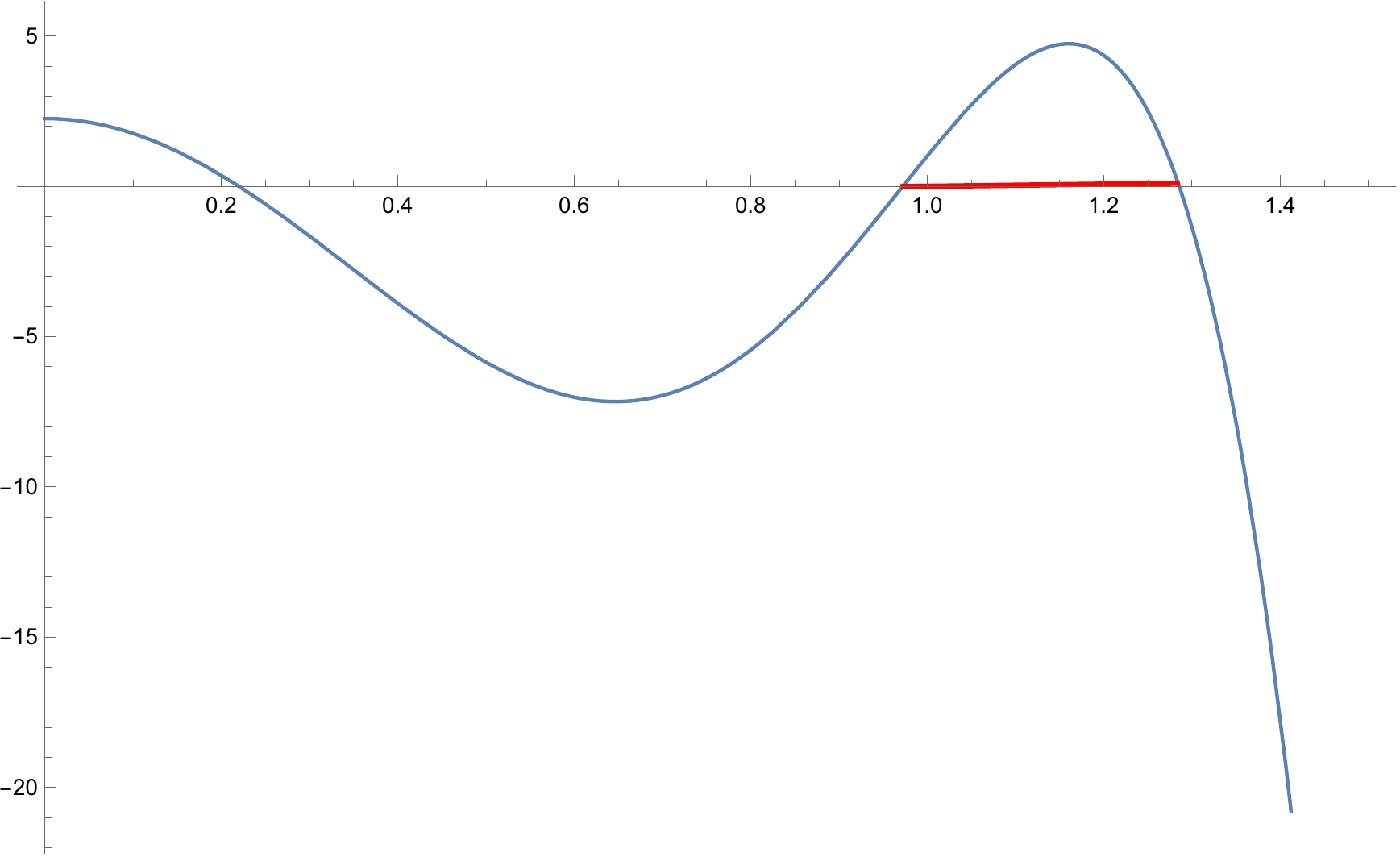}
		\caption{\centering Positive  regions for $f^{(3)}$ for $J=9$ marked in red.}
		\label{fig:szego}
	\end{subfigure}
	\hfill
	\begin{subfigure}[b]{0.46\textwidth}
		\centering
		\includegraphics[width=\textwidth]{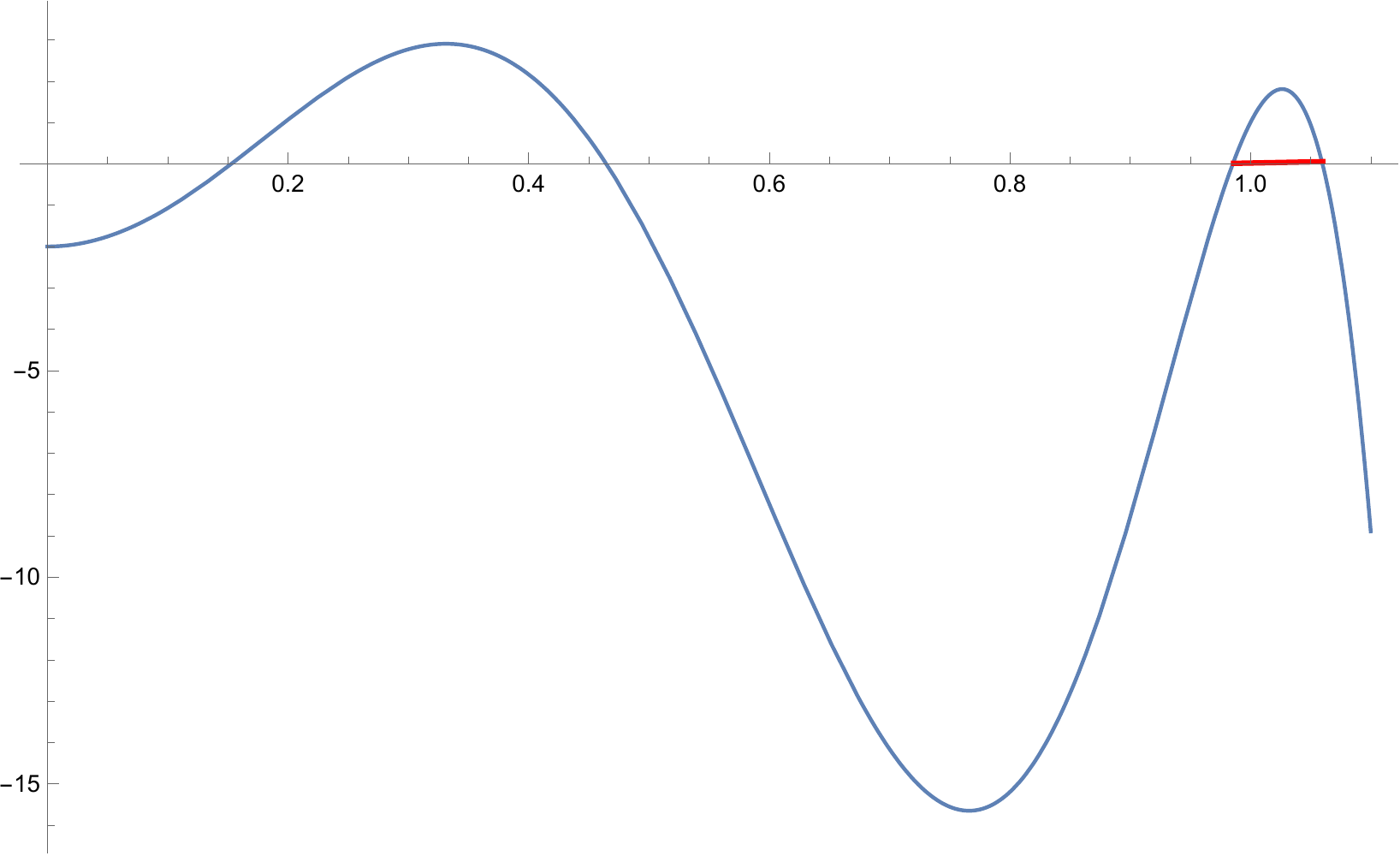}
		\caption{\centering Positive $\xi$ regions for $f^{(3)}$ for $J=11$ marked in red.}
		\label{fig:zetaminvsn}
	\end{subfigure}
	\caption{A comparative plot of changing regions of positivity in $\xi$ with spin for $f^{(3)}$.}
	\label{f3j911}
\end{figure}

We can also check this for known examples such the the 4-graviton amplitude in superstring theory.
\begin{figure}[H]
\begin{centering}
 \includegraphics[width=0.7\textwidth]{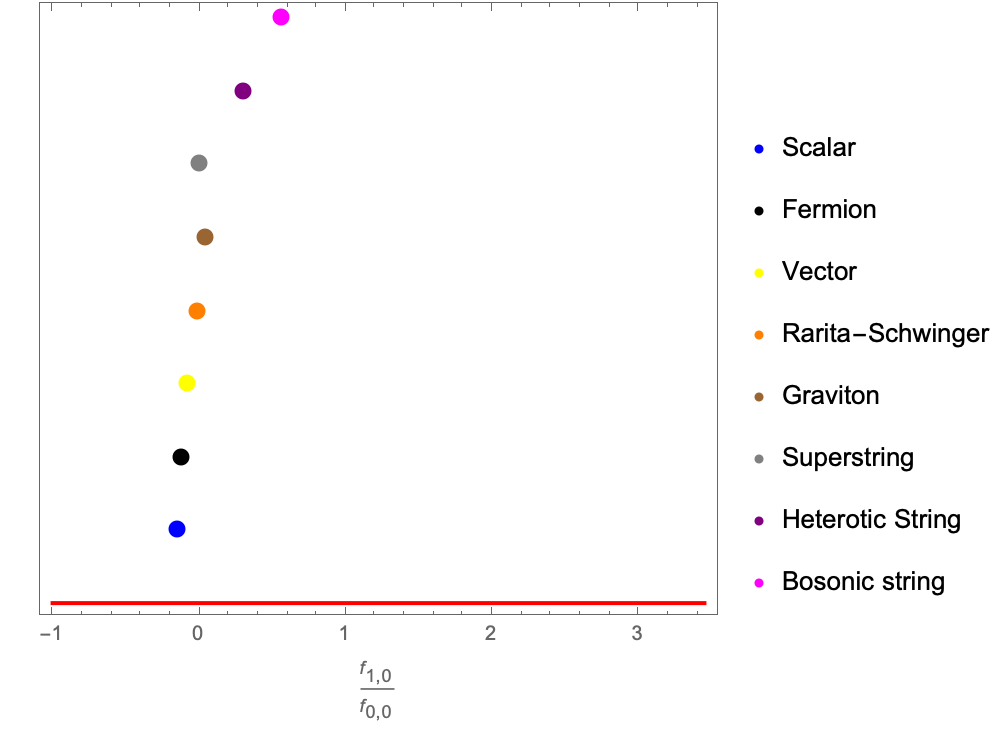}
  \vspace*{-10 pt}
 \caption{The Plot of $\rho_3^J(s_1)$ vs $J$ for $s_1=20$ seems to support the result that $J_c=2$.}
 \end{centering}
\end{figure}

Therefore for this range of $a$ in \eqref{rangeofagr}, we can impose the Bieberbach-Rogosinski bounds of subsection \ref{BRbounds} on the Wilson coefficients $\mw^{(f)}_{n-m,m}$ ---these constraints are called $TR^f_U$.

\subsection{Bounds}

In this section we put the bounds on the low energy EFT expansion which is parametrized by,
\be
\begin{split}
	\tilde{F}_2^{h}= 2 x f_{0,0}+3 y f_{1,0}+2 x^2 f_{2,0}+x y \left(2 f_{3,1}-f_{3,0}\right)+y^2 \left(-3 f_{4,0}-3 f_{4,1}+9 f_{4,2}\right)+2 x^3 f_{4,0}+\cdots\\
\end{split}
\ee 
where in terms of parametrization of \cite{sasha}, 
\be
T^h_3(s_1,s_2,s_3)= s_3^4 \left(\sum_{i=0}^\infty f_{2i,i} s_2^i s_1^i + \sum_{i=1}^\infty \sum_{j=0}^{\lfloor\frac{i}{2}\rfloor}f_{i,j}( s_2^{i-j}s_1^j+ s_1^{i-j}s_2^j) \right)\,.
\ee 
We have explicitly, 

\be
\begin{split}
	& w^{(f)}_{1,0}= 2f_{0,0},~~w^{(f)}_{0,1}=3f_{1,0},~~w^{(f)}_{2,0}=2f_{2,0},~~w^{(f)}_{1,1}= \left(2 f_{3,1}-f_{3,0}\right)\\
	& w^{(f)}_{0,2}=\left(-3 f_{4,0}-3 f_{4,1}+9 f_{4,2}\right),~~w^{(f)}_{3,0}=2f_{4,0}\,.
\end{split}
\ee 

We demonstrated in the previous subsection that due to positivity and typical realness of the amplitudes, we can put two sided bounds on Wilson coefficients. Using \eqref{w01} we have, \be
-1.5\leq w^f_{01}\leq 5.17331
\ee 
where $w^f_{01}=\frac{w^{(f)}_{0,1}}{w^{(f)}_{1,0}}$, which implies $-1\le \frac{f_{1,0}}{f_{0,0}}\le 3.44$. 

For $n=2$ $TR^h_U$ and $PB^h_C$, we obtain table \ref{mintgrav} (in units of $M^2=1$),
\begin{table}[H]
	\centering
	$\begin{array}{|c|c|c|}
		\hline
		w^f_{pq}=\frac{w^{(f)}_{p,q}}{w^{(f)}_{1,0}} & TR_U^{min}  & TR_U^{max}\\
		\hline
		w^f_{01} &-1.5 & 5.1733\\
		\hline
		w^f_{02} &-7.7600 & 3.8273\\
		\hline
		w^f_{20} & 0 & 1\\
		\hline
		w^f_{11} &-2.5 & 4.1734\\
		\hline
	\end{array} $
	\caption{A list of graviton bounds obtained using our results $TR_U$ up to $n=2$ in the normalisation $M^2=1$. }\label{mintgrav}
\end{table}
We note that terms such as $f_{2,1}$ or $f_{1,1}$ vanish when we consider fully crossing symmetric combinations as these are proportional to $s_1+s_2+s_3=0$ thus we will not be able to bound these using the current combinations\footnote{However by looking at $F_4(s_1,s_2,s_3)$ and $F_5(s_1,s_2,s_3)$ these terms do appear so we can bound them in principle. We do not attempt to do this in our current work.}we are looking at. However using our method we can bound coefficients like $\frac{f_{1,0}}{f_{0,0}}$ for which no non-trivial bounds were found using the fixed-$t$ dispersion relation, to the best of our knowledge.  
The region carved out by the Wilson coefficients with their respective data points for various theories are given. The data has been obtained from \cite{sasha}. 

\begin{figure}[H]
	\begin{centering}
		\includegraphics[width=0.9\textwidth]{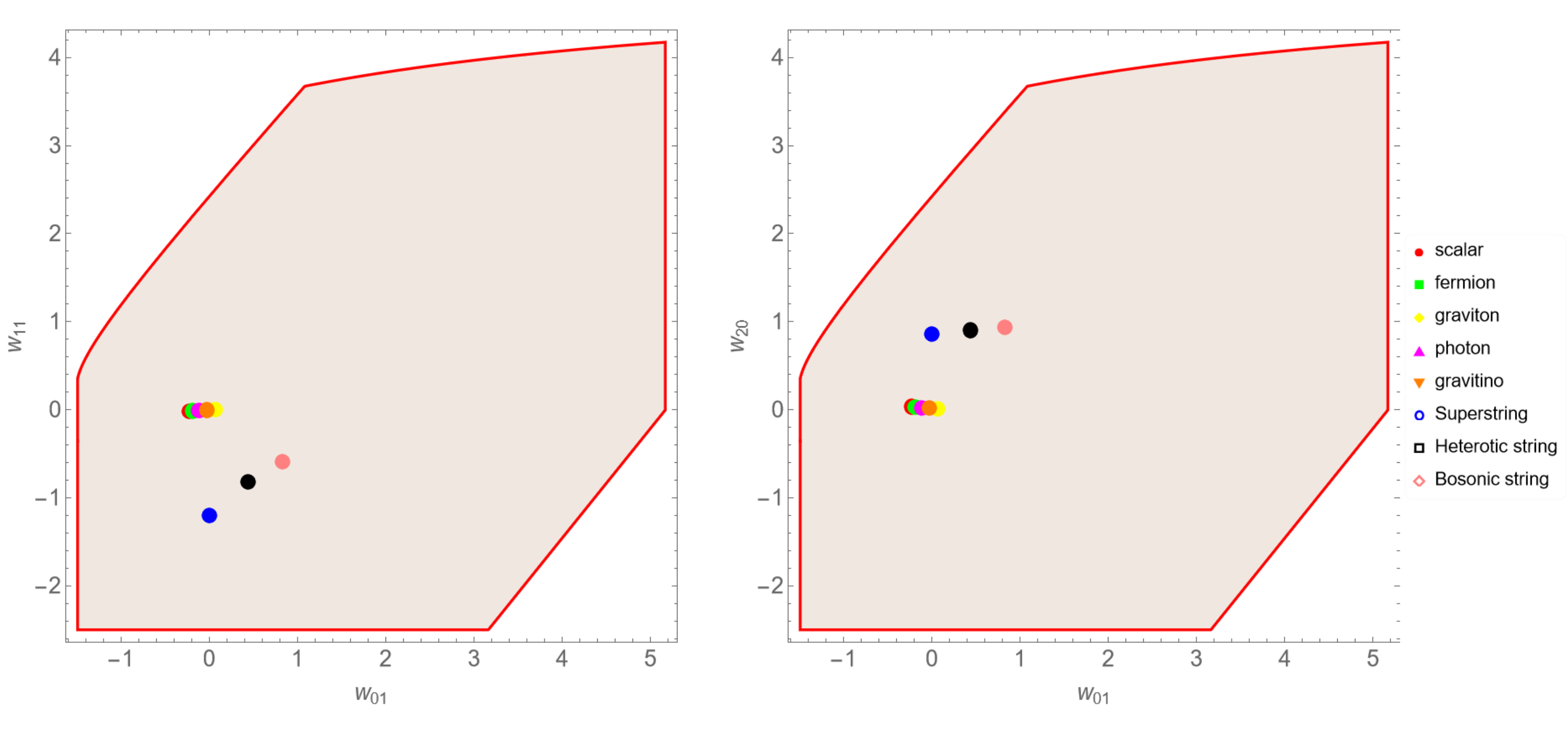}
		\vspace*{-10 pt}
		\caption{The allowed regions in the $\left(w_{11},w_{02}\right)$ vs $w_{01}$ space with scalar, fermion, photon, gravitino, graviton, superstring, Heterotic string and bosonic string sector benchmarked.}
	\end{centering}
\end{figure}

\section{Summary}\label{summary}
Let us summarize the principle findings of this chapter. 
\begin{itemize}
	\item Given a basis of amplitudes, which transform under crossing, we give a general prescription to construct crossing symmetric amplitudes relevant for CSDR from them. We demonstrated this construction explicitly for massless photons, gravitons and massive Majorana fermion helicity amplitudes in $d=4$. 
	\item Next, we used the CSDR for certain photons and graviton crossing symmetric amplitudes and put bounds on low energy Wilson coefficients. Our analysis suggested that the positivity of the absorptive part is dominated by partial waves of low lying spins---we found indications of spin-3 LSD for photons and spin-2 for gravitons (see \cite{sasha, nima2}). Using the typically-realness property of the amplitude, the Wilson coefficients satisfied the Bieberbach-Rogosinski ($BR$) bounds. We supplemented the $BR$ bounds using certain additional positivity conditions to get tighter bounds in some cases.  The photon bounds are in good agreement with existing results in literature. We dealt with the graviton amplitude separately since the low energy EFT expansion starts from eighth order  in derivatives for the crossing symmetric amplitude we consider. In order for the low energy expansion to be typically real, we considered a modified amplitude which then had the requisite properties. Similar to the photon case, we wrote down the locality constraints in closed form and analyzed certain bounds.\end{itemize}

\begin{subappendices}
	\section*{Appendix}
	\section{Representation theory of $S_3$: A crash course}\label{S3rep}
	In this appendix, we present a short self contained review of $S_3$ representations following \cite{SDC}. We can represent the three irreps of $S_3$ by the following young diagrams. 
	\begin{equation} \label{yts}
		{\bf 1_S}={\tyng(3)}\qquad\qquad {\bf 1_A}={\tyng(1,1,1)}\qquad \qquad {\bf 2_M}={\tyng(2,1)}.
	\end{equation}
	where ${\bf 1_S}$ is the one dimensional totally symmetric representation, ${\bf 1_A}$ is the one dimensional totally anti-symmetric representation and ${\bf 2_M}$ is the mixed symmetry two dimensional representation. Given an representation of $S_3$, we can easily decompose it to the irreducible sub spaces of ${\bf 1_S}$, ${\bf 1_A}$ and ${\bf 2_M}$ representations using the respective projectors. Denoting the generators for $S_3$ by $P_{12}$ and $P_{23}$ (where $P_{ij}$ denotes interchange of particles in $i$ and $j$ th position in a set $(123)$ ), the projectors for the totally symmetric and anti-symmetric subspaces are given by
	
	\begin{eqnarray}\label{projector}
		P_{{\bf 1_S}} = \frac{\left(1+P_{12}+P_{23}+P_{13}+P_{23}P_{12}+P_{12}P_{23}\right)}{6}\nonumber\\
		P_{{\bf 1_A}} = \frac{\left(1-P_{12}-P_{23}-P_{13}+P_{23}P_{12}+P_{12}P_{23}\right)}{6}
	\end{eqnarray}
	
	where $P_{13}= P_{23}P_{12}P_{23}$. The formulae \eqref{projector} make it clear that complete symmetrization  and anti symmetrization leads to projection onto the ${\bf 1_S}$ and ${\bf 1_A}$ subspace respectively, while the part that transforms in the ${\bf 2_M}$ representation is annihilated by both the symmetric and anti-symmetric projectors. The group theory for action of $S_3$ on the mandelstam invariants is given by the left action of $S_3$ on itself. The  ${\bf 6}_{\rm left}$ generated by the left action of $S_3$ onto itself can be decomposed as.  
	\begin{equation} \label{adjact}
		{\bf 6}_{\rm left}= {\bf 1_S}+ 2.{\bf 2_M} + {\bf 1_A}.
	\end{equation} 
	
	Note the appearance of two ${\bf 2_M}$ subspaces, which differ from one another in the sense that they have different $\Z_2$ charges. The explicit projectors for these two (two-dimensional) sub-spaces can be constructed as follows. The projectors for the two dimensional subspace of positive $\Z_2$ charge are  
	
	\begin{eqnarray}\label{projectorM+}
		P^{(1)}_{{\bf 2_{M+}}} &=& \frac{1+P_{23}}{2}-\frac{\left(1+P_{12}+P_{23}+P_{13}+P_{23}P_{12}+P_{12}P_{23}\right)}{6}\nonumber\\
		P^{(2)}_{{\bf 2_{M+}}} &=& \frac{P_{23}P_{12}+P_{13}}{2}-\frac{\left(1+P_{12}+P_{23}+P_{13}+P_{23}P_{12}+P_{12}P_{23}\right)}{6}\nonumber\\
	\end{eqnarray}
	
	Note that the above two projectors are respectively symmetric under the action of $\Z_2$ generator $P_{23}$ and $P_{13}$ and hence having a positive $\Z_2$ charge. We note that the projector $ P^{(1)}_{{\bf 2_{M+}}}$ projects to a subspace which is symmetric under $P_{23}$ while $P^{(2)}_{{\bf 2_{M+}}}$ projects to a subspace which is symmetric under $P_{12}$. The projectors for the two dimensional subspace for the negative $\Z_2$ charge (anti-symmetric under $P_{23}$ and $P_{13}$ respectively) are 
	\begin{eqnarray}\label{projectorM-}
		P^{(1)}_{{\bf 2_{M-}}} &=& \frac{1-P_{23}}{2}-\frac{\left(1-P_{12}-P_{23}-P_{13}+P_{23}P_{12}+P_{12}P_{23}\right)}{6}\nonumber\\ 
		P^{(2)}_{{\bf 2_{M-}}} &=& \frac{P_{23}P_{12}-P_{13}}{2}-\frac{\left(1-P_{12}-P_{23}-P_{13}+P_{23}P_{12}+P_{12}P_{23}\right)}{6}\nonumber\\ 
	\end{eqnarray} 
	To explicitly see the formalism in action, consider an arbitrary function of the Mandelstam invariants $F(s,t,u)$. The various irreducible subspaces are given by\footnote{We note that $f_{\rm Mixed +}(s,t,u)+ f_{\rm Mixed +}(t,u,s)+ f_{\rm Mixed +}(u,s,t)=f_{\rm Mixed -}(s,t,u)+ f_{\rm Mixed -}(t,u,s)+ f_{\rm Mixed -}(u,s,t)=0$ denoting that they form two dimensional subspaces.}, 
	\begin{eqnarray}\label{irrepfstu}
		f_{\rm Sym}(s,t,u) &=& \frac{1}{6}\left(F(s,t,u)+F(t,s,u)+F(s,u,t)+F(u,t,s)+F(t,u,s)+F(u,s,t)\right)\nonumber\\
		f_{\rm Anti-sym}(s,t,u) &=& \frac{1}{6}\left(F(s,t,u)-F(t,s,u)-F(s,u,t)-F(u,t,s)+F(t,u,s)+F(u,s,t)\right)\nonumber\\ 
		f_{\rm Mixed+}(s,t,u) &=& \frac{1}{6}\left(2F(s,t,u)-F(t,s,u)-F(u,s,t)+2F(s,u,t)-F(u,t,s)-F(t,u,s)\right)\nonumber\\
		f_{\rm Mixed-}(s,t,u) &=& \frac{1}{6}\left(2F(s,t,u)-F(t,s,u)-F(u,s,t)-2F(s,u,t)+F(u,t,s)+F(t,u,s)\right)\nonumber\\
	\end{eqnarray}
	
	We can easily write down examples of such functions built out of polynomials of Mandelstam invariants \cite{SDC, Chowdhury:2020ddc, Chowdhury:2021qfm}.
	\begin{eqnarray}
		f_{\rm Sym}(s,t,u) &=&(s^2+t^2+u^2)^m(s t u)^n\nonumber\\
		f_{\rm Anti-sym}(s,t,u)&=&(s^2t-t^2s-s^2u+s u^2-u^2t+t^2u) f_{\rm Sym}(s,t,u) \nonumber\\
		f_{\rm Mixed+}(s,t,u)&=&\{(2s-t-u) f_{\rm Sym}(s,t,u),~ (2s^2-t^2-u^2) f_{\rm Sym}(s,t,u)\} \nonumber\\
		f_{\rm Mixed+}(s,t,u)&=&\{(s-u) f_{\rm Sym}(s,t,u),~ (s^2-u^2) f_{\rm Sym}(s,t,u)\} \nonumber\\
	\end{eqnarray}

\section{Massless amplitudes: Examples} \label{appB}
\noindent We expect that the combinations $F_I$ also obey \eqref{crossdisp} since they satisfy all the necessary conditions. We can do some sanity checks by considering a couple of examples. In particular we look at $F_4$, $F_5$ since $F_I$ for $I=1,2,3$ the dispersion relation \eqref{crossdisp} is identical to the scalar case considered in \cite{Sinha:2020win}.\\

\noindent We first consider the Photon amplitude in superstring theory with a kinematic pre-factor being stripped off for appropriate Regge growth namely:
\bea\label{photont}
& \mt_1(s_1,s_2,s_3)=\frac{\Gamma \left[\frac{-s_2}{2}\right]~\Gamma \left[\frac{-s_3}{2}\right]}{\Gamma \left[1+\frac{s_1}{2}\right]}\,, \nonumber\\
& \mt_3(s_1,s_2,s_3)=\frac{\Gamma \left[\frac{-s_1}{2}\right]~\Gamma \left[\frac{-s_2}{2}\right]}{\Gamma \left[1+\frac{s_3}{2}\right]}\,,\nonumber\\
& \mt_4(s_1,s_2,s_3)=\frac{\Gamma \left[\frac{-s_1}{2}\right]~\Gamma \left[\frac{-s_3}{2}\right]}{\Gamma \left[1+\frac{s_2}{2}\right]}\,.
\eea\
We can construct \eqref{f4},\eqref{f5} from the above, we need to subtract out the massless poles and this is done by multipliying $F_4$ as defined in \eqref{f4} by an $s_1s_2s_3$ factor.
We can then check if \eqref{crossdisp} is satisfied by comparing the exact answer with the result obtained from \eqref{crossdisp} by computing the absorbtive part. Since \eqref{photont} has infinitely many poles at $s_1^{\prime}= k$ with $k=2,4,\cdots$, and each pole $p$ contributes a $-\pi \delta(s_1^{\prime} -p)$ factor in the absorbtive part thus \eqref{crossdisp} reduces to an infinite sum over all the poles $k$, $k \in 2 \mathbb{Z_{+}}$ which we call $G_I$. We can then compare the results by truncating this sum to some $k_{max}$ (say $k_{max}=100$) and the results are shown in first row of the plots below. 
\begin{figure}[H]
	\centering
	\includegraphics[scale=0.4]{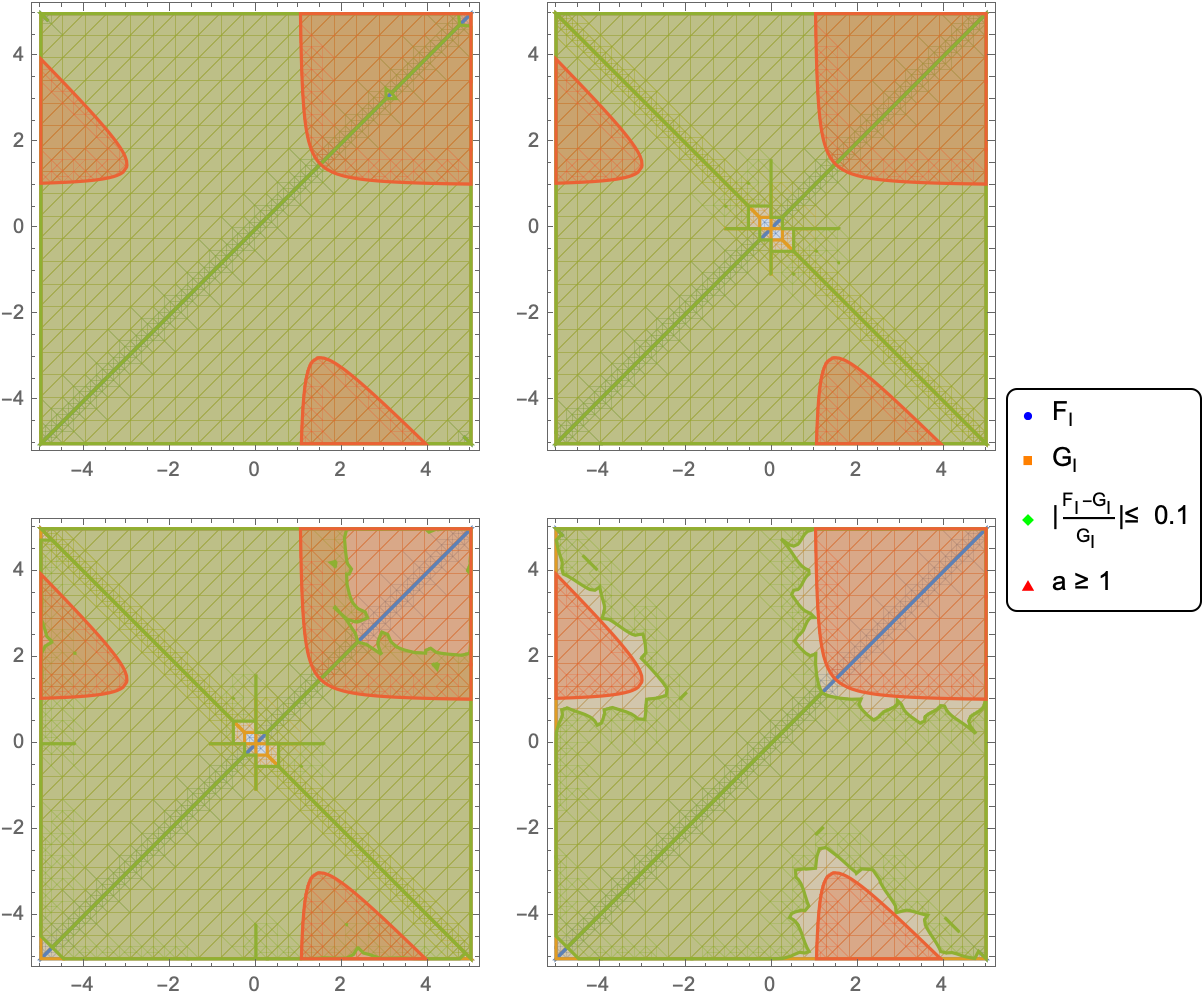}
	\caption{A comparison of the crossing symmetric dispersion for the Photon and Graviton cases which are shown in the first and second rows respectively. We have indicated the regions where $F_4,F_5$ differ from their dispersive analogues $G_4,G_5$ by less than $10 \% $ in green.}
\end{figure}

\noindent We can also consider the Graviton amplitude from superstring theory (again with appropriate kinematic pre-factors stripped off ) 
\bea \label{graviton}
& \mt_1(s_1,s_2,s_3)=\frac{\Gamma \left[-s_1\right]~\Gamma \left[-s_2\right] ~\Gamma\left[-s_3\right]}{\Gamma \left[1+s_1\right]~\Gamma \left[1+s_2\right] ~\Gamma\left[1+s_3\right]} \left(1-\frac{s_2 s_3}{s_1+1}\right)\,, \nonumber\\
& \mt_3(s_1,s_2,s_3)=\frac{\Gamma \left[-s_1\right]~\Gamma \left[-s_2\right] ~\Gamma\left[-s_3\right]}{\Gamma \left[1+s_1\right]~\Gamma \left[1+s_2\right] ~\Gamma\left[1+s_3\right]} \left(1-\frac{s_1 s_2}{s_3+1}\right)\,, \nonumber\\
& \mt_4(s_1,s_2,s_3)=\frac{\Gamma \left[-s_1\right]~\Gamma \left[-s_2\right] ~\Gamma\left[-s_3\right]}{\Gamma \left[1+s_1\right]~\Gamma \left[1+s_2\right] ~\Gamma\left[1+s_3\right]} \left(1-\frac{s_1 s_3}{s_2+1}\right)\,.
\eea
To remove the massless poles we now need to multiply $F_5$ as defined in \eqref{f5} by the factor $s_1s_2 s_3$ and we can follow the same procedure as the photon case with the only change being that we now have poles at $s_1^{\prime}=k$ with $k\in \mathbb{Z_{+}}$. \\
The results are shown in the second row of the figure above. We note that since \eqref{crossdisp} was written down assuming $o(s_1^2)$ behaviour for large $|s_1|$ and fixed $s_2$, examining the growth of $F_I$ restricts  $s_2$ to a region where we should trust the results, for e.g., $F_4$ in the photon case has a growth $s_1^{s_2}$ for large $s_1$ which implies we can strictly expect an agreement for only $s_2 <2$, though we have considered a bigger region in the figure we see that there is an excellent agreement between the dispersion relation and the exact answer. 

\noindent We have also verified that the crossing symmetric combination in eqn.\eqref{gravamp} namely $$ \tilde F^{h}_2(s_1,s_2,s_3) 
= \left(\frac{T^{h}_1(s_1,s_2,s_3)}{s_1^2}+ \frac{T^{h}_3(s_1,s_2,s_3)}{s_3^2}+ \frac{T^{h}_4(s_1,s_2,s_3)}{s_2^2}\right)$$ for tree-level 4-graviton scattering amplitudes in superstring, Heterotic string and bosonic string theories all obey that the crossing symmetric dispersion relation after subtracting out the massless poles. 

\section{Unitarity constraints}\label{uniconst}
In this section we review the unitarity constraints on partial wave amplitudes following \cite{spinpen, Henriksson:2021ymi}. Unitarity constraints can be summarized as positivity of norm of a state $\langle \psi |\psi \rangle \geq 0$. If we have multiple states (say of number $N$), this translates to \emph{positive semi-definiteness} of a $N\times N$ hermitian matrix. In order to see the relation of this statement in context of S-matrices, consider the incoming and outgoing particles as decomposed into irreps of the poincare group. To be more precise (eq (2.21) of \cite{spinpen}), 

\begin{eqnarray}\label{decom}
	| \k_1, \k_2 \rangle = \int \frac{d^4p}{(2\pi)^4} \theta(p^0) \theta(-p^2) \sum_{i,j} \sum_{\ell,\lambda} |c, \vec{p}; \ell, \lambda; \lambda_i, \lambda_j\rangle \langle c, \vec{p}; \ell, \lambda; \lambda_i, \lambda_j | \k_1,\k_2 \rangle.
\end{eqnarray}

$\ket{\k_1,\k_2}$ is generic $2-$particle momentum state \be\ket{\k_1,\k_2}:=\ket{m_1,\vec{p}_1;j_1,\l_1}\otimes \ket{m_2,\vec{p}_2;j_2,\l_2},\ee
where $\vec{p}_i, m_i$ are corresponding $3-$momentum, mass respectively, $p_i^\m p_{i\m}=-m_i^2$. $j_i, \l_i$ are spin and helicity respectively. For massive particles helicity takes $2j_i+1$ values, $\l_i\in\{-j_i,-j_i+1,\dots,j-1,j\},\,m_i\ne 0$, while for massless particles it takes two values, $\l_i=\pm j_i,\,m_i=0$. The Poincare $2-$particle irreps $\{\ket{c, \vec{p}; \ell, \lambda; \lambda_i, \lambda_j}\}$ are states of definite total momenta and total angular momenta, $\vec{p}$ being the total $3-$momentum and $\ell$ being the total angular momentum with $\l$ corresponding $3-$component taking $2\ell+1$ values, $\l\in\{-\ell,-\ell+1,\dots,\ell-1,\ell\}$. In particular,
\be \label{CG} 
\Braket{c, \vec{p}; \ell, \lambda; \lambda_i, \lambda_j | \k_1,\k_2}\propto (2\pi)^4\d^4(p^\m-p_1^\m-p_2^\m).
\ee  
Further, these states are normalized by 
\be \label{norm1}
\Braket{ c', \vec{p}';\ell',\lambda';\lambda_1', \lambda_2'| c, \vec{p};\ell,\lambda;\lambda_1, \lambda_2} = (2\pi)^4 \delta^4(p'^\mu-p^\mu)~ \delta_{l' l} \delta_{\lambda' \lambda}\delta_{\lambda_1' \lambda_1}\delta_{\lambda_2' \lambda_2}.
\ee For our purpose, we will work in CoM frame. Thus the states of interest to us are $\{\ket{c, \vec{0};\ell,\lambda;\lambda_1, \lambda_2}\}$. Under the action of parity operator $\mathcal{P}$, these states transform as \eqref{parityaction}
\be \label{parityaction}
\mathcal{P} \ket{ c, \vec{0};\ell,\lambda;\lambda_1, \lambda_2} = \eta_1\eta_2 (-1)^{\ell-j_1+j_2} \ket{ c, \vec{0};\ell,\lambda; -\lambda_1, -\lambda_2}.\ee Here  $\eta_1$ and $\eta_2$ are pure phases, also called \emph{intrinsic parity} associated with the particle, obey the constraint $\eta_i^2= \pm 1$ (with the negative sign only possible for fermions).

For identical particles we need to take care of the exchange symmetry. This prompts us to define the following states which we will use in our subsequent analysis
\begin{eqnarray} \label{iden}
	\ket{ c, \vec{0};\ell,\lambda;\lambda_1, \lambda_2}_{\text{id}} = \frac{1}{2}\left[\ket{ c, \vec{0};\ell,\lambda;\lambda_1, \lambda_2} + (-1)^{\ell+\lambda_1-\lambda_2} \ket{ c, \vec{0};\ell,\lambda;\lambda_1, \lambda_2}\right]
\end{eqnarray}  

We also note the relation between the 2-particle reducible state $\ket{\k_1, \k_2 }_{\text{COM}}$ in the COM frame and $\ket{ c, \vec{0};\ell,\lambda;\lambda_1, \lambda_2}_{\text{id}}$. This is essential in determining the range of $\ell$ for the irreducible 2-particle reps. Following \cite{spinpen} we can define the 2-particle reducible state in COM frame as product of eigenstates of the $J_z$ operator. 
\begin{eqnarray}\label{kcom}
	\ket{\k_1, \k_2 }_{\text{COM}} &\equiv& \ket{\left(\vec{p}, \theta, \phi\right); \lambda_1, \lambda_2 }\equiv \ket{m_1, \vec{p}; j_1, \lambda_1} \otimes \ket{m_2, -\vec{p}; j_2, \lambda_2}\nonumber\\
\end{eqnarray}
where $J_z=L_z + j^1_z + j^2_z$ ($L_z$ is the orbital angulam momentum, $j^i_z$ are the intrinsic spins). In the COM frame, therefore, \eqref{decom} can be expressed as follows (see (C.18) of \cite{spinpen}),
\begin{eqnarray}\label{pcom}
	\ket{\left(\vec{p}, \theta, \phi\right); \lambda_1, \lambda_2 }_{\text{id}}&=&\sqrt{2}\sum_{\ell}\sum^{\ell}_{ \lambda=-\ell} C_{\ell}(\vec{p}) e^{i \phi(\lambda_1+\lambda_2-\lambda)} d^{(l)}_{\lambda \lambda_{12}}(\theta)\ket{c,0; \ell, \lambda;\lambda_1,\lambda_2}_{\text{id}}
\end{eqnarray}  
The sum over $\ell$ is not unbounded and can be fixed as follows. Let us consider the case where the $\vec{p}$ is alligned along the $z$-axis (i.e $\theta=\phi=0$). 
The LHS is an eigenstate of $J_z= \lambda_1-\lambda_2$, since the orbital angular momentum is zero and the projection of intrinsic spin onto the direction of momenta now becomes the helicity itself. The RHS sum over $\lambda$ therefore must be therefore over only those states for which $\lambda=\lambda_1-\lambda_2$, and hence $\ell\geq |\lambda_1-\lambda_2|$.

\begin{eqnarray}
	\ket{\left(\vec{p},0, 0\right); \lambda_1, \lambda_2 }_{\text{id}}&=&\sqrt{2}\sum^{\infty}_{\ell=|\lambda_1-\lambda_2|} C_{\ell}(\vec{p}) \ket{c,0; \ell, \lambda=|\lambda_1-\lambda_2|;\lambda_1,\lambda_2}_{\text{id}}
\end{eqnarray}  

We can now apply the rotation matrix on both sides to bring it to the form \eqref{pcom}. Note that the rotation matrix does not change the casimir $J^2$ (and hence the $\ell$ sum). 

\begin{eqnarray}\label{finaldecom}
	\ket{\left(\vec{p},\theta, \phi\right); \lambda_1, \lambda_2 }_{\text{id}}&=&\sqrt{2}\sum^{\infty}_{\ell=|\lambda_1-\lambda_2|}\sum^{\ell}_{ \lambda=-\ell} C_{\ell}(\vec{p}) e^{i \phi(\lambda_1+\lambda_2-\lambda)} d^{(l)}_{\lambda \lambda_{12}}(\theta)\ket{c,0; \ell, \lambda;\lambda_1,\lambda_2}_{\text{id}} \nonumber\\
\end{eqnarray}  
To summarise, the additional symmetry due to identical nature of 2 particle states decides even or odd spins (or both) appear in the partial wave expansion. The cut off for the spin is decided by the helicities of the constituent states.

Corresponding to the incoming and the outgoing states, therefore, we can write down a basis of irreducible 2 particle states for each spin $l$ that appears in the decomposition \eqref{decom}. Generically they are denoted as, 
\begin{eqnarray}
	|1\rangle_{\text{in}}&=& |c, \vec{p}; \ell, \lambda; j_1, j_2 \rangle_{\text{in}}, \qquad |1\rangle_{\text{out}}= |c, \vec{p}; \ell, \lambda; j_3, j_4 \rangle_{\text{out}} \nonumber\\
	|2\rangle_{\text{in}}&=& |c, \vec{p}; \ell, \lambda; j_1-1, j_2 \rangle_{\text{in}}, \qquad |2\rangle_{\text{out}}= |c, \vec{p}; \ell, \lambda; j_3-1, j_4 \rangle_{\text{out}} \nonumber\\ 
	&&.\nonumber\\
	&&.\nonumber\\
	|N_{\text{in}}\rangle_{\text{in}}&=& |c, \vec{p}; \ell, \lambda; -j_1, -j_2 \rangle_{\text{in}}, \qquad |N_{\text{out}}\rangle_{\text{out}}= |c, \vec{p}; \ell, \lambda; -j_3, -j_4 \rangle_{\text{out}} \nonumber\\
\end{eqnarray}

Imposing the positivity constraints of hermitian matrix therefore translates to positivity of the following matrix (eq 2.118 of \cite{spinpen}) 

\begin{eqnarray}
	\mathcal{H}_\ell(s)\times (2\pi)^4 \delta^4 \delta_{\ell'\ell} \delta_{\lambda'\lambda}= \begin{pmatrix}
		{}_{\text{in}}\langle a'|b\rangle_{\text{in}}& {}_{\text{in}}\langle a'|b\rangle_{\text{out}} \\
		{}_{\text{out}}\langle a'|b\rangle_{\text{in}}& {}_{\text{out}}\langle a'|b\rangle_{\text{out}} \\
	\end{pmatrix}
\end{eqnarray} 

where $s=-p^2$. Using the normalisation conditions \eqref{norm1}, and 

\begin{eqnarray}
	{}_{\text{out}}\Braket{c', \vec{p}';\ell',\lambda';\lambda_1', \lambda_2'| c, \vec{p};\ell,\lambda;\lambda_1, \lambda_2}_{\text{in}} &=& (2\pi)^4 \delta^4(p'^\mu-p^\mu)~ \delta_{\ell' \ell} \delta_{\lambda' \lambda} ~{S_\ell}^{\lambda_1', \lambda_2'}_{\lambda_1, \lambda_2}(s)\nonumber\\
\end{eqnarray} 
where ${\spart}^{\lambda_1', \lambda_2'}_{\lambda_1, \lambda_2}(s)= (\delta_{\lambda_1'\lambda_1}\delta_{\lambda_2'\lambda_2} + (-1)^{\ell- (\lambda'_1-\lambda'_2)}\delta_{\lambda_2'\lambda_1}\delta_{\lambda_1'\lambda_2})+ i {\tpart}^{\lambda_1', \lambda_2'}_{\lambda_1, \lambda_2}(s)$ is the partial amplitude with spin $\ell$, we get, 

\begin{eqnarray}\label{unitarity}
	\begin{pmatrix}
		\delta_{a'b}& {\spart}^{*}_{a'b} \\
		{\spart}_{a'b}& \delta_{a'b} \\
	\end{pmatrix}\succeq 0
\end{eqnarray} 

\subsection{Massless bosons: Photons and gravitons} \label{A.1}
The two particle irreducible states can be labelled by the helicities ($\lambda_i=\pm m$ with $m=1, 2$ for photons and gravitons respectively ) as 
\begin{eqnarray}\label{photon}
	| c, \vec{p};\ell,\lambda;\lambda_1, \lambda_2\rangle &\equiv& |\lambda_1, \lambda_2 \rangle \nonumber\\
	|1 \rangle &=& \frac{1}{\sqrt{2}}\left( |++\rangle + |--\rangle \right), \qquad \ell= 0,2,4,\cdots \nonumber\\
	|2 \rangle &=& \sqrt{2}\left( |+-\rangle \right), \qquad \ell= 2m,2m+1, 2m+2, 2m+3,\cdots \nonumber\\
	|3 \rangle &=& \frac{1}{\sqrt{2}}\left( |++\rangle - |--\rangle \right), \qquad \ell= 0,2,4,\cdots \nonumber\\
\end{eqnarray}
We note that the states $|1 \rangle$ and $|3 \rangle$ only contain even spin. This is evident from the symmetry of the states ( see \eqref{iden}) and the discussion around \eqref{finaldecom}. In the following subsections, we try to impose the unitarity conditions assuming parity invariance and non-invariance respectively. For gravitons $\lambda_i=\pm 2$, so spins will change as $|\lambda_1-\lambda_2|$.

\subsubsection{ Parity invariant theories}

We note that the parity of these states: states $|1\rangle$ and $|2\rangle$ are parity even states while $|3\rangle$ is a parity odd state. This is due to the following\footnote{We work in the convention $\eta_i=1$ for photons.} 

\begin{eqnarray}
	\mathcal{P}| ++ \rangle = | -- \rangle,~~ \mathcal{P}| -- \rangle = | ++ \rangle,~~ \mathcal{P}| +- \rangle = (-1)^\ell | -+ \rangle = (-1)^\ell(-1)^{\ell-2} | +- \rangle= | +- \rangle \nonumber\\
\end{eqnarray} 

We are now in a position to evaluate the matrix \eqref{unitarity} for the set of states \eqref{photon}. Furthermore we assume parity invariance which implies that the states of definite parity do not mix. The following conditions are obtained for the parity even sector.

\begin{eqnarray}\label{photonpe}
	\begin{pmatrix}
		{}_{\text{in}}\langle 1'|1\rangle_{\text{in}}& {}_{\text{in}}\langle 1'|1\rangle_{\text{out}} \\
		{}_{\text{out}}\langle 1'|1\rangle_{\text{in}}& {}_{\text{out}}\langle 1'|1\rangle_{\text{out}} \\
	\end{pmatrix}&\succeq & 0,\qquad \ell=0, \nonumber\\
	\begin{pmatrix}
		{}_{\text{in}}\langle 2'|2\rangle_{\text{in}}& {}_{\text{in}}\langle 2'|2\rangle_{\text{out}} \\
		{}_{\text{out}}\langle 2'|2\rangle_{\text{in}}& {}_{\text{out}}\langle 2'|2\rangle_{\text{out}} \\
	\end{pmatrix}&\succeq & 0,\qquad \ell=2m+1, 2m+3,\dots \nonumber\\
	\begin{pmatrix}
		{}_{\text{in}}\langle 1'|1\rangle_{\text{in}}& {}_{\text{in}}\langle 1'|2\rangle_{\text{in}} & {}_{\text{in}}\langle 1'|1\rangle_{\text{out}}& {}_{\text{in}}\langle 1'|2\rangle_{\text{out}}\\
		{}_{\text{in}}\langle 2'|1\rangle_{\text{in}}& {}_{\text{in}}\langle 2'|2\rangle_{\text{in}} & {}_{\text{in}}\langle 2'|1\rangle_{\text{out}}& {}_{\text{in}}\langle 2'|2\rangle_{\text{out}}\\
		{}_{\text{out}}\langle 1'|1\rangle_{\text{in}}& {}_{\text{out}}\langle 1'|2\rangle_{\text{in}} & {}_{\text{out}}\langle 1'|1\rangle_{\text{out}}& {}_{\text{out}}\langle 1'|2\rangle_{\text{out}}\\
		{}_{\text{out}}\langle 2'|1\rangle_{\text{in}}& {}_{\text{out}}\langle 2'|2\rangle_{\text{in}} & {}_{\text{out}}\langle 2'|1\rangle_{\text{out}}& {}_{\text{out}}\langle 2'|2\rangle_{\text{out}}\\
	\end{pmatrix}&\succeq & 0,\qquad \ell=2m, 2m+2,\dots \nonumber\\
\end{eqnarray}

and for the parity odd sector, 
\begin{eqnarray}\label{photonpo}
	\begin{pmatrix}
		{}_{\text{in}}\langle 3'|3\rangle_{\text{in}}& {}_{\text{in}}\langle 3'|3\rangle_{\text{out}} \\
		{}_{\text{out}}\langle 3'|3\rangle_{\text{in}}& {}_{\text{out}}\langle 3'|3\rangle_{\text{out}} \\
	\end{pmatrix}&\succeq & 0,\qquad \ell=0,2,4,\dots\nonumber\\
\end{eqnarray}

Let us work out one of the conditions in detail: 1st matrix of \eqref{photonpe} gives us the following 

\begin{eqnarray}
	\begin{pmatrix}
		1& 1-i\left({\tpart}^{*~--}_{++}+{\tpart}^{*~++}_{++}\right) \\
		1+i\left({\tpart}^{~--}_{++}+{\tpart}^{~++}_{++}\right) & 1\\
	\end{pmatrix}&\succeq & 0,\qquad \ell=0, \nonumber\\
\end{eqnarray}  
Noting that the trace is positive trivially, the condition of positivity translates to the determinant of the matrix being positive:

\begin{eqnarray}
	2~\text{Im}(\mt_1^{\ell=0}+\mt_2^{\ell=0})\geq |\mt_1^{\ell=0}+\mt_2^{\ell=0}|^2\geq 0 \nonumber\\
\end{eqnarray}

Similarly, 2nd matrix of \eqref{photonpe} and \eqref{photonpo} gives us the following 
\begin{eqnarray}
	&\text{Im}T^\ell_3\geq |T^\ell_3|^2 \geq 0, \qquad \ell=2m+1, 2m+3,\dots.\nonumber\\
	&2\text{Im}(T^\ell_1-T^\ell_2)\geq |T^\ell_1-T^\ell_2|^2 \geq 0, \qquad \ell=0,2,4,\dots. 
\end{eqnarray}


Now let us consider the conditions coming from the third matrix in \eqref{photonpe}. We find that analysing  the $2 \times 2$ principal minors is sufficient for our purposes and we obtain,  
\begin{eqnarray}
	&\text{Im}T^\ell_3\geq |T^\ell_3|^2 \geq 0, \qquad \ell=2m, 2m+2,\dots,\nonumber\\
	&2\text{Im}(T^\ell_1+T^\ell_2)\geq |T^\ell_1+T^\ell_2|^2 \geq 0, \qquad \ell=2, 4,\dots,\nonumber\\
	&1\geq 4|T^\ell_5|^2\geq 0, \qquad \ell=2m, 2m+2,\dots\nonumber\\
\end{eqnarray}

\subsubsection{Parity violating theories}
For this case, the assumption that parity even and odd states do not mix no longer holds true. This leads to the modification of the unitarity equations, 

\begin{eqnarray}\label{oddphoton}
	\begin{pmatrix}
		{}_{\text{in}}\langle 1'|1\rangle_{\text{in}}& {}_{\text{in}}\langle 1'|3\rangle_{\text{in}} & {}_{\text{in}}\langle 1'|1\rangle_{\text{out}}& {}_{\text{in}}\langle 1'|3\rangle_{\text{out}}\\
		{}_{\text{in}}\langle 3'|1\rangle_{\text{in}}& {}_{\text{in}}\langle 3'|3\rangle_{\text{in}} & {}_{\text{in}}\langle 3'|1\rangle_{\text{out}}& {}_{\text{in}}\langle 3'|3\rangle_{\text{out}}\\
		{}_{\text{out}}\langle 1'|1\rangle_{\text{in}}& {}_{\text{out}}\langle 1'|3\rangle_{\text{in}} & {}_{\text{out}}\langle 1'|1\rangle_{\text{out}}& {}_{\text{out}}\langle 1'|3\rangle_{\text{out}}\\
		{}_{\text{out}}\langle 3'|1\rangle_{\text{in}}& {}_{\text{out}}\langle 3'|3\rangle_{\text{in}} & {}_{\text{out}}\langle 3'|1\rangle_{\text{out}}& {}_{\text{out}}\langle 3'|3\rangle_{\text{out}}
	\end{pmatrix}&\succeq & 0,\qquad \ell=0 \nonumber\\
	\begin{pmatrix}
		{}_{\text{in}}\langle 2'|2\rangle_{\text{in}}& {}_{\text{in}}\langle 2'|2\rangle_{\text{out}} \\
		{}_{\text{out}}\langle 2'|2\rangle_{\text{in}}& {}_{\text{out}}\langle 2'|2\rangle_{\text{out}} \\
	\end{pmatrix}&\succeq & 0,\qquad \ell=2m+1, 2m+3,\dots \nonumber\\
	\begin{pmatrix}
		{}_{\text{in}}\langle 1'|1\rangle_{\text{in}}& {}_{\text{in}}\langle 1'|2\rangle_{\text{in}} & {}_{\text{in}}\langle 1'|3\rangle_{\text{in}} & {}_{\text{in}}\langle 1'|1\rangle_{\text{out}} & {}_{\text{in}}\langle 1'|2\rangle_{\text{out}}& {}_{\text{in}}\langle 1'|3\rangle_{\text{out}}\\
		{}_{\text{in}}\langle 2'|1\rangle_{\text{in}}& {}_{\text{in}}\langle 2'|2\rangle_{\text{in}} & {}_{\text{in}}\langle 2'|3\rangle_{\text{in}} & {}_{\text{in}}\langle 2'|1\rangle_{\text{out}} & {}_{\text{in}}\langle 2'|2\rangle_{\text{out}}& {}_{\text{in}}\langle 2'|3\rangle_{\text{out}}\\
		{}_{\text{in}}\langle 3'|1\rangle_{\text{in}}& {}_{\text{in}}\langle 3'|2\rangle_{\text{in}} & {}_{\text{in}}\langle 3'|3\rangle_{\text{in}} & {}_{\text{in}}\langle 3'|1\rangle_{\text{out}} & {}_{\text{in}}\langle 3'|2\rangle_{\text{out}}& {}_{\text{in}}\langle 3'|3\rangle_{\text{out}}\\
		{}_{\text{out}}\langle 1'|1\rangle_{\text{in}}& {}_{\text{out}}\langle 1'|2\rangle_{\text{in}} & {}_{\text{out}}\langle 1'|3\rangle_{\text{in}} & {}_{\text{out}}\langle 1'|1\rangle_{\text{out}} & {}_{\text{out}}\langle 1'|2\rangle_{\text{out}}& {}_{\text{out}}\langle 1'|3\rangle_{\text{out}}\\
		{}_{\text{out}}\langle 2'|1\rangle_{\text{in}}& {}_{\text{out}}\langle 2'|2\rangle_{\text{in}} & {}_{\text{out}}\langle 2'|3\rangle_{\text{in}} & {}_{\text{out}}\langle 2'|1\rangle_{\text{out}} & {}_{\text{out}}\langle 2'|2\rangle_{\text{out}}& {}_{\text{out}}\langle 2'|3\rangle_{\text{out}}\\
		{}_{\text{out}}\langle 3'|1\rangle_{\text{in}}& {}_{\text{out}}\langle 3'|2\rangle_{\text{in}} & {}_{\text{out}}\langle 3'|3\rangle_{\text{in}} & {}_{\text{out}}\langle 3'|1\rangle_{\text{out}} & {}_{\text{out}}\langle 3'|2\rangle_{\text{out}}& {}_{\text{out}}\langle 3'|3\rangle_{\text{out}}\\
	\end{pmatrix}&\succeq & 0,\qquad \ell= 2,4,6\cdots \nonumber\\
\end{eqnarray}

The analysis of these matrices is tedious and we find the following constraints
\begin{eqnarray}
	\text{Im}T^\ell_3\geq |T^\ell_3|^2 \geq & 0, \qquad \ell=2m,2m+1\dots,\nonumber\\
	2\text{Im}(T^\ell_1+\frac{1}{2}(T^\ell_2+ T'^\ell_2))\geq |T^\ell_1+\frac{1}{2}(T^\ell_2+ T'^\ell_2)|^2 \geq & 0, \qquad \ell=0,2,4,6,\dots,\nonumber\\
	2\text{Im}(T^\ell_1-\frac{1}{2}(T^\ell_2+ T'^\ell_2))\geq |T^\ell_1-\frac{1}{2}(T^\ell_2+  T'^\ell_2)|^2 \geq & 0, \qquad \ell=0,2,4,6,\dots,\nonumber\\
	|T^\ell_2- T'^\ell_2|^2\leq & 4, \qquad \ell=0,1,2,\dots\nonumber\\
	|T^\ell_5- T'^\ell_5|\leq & 1, \qquad \ell=2m, 2m+1,\dots\nonumber\\
	|T^\ell_5+T'^\ell_5|\leq & 1 \qquad \ell=2m, 2m+1,\dots\nonumber\\
\end{eqnarray}

From the last three conditions listed above, it seems that linear unitarity analysis doesn't fix the sign of $\rho^\ell_2-\rho'^\ell_2, \rho^\ell_5 \pm \rho'^\ell_5$ and perhaps a more thorough investigation is required \cite{sasha}. Hence, in this work, we do not attempt to bound the parity violating amplitudes.

\subsection{Massive Majorana fermions}
The unitarity conditions for massive Majorana fermions were spelt out in \cite{spinpen}. Let us quickly review them for completeness. Recalling the fermion amplitudes $\{\Phi_i\}$ defined in \eqref{MMamp}, the corresponding partial amplitudes are denoted as $\{\Phi_i^\ell\}$. Then, following the similar arguements as in the previous subsection, one arrives at the  unitarity conditions as follows:

\begin{enumerate}
	\item 
	\be 
	\begin{pmatrix}
		1& 1-i\left[\Phi_1^{0*}(s)-\Phi_2^{0*}(s)\right]\\
		1+i\left[\Phi_1^{0}(s)-\Phi_2^{0}(s)\right]& 1 \\
	\end{pmatrix}\succeq  0,. \ee
	\\
	The positivity of the determinant of the matrix then gives 
	\be 
	2\,\text{Im.}\,\left[\Phi_1^{0}(s)-\Phi_2^{0}(s)\right]\ge \left|\Phi_1^{0}(s)-\Phi_2^{0}(s)\right|^2.
	\ee 
	\item 
	\be 
	\begin{pmatrix}
		1& 1-i\left[\Phi_1^{\ell*}(s)+\Phi_2^{\ell*}(s)\right]\\
		1+i\left[\Phi_1^{\ell}(s)+\Phi_2^{\ell}(s)\right]& 1 \\
	\end{pmatrix}\succeq  0,\qquad \ell\ge0\,\, (\text{even}), \ee \\
	implying 
	\be 
	2\,\text{Im.}\,\left[\Phi_1^{\ell}(s)+\Phi_2^{\ell}(s)\right]\ge \left|\Phi_1^{\ell}(s)+\Phi_2^{\ell}(s)\right|^2,\qquad \ell\ge0\,\, (\text{even}).
	\ee 
	\item 
	\be 
	\begin{pmatrix}
		1& 1-2i\,\Phi_3^{\ell*} \\
		1+2i\,\Phi_3^{\ell}& 1 \\
	\end{pmatrix}\succeq  0,\qquad \ell\ge 1\,\,(\text{odd}), \ee\\ with straightforward consequence 
	\be 
	\text{Im.}\,\Phi_3^\ell(s)\ge |\Phi_3^\ell(s)|^2\ge 0\qquad \ell\ge 1\,\,(\text{odd}).
	\ee 
	\item 
	\be 
	\begin{pmatrix}
		\mathbb{I}_{2\times 2} & \mathbb{S}_{2\times 2}^{\ell\dagger}\\
		\mathbb{S}_{2\times 2}^{\ell} & \mathbb{I}_{2\times 2}\\
	\end{pmatrix}\succeq  0,\qquad \ell=2,4,6,\dots 
	\ee with
	\be
	\mathbb{I}_{2\times 2}:=\begin{pmatrix}
		1& 0 \\
		0& 1 \\
	\end{pmatrix},\qquad \mathbb{S}_{2\times 2}^{\ell}(s):=\begin{pmatrix}
		1+i\left[\Phi_1^\ell(s)-\Phi_2^\ell(s)\right]& 2i\,\Phi_5^{\ell*} \\
		2i\,\Phi_5^\ell(s)& 1+2i\,\Phi_3^\ell(s) \\
	\end{pmatrix}
	\ee 
	We get, 
	
	$$\text{Det}[\mathbb{S}_{2\times 2}^{\ell}(s)]\geq 0$$
\end{enumerate}

\section{Representations of Wigner-$d$ functions}\label{repsofwigd} 
In this appendix, we give some convenient representations of the Wigner-$d$ functions that we used in the main text  for computational ease. For the photons we use, 

\begin{eqnarray}\label{wignerphoton}
	d^J_{0,0}\left(\cos^{-1}{\sqrt{\xi(s'_1,a)}}\right)&=&\, _2F_1\left(-J,J+1;1;\frac{1}{2} \left(1-\sqrt{\xi }\right)\right), \nonumber\\
	d^J_{2,2}\left(\cos^{-1}{\sqrt{\xi(s'_1,a)}}\right)&=&\frac{1}{24} \left(6 \left(\sqrt{\xi }+1\right)^2 \, _2F_1\left(2-J,J+3;1;\frac{1}{2} \left(1-\sqrt{\xi }\right)\right)\right.\nonumber\\
	d^J_{2,-2}\left(\cos^{-1}{\sqrt{\xi(s'_1,a)}}\right)&=&\frac{\left(1-\sqrt{\xi }\right)^2 \Gamma (J+3) \, _2F_1\left(2-J,J+3;5;\frac{1}{2} \left(1-\sqrt{\xi }\right)\right)}{96 \Gamma (J-1)} \nonumber\\
	d^J_{4,4}\left(\cos^{-1}{\sqrt{\xi(s'_1,a)}}\right)&=&\frac{1}{16} (\sqrt{\xi} +1)^4 \, _2F_1\left(4-J,J+5;1;\frac{1-\sqrt{\xi} }{2}\right)\nonumber\\
	d^J_{4,-4}\left(\cos^{-1}{\sqrt{\xi(s'_1,a)}}\right)&=&\frac{(1-\sqrt{\xi} )^4 \Gamma (J+5) \, _2F_1\left(4-J,J+5;9;\frac{1-\sqrt{\xi} }{2}\right)}{645120 \Gamma (J-3)}
\end{eqnarray}

\section{ Massive Majorana fermions: Locality constraints}
We would like to use our techniques to constrain the EFTs involving Majorana fermions. However in this work we will only spell out the locality constraints and leave a careful analysis for future work. Now we will list the locality constraint for the amplitude listed in \eqref{psi1}. Let us assume a low energy EFT expansion of the form 
\bea 
\Psi_1(s_1,s_2,s_3)=\sum_{n,m} \mw^\psi_{p,q} x^p y^q
\eea 
Tha partial wave decomposition reads, 
\bea\label{dispfermion} 
(\Psi_1(s_1,s_2)) =\alpha^\psi_0+\frac{1}{\pi}\int_{M^2}^{\infty}\frac{d s_1}{s'_1} \ma^\psi\left(s_1^{\prime} ; s_2^{(+)}\left(s_1^{\prime} ,a\right)\right) H\left(s_1^{\prime} ;s_1, s_2, s_3\right),\nonumber\\
\eea
where $H\left(s_1^{\prime} ;s_1, s_2, s_3\right)$ is defined in \eqref{kernel} and the partial wave decomposition reads
\bea\label{crossdispexpfermion}
\ma^\psi\left(s_1^{\prime} ; s_2^{(+)}\left(s_1^{\prime} ,a\right)\right)&=& \sum_{J=0,2,4,\cdots} 16 \pi (2 J+1) \rho^{1, \psi}_J d^{J}_{0,0}(\theta) +\sum_{J=1,2,3,\cdots} 64 \pi (2 J+1) (-1)^{J+1} \rho^{3, \psi}_J d^{J}_{1,-1}(\theta) \nonumber\\
&&-\sum_{J=0,2,4,\cdots} 16 \pi (2 J+1) \rho^{5, \psi}_J d^{J}_{0,1}(\theta)
\eea  
where $(\cos\theta)^2=\xi(s'_1, a)=\xi_0+ 4 \xi_0\left(\frac{a}{s'_1-a}\right)$ and $\xi_0=\frac{s_1^2}{s_1-M^2}$ while $\rho^{i, \psi}_J$ are the respective spectral functions which appear as coefficients in partial wave expansion of the absorptive parts $\ma^\psi$. We have also used that $\rho^{4, \psi}_J=(-1)^{J+1}\rho^{3, \psi}_J$\cite{spinpen}. The coefficients $\mw^\psi_{p,q}$ can be obtained from the amplitude via the inversion formula \cite{Sinha:2020win}
\begin{align} \label{Winvf}
\mw^\psi_{n-m,m}&=\int_{M^2}^\infty \frac{ds_1}{2\pi s_1^{2n+m+1}}16\pi\left[\sum_{J=0,2,4\cdots} (2J+1)\,\rho^{1, \psi}_J(s_1)\,\mathcal{\Psi}_{n,m}^{(1,J)}(s_1)\right. \nonumber\\
&\hspace{2.5 cm}\left.+\sum_{J=1,2,3\cdots} (2J+1)\,\rho^{3, \psi}_J(s_1)\,\mathcal{\Psi}_{n,m}^{(3,J)}(s_1) \sum_{J=0,2,4\cdots} (2J+1)\,\rho^{5, \psi}_J(s_1)\,\mathcal{\Psi}_{n,m}^{(5,J)}(s_1)\right],
\end{align}
with 
\be \label{bellmassivef}
\mathcal{\Psi}_{n,m}^{(i,J)}(s_1)=2\sum_{j=0}^{m}\frac{(-1)^{1-j+m}r_{J}^{(i,j)}\left(\x_0\right) \left(4\xi_0 \right)^{j}(3 j-m-2n)\Gamma(n-j)}{ j!(m-j)!\Gamma(n-m+1)},\qquad \x_0:=\frac{s_1^2}{(s_1-2\m/3)^2}.
\ee 
The functions $\{p_\ell^{(j)}(\x_0)\}$ are derivatives of respective Wigner-$d$ functions
\bea 
r_J^{(1, j)}(\x_0)&:=& \left.\frac{\partial^j d^J_{0,0}(\sqrt{z})}{\partial z^j}\right|_{z=\x_0}, \qquad r_J^{(3, j)}(\x_0):= (-1)^{J+1}\left.\frac{\partial^j d^J_{1,-1}(\sqrt{z})}{\partial z^j}\right|_{z=\x_0} \nonumber\\
&& r_J^{(5, j)}(\x_0):= \left.\frac{\partial^j d^J_{0,1}(\sqrt{z})}{\partial z^j}\right|_{z=\x_0}
\eea 
The locality constraints then are simply,
\bea
\mw^\psi_{n-m,m}=0,\qquad \forall~~ n<m \,.
\eea

\section{EFT expansion of the crossing basis elements}\label{eecbs}
In this section we give the low energy EFT expansion for the crossing basis elements \eqref{crosssymmbasis} for some special case used in the main text. 
For the photon case we have the following:
\begin{align}\label{EFTexpappphoton}
F_1(s_1,s_2,s_3)&= 2 f_2 x-f_3 y+4 f_4 x^2-2f_5 xy +f_{6,1} y^2+8f_{6,2}x^3+\cdots \,\nonumber\\
F_2(s_1,s_2,s_3)&= 2 g_2 x -3 g_3 y +2 (g_{41}+2 g_{42})x^2-(5 g_{5,1}+3 g_{5,2}) xy \nn\\
&\hspace{3.5 cm}+ 3 \left(g_{6,1}-g_{6,2}+g_{6,3}\right) y^2 + 2 g_{6,1}x^3 + \cdots  \,, \nonumber\\
F_3(s_1,s_2,s_3)&= 2 h_2 x-h_3 y+4 h_4 x^2-2h_5 xy +h_{6,1} y^2+8h_{6,2}x^3+\cdots \,\nonumber\\
F_4(s_1,s_2,s_3)&=\frac{1}{3}(g_2 +(g_{41}+2g_{42})x+ (g_{5,2}-g_{5,1})y + g_{6,1} x^2 + \cdots)\,,\nonumber\\
F_5(s_1,s_2,s_3)&=\frac{1}{3}(g_3 x - g_{4,1} y + g_{5,1}x^2-(2 g_{6,1}+g_{6,2})xy)+\cdots\,.
\end{align}
As alluded to in the main text in the discussion below eq.\eqref{crosssymmbasis} if an amplitude is $t-u$ symmetric then the crossing basis has only 3 elements $\{f(s_1,s_2,s_3),g_1(s_1,s_2,s_3),h_1(s_1,s_2,s_3)\}$ and $F_2,F_4,F_5$ above correspond to these for the $t-u$ symmetric amplitude, $\mt_1 (s_1, s_2, s_3)$, while $F_1$ and $F_3$ correspond to the fully crossing symmetric helicity amplitudes $\mt_2 (s_1, s_2, s_3)$ and $\mt_5 (s_1, s_2, s_3)$ respectively
\bea
\mt_1(s_1,s_2,s_3)&=&g_2 s_1^2+ g_3 s_1^3 +s_1^4 g_{4,1}+\left(s_1^2\right)(s_1^2+s_2^2+s_3^2) g_{4,2}+s_1^5 g_{5,1}+g_{5,2} \left(s_2^2 s_3+s_2 s_3^2\right) (s_2+s_3)^2 \nonumber\\
&&+ s_1^6 g_{6,1}+g_{6,2} \left(\left(s_2^3 s_3+s_2 s_3^3\right) (s_2+s_3)^2\right)+g_{6,3} \left(\left(s_2^2 s_3^2\right) (s_2+s_3)^2\right)\nonumber\\
\mt_2(s_1,s_2,s_3)&=&f_2 (s_1^2+s_2^2+s_3^2) + f_3 (s_1 s_2 s_3) + f_4 (s_1^2+s_2^2+s_3^2)^2 + f_5  (s_1^2+s_2^2+s_3^2)(s_1 s_2 s_3)+ f_{6,1}(s_1 s_2 s_3)^2\nonumber\\
&&f_{6,2}(s_1^2+s_2^2+s_3^2)^3\nonumber\\
\mt_5(s_1,s_2,s_3)&=&h_2 (s_1^2+s_2^2+s_3^2) + h_3 (s_1 s_2 s_3) + h_4 (s_1^2+s_2^2+s_3^2)^2 + h_5  (s_1^2+s_2^2+s_3^2)(s_1 s_2 s_3)+ h_{6,1}(s_1 s_2 s_3)^2\nonumber\\
&&h_{6,2}(s_1^2+s_2^2+s_3^2)^3
\eea 

For the parity violating case we have two additional elements corresponding to the crossing symmetric amplitudes $T_{2'}$ and $T_{5'}$. We note that $F_i$ for $i=1,2,3$ are the same ones considered in  \cite{vichi}, in this chapter additionally we also use $F_4$ and $F_5$. 

\section{Low spin dominance and Graviton scattering in String theory} \label{appG}
In this appendix we will show that the range of $a$, that we had obtained from our Locality constraint analysis in subsection \ref{gravitonboundTR}, is satisfied for type II string theory amplitude. For convenience we write the explicit amplitude
\begin{eqnarray}
	\mt_1 &=&\frac{s_1^3 \left(\frac{\Gamma \left(1-\frac{\alpha  s_1}{2}\right) \Gamma \left(1-\frac{\alpha  s_2}{2}\right) \Gamma \left(1-\frac{\alpha  s_3}{2}\right)}{\Gamma \left(\frac{s_1 \alpha }{2}+1\right) \Gamma \left(\frac{s_2 \alpha }{2}+1\right) \Gamma \left(\frac{s_3 \alpha }{2}+1\right)}-1\right)}{s_2 s_3} \nonumber\\
	\mt_3&=&\frac{s_3^3 \left(\frac{\Gamma \left(1-\frac{\alpha  s_1}{2}\right) \Gamma \left(1-\frac{\alpha  s_2}{2}\right) \Gamma \left(1-\frac{\alpha  s_3}{2}\right)}{\Gamma \left(\frac{s_1 \alpha }{2}+1\right) \Gamma \left(\frac{s_2 \alpha }{2}+1\right) \Gamma \left(\frac{s_3 \alpha }{2}+1\right)}-1\right)}{s_1 s_2}\nonumber\\
	\mt_4&=&\frac{s_2^3 \left(\frac{\Gamma \left(1-\frac{\alpha  s_1}{2}\right) \Gamma \left(1-\frac{\alpha  s_2}{2}\right) \Gamma \left(1-\frac{\alpha  s_3}{2}\right)}{\Gamma \left(\frac{s_1 \alpha }{2}+1\right) \Gamma \left(\frac{s_2 \alpha }{2}+1\right) \Gamma \left(\frac{s_3 \alpha }{2}+1\right)}-1\right)}{s_1 s_3}
\end{eqnarray}
Note that we have subtracted out the massless graviton pole. We obtain the $a_i^\ell(s_1)$ by considering the string amplitude as an infinite sum over poles at $s_1=\frac{2(n+1)}{\alpha^{\prime}}$. We want to verify that the constraints \eqref{ineqgraviton} are satisfied for our range of a: $-0.1933M^2\leq a\leq \frac{2}{3}M^2$.  
To be precise, we want to verify,
\bea
& &\mathcal{M}(s_i,a) = \nonumber \\
& &\int_{M^2}^{\infty}\frac{d s_1}{2\pi s_1} \sum_{J=0,2,4,\cdots}(2J+1)\tilde a^{(1)}_J(s_1)f^{(1)}_J(\xi)H(s_1',s_i) +\int_{M^2}^{\infty}\frac{d s_1}{2\pi s_1} \sum_{J=4,6,\cdots}(2J+1)\tilde a^{(2)}_J(s_1)f^{(2)}_J(\xi)H(s_1',s_i)\,\nonumber\\
& & + \int_{M^2}^{\infty}\frac{d s_1}{2\pi s_1} \sum_{J=5,7,\cdots}(2J+1)\tilde a^{(3)}_J(s_1)f^{(3)}_J(\xi)H(s_1',s_i)\,\geq 0,\nonumber\\
& &= \mathcal{M}^{\ell \le 2}(s_i,a)+\mathcal{M}^{\ell >2}(s_i,a) \ge 0
\eea
for the range of a which is where $\mathcal{M}^{\ell \le 2}(s_i,a) \ge 0$. This is what we called weak LSD in the main text. 


\begin{figure}[H]
	\centering
	\begin{subfigure}[b]{0.32\textwidth}
		\centering
		\includegraphics[width=\textwidth]{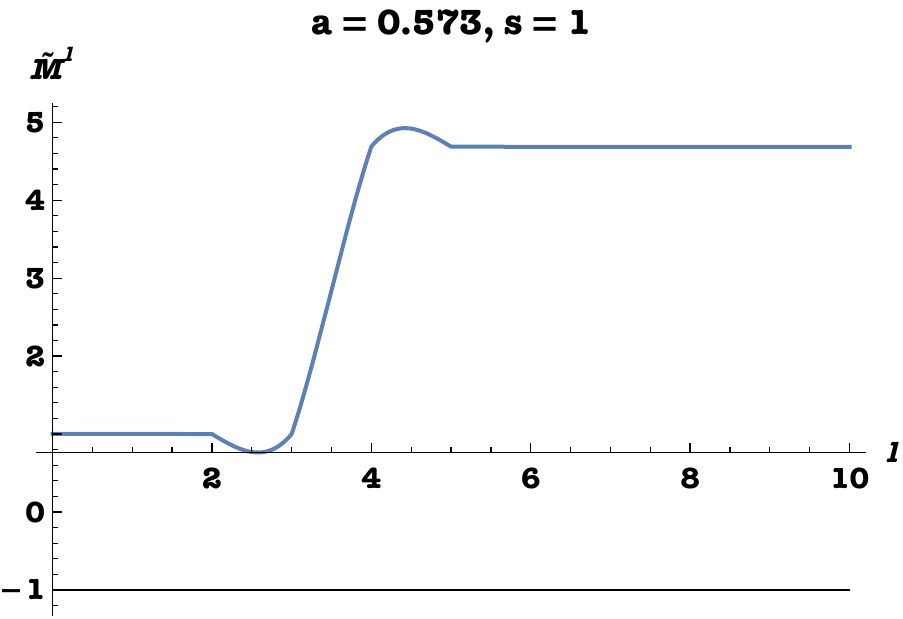}
		\caption{}
	\end{subfigure}
	\hfill
	\begin{subfigure}[b]{0.32\textwidth}
		\centering
		\includegraphics[width=\textwidth]{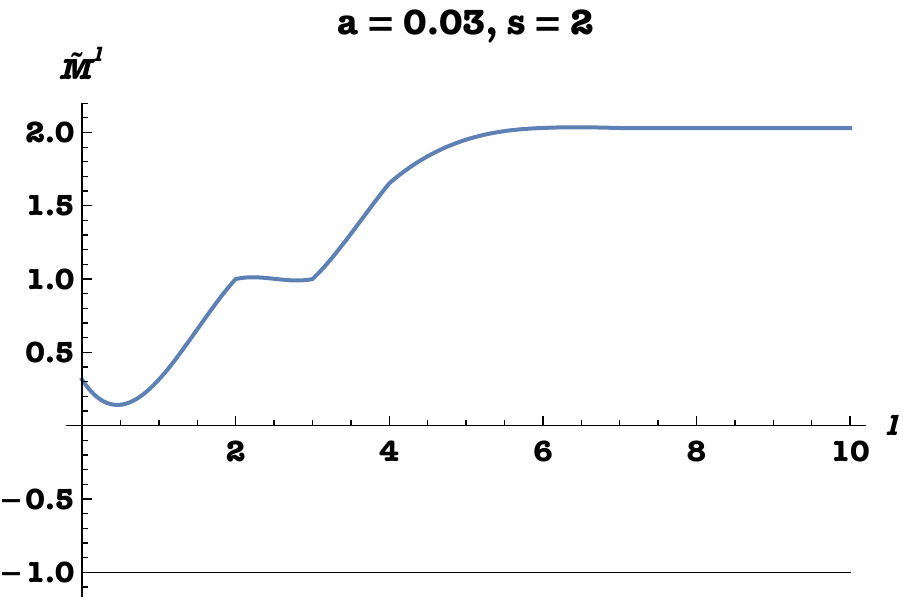}
		\caption{}
	\end{subfigure}
    \hfill
    \begin{subfigure}[b]{0.32\textwidth}
    \centering
    \includegraphics[width=\textwidth]{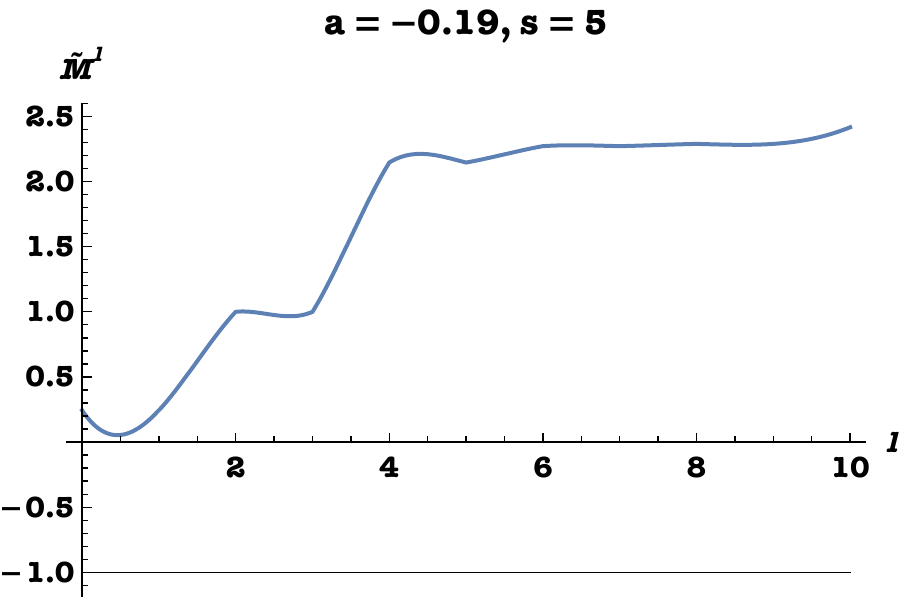}
    \caption{}
    \end{subfigure}
	\caption{$\tilde{\mathcal{M}}^{\ell}(s_i,a)$ vs $\ell$~for various values of $a$.}
	\label{LSDstring}
\end{figure}
\noindent We consider, where $\tilde a^{(i)}_J(s_1)=\frac{a^{(i)}_J(s_1)}{s_1^2}$ and $f^{(i)}_J(\xi)$ relevant for spin 2 have been defined around eq \eqref{gravitonbell}.Note that we have dropped the null constraints since a physical amplitude satisfies that by default. We have found that the positivity condition is satisfied for our range ``a". We present our analysis in the figure \ref{LSDstring} above for different values of $a$. The plots show $\tilde{\mathcal{M}}^{\ell}(s_i,a)=\frac{\mathcal{M}^{\ell}(s_i,a)}{\mathcal{M}^{\ell \le 2}(s_i,a)}$ as a function of $\ell$ for various values of $a$. We note that firstly, as advertised, our amplitude is positive for this region of $a$ since $\tilde{\mathcal{M}}^{\ell}(s_i,a)>-1$. Secondly, we note that the maximal contribution to the amplitude occurs between spins 2 and 4 which is consistent with our observation that our analysis in subsection \ref{gravitonboundTR}, indicated a spin 2 dominance. 

\end{subappendices}

%% file: mellin_bound.tex
\chapter{{Froissart-Martin Bound on Holographic S-Matrix} }\label{mellinbound}	
\section{Introduction}
\begin{figure}[!htb]
	\centering
	\begin{tikzpicture}[scale=0.65]
		\filldraw[
		color=Tomato4,
		fill=OliveDrab1!30,
		line width=1mm,
		] (-5,0)circle[radius=3 cm];
		\shade[shading=radial, outer color=Wheat1, inner color=Salmon1] (-5,0) circle (1.3 cm);
		\draw
		[
		line width=0.3mm,
		decoration={markings, mark=at position 0.40 with {\arrowreversed[line width=1 mm]{stealth}}},
		postaction={decorate}
		] (-4.081,0.919)--(-2.879,2.121);
		\draw
		[
		line width=0.3mm,
		decoration={markings, mark=at position 0.40 with {\arrowreversed[line width=1 mm]{stealth}}},
		postaction={decorate}
		] (-4.081,-0.919)--(-2.879,-2.121);
		\draw
		[
		line width=0.3mm,
		decoration={markings, mark=at position 0.40 with {\arrowreversed[line width=1 mm]{stealth}}},
		postaction={decorate}
		] (-5.919,0.919)--(-7.121,2.121);
		\draw
		[
		line width=0.3mm,
		decoration={markings, mark=at position 0.40 with {\arrowreversed[line width=1 mm]{stealth}}},
		postaction={decorate}
		] (-5.919,-0.919)--(-7.121,-2.121);
		\draw (-7.562,2.562) node{1};
		\draw (-7.562,-2.562) node{2};
		\draw (-2.638,2.562) node{3};
		\draw (-2.638,-2.462) node{4};
		\draw (2.738,2.562) node{1};
		\draw (2.738,-2.562) node{2};
		\draw (7.562,2.562) node{3};
		\draw (7.362,-2.362) node{4};
		\draw (-4.8,-4) node{Scattering in AdS};
		\draw (0.5,0.7) node{$\frac{R}{\ell_P}\to \infty$};
		\draw [-stealth,color=DodgerBlue3,{line width=0.7mm}] (-1.3,0)--(2.3,0);
		\shade[shading=radial, outer color=Wheat1, inner color=Salmon1] (5,0) circle (1.3 cm);
		\draw
		[
		line width=0.3mm,
		decoration={markings, mark=at position 0.40 with {\arrowreversed[line width=1 mm]{stealth}}},
		postaction={decorate}
		] (5.919,0.919)--(7.121,2.121);
		\draw
		[
		line width=0.3mm,
		decoration={markings, mark=at position 0.40 with {\arrowreversed[line width=1 mm]{stealth}}},
		postaction={decorate}
		] (5.919,-0.919)--(7.121,-2.121);
		\draw
		[
		line width=0.3mm,
		decoration={markings, mark=at position 0.40 with {\arrowreversed[line width=1 mm]{stealth}}},
		postaction={decorate}
		] (4.081,0.919)--(2.879,2.121);
		\draw
		[
		line width=0.3mm,
		decoration={markings, mark=at position 0.40 with {\arrowreversed[line width=1 mm]{stealth}}},
		postaction={decorate}
		] (4.081,-0.919)--(2.879,-2.121);
		\draw (5,-4) node{Scattering in Flat space};
	\end{tikzpicture}
	\caption{Transition from AdS to Flat Space}\label{fllim}
\end{figure}
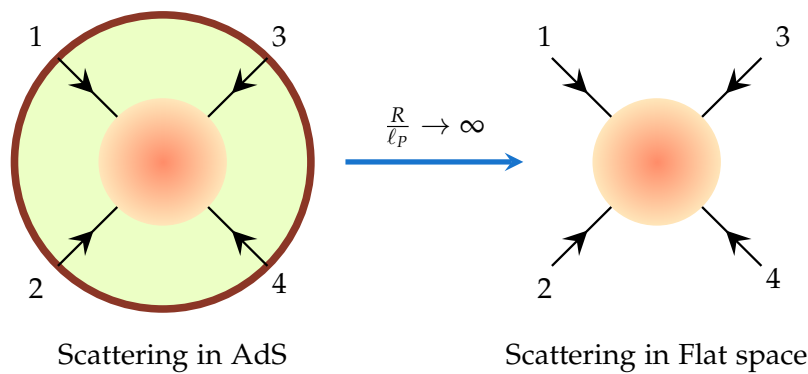
In the previous two chapters, we explored the space of S-matrices with the aid of various physical requirements. Now is it possible to give a concrete construction of such an abstract S-matrix? Lagrangian quantum field theory is obviously the most familiar way to achieve this, but usually, Lagrangian quantum field theory will be perturbative. Is there a non-perturbative way to construct S-matrix? One such answer to this question is provided by the conjectural $AdS/CFT$ holographic duality. In $AdS/CFT$ duality

Penedones \cite{Penedones:2010ue} conjectured that the Mellin representation of CFT correlation functions, also known as \emph{Mellin amplitude}, defines dual bulk flat space S-matrix via the vanishing curvature, or equivalently the large $AdS$ radius, limit of the $AdS/CFT$ correspondence. This limit is dubbed as \emph{flat space limit}. Penedones's conjectured construction gave a prescription for constructing flat space amplitudes for the scattering of massless particles. A similar prescription was conjectured by Paulos \emph{et al.} \cite{Paulos:2016fap} for obtaining amplitudes for the scattering of massive particles from the appropriate Mellin amplitudes. These conjectures were proved over multiple works \cite{Paulos:2011ie, Fitzpatrick:2011dm, Fitzpatrick:2011hu, Liflat, KomatsuSmat}. 
The key strength of this connection is that the analyticity properties of the S-matrix emerge from the analytic properties of the CFT Mellin amplitudes. Further, the unitarity is encoded in the OPE structure of the CFT. Due to the presence of extra symmetries, CFT is much more constrained than the flat space quantum field theories. CFT Mellin amplitudes are better understood non-perturbatively due to the impressive success of the conformal bootstrap programme, see \cite{slavarev, ASrev} for extensive reviews of this rather large body of work,  which skillfully exploits the extra constraints of CFTs. An S-matrix constructed out of flat space limits of CFT Mellin amplitudes will be called a \emph{holographic S-matrix}. Note that in constructing an S-matrix holographically, we have traded off one abstract formulation with another abstract formulation, abstract CFTs,  which we have better control over due to the powerful techniques of the conformal bootstrap programme. 

We now discuss the key points of this connection between  Mellin amplitudes and flat space scattering amplitudes. We will review the details in a while. The Mellin amplitude can be thought of as a representation of scattering in anti-de Sitter space. The CFT lives in $d$ spacetime dimensions while the scattering happens in $d+1$ dimensional AdS space. If the radius of curvature of the AdS space $R$ is much bigger than the Planck length $\ell_P$, then effectively, the cosmological constant is zero, and the scattering can be thought to be occurring in flat spacetime--see fig.\ref{fllim}. Note that the flat space scattering amplitude lives in one spacetime dimension higher than the CFT. In particular, S-matrix elements in Minkowski space $\mathbb{M}^{d,1}$ are obtained as the flat space limit of Mellin amplitudes in CFTs on $d$-dimensional Euclidean space, $\mathbb{E}^d$. We will be interested in $2-2$ flat space scattering of identical massive scalar particles. Such an amplitude can be obtained as the flat space limit of Mellin amplitude for $4-$point correlator of identical scalar primaries $\f$ having dimension $\D_\f$. The flat space limit corresponds to taking $AdS$ radius  $R_{AdS}$ much bigger than the Planck length $\ell_P$. Now the  $AdS/CFT$ dictionary gives the relation between the mass $m$ of the flat space particle and $\D_\f$ of the corresponding scalar primary: $m R_{AdS}\sim \Delta_\phi$. Them the flat space limit amounts to taking $\D_\f$ \emph{parametrically large} in the above relation with  $R_{AdS}/\ell_P\gg 1$ so as to keep $m$ fixed and finite.  

With the aid of this dictionary, various analyticity properties of flat space S-matrix can be understood in terms of analytic structures of the CFT \cite{Fitzpatrick:2011hu}. A natural question then arises as to how the various analytic structures of flat space S-matrix are encoded in a corresponding CFT Mellin amplitude. In this chapter, we attempt to answer this question for the Froissart-Martin bound satisfied by the total scattering cross-section for elastic $2-2$ scattering of massive particles. We prove a Froissart-Martin bound for the holographic S-matrix.  Let us tabulate key rationales of the derivation.

%

\begin{enumerate}
	\item As in the Froissart bound derivation, we start with the absorptive (imaginary) part of the amplitude. However, unlike in flat space where there is a cut in the complex $S$ plane, in the Mellin variable $s$, we have an infinite set of poles in the Mellin amplitude. In the imaginary part, these poles become a sum of delta functions \cite{Fitzpatrick:2011hu}.
	\item The flat space amplitude is expanded in terms of the Gegenbauer polynomials, which are the generalizations of the Legendre polynomials. The polynomials are indexed by the spin quantum number $\ell$. The Mellin amplitude is expanded in terms of the so-called Mack polynomials which are indexed by a spin quantum number $\ell$, as well as the dimension $\Delta$ of the exchanged conformal primary. In the flat space limit to be described below, however, the Mack polynomials go over to the Gegenbauer polynomials.
	\item In the flat space derivation, an assumption is made about the polynomial boundedness of the amplitude inside the Martin ellipse. We will make a similar assumption about the Mellin amplitude\footnote{We made explicit checks using the mean field theory OPE coefficients for the validity of this assumption.}. This assumption effectively leads to the sum over $\ell$ being cut-off. Typically the cut-off $L$ takes the form
	$$
	L\propto \sqrt{s}\ln \frac{s}{s_0}\,,
	$$
	with the proportionality constant depending on the power assumed in the polynomial boundedness.
	\item The key difference will be that unlike the partial wave unitarity bounds that are assumed in the flat space derivation, we will have to contend with the sum over $\Delta$. Here, we will make use of the fact that in order to reproduce the identity exchange in the crossed channel, the operator product expansion coefficients governing the sum over $\Delta$ are controlled by the so-called (complex) Tauberian theorems\cite{Mukhametzhanov:2018zja,Pappadopulo:2012jk}. 
	A final point to mention is that we will be dealing with averaged bounds, following, for example, \cite{Yndurain:1970mz}. This is because we will be dealing with distributions, and it makes more sense to talk about integrated quantities. This will turn out to be essential in making the range of twists in the $\Delta$ sum on the CFT side to be finite.
\end{enumerate}

The last point creates a very important difference in the form of the Froissart$_{AdS}$ bounds we will find. Summarizing our results:
\begin{itemize}
	\item We find that the coefficient in front of the bound, i.e., $\pi/\m^2$ is {\it exactly} this for 4 flat spacetime dimensions {\it except} that $\m$ here can also be the mass of the external particle while the mass parameter present  in the original Froissart bound formula 
	\be 
	\sigma_{tot}<\frac{\pi}{\m^2} \log^2 \frac{s}{s_0}
	\ee 
	is the mass of the lightest exchanged particle in $t$ channel, usually taken to be the pion mass.
	\item For $CFT_{d}$ with $d\leq 4$, the form of the bound is the same as flat space higher dimensional generalizations. However, for $d=2$ the coefficient in front is {\it lower} than the flat space derivation, for $d=3$ it is {\it identical} as mentioned above, while for $d=4$ it is {\it bigger}. 
	\item For $d>4$ the form of the bound is different, as we will discuss in our derivation below. This has important implications for the form of the polynomial boundedness. What happens is that first one assumes that the amplitude is $|\mm(s,t)|<s^n$ bounded for some unspecified $n$ for $t$ inside the Martin ellipse. Then, this leads to the Froissart bound (the $n$ enters in the coefficient in front). Suppose that the result for the absorptive part is $\ma_M(s,0)<c s^a \ln^{d-1} s/s_0$. At this stage, one can argue using the Phragmen-Lindeloff theorem \cite{Martin:1967zzd} that $n\leq \floor{a}+1$, where $\floor{\,\,}$ denotes the usual floor function. In the flat space derivation $\floor{a}=1$. However, we will find $\floor{a}> 1$, and hence $n>2$ for $d>6$.
\end{itemize}

The chapter is organized as follows. We give a flashing recapitulation of flat space results in section \ref{flatrev}. In section \ref{Mellinrev}, we review Mellin amplitudes in CFTs, including the flat space limit reviewing essential results from \cite{Paulos:2016fap}. In section \ref{bounder}, we turn to bound the absorptive part of the Mellin amplitude. We derive dispersion relations in section \ref{Mellindisp} and constrain the number of subtractions needed. There are appendices supplementing the calculations in the main text.
\vspace{ 0.5 cm}

{\it Warning: We will use the convention $h=d/2$ in many places. This unfortunate convention is somewhat standard in the CFT literature.}

\section{Flat space Froissart-Martin bound: A quick look back}\label{flatrev}
We have already discussed Froissart-Martin bound for $2-2$ elastic scattering in $4$ spacetime dimensions. The analysis can be generalized to  arbitrary spacetime dimensions in starightforward manner. In fact the analyticity properties, namely the polynomial boundedness and the existence of the Martin ellipses, remain unaltered in general spacetime dimensions \cite{Maha1, Maha2}. We write the general dimensional expressions for the Froissart-Martin bound. We will use a particular normalization. As before, we will be considering scattering of identical massive bosons of mass $m$ in $d+1$ spacetime dimensions.   We will, however, use a different normalization than that in chapter \ref{smatrev}. We will stick to the normalization of \cite{Chaichian:1987zt}. 

\subsection{Partial wave expansion} The $d+1$ dimensional flat space scattering amplitude admits a partial wave expansion in terms of generalized spherical functions spanning the representation space of $SO(d,1)$ corresponding to the unitary irreducible representations of maximally compact subgroup $SO(d)$.  The $2-2$ scattering amplitude $\mt(s,t)$ admits the following $S-$channel partial wave expansion
in a basis of Gegenbauer polynomials:
\be \label{partialF1}
\mt(s,t)=\f(s)\sum_{\substack{\ell=0\\\ell~\text{even}}}^\infty \widetilde{f}_\ell(s)~
\frac{C_\ell^{(h-1)}(1)}{N_\ell^{(h-1)}}
~C_\ell^{(h-1)}\left(Z_{s}=1+\frac{2s}{s-4m^2}\right),
\ee  
with $C_\ell^{(h-1)}$ being the Gegnbauer polynomial and $\{\widetilde{f}_\ell(s)\}$ are the \emph{partial wave coefficients}. Here,
\begin{align}\label{flatnorm}
	\begin{split}
		h&=\frac{d}{2},\\
		N_\ell^{(h-1)}&=\frac{2^{3-2h}\p\G(\ell+2h-2)}{\ell!(h-1+\ell)\G^2(h-1)},\\
		\f(s)&=\G\left(h-\frac{1}{2}\right)(16\p)^{\frac{2h-1}{2}}s^{\frac{3-2h}{2}},\\
		C_\ell^{(h-1)}(1)&=\frac{(2h-2)_\ell}{\G(\ell+1)},
	\end{split}
\end{align}  
where $(a)_b$ denotes the Pochhammer symbol.

The partial wave coefficients $\{f_\ell(s)\}$ satisfy the unitarity bound 
\be 
0\le |\widetilde{f}_\ell(s)|^2\le \text{Im}[\widetilde{f}_\ell(s)]\le 2\,.
\ee
\subsection{Martin analyticity and polynomial boundedness } 
 For $s$ in the cut plane, the absorptive part of the scattering amplitude, $\ma_s(s,t)$, is analytic inside Martin ellipse, $\me_M$ in complex $t$ plane. The ellipse has foci at $t=0,\, 4m^2-s$ and the right extremity at $t=R>0$, where $R$ is \emph{independent} of $s$. It can be proved that $R=4m^2$ \cite{Martin1}.
 
 We further have the polynomial boundedness satisfied by the scatterign amplitude. The $s-$channel polynomial boundedness reads 
 that there is a positive integer $N$ such that 
 \be 
 |\ma_s(s, t=t_0)|<c\, s^N,
 \ee where $t_0$ lies in the interior of the Martin ellipse $\me_M$. However, at the final stage we can take $t_0\to R$.  A more rigorous way to implement the polynomial boundedness is in the form of the following integral convergence criterion 
 \be 
 \int_{4m^2}^\infty \frac{d\bar{s}}{\bar{s}^{N+1}}\, \ma_s(s, t=t_0)<\infty.
 \ee
 \subsection{The Froissart-Martin bound}
 We can carry out the the same analysis as done in \ref{Yndurain} to obtain the general spacetime dimensional high energy bound on the absorptive part of the forward amplitude 
 \be \label{Imbound}
 \ma_s(s,t=0)\le 2^{4h-3}\p^{h-2}\frac{\G(h-1)\G^2(h-1)}{(2h-1)\G^2(2h-2)}\left(\frac{N-1}{\sqrt{R}\cos \varphi_0}\right)^{2h-1} s(\ln s)^{2h-1},\qquad |s|\to\infty. 
 \ee  Finally, one can obtain general dimensional Froissart-Martin bound on total scattering cross-section by making use of optical theorem.  
\section{Mellin amplitudes in CFTs and Holographic S-matrix}\label{Mellinrev}

\subsection{Definitions and conventions}
Mellin amplitudes for CFT correlators were introduced by Mack \cite{Mack:2009mi, Mack:2009gy}. In this section, we will review the analogy between conformal correlation function and scattering amplitude \cite{Penedones:2010ue,Paulos:2011ie,Costa:2012cb} via the AdS/CFT correspondence. In particular, one can consider the Mellin amplitude as the \enquote{scattering amplitude in AdS}. 

The Mellin amplitude associated with connected part of the $n$-point function of scalar primary operators is defined by, 
\be \label{Meldef1}
G(x_i)=\lan \mo_1(x_1)\dots\mo_n(x_n) \ran_c=\int[d\d] \mm(\d_{ij})~\prod_{1\le i<j\le n}\frac{\G(\d_{ij})}{(x_{ij}^2)^{\d_{ij}}},
\ee where the integral runs parallel to the imaginary axis and is to be understood in the sense of Mellin-Barnes contour integral. Conformal invariance constrains the integration variables $\{\d_{ij}\}$ to satisfy, 
\be \label{melconstr}
\d_{ij}=\d_{ji},~~~\d_{ii}=-\D_i,~~~\sum_{j=1}^n\d_{ij}=0,
\ee $\D_i$ being the scaling dimension of the operator $\mo_i$. Due to these constraints, there are $n(n-3)/2$ independent variables upon which the Mellin amplitude $\mm$ depends. Clearly, for four-point function i.e., $n=4$, the number of independent variables is $2$. 

Now let us focus on the problem at hand for which we will consider $4-$point correlator of identical scalar primaries $\f$ with dimension $\D_\f$. The reduced correlator $\mg(u,v)$ is defined by, 
\be \label{guvdef}
\lan\f(x_1)\f(x_2)\f(x_3)\f(x_4)\ran=\frac{1}{x_{12}^{2\D_\f}x_{34}^{2\D_\f}}\mg(u,v)\,,
\ee 
where $u, v$ are the conformal cross-ratios given by
$
u=\frac{x_{12}^2x_{34}^2}{x_{13}^2x_{24}^2},~v=\frac{x_{14}^2x_{23}^2}{x_{13}^2x_{24}^2}.
$
Now we define the \enquote{Mellin variables} $(s_M,t_M)$ as: 
\begin{align}\label{stdef}
	\begin{split}
		&\d_{12}=\d_{34}=\frac{\D_\f}{2}-s_M,\\
		&\d_{14}=\d_{23}=\frac{\D_\f}{2}-t_M,\\
		&\d_{13}=\d_{24}=s_M+t_M.
	\end{split}
\end{align} Note that this definition differs from the ones in \cite{Gopakumar:2018xqi} by a shift of $\D_\f/2$. The reduced correlator now has the Mellin space represenatation, 
\be \label{MElldef}
\mg(u,v)=\int_{-i\infty}^{i\infty}\frac{ds_M}{2\p i}\frac{dt_M}{2\p i} u^{s_M+\D_\f/2}v^{t_M-\D_\f/2}\m(s_M,t_M)\mm(s_M,t_M),
\ee
with 
\be \label{mellmeas}
\m(s_M,t_M)=\G^2\left(\frac{\D_\f}{2}-s_M\right)\G^2\left(\frac{\D_\f}{2}-t_M\right)\G^2(s_M+t_M),
\ee  
being a standard measure factor which has information about the double trace operators in the $N\rightarrow \infty$ limit in the context of the AdS/CFT correspondence.
\subsection{Conformal Partial Wave expansion}
Just as the flat space scattering amplitude admits a partial wave expansion in terms of Gegenbauer polynomials, the Mellin amplitude $\mm(s,t)$ admits the conformal partial wave expansion \cite{Mack:2009mi, Dolan:2011dv}. Our starting point is the Mellin space representation of the standard position space direct channel expansion \cite{Dolan:2011dv}. We are interested in the imaginary part of this which arises from the physical poles. We can either work directly with \cite{Dolan:2011dv} or a bit more conveniently, to make the pole structure manifest, following \cite{Gopakumar:2018xqi}, we can write an $s-$channel conformal partial wave expansion for the  Mellin amplitude, 
\be \label{melpart}
\mm(s_M,t_M)=\sum_{\substack{\t,\ell\\ \ell~\text{even}}}C_{\t,\ell}~f_{\t,\ell}(s_M)~\widehat{P
}_{\t,\ell}(s_M,t_M),
\ee 
where the equality is modulo some \emph{regular terms} --see appendix \ref{melpartdrv} for a derivation of how to go from the form in \cite{Dolan:2011dv} to the form in \cite{Gopakumar:2018xqi}. Here, we have defined $\t=\frac{\D-\ell}{2}$ and $\widehat{P}_{\t,\ell}(s_M,t_M)$ are the  Mack polynomials whose details we provide in appendix \ref{Mackasym1}. 
$C_{\t,\ell}$ are the squared OPE coefficients and 
\begin{align}\label{fdef}
	\begin{split}
		\scalebox{1}{
			$
		f_{\t,\ell}(s_M)=\frac{\mn_{\t,\ell}~\G^2(\t+\ell+\D_\f-h)}{\left(\t-s_M-\frac{\D_\f}{2  }\right)\G(2\t+\ell-h+1)}
		\frac{\sin^2\p\left(\frac{\D_\f}{2}-s\right)}{\sin^2\p\left(\D_\f-\t-\frac{\ell}{2}\right)}
		~_3F_2
		\left[
		\begin{matrix}
			\t-s_M-\frac{\D_\f}{2},1+\t-\D_\f,1+\t-\D_\f\\
			1+\t-s_M-\frac{\D_\f}{2},2\t+\ell-h+1
		\end{matrix}~\Big|1
		\right],$}
	\end{split}
\end{align}
with 
\be\label{Ndef}
\mn_{\t,\ell}:=2^{\ell}\frac{(2\t+2\ell-1)\G^2(2\t+2\ell-1)\G(2\t+\ell-h+1)}{\G(2\t+\ell-1)\G^4(\t+\ell)\G^2(\D_\f-\t)\G^2(\D_\f-h+\t+\ell)}.
\ee 
${}_3F_2$ is a generalized hypergeometric function. There are poles at $s_M=\tau-\frac{\D_\phi}{2}+q$ for $q\in {\mathbb Z}\geq 0$. This representation is suitable for the $s_M$ channel Witten diagram and the residues at the physical poles are identical to other standard ones used in the literature eg.\cite{Dolan:2011dv}. Note that the {\it full} Mellin amplitude includes the measure factor, which provides the $u_M$-channel poles. However, in what follows, we will be interested in averaging over positive values of $s_M$ so that these poles will not alter any of our conclusions.

\subsection{Flat space limit of the Mellin amplitude }
Now we will review the connection between the Mellin amplitude and the flat space scattering  via what is called the \enquote{flat space limit}. That such a connection exists was first conjectured in \cite{Penedones:2010ue} and was developed extensively in \cite{Fitzpatrick:2011hu,Fitzpatrick:2011dm}. In these papers, the Mellin amplitude was related to scattering amplitude of \emph{massless} particles. However, the limit  in which we are interested is the so called \enquote{massive flat space limit} and was first proposed in \cite{Paulos:2016fap}.  In this limit, the Mellin amplitude for the conformal correlator is related to the scattering amplitude for massive particles by taking the dimensions of the external operators to be parametrically large. Then via the AdS/CFT correspondence, the Mellin amplitude of conformal correlator is related to flat space scattering amplitude with external \emph{massive} particles in one higher space-time dimesnion i.e., the Mellin amplitude of a conformal correlator in $CFT_d$ is related to scattering amplitude $d+1$ dimensional flat space quantum field theory--to emphasise, the flat space QFT is not conformal. 

To understand better what we mean by parametrically large dimension, recall that in the AdS/CFT correspondence the scaling dimension of the boundary conformal operators in $CFT_{d}$, $\D_\phi$, and the mass of the corresponding dual bulk field in $AdS_{d+1}$, $m$,  are related by 
\be 
m^2R^2=\D_\phi(\D_\phi-d),
\ee $R$ being the AdS radius. Now in the flat space limit the conformal dimension $\D_\phi$ is taken to infinity along with $R$ so that $m$ remains finite i.e., 
\be \label{flat}
\lim_{\substack{\D_\phi\to\infty \\ R\to\infty}}\frac{\D_\phi^2}{R^2}=m^2.
\ee 
Since $R$ is dimensionful, we mean $R/\ell_P\gg 1$ where $\ell_P$ is the Planck length. Further, since we are taking the flat space limit to a massive theory, we also require $R/\ell_{s}$ with $\ell_s$ being the string length characterizing the string theory energy scale.
Now taking $R/\ell_P\to\infty$ takes the $AdS_{d+1}$ to $\mathbb{M}^{d,1}$. Thus, in this limit, we relate the Mellin amplitude for $CFT_d$ correlator to scattering amplitude of massive particles in flat spacetime $\mathbb{M}^{d,1}$. Now we turn to the explicit formulae relating the flat spacetime scattering amplitude and the Mellin amplitude.

The  $n-$point conformal correlator and $n-$particle scattering amplitude are related by:
\begin{align} \label{fl2}
	\begin{split}
		\left(m_{1}\right)^{a} \mathcal{T}\left(\{p_i^\n\}\right)
		=\lim _{\substack{\Delta_{i} \rightarrow \infty\\ R\rightarrow\infty}} \frac{\left(\Delta_{1}\right)^{a}}
		{\mathcal{N}} \mm\left(\d_{i j}=\frac{\Delta_{i} \Delta_{j}+R^{2}~ p_i^\n  p_{j\n}}{\Delta_{1}+\cdots+\Delta_{n}}+O\left(\Delta_{1}^0\right)\right),
\end{split}\end{align}
with
\be  
\mathcal{N}:=\frac{1}{2} \pi^{\frac{d}{2}} \Gamma\left(\frac{\sum \Delta_{i}-d}{2}\right) \prod_{i=1}^{n} \frac{\sqrt{\mathcal{C}_{\Delta_{i}}}}{\Gamma\left(\Delta_{i}\right)},\qquad\mathcal{C}_{\Delta} := \frac{\Gamma(\Delta)}{2 \pi^{\frac{d}{2}} \Gamma\left(\Delta-\frac{d}{2}+1\right)},\qquad a:=\frac{n(d-1)}{2} -d-1.
\ee
Here $\mathcal{T}\left(\{p_i^\n\}\right)$ is the $n$-particle $\mathbb{M}^{d,1}$ scattering amplitude with external Minkowski momenta $\lbrace p_i^\n\rbrace$. Note however, these $\lbrace p_i^\n\rbrace$ have momenta interpretation  after going to flat space amplitude only. On the Mellin amplitude side they are just $n$  vectors in $\mathbb{M}^{d,1}$ with the restriction,
\be\label{momvec}
\sum_{i=1}^{n}p_i^\n=0,\qquad\qquad p_i^\n p_{i\n}=-\frac{\D_i (\D_i-d)}{R^2}.
\ee
These restrictions are there for consistency with the momentum interpretation of $\lbrace p_i^\n\rbrace$ in the flat space limit. Note that the vector norm and inner product are usual $\mathbb{M}^{d,1}$ norm and inner product respectively. The parameterization holds for $\d_{i j} $ with $i \neq j $. $\d_{i i}$ should still be set to $-\D_i$ explicitly. Now the consistency with the third constraint in eq.\eqref{melconstr} can be met by adding following finite term
\be  \label{finite}
\frac{d}{n-2}\left[\frac{\Delta_{i}+\Delta_{j}}{\Delta_{1}+\cdots+\Delta_{n}}-\frac{1}{n-1}\right].
\ee

For the case of the 4-point conformal correlator of identical scalar primaries $\f$ with scaling dimensions $\D_\f$ and the corresponding flat space mass $m$ we have
\be \label{fl3}
{m^{a}~\mt(s,t)=\lim_{\dphi\to\infty}\frac{(\dphi)^a}{\mn}~\mm(s_M,t_M).}
\ee
with
\be  \label{Nexp}
\mathcal{N}
:=
\frac{\G(2\dphi-h)}{ 8\p^h\G^2(\dphi)\G^2(\dphi-h+1)},\qquad a:=2h-3.
\ee
Here  $(s_M,t_M)$ are defined as in eq.\eqref{stdef} and the flat space Mandelstam variables $(s,t)$ are those defined by $s:= -(p_1+p_2)^2,\, t=-(p_1+p_3)^2,\, \sum_{i}p_i=0$. The precise relation between  the flat space Mandelstam variables $(s,t)$ and $(s_M,t_M)$ is given by: 
\be\label{sSrel1}
s_M(t_M)=\frac{R^2}{8\dphi}s(t)+\frac{d}{6}.
\ee The term $d/6$ is the finite term eq.\eqref{finite} which we can ignore for all practical purpose in the flat space limit and thus we are going to use for all practical purposes, 
\be \label{sSrel2}
s_M(t_M)=\frac{R^2}{8\dphi}s(t)\,.
\ee On the same footing we use  
\be 
\widehat{u}_M=\frac{R^2}{8\dphi}u
\ee so that we can consider, 
\be \label{melsum}
s_M+t_M+\widehat{u}_M=\frac{\dphi}{2}.
\ee Note the consistency of the constraints eq.\eqref{melsum} and eq.\eqref{Mandelsum}. Then drawing parallels with the flat space understanding, the \enquote{physical domain} for $s_M-$channel of the Mellin amplitude in the current perspective is defined to be 
\be
s_M>\dphi/2;~~t_M,\widehat{u}_M<0.
\ee 
Since the flat space limit is $R\to\infty$, $R$ being the AdS radius, we would like to have a $1/R$ expansion around the flat space.  On dimensional grounds, in terms of the flat space $s$, we expect the dimensionless quantities $s/(m^4 R^2)$ and $1/(s R^2)$ to be small in order to allow a $1/R$ expansion. Here we are assuming only even powers of $m$ entering such expansion. In terms of $s_M$, this gives $(2\D_\phi)^3\gg s_M\gg 2\D_\phi\gg 1$ which is what is going to be used below.

\subsection{Absorptive part of Mellin amplitude}
The main goal in this work is to extract information about Mellin amplitude for conformal field theory by exploiting the structural analogy  between the former and the flat space scattering amplitude.  The absorptive part of scattering amplitude is the imaginary part of scattering amplitude.  In this spirit, we define the \emph{absorptive part of the Mellin amplitude} as,

\be \label{absdef}
\ma_M(s_M,t_M)=\text{Im}_{s_M}~\mm(s_M,t_M)=\sum_{\substack{\t,\ell\\ \ell~\text{even}}}C_{\t,\ell}~\text{Im}_{s_M}[f_{\t,\ell}(s_M)]~\widehat{P}_{\t,\ell}(s_M,t_M),
\ee where we have defined  
\be 
\text{Im}_{s_M}[g(s_M)]:=\lim_{\ve\to0}\frac{g(s_M+i\ve)-g(s_M-i\ve)}{2i}.
\ee 
Observe that the imaginary part comes only from the fucntion $f_{\t,\ell}$ because for unitary theories $C_{\t,\ell}\in\mathbb{R}^+$ and $\widehat{P}_{\t,\ell}(s_M,t_M)$ does not have poles. The imaginary part of the function $f_{\t,\ell}$ comes in a distributional sense at the pole locations which we will see in a while. Since  $\ma_M(s_M,t_M)$ is  a distribution,  we should handle quanitites involving integrals over $\ma_M(s_M,t_M)$. Towards that end, we define the following quantity \cite{Yndurain:1970mz} 

\be \label{fwddef}
\bar{\ma}_M(s_M,t_M)\equiv \frac{1}{s_M-\frac{\D_\f}{2}}\int_{\frac{\D_\f}{2}}^{s_M}ds'_M\,\ma_M(s'_M,t_M).
\ee 	
This can be viewed as an averaged absorptive Mellin amplitude.  For $d=3$, which leads to Froissart bounds for 4d flat space, the choice of the lower limit makes no difference.  We will introduce the quantity $x=1+2t_M/(s_M-\D_\phi/2)$. We will consider the problem of obtaining the asymptotic upper bound on this quantity in the limit $s_M\to\infty$ for two different scenarios: one is the \enquote{forward} limit i.e.,  $x\rightarrow 1$ and the other one is the \enquote{non-forward} limit i.e., with $x\neq 1$. 

\subsection{Polynomial boundedness of Mellin amplitude}\label{polymellin}
Now to proceed further, we need to assume something more about the analytic structure of the Mellin amplitude. Recall that the assumption of polynomial boundedness of the flat space scattering amplitude is extremely crucial in deriving the Froissart-Martin bound. In fact, it will be no exaggeration to say that the Froissart-Martin bound would not have existed without this additional boundedness property of the scattering amplitude. We will assume a similar polynomial boundedness for Mellin amplitudes as well. In close analogy with flat space case we assume the following polynomial boundedness condition upon $\ma_M(s_M,t_M)$: there exists at least an $n\in\mathbb{Z}^+$ such that the integral,
\be 
\mathfrak{a}_{n,\rho}:=\int_{\frac{\D_\f}{2}}^\infty\frac{d\bar{s}_M}{\bar{s}_M^{n+1}} \ma_M(\bar{s}_M,\rho \frac{\dphi}{2})
\ee 
exists. For our purpose we can assume  $\r\in\mathbb{R}^+$. In the flat space limit, this corresponds to $t=4\r m^2$ with $m$ being the mass of the external particle. In the flat space Froissart bound, one typically chooses $\r=\mu^2/m^2$ with $\mu\leq m$ being the mass of the lightest exchange in the crossed channel.
\subsection{The structure of $\ma_M(s_M,t_M)$}
We will now bound the quantity $\bar{\ma}_M(s_M)$. To do so, we will need to know the structure of $\ma_M(s_M,t_M)$. The most non-trivial component of the  same is $\text{Im}[f_{\t,\ell}]$.  From eq.\eqref{fdef} one has that  
\begin{align}\label{3F2pole}
	\begin{split}
		_3F_2
		\left[
		\begin{matrix}
			\t-s_M-\frac{\D_\f}{2},1+\t-\D_\f,1+\t-\D_\f\\
			1+\t-s_M-\frac{\D_\f}{2},2\t+\ell-h+1
		\end{matrix}~\Big|1
		\right]
		=\mathlarger{\sum}_{q=0}^{\infty}\,\frac{(1+\t-\D_\f)_q^2}{q!(2\t+\ell-h+1)_q}\frac{\t-s_M-\D_\f/2}{q+\t-s_M-\D_\f/2}\,.
	\end{split}
\end{align}
Clearly, we see that we have poles in the $s_M-$plane at the locations $s_M=\t-\D_\f/2+q$ for $q\ge 0$. Now we know that at the poles the imaginary part comes as a Dirac-delta distribution, i.e.
\be 
\text{Im.} \frac{1}{x-a}=\lim_{\ve\to0}\frac{1}{2i}\left(\frac{1}{x-a+i\ve}-\frac{1}{x-a-i\ve}\right)=-\p\d(x-a)\,.
\ee 
Note that this is a distributional statement and, hence, this equality \emph{holds under the integrals}. Specifically if  $f(x)$ is a continuous test function over $\mathbb{R}$ then we have 
\be 
\int_{-\infty}^{\infty}dx~f(x) ~\text{Im.}\left(\frac{1}{x-a}\right)=-\p f(a).
\ee 
So in this sense we can write
\begin{align}
&\text{Im.}~
	_3F_2
	\left[
	\begin{matrix}
		\t-s_M-\frac{\D_\f}{2},1+\t-\D_\f,1+\t-\D_\f\\
		1+\t-s_M-\frac{\D_\f}{2},2\t+\ell-h+1
	\end{matrix}~\Big|1
	\right]\nn\\
=&-\p\,\left(s_M+\frac{\D_\f}{2}-\t\right)
\mathlarger{\mathlarger{\sum}}_{q=0}^{\infty}\,\,\,
\frac{(1+\t-\D_\f)_q^2}{q!(2\t+\ell-h+1)_q}\d\left(-q-\t+s_M+\frac{\D_\f}{2}\right).
\end{align} 
Collecting everything together we have, 
\begin{align} 
\begin{split} \label{Imf}
\text{Im.}[f_{\t,\ell}(s_M)]&=\p\mn_{\t,\ell}\frac{\G^2(\t+\ell+\D_\f-h)}{\G(2\t+\ell-h+1)}\frac{\sin^2\p\left(\frac{\D_\f}{2}-s_M\right)}{\sin^2\p\left(\D_\f-\t-\frac{\ell}{2}\right)}\\
&\hspace{5 cm}\times\mathlarger{\mathlarger{\sum}}_{q=0}^{\infty}\,\,
\frac{(1+\t-\D_\f)_q^2}{q!(2\t+\ell-h+1)_q}\d\left(-q-\t+s_M+\frac{\dphi}{2}\right).
\end{split}
\end{align} 
We would like to note further that for $q=0$ corresponds to the contributions from the primary while the $q\neq 0$ corresponds to that coming from the descendants. In appendix \ref{primary}, we consider bounds on the primary contribution separately. This exercise is instructive, although the bounds thus obtained are exponentially smaller for large $\D_\phi$ compared to the full consideration in the next section.
\section{Bounds}\label{bounder}
\subsection{Obtaining the Froissart$_{AdS}$ bound: \enquote{Forward Limit} }
We start with the expression for the conformal partial wave expansion of the Mellin amplitude as defined in eq.\eqref{melpart}.  Further making use of eq.\eqref{3F2pole}, we can write the meromorphic structure of the Mellin amplitude as following, 
\be 
\mm(s_M,t_M)=-\sum_{\t,\ell}C_{\t,\ell}~\mn_{\t,\ell}
\G(2\dphi+\ell-h)
\frac{\sin^2\p\left[\frac{\dphi}{2}-s_M\right]}{\sin^2\p\left[\dphi-\t-\frac{\ell}{2}\right]}\Mack_{\t,\ell}(s_M,t_M)\left(\sum_{q=0}^{\infty}\frac{W_q}{s_M+\frac{\dphi}{2}-\t-q}\right),
\ee
with 
\be \label{wq}
W_q:=\frac{\G^2(\t+\ell+\dphi-h)}{\G(2\dphi+\ell-h)}\frac{(1+\t-\dphi)_q^2}{\G(q+1)\G(2\t+\ell-h+1+q)}\,.
\ee 
Now we will investigate a very specific limit. We will particularly look into the limit when $\t\gg1,~\dphi\gg1$. 
The flat space limit makes it necessary to consider $\dphi\gg 1$. Why we are  considering $\t\gg 1$ will become clear in a moment. 
We will also consider $\ell\gg 1$. The last assumption is for now a working assumption which will be justified in due course\footnote{In particular, the working assumption that, $\ell\gg1$ has really \emph{nothing} to do with flat space limit.}.

Now in the limit that $\dphi\gg1,~ \t\gg1$ the residue function $W_q$ is peaked around $q=q_{\star}\sim O(\t)$ Such an observation was first made in \cite{Paulos:2016fap}. In fact, in this limit we can approximate the residue by a Gaussian function\footnote{Here, we would like to mention that, this approximated expression is obtained by implicitly considering $\dphi\sim\t$ along with $\dphi\gg1,\,\,\t\gg1$.\label{dphitau} }, 
\be\label{flatpeak} 
W_q\approx ~\frac{1}{\sqrt{2\p}(\ell+2\dphi)\d q}~e^{-\frac{(q-q_\star)^2}{2\d q^2}},
\ee with 
\begin{align}\label{qdq}
	\begin{split}
		q_\star&=\frac{(\t-\dphi)^2}{\ell+2\dphi},\\
		\d q^2&=\frac{(\t-\dphi)^2(\ell+\t+\dphi)^2}{(\ell+2\dphi)^3}.
	\end{split}
\end{align}
From this above expression, note that while $q_\star\sim O(\t)$ in the limit of large $\t$ one has $\d q\sim O(\sqrt{\t})$ in the same limit. This suggests that in the limit $\t\to\infty$ we can in fact consider the above Gaussian as a Dirac Delta function to leading order. To see this explicitly, we introduce the \enquote{normalized variable}, 
\be 
\bar{q}=\frac{q}{q_\star}.\ee Now we define 
\be 
\e:=\left(\frac{\d q}{q_\star}\right)^2.
\ee 
In these new variables $(\bar{q}, \e)$ we have 
\be 
W_q\approx \frac{1}{q_\star(\ell+2\dphi)}~\frac{e^{-\frac{(\bar{q}-1)^2}{2\e}}}{\sqrt{2\p\e}}.
\ee Further note that  
\be 
\frac{\d q}{q_\star}\sim O(\t^{-1/2}),~~~~\t\to \infty,
\ee which further implies the equivalence of the limits $\t\to\infty$ and $\e\to0$. Thus, in this limit, 
\be 
\lim_{\e\to0}W_q\approx \frac{1}{q_\star(\ell+2\dphi)}~\lim_{\e\to0}\frac{e^{-\frac{(\bar{q}-1)^2}{2\e}}}{\sqrt{2\p\e}}=\frac{1}{q_\star(\ell+2\dphi)}~\d(\bar{q}-1)=\frac{1}{\ell+2\dphi}~\d(q-q_\star).
\ee Now we can use this to write the $q-$sum as
\be 
\sum_{q=0}^{\infty} \frac{W_q}{s_M+\frac{\dphi}{2}-\t-q}\approx \int dq\,\frac{W_q}{s_M+\frac{\dphi}{2}-\t-q}\approx \frac{1}{(\ell+2\dphi)\left(s_M+\frac{\dphi}{2}-\t-q_\star\right)}.
\ee Since we are considering the limit $\t\gg1$ and $s\gg\t\gg 1$ we can now use the Gegenbauer asymptotic of the Mack polynomials which is worked out in appendix \ref{Mackasym1}. Using this, we have
\begin{align}
	\begin{split} \label{mst}
\mm(s_M,t_M)\approx& -\mathlarger{\mathlarger{\sum}}_{\t,\ell}\left[ C_{\t,\ell}\mn_{\t,\ell} 
\frac{\G(2\dphi+\ell-h)}{2\dphi+\ell}
\frac{\sin^2\p\left[\frac{\dphi}{2}-s_M\right]}{\sin^2\p\left[\dphi-\t-\frac{\ell}{2}\right]}\right.\\
&\left.\hspace{4 cm}\times\left(\frac{s_M}{8}\right)^\ell \frac{\G(\ell+1)}{(h-1)_\ell}C_\ell^{(h-1)}(x)\left(\frac{1}{s_M+\frac{\dphi}{2}-\t-q_\star}\right)\right]
\end{split}\,,
\end{align} with 
\be 
x=1+\frac{2t}{s_M-\dphi/2}.
\ee Recalling
\be 
\ma_M(s_M,t)=\text{Im}_{s_M}\mm(s_M,t_M),
\ee we have 
\begin{align}
\begin{split} \label{Amst}
\ma_M(s_M,t_M)&\approx \p\mathlarger{\mathlarger{\sum}}_{\t,\ell}\,\,\left[ C_{\t,\ell}\mn_{\t,\ell}  \frac{\G(2\dphi+\ell-h)}{2\dphi+\ell}\frac{\sin^2\p\left[\frac{\dphi}{2}-s_M\right]}{\sin^2\p\left[\dphi-\t-\frac{\ell}{2}\right]}\right.\\
&\left.\hspace{4 cm}\times\left(\frac{s_M}{8}\right)^\ell \frac{\G(\ell+1)}{(h-1)_\ell}C_\ell^{(h-1)}(x)~\d\left(s_M+\frac{\dphi}{2}-\t-q_\star\right)\right].
\end{split}
\end{align}
Ultimately we are interested in quantities which are integrals of $\mathfrak{a}_{n,\dphi}$ and $\bar{\ma}_M(s)$. This integral over $s$ effectively truncates the $\t-$ sum due to presence of the Dirac delta function. As a consequence of this we have the following  expression, obtained in the forward limit $x\rightarrow 1$\footnote{In obtaining eq.\eqref{amb1}, we have made explicit use of the fact that, operators with  even spin ($\ell$), only, gets exchanged in the OPE channels of identical scalars. We have also used the fact that, $q_\star$ is an integer. Using these, one finds that the term $\sin^2\p[\dphi/2-s]/\sin^2\p[\dphi-\t-\ell/2]$ becomes unity on doing the $s$ integral in obtaining eq.\eqref{amb1}.}
\be\label{amb1}
	\bar{\ma}_M(s)\approx \frac{2\p}{2s_M-\dphi}\sum_{\substack{\ell\\\ell~\text{even}}} \frac{(2h-2)_\ell}{(h-1)_\ell}\frac{\G(2\dphi+\ell-h)}{2\dphi+\ell}\sum_{\t=\dphi}^{\t_\star}\left(\frac{1}{8}\right)^\ell C_{\t,\ell}~\mn_{\t,\ell}
	\left(\t+q_\star-\frac{\dphi}{2}\right)^\ell.
\ee

where $\t_\star$ satisfies, 
\be 
\t_\star+q_\star(\t_\star)=s_M+\frac{\dphi}{2}.
\ee Solving the equation and choosing the positive root for $\t_\star$,
\be 
\t_\star=\frac{1}{2} \left(\sqrt{( 2\Delta_\phi +\ell)(\ell+4 s_M)} -\ell\right).
\ee Now assuming $s\gg\ell$ we can approximate, 
\be 
\t_\star\approx \sqrt{(2\dphi+\ell) s_M}~.
\ee 
Thus, we will consider\footnote{In footnote \ref{dphitau} it was mentioned that, the analysis so far was carried out by implicitly considering $\dphi\sim\t$ along with $\dphi\gg1,\,\t\gg1$. However, observe that, the upper limit of the $\t-$sum really does not conform to $\t\sim \dphi$. In the upper limit one has, in fact, $\t\gg\dphi$. But, this does not cause any issue because, the center of our subsequent analysis, \eq{INEQimp}, is really independent of this. } 
\be\label{sigmam} 
\bar{\ma}_M(s_M)\approx \frac{2\p}{2s-\dphi}\sum_{\substack{\ell\\\ell~\text{even}}} \frac{(2h-2)_\ell}{(h-1)_\ell}\frac{\G(2\dphi+\ell-h)}{2\dphi+\ell}
\sum_{\t=\dphi}^{\sqrt{(2\dphi+\ell) s_M}}
\left(\frac{1}{8}\right)^\ell C_{\t,\ell}~\mn_{\t,\ell}~
\left(\t+q_\star-\frac{\dphi}{2}\right)^\ell \,.
\ee
Now observe that, 
\be
q_\star=\frac{(\t-\dphi)^2}{\ell+2\dphi}\le \frac{(\t-\dphi)^2}{2\dphi}.
\ee Further using this we can write, 
\be \label{INEQimp}
\left(\t+q_\star-\frac{\dphi}{2}\right)^\ell\le 
\left[\t+ \frac{(\t-\dphi)^2}{2\dphi}-\frac{\dphi}{2}\right]^\ell=\left(\frac{\t^2}{2\dphi}\right)^\ell
\ee because we have $\ell\ge 0$. 
Then using this we can write
\be\label{amtrunc} 
\bar{\ma}_M(s_M)\le \frac{2\p}{2s_M-\dphi}\sum_{\substack{\ell\\\ell~\text{even}}} \frac{(2h-2)_\ell}{(h-1)_\ell}\frac{\G(2\dphi+\ell-h)}{2\dphi+\ell}
\sum_{\t=\dphi}^{\sqrt{(2\dphi+\ell) s_M}}
\left(\frac{\t^2}{16\dphi}\right)^\ell C_{\t,\ell}~\mn_{\t,\ell}\,.
\ee
\subsubsection{Determining the $\ell-$cutoff}
Now we move on to the determination of the cutoff for the $\ell-$sum  in the expression for $\bar{\ma}_M(s_M).$ To do so we will take help of the \enquote{polynomial boundedness} condition that is expressed through the finiteness of the integral quantity $\mathfrak{a}_{n,\rho}$ for some positive integer $n$. Now since $\ma_M(s_M,t_M)$ is a positive distribution for unitary theories, we can write the following the chain of inequalities, 
\begin{align}
	\begin{split} 
\mathfrak{a}_{n,\rho}&=\int_{\frac{\dphi}{2}}^\infty\frac{d\bar{s}_M}{\bar{s}_M^{n+1}}\,\ma_M\left(\bar{s}_M,t=\frac{\r \D_\phi}{2}\right)\\
& > \int_{\frac{\dphi}{2}}^{s_M}\frac{d\bar{s}_M}{\bar{s}_M^{n+1}}\,\ma_M\left(\bar{s}_M,t=\frac{\r \D_\phi}{2}\right)\\
 &\ge s_M^{-(n+1)}\int_{\frac{\dphi}{2}}^\infty d\bar{s}_M\,\ma_M\left(\bar{s}_M,t=\frac{\r \D_\phi}{2}\right)
\end{split}
\end{align}

where the last inequality was possible because $n\ge 0$.  Thus we have the following inequality, 
\begin{align}\label{mpolyineq}
	\begin{split}
		\mathfrak{a}_{n,\rho}
		&\ge
		\p s_M^{-(n+1)}  
		\sum_{\substack{\ell\\ \ell~\text{even}}}\frac{\G(\ell+1)}{(h-1)_\ell} \frac{\G(2\dphi+\ell-h)}{2\dphi+\ell}C_{\ell}^{(h-1)}\left(1+\frac{\r\dphi}{s_M-\dphi/2}\right)
		\sum_{\t=\dphi}^{\t_\star}
		\left(\frac{1}{8}\right)^\ell 
		C_{\t,\ell} \mn_{\t,\tell} 
		\left(\t+q_\star-\frac{\dphi}{2}\right)^\ell\\
		&\ge 
		\p s_M^{-(n+1)} 
		\sum_{\substack{\ell=L+2\\ \ell~\text{even}}}\frac{\G(\ell+1)}{(h-1)_\ell} \frac{\G(2\dphi+\ell-h)}{2\dphi+\ell}C_{\ell}^{(h-1)}\left(1+\frac{\r\dphi}{s_M-\dphi/2}\right)
		\sum_{\t=\dphi}^{\t_\star}
		\left(\frac{1}{8}\right)^\ell 
		C_{\t,\ell} \mn_{\t,\tell} 
		\left(\t+q_\star-\frac{\dphi}{2}\right)^\ell\\
		&\ge 
		\p s_M^{-(n+1)} 
		\mathbf{C}_{L+2}^{(h-1)}\left(1+\frac{\r\dphi}{s_M-\dphi/2}\right)
		\sum_{\substack{\ell=L+2\\ \ell~\text{even}}}\frac{(2h-2)_\ell}{(h-1)_\ell} \frac{\G(2\dphi+\ell-h)}{2\dphi+\ell}
		\sum_{\t=\dphi}^{\t_\star}
		\left(\frac{1}{8}\right)^\ell 
		C_{\t,\ell} \mn_{\t,\tell} 
		\left(\t+q_\star-\frac{\dphi}{2}\right)^\ell
	\end{split}
\end{align}
where $\mathbf{C}_{\ell}^{(\a)}$ is the normalized Gegenbauer polynomial
\be \label{gegennorm}
{\Large \mathbf{C}_{\ell}^{(\a)}(x)=\frac{C_{\ell}^{(\a)}(x)}{C_\ell^{(\a)} (1)}=\frac{\G(\ell+1)}{(2\a)_\ell}C_{\ell}^{(\a)}(x)\,.}
\ee 
$L$ is some value of $\ell$ which is to be determined later and the last inequality is obtained using the fact that for $C_{\ell}^{(h-1)}(x)$ is an increasing function of $\ell$ for $x>1$ and also accounting for the correct normalization of the Gegenbauer polynomial.

Now we can split the sum in eq.\eqref{amb1} in the following manner, 
\be \label{trunc}
\bar{\ma}_M(s_M)\approx \frac{2\p}{2s_M-\dphi}\sum_{\substack{\ell\\\ell~\text{even}}}^{L} \frac{(2h-2)_\ell}{(h-1)_\ell}\frac{\G(2\dphi+\ell-h)}{2\dphi+\ell}\sum_{\t=\dphi}^{\t_\star}\left(\frac{1}{8}\right)^\ell C_{\t,\ell}~\mn_{\t,\ell}
\left(\t+q_\star-\frac{\dphi}{2}\right)^\ell+~ \mr(s)
\ee where 
\be 
\mr(s_M)=  \frac{2\p}{2s_M-\dphi}\sum_{\substack{\ell=L+2\\\ell~\text{even}}}^{\infty} \frac{(2h-2)_\ell}{(h-1)_\ell}\frac{\G(2\dphi+\ell-h)}{2\dphi+\ell}\sum_{\t=\dphi}^{\t_\star}\left(\frac{1}{8}\right)^\ell C_{\t,\ell}~\mn_{\t,\ell}
\left(\t+q_\star-\frac{\dphi}{2}\right)^\ell.
\ee 
Quite obviously, then, we can write 
\be 
\mr(s_M)\le  \frac{\mathfrak{a}_{n,\rho}~s^{n+1}}{(2s_M-\dphi)~\mathbf{C}_{L}^{(h-1)}\left(1+\frac{\r\dphi}{s_M-\dphi/2}\right)}.
\ee For $s_M\gg\dphi$ the above inequality effectively is, 
\be 
\mr(s_M)\le \frac{\mathfrak{a}_{n,\rho}s_M^n}{\mathbf{C}_{L}^{(h-1)}\left(1+\frac{\r\dphi}{s_M-\dphi/2}\right)}.
\ee 
Next will make use of  the following bounding relation satisfied by the Gegenbauer polynomials (see appendix \ref{bndderive} for a derivation),  
\be \label{Gegenbound}
\mathbf{C}_\ell^{(\a)}(z)\ge 2^{1-2\a}\frac{\G(2\a)}{\G^2(\a)}K(\varphi_0)\left(z+\sqrt{z^2-1}\cos \varphi_0\right)^\ell
\ee with
\begin{equation*}
	K(\varphi_0)=\int_0^{\varphi_0}(\sin\varphi)^{2\a-1} d\varphi     
\end{equation*}
for any $\varphi_0,~0<\varphi_0<\p,~x>1,~\a>0$.
Employing this we can constrain $\mr(s)$ as,  
\be 
\mr(s)\le 2^{3-2 h}~\mathfrak{a}_{n,\r}s^n \frac{\G(2h-2)}{\G^2(h-1)} \left(\tilde{x}+\sqrt{\tilde{x}^2-1}\cos\varphi_0\right)^{-L-2}
\ee with
\be 
\tilde{x}:=1+\frac{\r\dphi}{s_M-\dphi/2}.
\ee Now for $s_M\gg \dphi$ we have to leading order, 
\be 
\left(\tilde{x}+\cos\varphi_0\sqrt{\tilde{x}^2-1}\right)\sim 1+\cos\varphi_0\sqrt{\frac{2\r\dphi}{s_M}}.
\ee 
Thus we can write , 
\be \label{rineq1}
\mr(s)\le 2^{3-2 h}~\mathfrak{a}_{n,\rho}s_M^n \frac{\G(2h-2)}{\G^2(h-1)} \left(1+\cos\varphi_0\sqrt{\frac{2\r\dphi}{s_M}}\right)^{-L-2}.
\ee 
Now the optimal value for $L$ can be obtained by demanding that the remainder term be of exponentially suppressed magnitude. However there is a subtlety in this requirement. The important thing to keep in mind is that we need to have the remainder exponentially suppressed compared to the truncated sum eq.\eqref{amtrunc}. What this means is that we are keeping the possibility of certain overall growth behaviour (that of the truncated sum) for the remainder term but still sticking to the requirement that the growth be multiplied by a strong exponential suppression. Thus we are making the requirement a bit weaker than eq.\eqref{rineq1}. Assume a polynomial behaviour for the truncated sum $\sim s_M^{a}$ (here logarthmic terms may be present which we are ignoring  because they are in general much weaker than a polynomial behaviour). Then the optimal $L$ is given by the rather weaker constraint , 
\be \label{rineq2}
\mr(s_M)\le 2^{3-2 h}~\mathfrak{a}_{n,\rho}~s_M^{n-a} \frac{\G(2h-2)}{\G^2(h-1)} \left(1+\cos\varphi_0\sqrt{\frac{2\r\dphi}{s_M}}\right)^{-L-2}.
\ee 
The optimal $L$ is thus given to leading order by, 
\be \label{Lopt}
{	L=\frac{(n-a)}{\cos\varphi_0}\sqrt{\frac{s_M}{2\r\dphi}}\ln s_M.
}	\ee
We will truncate the $\ell-$sum in eq.\eqref{amtrunc} at $\ell=L$ as determined above to obtain the asymptotic bound,  
\be 
\bar{\ma}_M(s_M)\le \frac{2\p}{2s_M-\dphi}\sum_{\substack{\ell\\\ell~\text{even}}}^L \frac{(2h-2)_\ell}{(h-1)_\ell}\frac{\G(2\dphi+\ell-h)}{2\dphi+\ell}\sum_{\t=\dphi}^{\sqrt{(2\dphi+\ell) s_M}}\left(\frac{\t^2}{16\dphi}\right)^\ell C_{\t,\ell}~\mn_{\t,\ell}
\ee 
But to achieve the main goal of bounding $\bar{\ma}_M(s_M)$, we will need to have some information about the $\t-$sum appearing as in the above expression. This is what we turn to next.

\subsubsection{The final bounds}
In order to obtain the final bounding expression,  we need to have an estimate of the sum over $\t$ of  $C_{\t,\ell}\mn_{\t,\ell}$. We are concerned with the large $s$ asymptotic of the sum, 
\be \label{cnsum}
\sum_{\t=\dphi}^{\sqrt{(2\dphi+\ell) s_M}}
\left(\frac{\t^2}{16\dphi}\right)^\ell C_{\t,\ell}\mn_{\t,\ell}\,.
\ee 
It should be possible to do an analysis of this sum using the complex Tauberian theorem arguments used in \cite{Mukhametzhanov:2018zja}. However, we will content ourselves using a weaker result for now. To obtain the leading term in the asymptotic, 
we consider the generalized mean field theory (MFT) value for $C_{\t,\ell}$ and consider the large $\t$ limit  of the same.  The reason behind this is that the MFT operators are needed to reproduce the identity exchange in the crossed channel. This result is valid for spins greater than 2 and is a general result derived in \cite{caronhuot}. The large $\D_\phi$ limit that we consider does not affect the conclusions. Thus our results should be valid in any CFT with the identity operator.

So we consider the large $\t-$limit of the product, 
\be \label{summand}
\left(\frac{\t^2}{16\dphi}\right)^\ell C_{\t,\ell}\mn_{\t,\ell}
\sim 
\frac{2^{2 h+1}  \tau ^{4-2h} ~(2\dphi)^{-\ell}~\Gamma (\ell+h)}{\pi ^2 \Gamma^2 (\Delta_\phi ) \Gamma (\ell+1) \Gamma^2 (-h+\Delta_\phi +1)}\sin ^2\p[  \Delta_\phi -  \tau ]\,.
\ee 
Now at this point we have two separate cases at hand. As shown in appendix \ref{tausum} the sum eq.\eqref{cnsum} above has different asymptotes depending upon whether $h$ is greater, equal  or less than $5/2$. We have 
\be \label{tsasymp}
\frac
{\pi ^2 \Gamma^2 (\Delta_\phi ) \Gamma (\ell+1) \Gamma^2 (-h+\Delta_\phi +1)}
{2^{2 h+1}   ~(2\dphi)^{-\ell}~\Gamma (\ell+h)}
\sum_{\t=\dphi}^{\sqrt{(2\dphi+\ell) s_M}}\left(\frac{\t^2}{16\dphi}\right)^\ell C_{\t,\ell}\mn_{\t,\ell} \sim
\begin{cases}
	\frac{\left[s_M(2\dphi +\ell)\right]^{\frac{5}{2}-h}}{10-4h},\hspace{1 cm}h<\frac{5}{2};\\
	\\
	\frac{1}{4}\log s_M,\hspace{1.8 cm}h=\frac{5}{2};\\
	\\
	\frac{\dphi^{5-2 h}}{4h-10},\hspace{2.3 cm}h>\frac{5}{2}.
\end{cases}
\ee

With the aid of this expression we turn to the final step of obtaining the Froissart bound for the Mellin amplitude.
\subsubsection*{Case I. $h<\frac{5}{2}$}
First we start with the case $h<5/2$. Taking the large $\ell$ asymptotic of eq.\eqref{tsasymp} for $h<5/2$ we have ,
\be	
\sum_{\t=\dphi}^{\sqrt{(2\dphi+\ell) s_M}}
\left(\frac{\t^2}{16\dphi}\right)^\ell
C_{\t,\ell}\mn_{\t,\ell}
\sim~  
s_M^{\frac{5}{2}-h}~
\ell^{h-1}
\frac
{
	2^{2 h+1} 
	(2\dphi+\ell) ^{\frac{5}{2}- h}~(2\dphi)^{-\ell}
}
{
	\pi ^2 
	(10-4h)
	~	\Gamma^2 (\Delta_\phi ) 
	\Gamma^2 (-h+\Delta_\phi +1)
}\,.
\ee 
Next  putting this into eq.\eqref{amtrunc}, 
\be 
\bar{\ma}_M(s_M)\le 
\frac{2\p}{2s_M-\dphi}~s_M^{\frac{5}{2}-h}
\sum_{\substack{\ell=0\\ \ell~ \text{even}}}^{L}
\frac
{
	2^{ 2h+1} 
	(2\dphi+\ell) ^{\frac{3}{2}- h}~
	\G(2\dphi-h)(2\dphi-h)_\ell 
}
{
	\pi ^2 
	(10-4h)
	\Gamma^2 (\Delta_\phi )
	\Gamma^2 (-h+\Delta_\phi +1)
	(2\dphi)^\ell
}
\frac{(2h-2)_\ell}{(h-1)_\ell}
\ell^{h-1}
\ee 
where we have  used $\G(2\dphi-h+\ell)=(2\dphi-h)_\ell\G(2\dphi-h)$. 
Next, we will consider the large $\ell$ asymptotic\footnote{The dominant contribution to the $\ell$-sum comes from  the upper limit $\ell=L$ and since, $L$ is large we have used large $\ell$ approximation of for the $\ell$-summand. We observe that, this works because the $\ell$-summand behaves as power law with positive exponent in the large $\ell$ limit.  }
\be 
\frac
{
	(2h-2)_\ell
}
{
	(h-1)_\ell 
} \sim \ell^{h-1}\frac{ \Gamma (h-1)}{\Gamma (2 h-2)}.
\ee 
Now we  consider that $\dphi\gg h,~\dphi\gg 1$. At this point, to make progress (for $d=3$ we do not have to make this choice), we approximate $2\D_\phi+\ell\sim 2\D_\phi$ by assuming\footnote{This follows from the discussion in section 3.3. The (very interesting) case where $s\gg (2\D_\phi)^3$ and which will make a difference for $d\neq 3$ is beyond the scope of this work.} $\dphi\gg L$. Then one obtains, 

\be \label{ambsum}
\bar{\ma}_M(s_M)\le s^{\frac{3}{2}-h}
\frac
{
	2^{ 2h}    \G(2\dphi-h)
}
{
	\pi (5-2 h)  \Gamma^2 (\Delta_\phi ) \Gamma^2 (-h+\Delta_\phi +1)
}
\frac{\G(h-1)}{\G(2h-2)}
\sum_{\substack{\ell=0\\ \ell~ \text{even}}}^{L}\ell^{2h-2}~(2\dphi) ^{\frac{3}{2}- h}
\ee
where we have used  $s_M\gg\dphi/2$.\\
We can now use  eq.\eqref{ellsum3}
to obtain
\be
\bar{\ma}_M(s_M)\le 
\frac
{
	2^{ 2h}  (2\dphi) ^{\frac{3}{2}- h}  \G(2\dphi-h)
}
{
	\pi  \Gamma^2 (\Delta_\phi ) \Gamma^2 (-h+\Delta_\phi +1)
}
\frac{\G(h-1)}{(5-2h)\G(2h-2)}~s_M^{\frac{3}{2}-h}~\frac{L^{2h-1}}{4h-2}. 
\ee 
Now using eq.\eqref{Lopt} we have, 
\be \label{CFTFroissart1}
	\bar{\ma}_M(s_M)\le 
	\mb_1 s_M\ln^{2h-1}s_M.
\ee

with,
\begin{align} \label{CoeffB1}
	\begin{split}
		\mb_1=
		2^{ 2h-1}  (2\dphi) ^{2- 2h}~  
		\frac{8\p^{h-1}\mn~\G(h-1)}{(5-2h)(2h-1)\G(2h-2)}
		\left(\frac{n-1}{\sqrt{\r}\cos\varphi_0}\right)^{2h-1}.
	\end{split}
\end{align}
where $\mn$ is same as in eq.\eqref{Nexp} and we have put $a=1$ by observing that the leading power law dependency of bound is $\sim s_M$. 

$\mathbf{\underline{d=2}}:$ At this point we would like to comment upon the case of $d=2$, or equivalently $h=1$. Note that in this case the Gegenbauer polynomial $C_\ell^{(h-1)}$ is undefined. But this case can still be tackled following the analysis of \cite{CHAICHIAN1992151}. In fact on following the method one obtains  the bound in this case coincident with eq.\eqref{CFTFroissart1} if we put $h=1$ and $\cos\varphi_0=1$ formally into the same. Note that, while formally putting $h=1$ into eq.\eqref{CFTFroissart1} one has to consider doing so in the limiting sense if required.
\subsubsection*{Case II: $h=\frac{5}{2}$}
Next we turn to the case $h=\frac{5}{2}$. Considering the large $\ell$ limit as before one readily obtains from  eq.\eqref{tsasymp} 
\be 
\sum_{\t=\dphi}^{\sqrt{(2\dphi+\ell) s_M}}
\left(\frac{\t^2}{16\dphi}\right)^\ell
C_{\t,\ell}\mn_{\t,\ell}
\sim~
\frac
{
	16 ~  \log (s_M)~\ell^{\frac{3}{2}}
}
{
	\pi ^2\Gamma^2 \left(\Delta _{\phi }-\frac{3}{2}\right){} \Gamma^2 \left(\Delta _{\phi }\right)
}~(2\dphi)^{-\ell}.
\ee 
Thus we have, 
\be 
\bar{\ma}_M(s_M)\le \frac{\log s_M}{s_M}~\frac{16 \G(2\dphi-5/2)}{2\p\dphi\G^2(\dphi-3/2)\G^2(\dphi)}\sum_{\ell=0}^{L}\ell^3\sim \frac{\log s_M}{s_M}~\frac{16 \G(2\dphi-5/2)}{2\p\dphi\G^2(\dphi-3/2)\G^2(\dphi)}~\frac{L^4}{8}
\ee where the last equality follows by the large $L$ asymptotic of the $\ell-$sum. Using eq.\eqref{Lopt}, 

\be \label{CFTFroissart2}
	\bar{\ma}_M(s_M)\le \mb_2\, s_M\log^5 s_M,
\ee with 
\be \label{CoeffB2}
\mb_2:= 8\p^{\frac{3}{2}}\mn (2\dphi)^{-3} ~\left(\frac{(n-1)}{\sqrt{\r}\cos\varphi_0}\right)^4
\ee 
where
$a=1$ has been put in the last stage by the same logic as in the previous case. 
\subsubsection*{Case III: $h>\frac{5}{2}$}
Now we turn to the case when $h>5/2$. This is rather curious case. As shown in appendix \ref{tausum}, for this case the lower limit of the sum eq.\eqref{cnsum} dominates rather than the upper limit. As a consequence, we now have, 

\be 
\sum_{\t=\dphi}^{\sqrt{(2\dphi+\ell) s_M}}
\left(\frac{\t^2}{16\dphi}\right)^\ell
C_{\t,\ell}\mn_{\t,\ell}
\sim~  
\frac
{
	2^{4(h-1)} 
	(2\dphi) ^{5-2h-\ell}
	\Gamma (\ell+h)
}
{
	\pi ^2 
	(4h-10)
	\Gamma^2 (\Delta_\phi ) 
	\Gamma (\ell+1) 
	\Gamma^2 (-h+\Delta_\phi +1)
} ,~~~s\to \infty.
\ee 
Further considering large $\ell$ limit , 
\be 
\sum_{\t=\dphi}^{\sqrt{(2\dphi+\ell) s_M}}
\left(\frac{\t^2}{16\dphi}\right)^\ell
C_{\t,\ell}\mn_{\t,\ell}
\sim~  
\ell^{h-1} \frac
{
	2^{4h-4} 
	(2\dphi) ^{5-2h-\ell}
}
{
	\pi ^2 
	(4h-10)
	\Gamma^2  (\Delta_\phi )
	\Gamma^2 (-h+\Delta_\phi +1)
} 
\ee 
Further
putting this into eq.\eqref{sigmam} and following through the same steps as before we get the asymptotic bound,

\be 
\bar{\ma}_M(s_M)\le
\frac{1}{s_M}
\sum_{\substack{\ell=0\\ \ell~ \text{even}}}^{L}
\frac
{
	2^{ 4h-4} (2\dphi) ^{4-2h}  \G(2\dphi-h)\G(h-1)
}
{
	\pi  
	(4h-10)\G(2h-2)
	\Gamma^2 (\Delta_\phi )
	\Gamma^2 (-h+\Delta_\phi +1)
}
\ell^{2h-2}
\ee 
Now using eq.\eqref{ellsum2}, 
\be \label{CFTFroissart3}
	\bar{\ma}_M(s_M)\le\mb_3~ s_M^{h-\frac{3}{2}}\ln^{2h-1}s_M 
\ee
with, 
\be \label{CoeffB3}
\mb_3:=
2^{ 4h-6} (2\dphi) ^{\frac{9}{2}-3h} 
\frac{8\p^{h-1}\mn~\G(h-1)}{(2h-5)(2h-1)\G(2h-2)}
\left(\frac{2(n-h)+3}{\sqrt{\r}\cos\varphi_0}\right)^{2h-1}
\ee 
\subsubsection{On the number of subtractions of Mellin Amplitude dispersion relation}
The polynomial boundedness assumption for the Mellin amplitude is closely tied to the question of writing a dispersion relation for the Mellin amplitude. The key point in this regard is how many subtractions are sufficient to write such a dispersion relation. The assumption of finiteness of $\mathfrak{a}_n$ naively suggests the possibility of writing a dispersion relation for Mellin amplitude with $n-$subtractions. Then the question is what can be the value of $n$. In the above analysis we have kept $n$ arbitrary. $n$ will be determined by the leading power law behaviour of the bound. What we mean by this is that the we have already seen that the generic structure of the Froissart$_{AdS}$ bound for $\bar{\ma}_M(s_M)$ is of the form $\bar{\ma}_M(s_M)\le C s_M^a\ln^bs_M $. Now it turns out that \emph{the value of  $n$ is controlled by $a$}.  This control happens in two ways. 

First observe the expression for the optimal value of the $\ell-$cutoff in eq.\eqref{Lopt}. There sits a factor of $(n-a)$. Now, in our analysis  we have extensively used the assumption $L\gg 1$. Then for the consistency of this assumption we require necessarily $n>a$. 

While this simple consideration puts a lower bound on the magnitude of $n$, it is also possible to obtain an upper bound on the same. The way to have so is by using a theorem from complex analysis called Phragmen-Lindeloff theorem (see, for example, \cite{Rademacher:1973}). The general logic goes as follows: assuming the polynomial boundedness condition as in section \ref{polymellin} and using the  Froissart$_{AdS}$ bound it is possible to show by the use of Phragmen-Lindeloff theorem that $n\le \floor{a}+1$. This thus puts an upper bound on $n$. Now we analyse the individual cases of different $h$ values.
\begin{itemize}
	\item[I.] \underline{$h< 5/2$}: 
	For $h<5/2$ we have from  eq.\eqref{CFTFroissart1} $a=1$.
	Then following logic chalked out above we have clearly $n=2$.
	\item[II.] \underline{$h=5/2$}: For this case as well we have $a=1$ from eq.\eqref{CFTFroissart2}. Thus again we will have   $n=2$.
	\item[III.] \underline{$h>5/2$}: This case is rather interesting. From eq.\eqref{CFTFroissart3} we have $a=(2h-3)/2$ thus leading immediately to, 
	\be 
	\frac{2h-3}{2}<n\le \lfloor{\frac{2h-3}{2}}\rfloor+1
	\ee  
	What this implies is that, while for $h=3$( equivalently $d=6$) one has to have $n=2$, \emph{one must have \underline{$n\ge 3$}}  for $h>3$ i.e., $d>6$. The number goes to infinity as $d$ goes to infinity. In section 5, we will provide an alternative derivation of these results without invoking the Phragmen-Lindeloff theorem.
\end{itemize}
\subsection{Connection to flat space Froissart bound} 
\subsubsection{$d=3$}
Now that we have bounds on Mellin amplitude  we would like to consider the flat space limit of the above bound. 
It is quite straightforward that  $\ma_M(s,t)$ is related to the absorptive part of the flat spacetime scattering amplitude $\mt(S,T)$ as in eq.\eqref{fl3} and similar relations follow for all averaged quantities. 

Now for the flat space limit we will focus our attention upon $h=3/2$ i.e., $3d$ CFT which in the flat space limit connects to $(3+1)d$ flat spacetime quantum field theory where the original Froissart bound was proved. We start with the bounding relation eq.\eqref{CFTFroissart1}. Considering  $n=2$ and $\cos\varphi_0=1$ in a limiting sense for strongest bound in eq.\eqref{CoeffB1} the Froissart$_{Ads}$ bound, eq.\eqref{CFTFroissart1}, becomes for $h=3/2$, 
\be 
\bar{\ma}_M(s_M)\le 4\p{\mn} \frac{s_M}{\r\dphi}\ln^2s_M.
\ee 
Next taking the flat space limit and making use of eq.\eqref{sSrel2}, one obtains\footnote{$s_0$ here has been put in on dimensional grounds.}, 
\be 
\bar{\ma}(s)\le \frac{\p}{2\r m^2}\,s\ln^2 \frac{s}{s_0}.
\ee
The known bound on $\bar{\ma}$ from literature is (recall we are averaging so there is an extra $1/2$), 
\be 
\bar{\ma}(s)\le \frac{\p}{2\m^2}\,s\ln^2 \frac{s}{s_0}.
\ee 
Thus what we find is an exact match provided we identify $\rho=\mu^2/m^2$, identifying $\mu$ as the lightest exchange in the t-channel. However here comes a crucial difference. One can check that using the MFT asymptotics, the sum over conformal partial waves converges\footnote{In the large $\ell$ limit, one can show that the summand in the $\ell$-sum goes like $\ell^{-2\sqrt{s}\ell} s^\ell/\ell!$ for any $\rho$. We need $\rho\leq 1$ since the measure factor $\m(s_M,t_M)$ in the Mellin representation has $\G^2(\D_\phi/2-t_M)$ so there is a double pole at $t_M=\D_\phi/2$.} even for $\rho=1$. This means we can set $\mu=m$, which is the mass of the external particle. If we do this, then in fact we will have better agreement with the existing numerical fits of the proton-proton data. This needs to be checked carefully of course, which we will leave for future work.

\subsubsection{$d\neq 3$}
Now we consider the case of $d>3$. Here we would like to make a comparison of the flat space limit of the Mellin amplitude bounding relation with standard result of bounds on flat space scattering amplitude in general spacetime dimensions \cite{Chaichian:1987zt} given in eq.(\ref{Imbound}). Upon comparison one finds that ratio of the the  frontal coefficient that is obtained on taking the flat space limit of eq.\eqref{CFTFroissart1}  to the frontal coefficient that appears in eq.(\ref{Imbound}) is,
\be 
\frac{2}{(5-2h)}.
\ee This coefficient is unity for $d=3$. Further for $d=4$ we find a weaker flat space bound by taking flat space limit of Mellin amplitude. For $d=2$ it is stronger.

\subsubsection{On  flat space limit of eq.\eqref{CFTFroissart2} and eq.\eqref{CFTFroissart3}}
We can consider taking  flat space limit of the Froissart$_{AdS}$ bound for $h\ge 5/2$, eq.\eqref{CFTFroissart2} and eq.\eqref{CFTFroissart3}, using the dictionary eq.\eqref{fl3}. 
\begin{itemize}
	\item[I.] \underline{$h=5/2$} : Upon applying the flat space limit translation on eq.\eqref{CFTFroissart2} one obtains, 
	\be \label{flat5}
	\bar{\ma}(s)\le \frac{\p^{\frac{3}{2}}}{8 m^2}~\left(\frac{n-1}{\sqrt{\r}\cos\varphi_0}\right)^4\left(\frac{s}{m^2}\right)\ln^5 \frac{s}{s_0}.
	\ee Recall that this supposed to be corresponding to flat space scattering in 6 spacetime dimensions. Now if we compare the above with standard Froissart-Martin bound in 6 spacetime dimensions (c.f. eqn (24)  of \cite{Chaichian:1987zt}) for the dependency upon the Mandelstam variable $s$ then we realize that the bound eq.\eqref{flat5} above is a weaker one due to the presence of one extra power of $\ln S$.
	\item[II.] \underline{$h>5/2$} : Taking the flat space limit of eq.\eqref{CFTFroissart3} yields , 
	\be \label{flat6}
	\bar{\ma}(s)\le 
	\frac{2^{h+\frac{3}{2}}~ \p^{h-1}~\G(h-1)}{(2h-5)(2h-1)\G(2h-2)}
	\left(\frac{2(n-h)+3}{\sqrt{\r}\cos\varphi_0}\right)^{2h-1}~\left(\frac{s}{m^2}\right)^{h-\frac{3}{2}}\ln^{2h-1}\frac{s}{s_0}
	\ee 
	Now if one  to compare the $s$ dependency of this bound with that of the standard Froissart-Martin bound one readily observes that  the bound eq.\eqref{flat6} becomes weaker with increasing $h$. 
\end{itemize}

\subsubsection{Is the difference in form for $d>4$ expected?}
We can give a heuristic reason to justify, that a crossover at some value of $d$ is expected, in the behaviour of the Froissart bound. Froissart in his original paper and Feynman independently \cite{Khuri:1969vlg} had a heuristic argument for the $\ln^2$ behaviour. The argument goes as follows. Imagine that the interaction is well approximated by a Yukawa type potential
$$ V\sim g\frac{e^{-\mu r}}{r}\,. $$
Now the maximum interaction happens when $g e^{-\mu r_*}\sim 1$, giving $r_*\sim \frac{\ln g}{\mu}$. Now assuming that the coupling $g$ depends on the energy $E$ polynomially, i.e., $g\propto E^N$ and also assuming that $\mu$ does not depend on $E$, we will find $r_*\sim \frac{\ln E}{\mu}$. Thus, the scattering cross-section in $d+1$ dimensions is $$\sigma\sim r_*^{d-1}\propto \ln^{d-1} E\,.$$
Now let us assume that this $\ln^{d-1}$ behaviour is to be expected (which is what flat space calculations give). In our calculation, since $L\sim \sqrt{s}\ln s$, this can happen from a factor of $L^{d-1}$. Now in Mellin space considerations, the extra powers of $s$ are given by the twist sum. If we assume that the asymptotic growth of twist is of the form $\tau^a$ (so that no extra powers of $\ln s$ can come from here), then from the upper limit of the $\tau$ integral we will get $s^{a/2+1/2}/(a+1)$. So the overall power of $s$ (taking into account the $1/{s_M}$ in the definition) in $\bar A_M$ is then 
$$
\frac{s_M^{\frac{a+d-2}{2}}}{a+1}\,,
$$
so that to match with the existing flat space answers in the literature we must have $a=4-d$. This makes the denominator $5-d$ so that for $d\geq 5$ there is a change in behaviour than what is expected since here the dominant contribution comes from the lower limit of the twist integral, which is independent of $s_M$. This is essentially what we find. 

\subsection{Obtaining the Froissart$_{AdS}$ bound: \enquote{Non-forward limit}} 
In the previous section the we tackled the problem of obtaining an asymptotic upper bound to $\bar{\ma}_M(s_M,t_M)$ for $s_M$ large and $t_M=0$ i.e., the forward limit. Now we turn to the same problem for $t\ne 0$. In fact, the main task i.e., that of obtaining an $\ell-$ cutoff, is already done. The new piece of information that we need now is an upper bound for the Gegenbauer polynomial $C_\ell^{(\l)}(x)$ for $x\in [-1,1]$. At this point it is worth of mentioning that we are considering \enquote{physical} values of $t_M$ i.e., $t_M<0$.
\paragraph{} Starting with eq.\eqref{Amst} we have, 
\begin{align}
	\begin{split} \label{ambst}
\bar{\ma}_M(s_M,t_M)\equiv \bar{\ma}_M(s_M,x)\approx  \frac{2\p}{2s_M-\dphi}&\mathlarger{\mathlarger{\sum}}_{\substack{\ell\\\ell~\text{even}}} \,\left[\,\frac{\G(\ell+1)}{(h-1)_\ell}
\frac{\G(2\dphi+\ell-h)}{2\dphi+\ell}
C_\ell^{(h-1)}(x)\right.\\
&\hspace{ 2.5 cm}\times
\left.\sum_{\t=\dphi}^{\t_\star}\left(\frac{1}{8}\right)^\ell C_{\t,\ell}~\mn_{\t,\ell}
\left(\t+q_\star-\frac{\dphi}{2}\right)^\ell\right].
\end{split}
\end{align} 
Below we will write $\bar{\ma}_M(s_M,t_M)$ and $\bar{\ma}_M(s_M,x)$ interchangebly. However  we will later explain that there is a subtle difference between holding $t_M$ fixed and holding $x$ fixed while considering $s_M\to\infty$.   

\paragraph{} Start with the following inequlity for Jacobi polynomial \cite{sz75-1}, 
\be \label{jacless}
P_\ell^{(\a,\b)}(\cos \theta)<\frac{K}{\theta^{\a+1/2}~\ell^{1/2}},~~~\a\ge -\frac{1}{2}
\ee where $K$ is a constant. 
Now using the definition of the Gegenbauer polynomial in terms of Jacobi polynomials, 
\be 
C_\ell^{(\l)}(x)=\frac{(2\l)_\ell}{(\l+1/2)_\ell}P_\ell^{(\l-1/2,\l-1/2)}(x).
\ee we obtain by eq.\eqref{jacless} above, 
\be 
C_\ell^{(\l)}(\cos \th)<\frac{(2\l)_\ell}{(\l+1/2)_\ell}\frac{K}{\theta^{\l}~\ell^{1/2}}
\ee Further considering the large $\ell$ limit, 
\be \label{gegnless}
C_\ell^{(\l)}(\cos \th)< \widehat{K} \ell^{\l-1}\th^{-\l}. 
\ee with, 
\be 
\widehat{K}=\frac{ \Gamma \left(\lambda +\frac{1}{2}\right)}{\Gamma (2 \lambda )} K.
\ee Now we can use this inequality eq.\eqref{gegnless} into eq.\eqref{ambst} to obtain [putting $x=\cos\th$]\footnote{We have pulled out the $\theta^{1-h}$ factor outside the $\tau$ sum since for $d\geq 2$, it behaves like $s^a$ with $a>0$ so we can replace the $s'$ dependence by $s$ at the level of the $s'$ integral.}, 
\begin{align}
	\begin{split}
		\bar{\ma}_M(s_M,\cos\th)&\le  \widehat{K}\frac{2\p}{2s_M-\dphi}
		\th^{1-h}
		\sum_{\substack{\ell\\\ell~\text{even}}}
		\frac{\Gamma (\ell+1)}{(h-1)_\ell}~\ell^{h-2}  \frac{(2\dphi-h)_\ell}{2\dphi+\ell}~\G(2\dphi-h)\sum_{\t=\dphi}^{\t_\star}\left(\frac{1}{8}\right)^\ell C_{\t,\ell}\mn_{\t,\ell}\left(\t+q_\star-\frac{\dphi}{2}\right)^\ell\\
		&\approx \widehat{K}\frac{2\p}{2s_M-\dphi}
		\th^{1-h}
		\sum_{\substack{\ell\\\ell~\text{even}}}
		\frac{\Gamma (\ell+1)}{(h-1)_\ell}~\ell^{h-2}  \frac{(2\dphi)^\ell}{2\dphi+\ell}~\G(2\dphi-h)\sum_{\t=\dphi}^{\t_\star}\left(\frac{1}{8}\right)^\ell C_{\t,\ell}\mn_{\t,\ell}\left(\t+q_\star-\frac{\dphi}{2}\right)^\ell
	\end{split}
\end{align} where we have considered the large  $\dphi$ limit. Next mimicking the same steps as in forward limit we can use eq.\eqref{tsasymp} for the $\t$ sum in the above. Thus again as before we have three cases depending upon the value of $h$. Further the optimal value of $L$ where the $\ell-$sum will be truncated is same as before i.e that given by eq.\eqref{Lopt}. 
\begin{enumerate}
	\item[1)] \textbf{Case I:}~$h<\frac{5}{2}$\\
	\be\label{ellnf}
	\bar{\ma}_M(s_M,\cos\th)\le K_1~ s_M^{\frac{3}{2}-h}~\th^{1-h}\lsumt\ell^{h-1}(2\dphi+\ell)^{\frac{3-2h}{2}}
	\ee with, 
	\be 
	K_1=8\p^h \mn
	\frac{2^{2h}~\G(h-1)}{\p(5-2h)\G(2h-2)}\widehat{K} .
	\ee 
	Now using the result eq.\eqref{ellsum4} we obtain, 
	\be 
	\bar{\ma}_M(s_M,\cos\th)\le 
	K_1~(2\dphi)^{\frac{3-2h}{2}} s^{\frac{3}{2}-h}\th^{1-h}~\frac{L^{h}}{2 h},~~~~2\dphi\gg L\gg 1.
	\ee Now using the  optimal value of L eq.\eqref{Lopt}, 
	
	\be
	\bar{\ma}_M(s_M,\cos\th)\le 
	\mc~s^{\frac{3-h}{2}}~ \ln^{h}s_M~\th^{1-h}\ee
	with, 
	\begin{align}
		\begin{split} 
			\mc=\frac{K_1}{4h}~(2\dphi)^{\frac{3}{2}(1-h)}\left(\frac{n-3(1-h)/2}{\sqrt{\r}\cos\varphi_0}\right)^{h}
		\end{split}
	\end{align}
	where we have put the apt values of $a$ following the same logic as in the forward case. Now for fixed $t_M$ with $s_M\gg\dphi\gg 1$ we can rewrite the bounding expression above in terms of $t$ by using
	\be 
	\th\approx 2\sqrt{\frac{|t_M|}{s_M}}.
	\ee  Thus we have the bound, 
	\be\label{astbnd}
		\bar{\ma}_M(s_M,t_M)\le 
		\mc_1~s\ln^{h}s~|t|^{\frac{1-h}{2}}~ 
	\ee where $\mc_1=2^{1-h} \mc$. Note that $|t|$ takes care of the fact we are considering $t<0$.
	
	Using the flat space limit dictionary eq.\eqref{fl3} obtain the following bound, 
	\be 
	\ma(s,t)\le \mk_1~{m^{3-2h}} \left(\frac{s}{m^2}\right)\ln^h \frac{s}{s_0}\left(\frac{|T|}{m^2}\right)^{\frac{1-h}{2}}
	\ee where $\mk_1$ is constant.
	\item[2)] \textbf{Case II:} $h=\frac{5}{2}$\\
	\be\label{ellsum5}
	\bar{\ma}_(s_M,\cos\theta)\le K_2~\frac{\ln s}{s}~\th^{-\frac{3}{2}}
	\lsumt  \ell^{3/2}
	\ee
	with, 
	\be 
	K_2=32\p^2\frac{\mn}{\dphi} \widehat{K}
	\ee Now doing the $\ell-$sum, 
	\be 
	\lsumt  \ell^{3/2}= 2 \sqrt{2} H_{\frac{L}{2}}^{\left(-\frac{3}{2}\right)}\sim \frac{L^{5/2}}{5},~~~L\to\infty.
	\ee Putting this into eq.\eqref{ellsum5} we obtain, 
	\be 
	\bar{\ma}_M(s_M,\cos\theta)\le \frac{K_2}{5}\left(\frac{n-1/4}{\sqrt{\r}\cos\varphi_0}\right)^{5/2}(2\dphi)^{-\frac{5}{4}}~s^{\frac{1}{4}}\ln^{\frac{7}{2}}s~\th^{-\frac{3}{2}}\ee
	Again we can express $\th$ in terms of $t$ to obtain, 
	\be
		\bar{\ma}_M(s_M,t_M)\le \mc_2~s_M\ln^{\frac{7}{2}}s_M~|t_M|^{-\frac{3}{4}}
	\ee
	with 
	\be 
	\mc_2=\frac{\p^2\mn}{5}\left(\frac{2}{\dphi}\right)^{\frac{9}{4}}\left(\frac{n-1/4}{\sqrt{\r}\cos\varphi_0}\right)^{5/2}\widehat{K}.
	\ee If we now consider taking the flat space limit of the above bound using eq.\eqref{fl3} then we get, 
	\be 
	\bar{\ma}(s,t)\le \frac{\mk_2}{m^2}\left(\frac{s}{m^2}\right)\ln^{\frac{7}{2}}s\left(\frac{|t|}{m^2}\right)^{-\frac{3}{4}}
	\ee  where, 
	\be 
	\mk_2=\frac{{2}^{\frac{9}{4}}\p^2}{5}\left(\frac{n-1/4}{\sqrt{\r}\cos\varphi_0}\right)^{5/2}\widehat{K}.
	\ee 
	\item[3)]\textbf{Case III:} $h>\frac{5}{2}$\\
	Finally we come to to case of $h>5/2$. Going thorugh the same steps as before we reach the following bounding relation, 
\be
		\bar{\ma}_M(s_M,t_M)\le K_3 ~s^{h-\frac{3}{2}}\ln^hs_M~|t_M|^{\frac{1-h}{2}}\ee
where $K_3$ is a constant.
\end{enumerate} 	

\section{Dispersion relations}\label{Mellindisp}

In this section, we will follow \cite{Khuri:1969vlg} and write down dispersion relations for the Mellin amplitudes. The bounds derived in the previous section for the non-forward limit will prove useful here. We begin by writing an $N$-subtracted dispersion relation
\be\label{disp1}
{\mathcal M}(s,t)=\sum_{m=0}^{N-1}C_m(t)s^m+\frac{s^N}{\pi}\int_{\frac{\D_\phi}{2}}^\infty \, ds' \frac{\ma_M^{(s)}(s',t)}{s'{}^N(s'-s)}+\frac{u^N}{\pi}\int_{\frac{\D_\phi}{2}}^\infty \, du' \frac{\ma_M^{(u)}(u',t)}{u'{}^N(u'-u)}\,,
\ee
where $s+t+u=\D_\phi/2$ and $C_n(t)$'s are analytic in $t$ for $t<\D_\phi/2$. The number of subtractions is related to the number of  $C_m(t)$'s that one will need to take as input. For the identical scalar case that we have been considering so far, $\ma^{(s)}_M(u,t)=\ma^{(u)}_M(u,t)$, but we will keep the discussion more general. The bounds in the previous section, although derived for $t<0$ will continue to hold for $t>0$ for sufficiently small\footnote{This can be explicitly checked using the expressions in \cite{Khuri:1969vlg} where an alternative derivation can be found.} $t$. The bounds are of the form $\bar\ma(s,t)\leq K s^a \ln^b s |t|^{\frac{1-h}{2}}$. This suggests that there exists a $t=t_0$ with  $0<t_0\ll \D_\phi/2$ such that  $\ma(s,t)\leq c s^{n-1+\epsilon}$ can be used inside an integral with $\epsilon<1$. For instance, for $h\leq 5/2$ we have $n=2$ while for $h>5/2$ we have $n=\lfloor{\frac{2h-3}{2}}\rfloor+1.$
This means that for $0\leq t\leq t_0$ we can write the dispersion relation
\be\label{disp2}
{\mathcal M}(s,t)=\sum_{m=0}^{n-1}C_m(t)s^m+\frac{s^n}{\pi}\int_{\frac{\D_\phi}{2}}^\infty \, ds' \frac{\ma_M^{(s)}(s',t)}{s'{}^n(s'-s)}+\frac{u^n}{\pi}\int_{\frac{\D_\phi}{2}}^\infty \, du' \frac{\ma_M^{(u)}(u',t)}{u'{}^n(u'-u)}\,.
\ee
Comparing eq.\eqref{disp1} and eq.\eqref{disp2}) (assuming $N>n$) we get an equation
\be
\sum_{m=0}^{N-1}C_m(t)s^m=\sum_{m=0}^{n-1}C_m(t)s^m+\sum_{m=n}^{N-1}\left(\frac{s^m}{\pi}\int_{\frac{\D_\phi}{2}}^\infty \, ds' \frac{\ma_M^{(s)}(s',t)}{s'{}^{m+1}}+\frac{u^m}{\pi}\int_{\frac{\D_\phi}{2}}^\infty \, du' \frac{\ma_M^{(u)}(u',t)}{u'{}^{m+1}}\right)\,.
\ee
Comparing the highest power of $s$ for large $s$ (assuming $N$ is even), we have
\be \label{Cint}
C_{N-1}(t)=\frac{1}{\pi}\int_{\frac{\D_\phi}{2}}^\infty \, ds' \frac{\ma_M^{(s)}(s',t)}{s'{}^{N}}+\frac{1}{\pi}\int_{\frac{\D_\phi}{2}}^\infty \, du' \frac{\ma_M^{(u)}(u',t)}{u'{}^{N}}\,.
\ee
We can Taylor expand the integrand around $t=0$ writing $\ma_M^{(s)}(s',t)=\sum_{k=0}^\infty A^{(s)}_{n}(s') t^n$ and $\ma_M^{(u)}(u',t)=\sum_{k=0}^\infty A^{(u)}_{n}(u') t^n$ where it can be shown that the coefficients are all positive which follows from $\frac{d^n C_\ell^{(\lambda)}(x)}{d x^n}\geq 0$. Using this and the fact that $C_{N-1}(t)$ was analytic for $t<\D_\phi/2$, it follows that each integral on the rhs of eq.\eqref{Cint} is finite for $t<\D_\phi/2$. This in turn 
implies that 
\be
\frac{|{\mathcal M}(s,t)|}{s^N}\rightarrow 0\,,
\ee 
as $s\rightarrow \infty$ for $t<\D_\phi/2$. As a result we can consider one less subtraction in eq.\eqref{disp1} than what we started off with. For $N$ odd, the situation is similar with the number of subtractions going down by two. This can be repeated until we reach the conclusion that 
$$
\int_{\frac{\D_\phi}{2}}^\infty \, ds' \frac{\ma_M^{(s)}(s',t)}{s'{}^{n+1}}\,,
$$ 
and analogously the $u$-channel integral is finite for $t<\D_\phi/2$, for $n$ specified above. This essentially leads to eq.(\ref{disp2}) being the appropriate dispersion relation for $t<\D_\phi/2$. This is another way of deriving our conclusions stated in section 4.1.3. 

\begin{subappendices}
	\section*{Appendix}
	\section{ Mack polynomials: Conventions and properties}\label{Mackasym1}
	In this appendix, we will show explictly that in the \enquote{flat space limit} the leading asymptotic of  Mack polynomial is  Gegegnabuer polynomial. For that we need the explicit form of the Mack polynomial. There exists varied representation of Mack polynomials \cite{Mack:2009mi,Costa:2012cb,Dolan:2011dv}. The normalization that we deploy  for our cause is given by , 
	\begin{align}\label{MAckdef1}
		\begin{split}
			\widehat{P}_{\t,\ell}(s,t)&=\sum_{n=0}^{\ell}\sum_{m=0}^{\ell-n}\m_{m,n}^{(\ell)}\left(\t-s-\frac{\D_\f}{2}\right)_m\left(-t+\frac{\D_\f}{2}\right)_n
		\end{split}
	\end{align} 
	with
	\begin{align}\label{mackmudef}
		\begin{split}
			\m_{m,n}^{(\ell)}=2^{-\ell}(-1)^{m+n}
			&\left(
			\begin{matrix}
				\ell\\
				m,n
			\end{matrix}\right)(\t+\ell-m)_m(\t+n)_{\ell-n}(\t+m+n)_{\ell-m-n}(\ell+h-1)_{-m}(2\t+2\ell-1)_{n-\ell}\\
			& 
			~\hspace{4.5 cm}\times~_4F_3
			\left[
			\begin{matrix}
				-m,~1-h+\t,~1-h+\t,~n-1+2\t+\ell\\
				\t+\ell-m,~\t+n,~2-2h+2\t
			\end{matrix}
			~\Big|1\right].
		\end{split}
	\end{align}
	Now we introduce the variable, 
	\be \label{Gegvar}
	x=1+\frac{2t}{s-\D_\f/2}.
	\ee 
	Using this variable we rewrite the Mack polynomial as a function of $(s,x)$,
	\be \label{mnewdef}
	\widehat{P}_{\t,\ell}(s,x)=\sum_{n=0}^{\ell}\sum_{m=0}^{\ell-n}\m_{m,n}^{(\ell)}\left(\t-\left(s+\frac{\D_\f}{2}\right)\right)_m\left(\frac{1-x}{2}(s-\D_\f/2)+\frac{\D_\f}{2}\right)_n.
	\ee 
	Next we move on to giving the prescription for flat space limit. In the \enquote{flat space limit} we will consider, 
	\be 
	s\gg \t\gg 1
	\ee 
	The reason for this is that in our analysis, the twist sum lies between $\D_\phi$ and $\sqrt{2\D_\phi s}$ and since $s\gg \D_\phi/2$, the above consideration follows. 
	In this limit we have for the leading asymptotic, 
	\be 
	\left(\t-s-\frac{\D_\f}{2}\right)_m\left(\frac{1-x}{2}(s-\D_\f/2)+\frac{\D_\f}{2}\right)_n\sim (-1)^m s ^{m+n}\left(\frac{1-x}{2}+\frac{\D_\f}{(2s-\D_\f)}\right)^n
	\ee
	Clearly in the limit $s\gg1$ the leading contribution in eq.\eqref{mnewdef} comes from $m=\ell-n$ so that we have, 
	\be \label{limitform}
	\Mack_{\t,\ell}(s,x)\sim s^{\ell}\sum_{n=0}^{\ell}(-1)^{\ell-n}\m_{\ell-n,n}^{(\ell)}\left(\frac{1-x}{2}+\frac{\D_\f}{(2s-\D_\f)}\right)^n
	\ee 
	Now we will focus upon the $\t\to \infty$ asymptotic of $(-1)^{\ell-n}\m_{\ell-n,n}^{(\ell)}$. To start with, the leading large $\t$ asymptotic of the factor premultiplying the hypergeometric function in eq.\eqref{mackmudef} is given by 
	\be \label{prefactorasymp}
	2^{-2\ell+n}\t^{\ell-n}\frac{(-\ell)_n}{n!}(\ell+h-1)_{n-\ell}
	\ee 
	where we have used the relation 
	\be 
	(-1)^n \left(\begin{matrix}
		\ell\\
		n
	\end{matrix}\right)=\frac{(-\ell)_n}{n!}.
	\ee 
	Next we focus upon the hypergeometric function above. Note that the $_4F_3$ above is balanced. Therefore we  can use the following transformation due to Whipple to convert one balanced $_4F_3$ into another balanced $_4F_3$, 
	\be \label{whipple1}
	_4F_3\left[
	\begin{matrix}
		-p,~a,~b,~c\\
		d,~e,f
	\end{matrix}
	~\Big|1\right]=\frac{(e-a)_p(f-a)_p}{(e)_p(f)_p}
	~_4F_3\left[
	\begin{matrix}
		-p,~a,~d-b,~d-c\\
		d,~a+1-p-e,a+1-p-f
	\end{matrix}
	~\Big|1\right].
	\ee 
	Using this we convert the $_4F_3$ in eq.\eqref{mackmudef} into , 
	\be 
	\frac{(h+n-1)_{\ell-n} (-h+\tau +1)_{\ell-n}}{(n+\tau )_{\ell-n} (2 \tau-2 h +2)_{\ell-n}}
	~_4F_3
	\left[
	\begin{matrix}
		-(\ell-n),~h+n-1,~1-\ell-\t,~1-h+\t\\
		2-\ell-h,~n+\t,n+h-\tau-\ell
	\end{matrix}
	~\Big|1\right].
	\ee 
	Next we consider the limit $\t\to\infty$ keeping $\ell$ fixed. The leading asymptotic is given by, 
	\be \label{hypgeomasymp}
	2^{n-\ell}\frac{(h+n-1)_{\ell-n}}{\t^{\ell-n}}~
	_2F_1\left[
	\begin{matrix}
		-(\ell-n),~h+n-1\\
		2-\ell-h
	\end{matrix}
	~\Big|1\right].
	\ee 
	Thus clubbing together eq.\eqref{hypgeomasymp} and eq.\eqref{prefactorasymp} we have, 
	\be \label{muinter}
	(-1)^{\ell-n}\m_{\ell-n,n}^{(\ell)}\sim 8^{-\ell} 2^{2 n}\frac{(-\ell)_n}{n!}~_2F_1\left[
	\begin{matrix}
		-(\ell-n),~h+n-1\\
		2-\ell-h
	\end{matrix}
	~\Big|1\right].
	\ee 
	Next using Chu-Vandermonde identity
	\be 
	_2F_1\left[
	\begin{matrix}
		-p,~a\\
		c
	\end{matrix}
	~\Big|1\right]=~\frac{(c-a)_p}{(c)_p},~~~~p\in\mathbb{Z}\backslash\mathbb{Z}^-.
	\ee 
	to obtain further 
	\begin{align}
		\begin{split} 
			_2F_1\left[
			\begin{matrix}
				-(\ell-n),~h+n-1\\
				2-\ell-h
			\end{matrix}
			~\Big|1\right]=&~\frac{(3-\ell-2h-n)_{\ell-n}}{(2-\ell-h)_{\ell-n}}\\
			=&~\frac{\G(2h-2+\ell+n)}{\G(2h-2+2n)}\times\frac{\G(h-1+n)}{\G(h-1+\ell)}\\
			=&~\frac{(2h-2)_\ell}{(h-1)_\ell}\times\frac{(h-1)_n(2h-2+\ell)_n}{(2h-2)_{2n}}
		\end{split}
	\end{align} 
	Further using the identity, 
	\be 
	\left(x+\frac{1}{2}\right)_n=~2^{2n}\frac{(2x)_{2n}}{(x)_n},~~~n\in\mathbb{Z}\backslash\mathbb{Z}^-
	\ee 
	we reach at, 
	\be 
	2^{2n}~_2F_1\left[
	\begin{matrix}
		-(\ell-n),~h+n-1\\
		2-\ell-h
	\end{matrix}
	~\Big|1\right]=~\frac{(2h-2)_\ell}{(h-1)_\ell}\times \frac{(2h-2+\ell)_n}{\left(h-\frac{1}{2}\right)_n}
	\ee 
	Thus collecting everything, 
	\be \label{mufinal}
	(-1)^{\ell-n}\m_{\ell-n,n}^{(\ell)}\sim \frac{8^{-\ell}}{(h-1)_\ell}(2h-2)_\ell~\frac{(-\ell)_n(\ell+2 h-2)_n}{n!\left(h-\frac{1}{2}\right)_n}
	\ee 
	Putting this into eq.\eqref{limitform} we have in the limit $s\gg\t\gg 1,~ s\gg\D_\f$, with $x$ fixed, 
	\begin{align}
		\begin{split} 
			\Mack_{\t,\ell}(s,x)&\sim \frac{s^\ell}{8^\ell(h-1)_\ell}(2h-2)_\ell\sum_{n=0}^{\ell}\frac{(-\ell)_n(\ell+2 h-2)_n}{n!\left(h-\frac{1}{2}\right)_n}\left(\frac{1-x}{2}+\frac{\D_\f}{(2s-\D_\f)}\right)^n\\
			&=\frac{s^\ell}{8^\ell(h-1)_\ell}(2h-2)_\ell~
			_2F_1\left[
			\begin{matrix}
				-\ell,~2(h-1)+\ell\\
				(h-1)+\frac{1}{2}
			\end{matrix}~\Big|\frac{1-x}{2}+\frac{\D_\f}{(2s-\D_\f)}\right]\\
			&=\frac{s^\ell}{n_\ell}C_\ell^{(h-1)}(x)+O(s^{\ell-1})\,
		\end{split}
	\end{align} 
	where 
	\be 
	n_\ell=8^\ell \frac{(h-1)_\ell}{\ell!}.
	\ee 
	Note that even for $x=1$ the above expression holds true in the large $s$ limit.
	Thus we have the final asymptotic equivalence, 
	\be \label{finalasymp}
	{\Mack_{\t,\ell}(s,x)\sim \left(\frac{s}{8}\right)^\ell\frac{\ell!}{(h-1)_\ell}C_{\ell}^{(h-1)}(x),~~~s\gg\t\gg1.}
	\ee 
	
	\section{Bounds on Gegenbauer polynomial}\label{bndderive}
	In this appendix we provide with a proof of the bounding relation eq.\eqref{Gegenbound}. The proof follows that given in \cite{Chaichian:1987zt}. We start with the following integral representation of the Gegenbauer polynomial $C_\ell^{(\a)}(x)$ for $x>1$ (see for example \cite{andrews_askey_roy_1999}), 
	\begin{align} 
		C_{\ell}^{(\a)}(x)&=\frac{\G(2\a+\ell)}{2^{2\a-1}\G(\ell+1)\G^2(\a)}\int_0^\p~\left[x+\sqrt{x^2-1}\cos\varphi\right]^\ell~\sin^{2\a-1}\varphi~d\varphi\\
		&=C_\ell^{(\a)}(1)\frac{\G(2\a)}{2^{2\a-1}\G^2(\a)}\int_0^\p~\left[x+\sqrt{x^2-1}\cos\varphi\right]^\ell~\sin^{2\a-1}\varphi~d\varphi
	\end{align}
	Then  we have the following integral representation for the normalized Gegenbauer polynomial [c.f. eq.\eqref{gegennorm}], 
	\be 
	\mathbf{C}_\ell^{(\a)}(x)=\frac{\G(2\a)}{2^{2\a-1}\G^2(\a)}\int_0^\p~\left[x+\sqrt{x^2-1}\cos\varphi\right]^\ell~\sin^{2\a-1}\varphi~d\varphi
	\ee 
	Introducing the variable, 
	\be 
	y=\frac{\sqrt{x^2-1}}{x},~~~x>1
	\ee define, 
	\begin{align}
		\begin{split}
			H_\ell(y,\varphi_0,\varphi)&:=\frac{(1+y\cos\varphi)^\ell}{(1+y\cos\varphi_0)^\ell},\\
			G_{\ell}(y,\varphi_0)&:=\int_{0}^\p  H_\ell(y,\varphi_0,\varphi) (\sin\varphi)^{2\a-1} d\varphi.
		\end{split}
	\end{align}
	Since $H_\ell(y,\varphi_0,\varphi)$ is an increasing function of $y$ for $0<\varphi<\varphi_0<\p$ and decreasing function of $y$ for $0<\varphi_0<\varphi<\p$, we can obtain the following inequality quite easily, 
	\be \label{kinq}
	G_{\ell}(y,\varphi_0)\ge K(\varphi_0)
	\ee 
	where 
	\be 
	K(\varphi_0)=\int_{0}^{\varphi_0}d\varphi~(\sin\varphi)^{2\a-1}.
	\ee Thus using eq.\eqref{kinq} we can obtain quite straightforwardly, 
	\be \label{gegeninq}
	\mathbf{C}_\ell^{(\a)}(x)\ge \frac{\G(2\a)}{2^{2\a-1}\G^2(\a)}K(\varphi_0) \left[x+\sqrt{x^2-1}\cos\varphi_0\right]^\ell
	\ee 
	\section{The Conformal Partial Wave Expansion in Mellin Space} \label{melpartdrv}
	In this appendix, we give a short account of the conformal partial wave expansion in Mellin space, leading to eq.\eqref{melpart}. 
	First, consider the usual conformal partial wave expansion in position space. For a $4-$point correlator of identical scalar primaries of conformal dimension $\dphi$, the $s-$channel conformal partial wave expansion is given by, 
	\be \label{cblockexp}
	\mg(u,v)=\sum_{\t,\ell} C_{\t,\ell}\,G_{\t,\ell}(u,v)
	\ee where $u, v, \mg(u,v)$ are defined in eq.\eqref{guvdef} and $\t=(\D-\ell)/2$ with $\D, \ell$ being respectively the scaling dimension and spin of the exchanged primary and $C_{\t,\ell} $ is the corresponding OPE coefficient squared. We are considering unitary theories where $C_{\t,\ell}\ge 0$. Here, $G_{\t,\ell}(u,v)$ is the $s-$channel conformal block  with the normalization that, in the limit $v\ll u\ll 1$ \cite{Alday:2015ewa} the conformal block has the asymptotic form \ 
	\be \label{cblock}
	G_{\t,\ell}(u,v)\sim u^{\t} f_{\t,\ell}(v).
	\ee 
	\paragraph{}  With our definition for the Mellin amplitude, eq.\eqref{MElldef}, the $\mm(s_M,t_M)$ admits a partial wave expansion \cite{Dolan:2011dv}
	\be \label{mell1}
	\mm(s_M,t_M)=\sum_{\t,\ell}C_{\t,\ell}\,
	\widehat{\mn}_{\t,\ell}\,
	\frac{\sin^2\p\left(\frac{\D_\f}{2}-s_M\right)}{\sin^2\p\left(\D_\f-\t-\frac{\ell}{2}\right)}\,\,
	\frac
	{
		\G\left(
		\t-s_M-\frac{\dphi}{2}
		\right)
		\G\left(
		h-\t-\ell-\frac{\dphi}{2}-s_M
		\right)
	}
	{
		\G^2\left(\frac{\dphi}{2}-s_M\right)
	}
	\,
	\Mack_{\t,\ell}(s_M,t_M)
	\ee with 
	\be 
	\widehat{\mn}_{\t,\ell}=2^{\ell}\frac{(2\t+2\ell-1)\G^2(2\t+2\ell-1)}{\G(2\t+\ell-1)\G(h-2\t-\ell)\G^4(\t+\ell)}.
	\ee    This expansion above is just the Mellin space version of the $s-$channel conformal block expansion eq.\eqref{cblockexp} above  with the normalization eq.\eqref{cblock}. Now we will massage this expression into eq.\eqref{MElldef}.
	\paragraph{}
	The starting point is the observation that
	the Euler-beta function $B(x,y)=\frac{\G(x)\G(y)}{\G(x+y)}$ admits the following expansion, 
	\be \label{Bexp}
	B(x,y)=\sum_{n=0}^{\infty}\frac{(-1)^n(y-n)_n}{n!(x+n)}=\sum_{n=0}^{\infty}\frac{(-1)^n(x-n)_n}{n!(y+n)}.
	\ee 
	In other words, the Euler-beta function can be expanded in terms of the poles of either of the Gamma
	functions. This is not true for the usual Gamma function which needs, in addition, a regular piece for
	the expansion to be valid. Using this, we will massage the $s-$dependent combination of Gamma functions appearing in eq.\eqref{mell1} into the form,
	\be 
	\frac
	{
		\G\left(
		\t-s_M-\frac{\dphi}{2}
		\right)
		\G\left(
		h-\t-\ell-\frac{\dphi}{2}-s_M
		\right)
	}
	{
		\G^2\left(\frac{\dphi}{2}-s_M\right)
	}
	=
	B
	\left(
	\t-s_M\frac{\dphi}{2}, \dphi-\t
	\right) 
	\frac{
		\G\left(h-\t-\ell-s_M-\frac{\dphi}{2}\right)
	}
	{
		\G\left(\frac{\dphi}{2}-s_M\right)\G(\dphi-\t)
	}
	\ee Next, we use eq.\eqref{Bexp} to expand the beta function leading to\footnote{Here we have defined $\mf_{\t,\ell}(s):=C_{\t,\ell}\,
		\widehat{\mn}_{\t,\ell}\,
		\frac{\sin^2\p\left(\frac{\D_\f}{2}-s\right)}{\sin^2\p\left(\D_\f-\t-\frac{\ell}{2}\right)}$ to avoid cumbersome notation.}
	\begin{align}
		\begin{split}
			\mm(s_M,t_M)=&\sum_{\t,\ell} \mf_{\t,\ell}(s_M) 
			\sum_{n=0}^{\infty}
			\frac{(-1)^n(\dphi-\t-n)_n}{n!(\t-s_M-\dphi/2+n)}\,
			\frac{\G(h-\t-\ell-s_M-\dphi/2)}{\G(\dphi/2-s_M)\G(\dphi-\t)}
			\Mack_{\t,\ell}(s_M,t_M)\\
			=&\sum_{\t,\ell} \mf_{\t,\ell}(s_M) 
			\sum_{n=0}^{\infty}
			\frac{(-1)^n(\dphi-\t-n)_n}{n!(\t-s_M-\dphi/2+n)}\,
			\frac{\G(h-2\t-\ell-n)}{\G(\dphi-\t-n)\G(\dphi-\t)}
			\Mack_{\t,\ell}(s_M,t_M)+\dots\\
			=& \sum_{\t,\ell} \mf_{\t,\ell}(s_M)
			\frac{\G(h-2\t-\ell)\Mack_{\t,\ell}}{\G^2(\dphi-\t)\left(\t-s_M-\frac{\dphi}{2}\right)}
			~_3F_2
			\left[
			\begin{matrix}
				\t-s_M-\frac{\D_\f}{2},1+\t-\D_\f,1+\t-\D_\f\\
				1+\t-s_M-\frac{\D_\f}{2},2\t+\ell-h+1
			\end{matrix}~\Big|1
			\right]+\dots
		\end{split}
	\end{align}
	where in the second line we have Taylor expanded the ratio of Gamma functions around the $s-$pole $s_M=\t-\frac{\dphi}{2}+n$. The dots represent regular terms and we assume that,
	this final form will give the same residues at the location of the physical poles as eq.\eqref{mell1}. However, notice that $_3F_2$ form now is devoid of the zeros coming from the inverse $\G^2\left(\frac{\dphi}{2}-s_M\right)$ factor and
	differs from the form in eq.\eqref{mell1} by regular terms hidden in the dots. We assume that, these regular terms converge so that they will not contribute to the absorptive part of $\mm(s_M,t_M)$. Thus we can consider this form of $\mm(s_M,t_M)$ modulo the regular terms. Finally restoring all the factors we reach the desired form eq.\eqref{melpart}. Again to emphasise, we could have started with the Dolan-Osborn form in eq.(\ref{mell1}) and obtained the imaginary part directly from there as well--the results would be identical.
	
	\section{Asymptotic analysis of the sum eq.\eqref{cnsum}}\label{tausum}
	The summand of the $\t-$ sum is given by eq.\eqref{summand}. Barring all the prefactors, we will concentrate upon the following sum, 
	\be 
	\sum_{\t=\dphi}^{\t_\star}\t^{4-2 h} \sin^2[\p(\t-\dphi)],~~~\t_\star=\sqrt{(2\dphi+\ell)s_M}.
	\ee 
	We need to be able to this sum. Generally we can resort to integrals assuming that the $\t-$spectrum can be considered to be continuous so that we can resort to the integral. However this is actually not true. But for the purpose of the asymptotic evaluations we can resort to integral. Thus we are interested in the integral, 
	\be \label{tint}
	\int_{\dphi}^{\t_\star}d\t~\t^{4-2 h}\sin^2[\p(\t-\dphi)]
	\ee We can start with the  indefinite version of the above integral,
	\begin{align}
		\begin{split}
			F(\t):=&\int d\t~\t^{4-2 h}\sin^2[\p(\t-\dphi)]\\
			&=
			\frac{e^{-2 i \pi  \Delta _{\phi }} }{4 (4 (h-3) h+5)}
			\left[(2 h-1) \tau ^{5-2h} \left[-2 e^{2 i \pi  \Delta _{\phi }}+(2 h-5) \left\lbrace e^{4 i \pi  \Delta _{\phi }} E_{2 h-4}(2 i \pi  \tau )+ E_{2 h-4}(-2 i \pi  \tau )\right\rbrace\right]\right.\\
			& \left. \hspace{10.5 cm}+8 e^{2 i \pi  \Delta _{\phi }} \Delta _{\phi }^{5-2h} \right]   
		\end{split}
	\end{align} 
	Now we will consider two asymptotic limits. First we will consider the asymptotic of $F(\dphi)$ in the limit of large $\dphi$. This limit is given by, 
	\be 
	F(\dphi)\sim \frac{\Delta _{\phi }^{5-2 h}}{10-4 h}
	\ee 
	Also in the limit $\t_\star\gg 1$, 
	\be 
	F(\t_\star)\sim \frac{\t_\star^{5-2 h}}{10-4 h}.
	\ee 
	Note that since we are considering $s\gg\dphi$ hence for $h>\frac{5}{2}$,  $F(\dphi)$ dominates over $F(\t_\star)$.
	Thus for $h<5/2$ we can write, 
	\be \label{hless}
	\int_{\dphi}^{\t_\star}d\t~\t^{4-2h}\sin^2[\p(\t-\dphi)]\sim \frac{\t_\star^{5-2 h}}{10-4 h}=\frac{1}{10-4h}\left[s(2\dphi +\ell)\right]^{\frac{5}{2}-h},~~~\t_\star\to\infty.
	\ee 
	and for $h>5/2$ we can use, 
	\be \label{hgrt}
	\int_{\dphi}^{\t_\star}d\t~\t^{4-2h}\sin^2[\p(\t-\dphi)]\sim \frac{\dphi^{5-2 h}}{4h-10},~~~\dphi\to\infty
	\ee   
	
	Now we would like to consider the case $h=5/2$. Note that we can not put $h=5/2$ directly into the either asymptotic expressions that we have obtained so far, eq.\eqref{hless} or eq.\eqref{hgrt}, because the denominator vanishes identically thus resulting into a non-removable singular structure. This is however expected because on putting $h=5/2$ into  eq.\eqref{tint} we see that the integrand being $\sim\frac{1}{\t}$ has a logarithmic singularity. Thus to tackle this case we will start with the integral eq.\eqref{tint} and put $h=5/2$ into it so that the integral we need to do is, 
	\be 
	\int_{\dphi}^{\t_\star}\frac{d\t}{\t}\sin^2[\p(\t-\dphi)].
	\ee 
	Therefore, 
	\be 
	F(\t)\bigg|_{h=\frac{5}{2}}=\frac{1}{2} \left[\log (\pi  \tau)-\text{Ci}(2 \pi  \tau) (-\cos (2 \pi  \dphi))-\sin (2 \pi  \dphi) \text{Si}(2 \pi  \tau)\right]
	\ee Clearly we can write the asymptotic expression, 
	\be 
	F(\t)\bigg|_{h=\frac{5}{2}}\sim \frac{1}{2}\log(\t),~~~\t\to\infty.
	\ee Thus in the limit $s_M\gg\dphi\gg1$ and also $s\gg\ell$, we can write the leading order asymptotic expression, 
	\be \label{5dint}
	\int_{\dphi}^{\t_\star}\frac{d\t}{\t}\sin^2[\p(\t-\dphi)]\sim \frac{1}{4}\log s_M. 
	\ee 
	
	\section{Asymptotic evaluations of various sums}
	In this appendix we provide certain asymptotes of summation expressions.
	\begin{itemize}
		\item[1)] We come across the following sum in various occasions of our analysis  , 
		\be 
		\lsumt \ell^{2h-2}.
		\ee We are mostly interested in the large $L$ asymptotic of the above sum. To obtain so, first we we have
		\be \label{ellsum1}
		\lsumt \ell^\a=2^{\alpha } H_{\frac{L}{2}}^{(-\alpha )}.
		\ee 
		Next, considering the asymptotic of the $r^{th}$ order Harmonic number
		\be 
		H_{x}^{(r)}\sim \frac{x^{1-r}}{1-r},~~~x\to\infty\ee
		we can write, 
		\be 
		\lsumt \ell^\a \sim \frac{L^{1+\a}}{2+2\a}.
		\ee Thus we have finally, 
		\be \label{ellsum2}
		\lsumt \ell^{2h-2}\sim \frac{L^{2h-1}}{4h-2}.
		\ee  
		\item[2)] Next, we consider the $\ell-$sum appearing in the eq.\eqref{ambsum} and eq.\eqref{ellnf}. The sum is of the generic form, 
		\be 
		\lsumt \ell^{b}~(a+\ell)^{c},~~~~a>0;~b,c\in\mathbb{R}.
		\ee Now for our purpose, we are generally interested in large $a$, large $L$ asymptotic of the sum. The case that interests us is the one where we consider $a\gg L$. To tackle the sum we can take help of the Euler-Maclaurin formula and to \emph{leading order in $L$} we can replace the sum by integral, 
		\be \label{lsum2}
		\lsumt \ell^{b}~(a+\ell)^{c}\sim 2^b \int_{0}^{L/2} dx~x^b~(a+2 x)^{c}
		\ee  This integral can be expressed in terms of incomplete Beta function as, 
		\be 
		2^b \int_{0}^{L/2} dx~x^b~(a+2 x)^{c}=
		\frac{1}{2} a^c L^{b+1} \Gamma (b+1) \, _2\tilde{F}_1\left(b+1,-c;b+2;-\frac{L}{a}\right)
		\ee where, 
		\be 
		_2\tilde{F}_1\left(a,b,c;z\right)=\frac{_2{F}_1\left(a,b,c;z\right)}{\G(c)}.
		\ee Now we will consider further $a\gg L\gg 1$. The desired asymptotic in this limit is given by,
		\be 
		2^b \int_{0}^{L/2} dx~x^b~(a+2 x)^{c}\sim a^c\frac{ L^{b+1}}{2 b+2}
		\ee 
		so that we can write finally wrapping up everything, 
		\be 
		\lsumt \ell^b~(a+\ell)^{c}\sim a^c\frac{ L^{b+1}}{2 b+2}.
		\ee 
		As for the $\ell-$sum appearing in eq.\eqref{ambsum} we put the values $a=2\dphi,~b=2 h-2$ and $c=(3-2h)/2$ to obtain, 
		\be \label{ellsum3}
		\lsumt \ell^{2h-2}~(2\dphi+\ell)^{\frac{3-2h}{2}}\sim 
		(2\dphi)^{\frac{3-2 h}{2}}~\frac{ L^{2 h-1}}{4h-2},~~~~2\dphi\gg L\gg 1.
		\ee 
		Similarly for the $\ell-$sum appearing in eq.\eqref{ellnf}, one puts $a=2\dphi,~b=h-1,~c=(3-2h)/2$ to obtain, 
		\be \label{ellsum4}
		\lsumt \ell^{h-1}~(2\dphi+\ell)^{\frac{3-2h}{2}}\sim 
		(2\dphi)^{\frac{3-2 h}{2}}~\frac{ L^{ h}}{2h},~~~~2\dphi\gg L\gg 1.
		\ee 
	\end{itemize}
	
	\section{Considering the primaries only}\label{primary}
	We have seen that, in the flat space limit, the contribution towards the Mellin amplitude of a conformal family corresponding to a primary operator with dimension $\D$ is peaked not at the primary, rather at a descendant as dictated by eq.\eqref{flatpeak}, eq.\eqref{qdq}. Thus the natural expectation is that, if we consider just the contributions of  the primaries then it should have vanishingly small contribution towards the full result in the flat space limit. It is a worthwhile exercise to look into this explicitly. In this section, we take up this job and  and bound the primary contribution towards the Mellin amplitude. 
	
	We start  with the following definition which isolate the contribution of the primaries towards $\ma_M(s_M,t_M)$, 
	\be 
	\ma_M^{(p)}(s_M,t_M):=\sum_{\substack{\t,\ell\\ \ell~\text{even}}}C_{\t,\ell}~\text{Im}[f_{\t,\ell}(s_M)]^{(p)}~\widehat{P}_{\t,\ell}(s_M,t_M)
	\ee  with
	\be \label{impm}
	\text{Im}[f_{\t,\ell}(s_M)]^{(p)}:=\p\mn_{\t,\ell}\frac{\G^2(\t+\ell+\D_\f-h)}{\G(2\t+\ell-h+1)}\frac{\sin^2\p\left(\frac{\D_\f}{2}-s_M\right)}{\sin^2\p\left(\D_\f-\t-\frac{\ell}{2}\right)}~\d(s_M+\D_\f/2-\t).
	\ee 
	With the help of this expression we can define $\bar{\ma}_M^{(p)}(s)$ and $\mathfrak{a}_{n,\dphi/2}^{(p)}.$
	We have
	\begin{align}\label{sigp}
		\begin{split}
			\bar{\ma}_M^{(p)}(s_M)&=\frac{1}{s_M-\frac{\D_\f}{2}}\int_{\D_\f/2}^{s_M}ds'\ma_M^{(p)}(s',t=0),\\
			&=\frac{\p}{s-\frac{\D_\f}{2}}\int_{\D_\f/2}^{s_M}ds'\sum_{\substack{\ell\\ \ell~\text{even}}}\sum_{\t=h-1}^{\infty}\left[C_{\t,\ell}\mn_{\t,\ell}\frac{\G^2(\t+\ell+\D_\f-h)}{\G(2\t+\ell-h+1)}\right.\\
			&\left.\hspace{6 cm}\times\frac{\sin^2\p\left(\frac{\D_\f}{2}-s\right)}{\sin^2\p\left(\D_\f-\t-\frac{\ell}{2}\right)}~\d(s+\D_\f/2-\t)\Mack_{\t,\ell}(s',0)\right]
		\end{split}
	\end{align}
	Next we will do a small trick to handle the above quantity. We introduce an integration over a sum of delta functions in $x$ to rewrite,
	\be\label{taudelta}
	\int_{\D_\f/2}^{s}ds'\sum_{\substack{\ell\\ \ell~\text{even}}}
	\int_{0}^{\infty}dx\sum_{\t=h-1}^{\infty}\d(x-\t)C_{\ell}(x)\mn_{\ell}(x)\frac{\G^2(x+\ell+\D_\f-h)}{\G(2x+\ell-h+1)}\frac{\sin^2\p\left(\frac{\D_\f}{2}-s'\right)}{\sin^2\p\left(\D_\f-x-\frac{\ell}{2}\right)}~\d\left(s'+\frac{\D_\f}{2}-x\right).
	\ee
	where, 
	\be 
	C_\ell(\t)\equiv C_{\t,\ell},~~\mn_\ell(\t)\equiv\mn_{\t,\ell}.
	\ee 
	Now we will change the order of integration and do the $s'$ integral first. Note that the $s'$ dependent delta function will give nonzero contribution when, 
	\be 
	s'=x-\frac{\D_\f}{2}
	\ee 
	But further we have $\D_\f<s'<s_M$ which gives us a condition on $x$, 
	\be \label{res1}
	\D_\f\le x\le s_M+\frac{\D_\f}{2}
	\ee 
	This condition effectively truncates the $\t$ sum above and produces, 
	\be 
	\sum_{\substack{\ell\\ \ell~\text{even}}}\sum_{\t=\D_\f}^{s+\frac{\D_\f}{2}}C_{\t,\ell}~\mn_{\t,\ell}\frac{\G^2(\t+\ell+\D_\f-h)}{\G(2\t+\ell-h+1)}\frac{\sin^2\left[\p\left(\D_\f-\t\right)\right]}{\sin^2\left[\p\left(\D_\f-\t-\frac{\ell}{2}\right)\right]}\Mack_{\t,\ell}\left(\t-\frac{\D_\f}{2},0\right)\,.
	\ee 
	Next we would like to investigate the Mack polynomial $\Mack_{\t,\ell}(\t-\D_\f/2,t)$. We have, 
	\be 
	\Mack_{\t,\ell}\left(\t-\frac{\D_\f}{2},t\right)=\sum_{n=0}^\ell \m_{0,n}^{(\ell)}\left(\frac{\D_\f}{2}-t\right)_n
	\ee 
	Further, 
	\begin{align}
		\begin{split} 
			\m_{0,n}^{(\ell)}&=2^{-\ell}(-1)^n\left(
			\begin{matrix}
				\ell\\
				n
			\end{matrix}\right) (\t+n)_{\ell-n}^2~(2\t+2\ell-1)_{n-\ell}\\
			&=2^{-\ell}\frac{(-\ell)_n}{n!}\frac{\G^2(\t+\ell)}{\G^2(\t+n)}\times (2\t+\ell-1)_n\times\frac{\G(2\t+\ell-1)}{\G(2\t+2\ell-1)},\\
			&=2^{-\ell}\frac{(\t)_\ell^2}{(2\t+\ell-1)_\ell}\times \frac{(-\ell)_n~(2\t+\ell-1)_n}{n!~(\t)_n^2}
		\end{split}
	\end{align}
	Thus we have,
	
	\be \label{macktau}
	\Mack_{\t,\ell}\left(\t-\frac{\D_\f}{2},t\right)=2^{-\ell}\frac{(\t)_\ell^2}{(2\t+\ell-1)_\ell}~_3F_2\left[
	\begin{matrix}
		-\ell,~2\t+\ell-1,~\frac{\D_\f}{2}-t\\
		\t,~\t
	\end{matrix}~\Big|1\right].
	\ee 
	Also we have, 
	\be \label{Atauell}
	A_{\t,\ell}:=\mn_{\t,\ell}\frac{\G^2(\t+\ell+\D_\f-h)}{\G(2\t+\ell-h+1)}=2^{\ell}\frac{(2\t+2\ell-1)\G^2(2\t+2\ell-1)}{\G(2\t+\ell-1)\G^4(\t+\ell)\G^2(\D_\f-\t)}.
	\ee 
	Putting everything together we have, 
	
	\be \label{sigmap}
	\bar{\ma}_M^{(p)}(s_M)=\frac{\p}{s_M-\frac{\D_\f}{2}}\sum_{\substack{\ell\\ \ell~\text{even}}}\sum_{\t=\D_\f}^{s+\frac{\D_\f}{2}}C_{\t,\ell}~\frac{\G(2\t+2\ell)}{\G^2(\t)\G^2(\t+\ell)\G^2(\D_\f-\t)}~_3F_2\left[
	\begin{matrix}
		-\ell,~2\t+\ell-1,~\frac{\D_\f}{2}\\
		\t,~\t
	\end{matrix}~\Big|1\right]
	\ee 
	where we have used the fact of $\ell$ being even to get rid of the ratio of the $\sin$ squares. To bound this, we will have to \enquote{effectively cut} the $\ell$ sum to some finite summation. To determine that \enquote{$\ell$ cutoff} we will exploit the information of the polynomial boundedness of the Mellin amplitude that we have assumed. To do so we will take help of $\mathfrak{a}_{n,\D_\f/2}$ as follows
	\be
	\mathfrak{a}_{n,\dphi/2}^{(p)}=\int_{\frac{\dphi}{2}}^\infty \frac{d\bar{s}}{\bar{s}^{n+1}}\,\ma_M^{(p)}(\bar{s},t=\dphi/2)>\int_{\frac{\dphi}{2}}^{s_M} \frac{d\bar{s}}{\bar{s}^{n+1}}\,\ma_M^{(p)}(\bar{s},t=\dphi/2)	
	>s_M^{-(n+1)}\int_{\frac{\dphi}{2}}^{s_M} d\bar{s}\,\ma_M^{(p)}(\bar{s},t=\dphi/2)
	\ee
	The first  inequality follows from the integrand being  always positive\footnote{the positivity is in the sense of distribution i.e, on integrating against a Schwartz function the sign of the function remains unaltered.} for unitary theories because for unitary theories $C_{\t,\ell}\in\mathbb{R}^{\ge}$. The second inequality follows because $s^{-(n+1)}$ is a monotonically decreasing function of $s$ for $n>0$. Next ploughing through the same steps as before we have, 
	\be  \label{tailbound}
	\mathfrak{a}_{n,\D_\f/2}^{(p)}~s_M^{n+1}>\p\sum_{\substack{\ell\\ \ell~\text{even}}}\sum_{\t=\D_\f}^{s_M+\frac{\D_\f}{2}}C_{\t,\ell}\frac{\G(2\t+2\ell)}{\G^2(\t)\G^2(\t+\ell)\G^2(\D_\f-\t)}>\p\sum_{\substack{\ell=L+2\\ \ell~\text{even}}}^{\infty}\sum_{\t=\D_\f}^{s_M+\frac{\D_\f}{2}}C_{\t,\ell}\frac{\G(2\t+2\ell)}{\G^2(\t)\G^2(\t+\ell)\G^2(\D_\f-\t)}.
	\ee 
	where the last inequality follows on using the positivity of the summand.
	Here $L$ is some $\ell$ value which is presumably large. This  basically defines the \emph{tail} of the series. 
	\subsection{Determining the $\ell$-cutoff}To make use of this above inequality eq.\eqref{tailbound} in order to find out the \enquote{$\ell$ cutoff} as mentioned above we turn our attention to once again to eq.\eqref{sigmap} and write the same as follows, 
	\be 
	\bar{\ma}_M^{(p)}(s)=\frac{\p}{s-\frac{\D_\f}{2}}\sum_{\substack{\ell=0\\\ell~\text{even}}}^{L}\sum_{\t=\D_\f}^{s+\frac{\D_\f}{2}}C_{\t,\ell}\frac{\G(2\t+2\ell)}{\G^2(\t)\G^2(\t+\ell)\G^2(\D_\f-\t)}~_3F_2\left[
	\begin{matrix}
		-\ell,~2\t+\ell-1,~\frac{\D_\f}{2}\\
		\t,~\t
	\end{matrix}~\Big|1\right]+\mathfrak{R}(s)
	\ee 
	with the \enquote{remainder} term being, 
	\be \label{rem}
	\mathfrak{R}(s)=\frac{\p}{s-\frac{\D_\f}{2}}~\sum_{\substack{\ell=L+2\\\ell~\text{even}}}^{\infty}~\sum_{\t=\D_\f}^{s+\frac{\D_\f}{2}}C_{\t,\ell}~\frac{\G(2\t+2\ell)}{\G^2(\t)\G^2(\t+\ell)\G^2(\D_\f-\t)}~_3F_2\left[
	\begin{matrix}
		-\ell,~2\t+\ell-1,~\frac{\D_\f}{2}\\
		\t,~\t\end{matrix}~\Big|1\right]
	\ee 
	Next we will use certain properties of the $_3F_2$ polynomial appearing above. For future reference let us introduce the following defining notation
	\be 
	\widehat{Q}_{\ell}(\t)=~_3F_2\left[
	\begin{matrix}
		-\ell,~2\t+\ell-1,~\frac{\D_\f}{2}\\
		\t,~\t\end{matrix}~\Big|1\right]\,.
	\ee 
	This polynomial has two crucial properties that will come to our use to a great extent. These are the following, 
	\begin{itemize}
		\item[I.] The first useful property that we have is that $\widehat{Q}_\ell(\t)$ is a decreasing function of $\ell$. The most general \enquote{observation}\footnote{This has been checked numerically on Mathematica.} is that this is true \emph{separately} for even spins and odd spins.
		Using this property therefore we can write, 
		\be \label{rest1}
		\mathfrak{R}(s)<\frac{\p}{s-\frac{\D_\f}{2}}~\sum_{\substack{\ell=L+2\\\ell~\text{even}}}^{\infty}~\sum_{\t=\D_\f}^{s+\frac{\D_\f}{2}}C_{\t,\ell}~\frac{\G(2\t+2\ell)}{\G^2(\t)\G^2(\t+\ell)\G^2(\D_\f-\t)}~\widehat{Q}_{L+2}(\t)\,.
		\ee 
		\item[II.] The second property that we will make use of is that generally for large enough $\t$ one has $\widehat{Q}_{\ell}(\t)$ an increasing function of $\t$. Now the important part of this statement is \emph{ large enough $\t$}. For practical reasons this is synonymous with $\t\gg\D_\f$ for our case. The reason for emphasizing this is that in general vary near to $\t=\D_\f$ the polynomial $\widehat{Q}_{\ell}(\t)$  decreases for some time reaching a minimum and then once again starts increasing and maintains the increasing trend with increasing $\t$. Since we are ultimately interested in $s\gg \D_\f/2$ we can safely use this property of $\widehat{Q}_\ell(\t)$ to write, 
		\be\label{rest2} 
		\mathfrak{R}(s)<\frac{\p}{s-\frac{\D_\f}{2}} \widehat{Q}_{L+2}\left(s+\frac{\D_\f}{2}\right)\sum_{\substack{\ell=L+2\\\ell~\text{even}}}^{\infty}~\sum_{\t=\D_\f}^{s+\frac{\D_\f}{2}}C_{\t,\ell}~\frac{\G(2\t+2\ell)}{\G^2(\t)\G^2(\t+\ell)\G^2(\D_\f-\t)}
		\ee 
	\end{itemize} 
	Now using eq.\eqref{tailbound} in the above equation, we obtain, 
	\be 
	\mathfrak{R}(s)\le\mathfrak{a}_{n,\D_\f/2}^{(p)}~\frac{s^{n+1}}{s-\frac{\D_\f}{2}}\widehat{Q}_{L+2}\left(s+\frac{\D_\f}{2}\right)\,.
	\ee 
	Next we will analyze $\widehat{Q}_{L+2}\left(s+\D_\f/2 \right)$. We will analyze this in the limit of large $L$ first. For this purpose we will look into large $L$ asymptotic of the $\widehat{Q}_{L+2}(s+\D_\f/2)$. This asymptotic was worked out in \cite{Dey:2017fab} and is given by the equation (A.23) therein. Using the formula we have, 
	\be 
	\widehat{Q}_{L+2}(s+\D_\f/2)\sim (s)_{\frac{\D_\f}{2}}^2 \left[\left(L+2+s+\frac{\D_\f}{2}\right)\left(L+1+s+\frac{\D_\f}{2}\right)\right]^{-\frac{\D_\f}{2}}
	\ee 
	Thus can write asymptotically, 
	\be \label{rembn}
	\mathfrak{R}(s)\le ~\mathfrak{a}_{n,\D_\f/2}^{(p)}~s^{n+\D_\f}(L+s)^{-\D_\f} \,.
	\ee
	We can use this inequality to find the optimal value of $L$. The idea is that the remainder term is exponentially small. Explicitly, first we cast the RHS of the inequality above in the following form, 
	\be \label{exponen}
	e^{\ln\mathfrak{a}_{n,\D_\f/2}^{(p)}+(n+\D_\f)\ln s-\D_\f\ln(L+s)}
	\ee  
	which leads to
	\be \label{Lsol}
	L=s\left[(s^n\mathfrak{a}_{n,\D_\f/2}^{(p)})^{\frac{1}{\D_\f}}-1\right]
	\ee  
	We note that if $\D_\f\gg1$ then we have essentially the leading asymptotic for $L$, 
	\be\label{dphilarge} 
	\boxed{	L\approx \frac{n}{\D_\f}s\ln s}\,.
	\ee 
	Interestingly, this $s\ln s$ behavior was also found in \cite{Greenberg:1961zz} giving rise to the so called Greenberg-Low bound, which is weaker than the Froissart bound.

	\subsection{Summing over twists}
	With this, next we move on to bounding $\bar{\ma}_M^{(p)}(s)$. Now if we assume that $L$ is such that in the large $s$ limit the remainder term $\mathfrak{R}(s)$ is vanishingly small then we can effectively cut the $\ell$ sum at $\ell=L$. Thus we have, 
	\be
	\bar{\ma}_M^{(p)}(s)=\frac{\p}{s-\frac{\D_\f}{2}}\sum_{\substack{\ell=0\\\ell~\text{even}}}^{L}\sum_{\t=\D_\f}^{s+\frac{\D_\f}{2}}C_{\t,\ell}\frac{\G(2\t+2\ell)}{\G^2(\t)\G^2(\t+\ell)\G^2(\D_\f-\t)}~\widehat{Q}_{\ell}(\t)
	\ee 
	Next using the fact $\widehat{Q}_{\ell}(\t)\le1$ we can write, 
	\be \label{sb1}
	\bar{\ma}_M^{(p)}(s)\le \frac{2\p}{2s-\D_\f}\sum_{\substack{\ell=0\\\ell~\text{even}}}^{L}\sum_{\t=\D_\f}^{s+\frac{\D_\f}{2}}C_{\t,\ell}\frac{\G(2\t+2\ell)}{\G^2(\t)\G^2(\t+\ell)\G^2(\D_\f-\t)}
	\ee 
	We follow the same strategy as the one in the main text. What we will do is to put for the conformal block coefficient its MFT value
	\begin{align}\label{CMFT}
		\begin{split}
			C_{\t,\ell}^{MFT}=&\frac{
				2 
				\Gamma (\ell+h) 
				\Gamma^2 (\ell+\tau ) 
				\Gamma (\ell+2 \tau -1) 
				\Gamma^2 (-h+\tau +1) 
			}
			{
				\Gamma (\ell+1) 
				\Gamma (2 \ell+2 \tau -1) 
				\Gamma (-2 h+2 \tau +1) 
				\Gamma (\ell-h+2 \tau ) 
			}\\
			&\hspace{5 cm}\times 
			\frac 
			{
				\Gamma \left(-2 h+\tau +\Delta _{\phi }+1\right) 
				\Gamma \left(\ell-h+\tau +\Delta _{\phi }\right)
			}
			{
				\Gamma^2 \left(\Delta _{\phi }\right)
				\Gamma \left(\tau -\Delta _{\phi }+1\right) 
				\Gamma^2 \left(-h+\Delta _{\phi }+1\right) 
				\Gamma \left(\ell+h+\tau -\Delta _{\phi }\right)
			}
		\end{split}
	\end{align} 
	and do the analysis. In the limit $\t\gg \ell$ while also considering $\t\gg 1$ the $\t$ summand asymptotes to
	\be 
	C_{\t,\ell}^{MFT}\frac{\G(2\t+2\ell)}{\G^2(\t)\G^2(\t+\ell)\G^2(\dphi-\t)}\sim \frac{2^{3 h-2 \tau +1} \Gamma (\ell+h)  \tau ^{2 \Delta _{\phi }-3 h+\frac{5}{2}}}{\pi ^{3/2} \Gamma (\ell+1) \Gamma \left(\Delta _{\phi }\right){}^2 \Gamma^2 \left(-h+\Delta _{\phi }+1\right)}\sin ^2\left[\pi  \left(\Delta _{\phi }-\tau \right)\right].
	\ee Now as before we will replace the sum over twist by an integral so that basically we are left with,
	\be 
	\int_{\dphi}^{s+\frac{\dphi}{2}}d\t ~e^{-(\ln 4)\t}\t^{2\dphi-3h+\frac{5}{2}}\sin^2\left[\p(\t-\dphi)\right] ~\le \int_{\dphi}^{s+\frac{\dphi}{2}}d\t ~e^{-(\ln 4)\t}\t^{2\dphi-3h+\frac{5}{2}}
	\ee where we have used $\sin^2\left[\p(\t-\dphi)\right]\le 1$ in the last step. Now considering $s\gg \dphi$ we introduce the rescaled variable $\hat{\t}$ defined by
	\be 
	\hat{\t}:=\frac{\t}{s}.
	\ee In terms of this variable the integral above translates into, 
	\be 
	s^{\frac{7}{2}+2\dphi-3h}\int_{0}^{1} d\hat{\t}~ e^{-s\ln4\hat{\t}}~\hat{\t}^{\frac{5}{2}+2\dphi-3h}.
	\ee Now  the large $s$ asymptotic of the integral is, 
	\be 
	s^{\frac{7}{2}+2\dphi-3h}\int_{0}^{1} d\hat{\t}~ e^{-s\ln4\hat{t}}~\hat{\t}^{\frac{5}{2}+2\dphi-3h}\sim 
	(\log4)^{-2 \Delta _{\phi }+3 h-\frac{7}{2}} \Gamma \left(-3 h+2 \Delta _{\phi }+\frac{7}{2}\right)-\frac{4^{-s} s^{2 \Delta _{\phi }-3 h+\frac{5}{2}}}{\log (4)}.
	\ee Now clearly the first term dominates over the second above in the limit $s\gg \dphi\gg 1$ so that we can finally use
	\be \label{sb3}
	\sum_{\t=\dphi}^{s+\frac{\dphi}{2}} C_{\t,\ell}\frac{\G(2\t+2\ell)}{\G^2(\t)\G^2(\t+\ell)\G^2(\dphi-\t)}
	\sim 
	(\log4)^{-2 \Delta _{\phi }+3 h-\frac{7}{2}} \Gamma \left(-3 h+2 \Delta _{\phi }+\frac{7}{2}\right)
	\ee 
	\subsection{Finally the bound}
	Putting this crucial piece of information into eq.\eqref{sb1}  we obtain, 
	\be 
	\bar{\ma}_M^{(p)}\le \frac{\p^{-1}}{2s-\D_\f}(\ln4) ^{3 h-2 \Delta _\phi -\frac{7}{2}}~ \Gamma \left(2 \Delta _\phi +\frac{7}{2}-3 h\right)\frac{2^{2h-3}\sqrt{\p}}{\G^2(\D_\f)\G^2(1-h+\D_\f)}
	\sum_{\substack{\ell=0\\\ell~\text{even}}}^{L}(\ell+1)_{h-1}
	\ee  Now we do the sum, 
	\be 
	\sum_{\substack{\ell=0\\\ell~\text{even}}}^{L}(\ell+1)_{h-1}=\frac{(L+2) \Gamma \left(\frac{2 h+L+2}{2}\right)}{2 h \Gamma \left(\frac{L+4}{2}\right)}
	\ee 
	Now since in general $L$ is large hence we can consider the large $L$ asymptotic of the above sum and thus we can write to the leading order in the large $L$ asymptotic, 
	\be 
	\sum_{\substack{\ell=0\\\ell~\text{even}}}^{L}(\ell+1)_{h-1}\sim \frac{2^{-h} L^h}{h}
	\ee  Note that we have made extensive use of  the assumption $L\ll s$ in the previous section  to reach upto this point. As explained before this is possible when $\dphi\gg 1$. Thus we will now put this constraint into its place.  Thus using  eq.\eqref{dphilarge} and considering the limit $s\gg\D_\f/2$ one has the following asymptotic bound on $\bar{\ma}_M^{(p)}$, 
	\be
		\bar{\ma}_M^{(p)}\le \ma_0~s^{h-1}\ln^hs,
\ee
	with
	\be 
	\ma_0=(\ln4) ^{-2 \Delta _\phi }~ \Gamma \left(-3 h+2 \Delta _{\phi }+\frac{7}{2}\right)\frac{2^{h-4}\p^{-1/2}h^{-1}}{\G^2(\D_\f)\G^2(1-h+\D_\f)}\left(\frac{n}{\dphi}\right)^{h}.
	\ee 
	
	Now at this point we would like to comment on the main purpose of this exercise. To do so we compare $\ma_0$ above with $\mb_1,\mb_2,\mb_3$ from eq.\eqref{CoeffB1}, eq.\eqref{CoeffB2}, eq.\eqref{CoeffB3} respectively. In each case we observe that $\ma_0$ is exponentially suppressed in the limit $\dphi\to\infty$ i.e., the flat space limit under consideration. Thus, this matches with our expectation as described at the beginning of this appendix.
	
\end{subappendices}

\bibliographystyle{utphys}


%% file: conclusion.tex
\chapter{Conclusion}\label{conclusion}
\section*{Summary of results}
In this thesis, we have explored various mathematical properties of the S-matrix that follow as consequences of physical principles of quantum mechanics and relativistic dynamics. Since S-matrix can be used to define an abstract space of theory, our exploration is also understood as an exploration of the theory space. We took two lines of endeavours towards that end.

\paragraph{} We consider amplitudes $\mt(s,t)$ for $2-2$ elastic scattering of identical particles. Scattering amplitudes are known to satisfy various analyticity properties, which follow from various considerations like unitarity, causality, Poincare invariance etc. Our analysis offers new insights into these properties and unearths new and exciting connections with the mathematical field of geometric function theory. In particular, the scattering amplitudes are found to have structures related to those of univalent and typically real functions, two kinds which fare prominently in the study of geometric function theory. A novel dispersive representation of $\mt(s,t)$, called the crossing-symmetric dispersion relation, first derived in \cite{AK} and recently revived in \cite{ASAZ},  helps us find these connections. The key findings are that the dispersion kernel, which is a crossing-symmetric function, when expressed in terms of complex variables $\tilde{z},\, a$ each of which is crossing-symmetric function of the Mandelstam invariants $s,\, t,\, u$, turns out to be univalent and typically real in the unit disc $|\tilde{z}|<1$ for a certain range of real values $a$. Unitarity, used in the form of positivity of the imaginary parts of the partial wave amplitude, further implies that $\mt(s,t)$ itself is a typically real function in the unit disc for the same range of real values of $a$. 

The Taylor coefficients of the Univalent and typically real functions are bounded. We applied these bounds to put bounds on Taylor coefficients of $\mt$ when expanded about $\tilde{z}=0$. But such an expansion is a low energy expansion and can be considered as effective field theory amplitudes. Thus the bounds obtained translate to $2-$sided bounds on ratios of Wilson coefficients which parametrize an effective field theory. We would like to emphasize that our analysis gives a robust mathematical foothold to these $2-$sided bounds, which have been obtained before primarily by numerical means. 

For specific cases, we apply our formalism to pion scattering, photon scattering and graviton scattering. For photon and graviton scattering, we worked with Jacob-Wick helicity amplitudes. Helicity amplitudes are not crossing symmetric by themselves and, therefore, do not admit the crossing symmetric dispersive representations automatically. We constructed fully crossing symmetric amplitudes out of the helicity amplitudes using the representation theory of $S_3$, the group implementing crossing symmetry. This, however, affects the positivity properties that follow from the unitarity. Nevertheless, we applied our GFT tools to the appropriate crossing-symmetric combinations. We further explored the consequences of the locality. Locality demands the absence of negative powers of Mandelstam invariants. However, the crossing symmetric dispersion relation, by itself, does not forbid such terms, and therefore, locality has to be imposed as an external requirement. By analyzing the constraints that result from the requirement of locality, the \emph{locality constraints}, we provided a proof of the principle of \emph{weak low spin dominance} (wLSD) \cite{SashaBern}, which tells that contributions from partial waves with $\ell>\ell_0$ are strongly suppressed. For photon scattering, we find that $\ell_0=3$ and for gravitons, $\ell_0=2$. We also find no such low spin dominance in the case of the scattering of massive particles.  

\paragraph{} While the previous exploration is centred on working with scattering amplitudes directly, in the second part of the thesis, we focused on holographic S-matrix. A holographic S-matrix is one which can be obtained from the correlation functions of conformal field theories (CFT)  via $AdS/CFT$ holographic duality in what is called the flat space limit. This limit correspond to taking $R_{AdS}/\ell_P\gg1$, where $R_{AdS}$ is the $AdS$ radius and $\ell_P$ is the Planck length. We work with CFT Mellin amplitudes to analyze this limit. Our main result is a Froissart-Martin bound for the holographic S-matrix that we derive starting with CFT Mellin amplitude in the flat space limit. We work with Mellin amplitude for conformal correlator of identical scalar primaries with dimension $\D_\f$. The flat space limit we consider involves taking $\D_\f$ to be parametrically large \cite{Sboot1}. In this limit, the Mellin amplitude is mapped to flat space scattering amplitude of identical massive particles. We obtained Froissart-Martin bound in both forward and non-forward limits. One of the important observations we make is that beyond $6$ spacetime dimensions, the holographic Froissart-Martin bound is weaker than the usual Froissart-Martin bound. 

\paragraph{} Now that we have summarized the work in this thesis, below we outline multiple paths for future explorations.

\section*{More mathematical explorations}
The novel connections with the mathematical field of geometric function theory provided robust mathematical foundations to various $2-$sided bounds on the Wilson coefficients for effective field theories. However, this is only the tip of the iceberg. This connection crucially depended on a particular \emph{real} domain of the parameter $a$ of the scattering amplitude, which followed from the unitarity in the form of positivity of the imaginary part of the partial wave amplitudes, $\{\text{Im.}\,[f_\ell]\}$. Two imminent questions follow. The first one is the extension of the analysis to complex domains of $a$, which is natural in the context that there exist analyticity domains of the scattering amplitude in $(s,t)\in \mathbb{C}^2$. When translated to our variables $(z, a)$, these domains correspond to complex domains of $a$. It is worthwhile to explore whether the same mathematical structures survive in the complex domain of $a$. One natural hurdle in this direction is how to employ unitarity because the straightforward use of the positivity argument that we presented breaks down for complex $a$ in general. Therefore, it warrants an investigation into whether going to complex domains of $a$ unearths new mathematical structures that give back the already established properties when restricted to real domains of $a$. 

Secondly, we have used unitarity only in the form of positivity of the partial wave amplitudes for physical $s$, a linear condition, 
$$ \text{Im.}\,[f_\ell](s) \ge0.$$ However there exists a stronger non-linear condition that follows from  unitarity 
$$ 
2 \, \text{Im.}\,[f_\ell](s) \ge |f_\ell(s)|^2.
$$  A natural line of investigation considers finding out what kind of mathematical properties of scattering amplitude follow when this non-linear unitarity condition is considered. 

Finally, a rather ambitious question would be whether we can use these mathematical properties to gain momentum into a possible mathematically rigorous proof of Mandelstam analyticity of the scattering amplitudes, which is still an open problem. 

\section*{Connections with positive geometry}
Positive geometry is the study of convex hulls of geometric spaces. The principal objects of study are polytopes, the higher dimensional generalizations of the planar polygons, and generalizations like Grassmannians. In the past decade, fascinating connections between positive geometries and scattering amplitudes have been unearthed in various forms. The first such connection was obtained in the form of the so-called Amplituhedron \cite{Amplituhedron} in the context of amplitudes for $\mn=4$ SYM  theory. Subsequently, more such connections were obtained in varied contexts, starting from conformal field theories \cite{CFThedron} to cosmological correlation functions \cite{cosmopoly}. Recently, such a positive geometry construct was obtained for low energy amplitudes in effective field theories. This geometric structure is called the \emph{EFT-hedrons} \cite{EFThedron}. EFT-hedron is a positive geometric structure in the space of EFTs, which is constructed employing unitarity in the form of positivity, analyticity properties of the scattering amplitude and the crossing symmetry. This space is parametrized by the Wilson coefficients. The EFTs consistent with basic requirements of unitarity and analytic structure of the S-matrix lie inside the EFT-hedron. More concretely, this translates to \emph{$2-$sided bounds} on Wilson coefficients. Thus EFT-hedron gives a geometric understanding of $2-$sided bounds on Wilson coefficients. Since our analysis employing the tools from geometric function theory also yields $2-$sided bounds on Wilson coefficients, an imminent question is what is a connection between our analysis and the EFT-hedron. In particular, can we derive EFT-hedron using the geometric function theoretic analysis?

\section*{Spinning amplitudes: Further explorations}
Our analysis of spinning S-matrices offer some natural future goals to chase for. Let us first delineate a few technical problems worth exploring. Our analysis is anchored on constructing fully-crossing symmetric combinations out of helicity amplitudes and then applying the crossing-symmetric dispersion formalism to these fully crossing-symmetric functions. We had to resort to this way because helicity amplitudes are \emph{not} invariant under crossing-symmetry. Instead, crossing symmetry acts as a non-trivial automorphism of the set of helicity amplitudes. One can write the generic crossing equation as follows: 
$$ 
\mathbf{T}(s,t,u)=\mathbf{C}_{\p}\cdot \mathbf{T}(\p(s),\p(t),\p(u))
$$ where $\mathbf{T}(s,t,u)$ is a column vector constructed out of the (regularized)helicity amplitudes, $\textbf{C}_{\p}$ is the crossing matrix corresponding to the permutation $\p$ of the Mandelstam invariants. 
A natural goal would be to generalize the crossing-symmetric dispersion relation to this amplitude vector $\mathbf{T}$ directly to obtain a  \emph{matrix dispersion relation} of the form 
$$ 
\mathbf{T}(s,t,u)=\int ds' \, \mathbf{H}(s'\,;\,s,t,u) \cdot \mathbf{A}(s'),
$$ where  $\mathbf{A}$ is a column vector constructed of the corresponding absorptive parts and $\mathbf{H}$ is now a \emph{matrix-valued dispersion kernel} which respects the crossing-properties of the amplitude vector.   This will pave way for mathematical explorations along the lines presented in this thesis.     

\paragraph{}Another technical point requiring attention is the possible application of crossing-symmetric dispersion relation to other formulations of spinning amplitudes. One such is transversity amplitudes \cite{Kotanski}. In
transversity formalism, the spin is quantized normal to the plane of scattering. In this formalism,
the crossing equations are diagonalized. This, however, comes at the price that the unitarity is
now straightforward. However, one can still work out the consequences of the unitarity by exploiting the relation between
the transversity amplitude and the helicity amplitude, the former being a linear combination
of the latter. The unitarity consideration, along with fixed transfer dispersion relations, was
employed to obtain positivity bounds for transversity amplitudes for EFTs in \cite{deRham:2017zjm}. However,
it is unclear how to translate these positivity bounds to constraints on EFT parameters like
Wilson coefficients. Therefore, it is worth investigating how these positivity bounds can be used
to constrain the EFT parametric space. Further, it will be interesting to consider applying the
crossing-symmetric dispersive techniques to these transversity amplitudes.

\section*{EFTs beyond tree level}
An important assumption of our work is that we are only analyzing low energy effective field
theories at the tree level. This is justifiable for EFTs having weakly coupled UV completion. Even
in this situation, it will be interesting to know how these bounds get modified, including massless
loops \cite{Bellazzini:2020cot}. The low energy expansion of effective amplitude 
about $s, t = 0$ is not ideal for tackling the loops. However, in crossing symmetric dispersion relation it is not natural
to expand in this forward limit. So our set up might be better suited to address this issue, and
this exciting possibility necessitates future exploration.

\section*{Regge boundeness of holographic S-Matrix}

One of the key findings of our derivation of the holographic Froissart-Martin bound is that in the Regge limit of $|s|\to\infty$ with $t$ fixed, for spacetime dimensions greater than $6$, the holographic $2-2$ scattering amplitude of massive particles is polynomially bounded with an order that increases with spacetime dimensions. However, from axiomatic quantum field theory, it is known that the $2-2$ scattering amplitude of massive particles is quadratically bounded independent of spacetime dimensions. The reason why the holographic S-matrix shows this weaker behaviour is not clear to us. A careful analysis is warranted to understand this issue. Also, the physical significance of dimension dependent polynomial boundedness needs to be analysis 

Further, we note that the Regge boundedness of the holographic scattering amplitude is tied to the Regge bound on the corresponding CFT Mellin amplitude. Now it was shown in \cite{sashajoao} that Mellin amplitude is also quadratically Regge bounded independent of dimension. It is then natural to ask why we obtain dimension dependent polynomial boundedness in the flat space limit.

\section*{Flat space limit for spinning amplitude}
So far all the flat space limit has been considered for scalar Mellin amplitudes. A natural extension should be in the realm of spinning amplitudes. One of the subtlety to be tackled in establishing flat space limit of spinning Mellin amplitudes is the mapping between tensor structures of the Mellin amplitude and the corresponding flat space scattering amplitude \footnote{Recently an attempt in this direction was initiated in \cite{Liflat} where the recently introduced \cite{CaronHelicity} helicity formalism CFTs  was used.}. 

 Moreover, there are multiple formalism to study flat space scattering amplitudes  of spinning particles like helicity, transeversality, invariant amplitudes etc. One question then arises naturally how these different frameworks are tied to  corresponding spinning Mellin amplitude. Finally,  obtaining flat space limit of spinning Mellin amplitudes will pave way for analyzing analyticity and unitarity structures of the holographic spinning S-matrix in terms of corresponding CFT structures, in the spirit of \cite{FitzKaplan1, FitzKaplan2}.  In particular, such an analysis will have potential to offer insights into flat space graviton scattering from stress tensor amplitude in CFT.

\bibliographystyle{utphys}